\pdfoutput=1
\documentclass[runningheads]{llncs}
\usepackage[T1]{fontenc}

\usepackage{datetime}
\usepackage{url}
\usepackage{paralist}

\usepackage{multicol}
\usepackage[font=small,labelfont=bf]{caption}

\usepackage{hyperref}
\usepackage{url}

\usepackage{amsmath}
\usepackage{amsthm}

\newtheorem{myclaim}{Claim}

\theoremstyle{definition}

\theoremstyle{definition}

\usepackage{graphicx}

\usepackage{array}
\newcolumntype{P}[1]{>{\centering\arraybackslash}p{#1}}
\newcolumntype{M}[1]{>{\centering\arraybackslash}m{#1}}

\usepackage{arydshln}

\usepackage{multirow}

\usepackage{threeparttable}

\usepackage{cleveref}

\crefname{section}{\S}{\S\S}
\Crefname{section}{\S}{\S\S}
\crefformat{section}{\S#2#1#3}

\crefname{appendix}{\S}{\S\S}
\Crefname{appendix}{\S}{\S\S}
\crefformat{appendix}{\S#2#1#3}

\Crefname{assumption}{Assumption}{Assumptions}

\Crefname{invariant}{Invariant}{Invariants}

\Crefname{observation}{Observation}{Observations}

\Crefname{figure}{Fig.}{Figs.}
\Crefname{table}{Tab.}{Tabs.}
\Crefname{algorithm}{Alg.}{Algs.}
\Crefname{equation}{Eq.}{Eqs.}
\Crefname{definition}{Def.}{Defs.}
\Crefname{lemma}{Lem.}{Lems.}
\Crefname{theorem}{Thm.}{Thms.}

\usepackage{soul}

\usepackage{algorithmicx}
\usepackage{algorithm}
\usepackage[noend]{algpseudocode}

\newfloat{module}{htbp}{loa}
\floatname{module}{Module}
\crefname{module}{module}{modules}
\Crefname{module}{Module}{Modules}

\usepackage{algorithm}
\usepackage{algorithmicx}
\usepackage[noend]{algpseudocode}

\algnewcommand{\algorithmicswitch}{\textbf{switch}}
\algdef{SE}[SWITCH]{Switch}{EndSwitch}[1]{\algorithmicswitch\ #1\ \algorithmicdo}{\algorithmicend\ \algorithmicswitch}%
\algtext*{EndSwitch}%

\algnewcommand{\algorithmiccase}{\textbf{case}}
\algdef{SE}[CASE]{Case}{EndCase}[1]{\algorithmiccase\ #1}{\algorithmicend\ \algorithmiccase}%
\algtext*{EndCase}%

\algnewcommand{\algorithmicupon}{\textbf{upon}}
\algdef{SE}[UPON]{Upon}{EndUpon}[1]{\algorithmicupon\ #1\ \algorithmicdo}{\algorithmicend\ \algorithmicupon}%
\algtext*{EndUpon}%

\algnewcommand{\algorithmicforeach}{\textbf{for each}}
\algdef{SE}[FOREACH]{ForEach}{EndForEach}[1]{\algorithmicforeach\ #1\ \algorithmicdo}{\algorithmicend\ \algorithmicforeach}%
\algtext*{EndForEach}%

\algdef{SE}[GENERICBLOCK]{GenericBlock}{EndGenericBlock}[1]{#1}{}%
\algtext*{EndGenericBlock}%

\algrenewcommand{\algorithmicdo}{\kern-0.3em:}
\algrenewcommand{\algorithmicthen}{\kern-0.3em:}

\algnewcommand{\algorithmicgoto}{\textbf{goto}}%
\algnewcommand{\Goto}[1]{\algorithmicgoto~\ref{#1}}%

\algnewcommand{\algorithmictrigger}{\textbf{trigger}}%
\algnewcommand{\Trigger}[1]{\algorithmictrigger~{#1}}%

\algnewcommand{\algorithmicinvoke}{\textbf{invoke}}%
\algnewcommand{\Invoke}[1]{\algorithmicinvoke~{#1}}%

\algnewcommand{\algorithmicassert}{\textbf{assert}}%
\algnewcommand{\Assert}[1]{\algorithmicassert~{#1}}%

\algnewcommand{\algorithmicbreak}{\textbf{break}}%
\algnewcommand{\Break}[0]{\algorithmicbreak}%

\algnewcommand{\algorithmicwaitfor}{\textbf{wait for}}%
\algnewcommand{\WaitFor}[1]{\algorithmicwaitfor~{#1}}%

\algnewcommand{\InlineRequire}[1]{\textbf{require} {#1}}

\algrenewcommand\alglinenumber[1]{\tiny\textcolor{gray}{#1}}

\usepackage{xspace}

\newcommand{\comm}{\textsf{Reducer}\xspace}
\newcommand{\commplus}{\textsf{Reducer++}\xspace}
\newcommand{\first}{p_{\mathrm{first}}}
\newcommand{\dfirst}{\mathcal{D}_{\mathrm{first}}}

\usepackage[dvipsnames]{xcolor}
\usepackage{framed}
\definecolor{lightgray}{gray}{0.90}
\renewenvironment{leftbar}[1][\hsize]
{%
\MakeFramed{\hsize#1\advance\hsize-\width\FrameRestore}%
}
{\endMakeFramed}

\algnewcommand{\BlueComment}[1]{\textcolor{jnSUDigitalRedLight}{\hfill\(\triangleright\) #1}}
\algnewcommand{\LineComment}[1]{\textcolor{jnSUDigitalRedLight}{\(\triangleright\) #1}}

\newcommand{\valuemvba}{\mathsf{Value}_{\mathsf{MVBA}}\xspace}
\newcommand{\valuemba}{\mathsf{Value}_{\mathsf{MBA}}\xspace}
\newcommand{\botmba}{\bot_{\mathsf{MBA}}\xspace}
\newcommand{\valuedigest}{\mathsf{Digest}\xspace}
\newcommand{\smba}{$\textsf{SMBA}_\kappa$\xspace}
\newcommand{\mbasimple}{\textsf{MBA}_\ell\xspace}
\newcommand{\botcrb}{\bot_{\mathsf{CRB}}\xspace}
\newcommand{\crb}{$\textsf{CRB}_\kappa$\xspace}

\usepackage{cleveref}
\crefname{section}{\S}{\S\S}
\Crefname{section}{\S}{\S\S}
\crefformat{section}{\S#2#1#3}
\Crefname{line}{Line}{line}
\crefname{line}{Line}{line}
\Crefname{assumption}{Assumption}{assumption}
\Crefname{@theorem}{Theorem}{theorem}
\crefname{@theorem}{Theorem}{theorem}
\crefname{module}{Module}{module}
\Crefname{module}{Module}{module}
\crefname{myclaim}{Claim}{claim}
\Crefname{myclaim}{Claim}{claim}

\newcommand{\remove}[1]{}

\usepackage{bold-extra} %

\usepackage{listings}
\crefname{lstlisting}{listing}{listings}
\Crefname{lstlisting}{Listing}{Listings}

\crefname{code}{line}{lines}
\Crefname{code}{Line}{Lines}

\definecolor{mygreen}{rgb}{0.254,0.572,0.294}
\definecolor{mygray}{rgb}{0.5,0.5,0.5}
\definecolor{myorange}{rgb}{1,0.35,0}
\definecolor{mymauve}{rgb}{0.58,0,0.82}
\definecolor{myblue}{rgb}{0.2,0.4,0.6}

\definecolor{rakos4orange}{RGB}{255,165,0}
\definecolor{rakos4blue}{RGB}{14,48,173}
\definecolor{rakos4lblue}{RGB}{92,172,238}
\definecolor{rakos4dgray}{RGB}{77,77,77}
\definecolor{plainred}{RGB}{211,63,63}
\definecolor{plainorange}{RGB}{221,105,41}

\lstdefinelanguage{Golang}%
  {morekeywords=[1]{package,import,struct,defer,panic,%
     recover,select,var,const,iota, class},%
   morekeywords=[2]{string,uint,uint8,uint16,uint32,uint64,int,int8,int16,%
     int32,int64,bool,float32,float64,complex64,complex128,byte,rune,uintptr,%
     error,interface,node},%
   morekeywords=[3]{map,slice,make,new,nil,len,cap,copy,close,%
     delete,append,real,imag,complex,chan,},%
   morekeywords=[4]{break,continue,goto,switch,case,fallthrough,%
    default,},%
   morekeywords=[5]{Println,Printf,Error,Send},%
   sensitive=true,%
   morecomment=[l]{//},%
   morecomment=[s]{/*}{*/},%
   morestring=[b]",%
   morestring=[s]{`}{`},%
   }

\usepackage{algorithmicx}
\usepackage{algorithm}
\usepackage{algpseudocode}
\usepackage{xifthen}
\usepackage{graphicx}

\usepackage{stmaryrd}
\usepackage{amsmath}
\usepackage{amsthm}
\usepackage{amsfonts}
\usepackage{amssymb}

\usepackage{shortsym}

\usepackage{cancel}
\usepackage[normalem]{ulem}
\usepackage{soul}

\usepackage{booktabs}
\usepackage{makecell}

\usepackage{pifont}
\newcommand{\cmark}{\ding{52}}%
\newcommand{\xmark}{\ding{56}}%

\definecolor{myParula01Blue}{RGB}{0,114,189}
\definecolor{myParula02Orange}{RGB}{217,83,25}
\definecolor{myParula03Yellow}{RGB}{237,177,32}
\definecolor{myParula04Purple}{RGB}{126,47,142}
\definecolor{myParula05Green}{RGB}{119,172,48}
\definecolor{myParula06LightBlue}{RGB}{77,190,238}
\definecolor{myParula07Red}{RGB}{162,20,47}

\definecolor{jnSUDigitalRed}{HTML}{B1040E}
\definecolor{jnSUDigitalRedLight}{HTML}{E50808}
\definecolor{jnSUDigitalRedDark}{HTML}{820000}
\definecolor{jnSUDigitalBlue}{HTML}{006CB8}
\definecolor{jnSUDigitalBlueLight}{HTML}{6FC3FF}
\definecolor{jnSUDigitalBlueDark}{HTML}{00548f}
\definecolor{jnSUDigitalGreen}{HTML}{008566}
\definecolor{jnSUDigitalGreenLight}{HTML}{1AECBA}
\definecolor{jnSUDigitalGreenDark}{HTML}{006F54}

\definecolor{jnSUAccentIlluminating}{HTML}{FEDD5C}
\definecolor{jnSUAccentIlluminatingLight}{HTML}{FFE781}
\definecolor{jnSUAccentIlluminatingDark}{HTML}{FEC51D}
\definecolor{jnSUAccentPoppy}{HTML}{E98300}
\definecolor{jnSUAccentPoppyLight}{HTML}{F9A44A}
\definecolor{jnSUAccentPoppyDark}{HTML}{D1660F}

\newcommand{\tablefontsize}{\scriptsize}

\algrenewcommand{\algorithmicindent}{1em}

\usepackage{tikz}
\usetikzlibrary{calc}
\usetikzlibrary{arrows}
\usetikzlibrary{arrows.meta}
\usetikzlibrary{patterns}
\usetikzlibrary{positioning}
\usetikzlibrary{decorations.pathreplacing}
\usetikzlibrary{calligraphy}
\usetikzlibrary{shapes.misc}
\usetikzlibrary{shapes.symbols}
\usetikzlibrary{shapes.arrows}
\usetikzlibrary{shapes.callouts}
\usetikzlibrary{spy}
\usetikzlibrary{shapes.geometric}
\usetikzlibrary{intersections}
\usetikzlibrary{fit}

\usepackage{pgfplots}
\pgfplotsset{compat=1.17}
\usepgfplotslibrary{fillbetween}

\usepackage{fullpage}

\usepackage{color}

\titlerunning{Toward Optimal Hash-Based Asynchronous MVBA}
\title{Toward Optimal-Complexity Hash-Based Asynchronous MVBA with Optimal Resilience}

\author{Jovan Komatovic\inst{1} \and
Joachim Neu\inst{2} \and
Tim Roughgarden\inst{2,3}}

\authorrunning{J. Komatovic, J. Neu, T. Roughgarden}

\institute{%
Category Labs\\%
\email{jkomatovic@category.xyz}%
\and%
a16z Crypto Research\\%
\email{\{jneu,troughgarden\}@a16z.com}%
\and%
Columbia University\\%
\email{tim.roughgarden@gmail.com}%
}%

\let\oldkappa\kappa
\let\oldlambda\lambda
\renewcommand{\kappa}{\oldlambda}
\renewcommand{\lambda}{\oldkappa}

\date{}
\begin{document}
\maketitle
\begin{abstract}
    Multi-valued validated Byzantine agreement (MVBA),
a fundamental primitive of distributed computing,
allows $n$ processes to agree on a valid $\ell$-bit value, despite $t$ faulty processes behaving maliciously.
Among hash-based solutions for the asynchronous setting with adaptive faults, the state-of-the-art HMVBA protocol~(\emph{DSN'25}) achieves optimal $O(n^2)$ message complexity, (near-)optimal $O(n \ell + n^2 \kappa \log n)$ bit complexity, and optimal $O(1)$ time complexity.
However, it only tolerates 
$t < \frac15 n$ 
failures.
In contrast, the best-known optimally-resilient protocol, SQ~(\emph{S\&P'25}), 
incurs a higher bit complexity of $O(n^2 \ell + n^3 \kappa)$.
This 
poses
a fundamental question: 
Can a hash-based protocol be designed for the asynchronous setting with adaptive faults that 
simultaneously 
achieves 
optimal complexity and optimal resilience?

This paper takes a significant step toward answering this question.
Namely, we introduce \comm, an MVBA protocol that retains HMVBA's optimal complexity while improving its resilience to $t < \frac14 n$.
Like HMVBA and SQ, \comm relies exclusively on collision-resistant hash functions.
A key innovation in \comm's design is its internal use of strong multi-valued Byzantine agreement (SMBA), a new variant of Byzantine agreement we introduce and construct, which ensures that the decided value was proposed by a correct process.
To further advance resilience toward the optimal one-third bound, we then propose \commplus,
an MVBA protocol that tolerates up to $t < (\frac13 - \epsilon)n$ adaptive failures, for any fixed constant $\epsilon > 0$.
Unlike \comm, \commplus does not rely on SMBA.
Instead, it employs a novel approach involving hash functions modeled as random oracles to ensure termination.
\commplus maintains constant time complexity, quadratic message complexity, and quasi-quadratic bit complexity, with constants dependent on $\epsilon$.

\end{abstract}
\section{Introduction}\label{section:introduction}

Multi-valued validated Byzantine agreement (MVBA), first introduced in~\cite{CKPS01}, has become a fundamental building block for secure distributed systems, such as fault-tolerant replicated state-machines~\cite{abd2005fault,adya2002farsite,amir2006scaling,CL02,Kotla2009,kotla2004high}.
It also plays a central role in distributed key generation~\cite{AbrahamJMMST21,DasYXMK022,Kokoris-KogiasM20} and secure multiparty computation~\cite{DBLP:conf/tcc/DeligiosHL21,DBLP:conf/eurocrypt/FitziGMR02,DBLP:conf/crypto/GennaroIKR02}.
In essence, MVBA protocols enable $n$ processes in a message-passing model
to agree on an $\ell$-bit value that satisfies a fixed external validity predicate,
despite $t$ faulty processes deviating from the protocol in any arbitrary, possibly coordinated, but computationally bounded manner.
Of particular interest, due to superior robustness to network conditions, is MVBA under asynchrony, where messages are guaranteed eventual delivery, but no assumptions are made about the timing.
The design of \emph{asynchronous MVBA protocols} is the subject of this paper.

The seminal FLP impossibility result~\cite{fischer1985impossibility} implies that no deterministic algorithm can solve asynchronous MVBA.
In other words, any asynchronous MVBA protocol must employ randomness.
Since then, it has become standard practice~\cite{hmvba,finmvba,sq,MostefaouiMR15,DBLP:journals/acta/MostefaouiR17,chen2024ociormvbanearoptimalerrorfreeasynchronous,pace} to assume an \emph{idealized common coin} (i.e., Rabin's dealer~\cite{Rabin83}) as a given primitive and build the MVBA protocol around it.
The common coin
encapsulates the randomness,
and upon invocation by sufficiently many processes provides the same unpredictable and unbiasable random sequence to all processes.
The rest of the protocol is the actual distributed-computing ``core mechanism''.
We adopt this blueprint as well: we assume an idealized common-coin object, whose cost is omitted when computing complexity metrics, and design our algorithms around this abstraction.

Moreover, we pursue a \emph{hash-based} MVBA protocol that is,
in particular, setup-free and signature-free.
Hash-based protocols rely on relatively cheap ``unstructured'' operations like hashes,
instead of the relatively expensive ``highly-structured'' algebraic operations that underlie, for instance, public-key or threshold cryptography.
As a result, hash-based protocols are plausibly post-quantum secure, and, ceteris paribus, tend to be more performant.
They also avoid the complexity overhead and trust issues that come with trusted or private setups.

Finally, \cite{CKPS01} has hinted and \cite[Sec.~1.2]{hmvba} has shown that the design of ``good'' hash-based asynchronous MVBA protocols is easy if security is required only against a static adversary, who determines which processes to corrupt at the beginning of the execution before any randomness of the protocol's common coin is sampled.
The gold standard, however, is security against \emph{adaptive} adversaries,
who are at liberty to decide which processes to corrupt during the protocol execution as randomness is revealed.
We thus pursue \emph{adaptive security}.

\begin{table}[tb]
    \caption{%
        State-of-the-art
        adaptively-secure 
        asynchronous MVBA algorithms.
    }%
    \setlength{\tabcolsep}{3pt}%
    \tablefontsize%
    \begin{center}
    \vspace{-0.5em}
    \begin{tabular}{lcccccc}
        \toprule
        Algorithm
        & Hash-based
        & Messages
        & Bits
        & Time
        & Resilience
        \\
        \midrule
        CKPS01-MVBA~\cite{CKPS01} \textsuperscript{TS}
        & \textcolor{jnSUDigitalRed}{\xmark{}}
        & \textcolor{jnSUDigitalGreen}{$O(n^2)$}
        & \textcolor{jnSUDigitalRed}{$O(n^2 \ell + n^2 \kappa + n^3)$}
        & \textcolor{jnSUDigitalGreen}{$O(1)$}
        & \textcolor{jnSUDigitalGreen}{$t < \frac13 n$}
        \\[2pt]
        CKPS01-MVBA/HS~\cite{CKPS01} \textsuperscript{H-CR}
        & \textcolor{jnSUDigitalGreen}{\cmark{}}
        & \textcolor{jnSUDigitalGreen}{$O(n^2)$}
        & \textcolor{jnSUDigitalRed}{$O(n^2 \ell + n^3 \kappa)$}
        & \textcolor{jnSUDigitalGreen}{$O(1)$}
        & \textcolor{jnSUDigitalGreen}{$t < \frac13 n$}
        \\[2pt]
        
        VABA~\cite{vaba} \textsuperscript{TS}
        & \textcolor{jnSUDigitalRed}{\xmark{}}
        & \textcolor{jnSUDigitalGreen}{$O(n^2)$}
        & \textcolor{jnSUDigitalRed}{$O(n^2 \ell + n^2 \kappa)$}
        & \textcolor{jnSUDigitalGreen}{$O(1)$}
        & \textcolor{jnSUDigitalGreen}{$t < \frac13 n$}
        \\[2pt]
        VABA/HS~\cite{vaba} \textsuperscript{H-CR}
        & \textcolor{jnSUDigitalGreen}{\cmark{}}
        & \textcolor{jnSUDigitalGreen}{$O(n^2)$}
        & \textcolor{jnSUDigitalRed}{$O(n^2 \ell + n^3 \kappa)$}
        & \textcolor{jnSUDigitalGreen}{$O(1)$}
        & \textcolor{jnSUDigitalGreen}{$t < \frac13 n$}
        \\[2pt]
        
        Dumbo-MVBA~\cite{dumbomvba} \textsuperscript{TS}
        & \textcolor{jnSUDigitalRed}{\xmark{}}
        & \textcolor{jnSUDigitalGreen}{$O(n^2)$}
        & \textcolor{jnSUDigitalGreen}{$O(n \ell + n^2 \kappa)$}
        & \textcolor{jnSUDigitalGreen}{$O(1)$}
        & \textcolor{jnSUDigitalGreen}{$t < \frac13 n$}
        \\[2pt]
        Dumbo-MVBA/HS~\cite{dumbomvba} \textsuperscript{H-CR}
        & \textcolor{jnSUDigitalGreen}{\cmark{}}
        & \textcolor{jnSUDigitalGreen}{$O(n^2)$}
        & \textcolor{jnSUDigitalRed}{$O(n \ell + n^3 \kappa)$}
        & \textcolor{jnSUDigitalGreen}{$O(1)$}
        & \textcolor{jnSUDigitalGreen}{$t < \frac13 n$}
        \\[2pt]
        
        sMVBA~\cite{smvba} \textsuperscript{TS}
        & \textcolor{jnSUDigitalRed}{\xmark{}}
        & \textcolor{jnSUDigitalGreen}{$O(n^2)$}
        & \textcolor{jnSUDigitalRed}{$O(n^2 \ell + n^2 \kappa)$}
        & \textcolor{jnSUDigitalGreen}{$O(1)$}
        & \textcolor{jnSUDigitalGreen}{$t < \frac13 n$}
        \\[2pt]
        sMVBA/HS~\cite{smvba} \textsuperscript{H-CR}
        & \textcolor{jnSUDigitalGreen}{\cmark{}}
        & \textcolor{jnSUDigitalGreen}{$O(n^2)$}
        & \textcolor{jnSUDigitalRed}{$O(n^2 \ell + n^3 \kappa)$}
        & \textcolor{jnSUDigitalGreen}{$O(1)$}
        & \textcolor{jnSUDigitalGreen}{$t < \frac13 n$}
        \\[2pt]

        FIN-MVBA~\cite{finmvba} \textsuperscript{H-CR}
        & \textcolor{jnSUDigitalGreen}{\cmark{}}
        & \textcolor{jnSUDigitalRed}{$O(n^3)$}
        & \textcolor{jnSUDigitalRed}{$O(n^2 \ell + n^3 \kappa)$}
        & \textcolor{jnSUDigitalGreen}{$O(1)$}
        & \textcolor{jnSUDigitalGreen}{$t < \frac13 n$}
        \\[2pt]
        FIN-MVBA~\cite{finmvba} \textsuperscript{NO}
        & \textcolor{jnSUDigitalGreen}{\cmark{}}
        & \textcolor{jnSUDigitalRed}{$O(n^3)$}
        & \textcolor{jnSUDigitalRed}{$O(n^2 \ell + n^2 \kappa + n^3 \log n)$}
        & \textcolor{jnSUDigitalGreen}{$O(1)$}
        & \textcolor{jnSUDigitalGreen}{$t < \frac13 n$}
        \\[2pt]
        HMVBA~\cite{hmvba} \textsuperscript{H-CR}
        & \textcolor{jnSUDigitalGreen}{\cmark{}}
        & \textcolor{jnSUDigitalGreen}{$O(n^2)$}
        & \textcolor{jnSUDigitalGreen}{$O(n \ell + n^2 \kappa \log n)$}
        & \textcolor{jnSUDigitalGreen}{$O(1)$}
        & \textcolor{jnSUDigitalRed}{$t < \frac15 n$}
        \\[2pt]

        SQ~\cite{sq} \textsuperscript{NO}
        & \textcolor{jnSUDigitalGreen}{\cmark{}}
        & \textcolor{jnSUDigitalGreen}{$O(n^2)$}
        & \textcolor{jnSUDigitalRed}{$O(n^3 \ell)$}
        & \textcolor{jnSUDigitalGreen}{$O(1)$}
        & \textcolor{jnSUDigitalGreen}{$t < \frac13 n$}
        \\[2pt]

        SQ~\cite{sq} \textsuperscript{H-CR}
        & \textcolor{jnSUDigitalGreen}{\cmark{}}
        & \textcolor{jnSUDigitalGreen}{$O(n^2)$}
        & \textcolor{jnSUDigitalRed}{$O(n^2 \ell + n^3 \kappa)$}
        & \textcolor{jnSUDigitalGreen}{$O(1)$}
        & \textcolor{jnSUDigitalGreen}{$t < \frac13 n$}
        \\[2pt]
        OciorMVBArr~\cite{chen2024ociormvbanearoptimalerrorfreeasynchronous} \textsuperscript{NO} \textsuperscript{\textdagger}
        & \textcolor{jnSUDigitalGreen}{\cmark{}}
        & \textcolor{jnSUDigitalGreen}{$O(n^2)$}
        & \textcolor{jnSUDigitalGreen}{$O(n \ell + n^2 \log n)$}
        & \textcolor{jnSUDigitalGreen}{$O(1)$}
        & \textcolor{jnSUDigitalRed}{$t < \frac15 n$}
        \\[2pt]
        OciorMVBA~\cite{chen2024ociormvbanearoptimalerrorfreeasynchronous} \textsuperscript{NO} \textsuperscript{\textdagger}
        & \textcolor{jnSUDigitalGreen}{\cmark{}}
        & \textcolor{jnSUDigitalGreen}{$O(n^2)$}
        & \textcolor{jnSUAccentPoppyLight}{$O(n \ell \log n + n^2 \log n)$}
        & \textcolor{jnSUDigitalRed}{$O(\log n)$}
        & \textcolor{jnSUDigitalGreen}{$t < \frac13 n$}
        \\[2pt]
        OciorMVBAh~\cite{chen2024ociormvbanearoptimalerrorfreeasynchronous} \textsuperscript{H-CR} \textsuperscript{\textdagger}
        & \textcolor{jnSUDigitalGreen}{\cmark{}}
        & \textcolor{jnSUDigitalRed}{$O(n^3)$}
        & \textcolor{jnSUDigitalRed}{$O(n \ell + n^3 \kappa)$}
        & \textcolor{jnSUDigitalGreen}{$O(1)$}
        & \textcolor{jnSUDigitalGreen}{$t < \frac13 n$}
        \\[2pt]
        FLT24-MVBA~\cite{flt24mvba} \textsuperscript{H-CR} \textsuperscript{\textdagger}
        & \textcolor{jnSUDigitalGreen}{\cmark{}}
        & \textcolor{jnSUAccentPoppyLight}{$O(n^2 \lambda)$}
        & \textcolor{jnSUAccentPoppyLight}{$O(n \ell + n^2 \kappa \log n + n^2 \lambda \kappa)$}
        & \textcolor{jnSUAccentPoppyLight}{$O(\log \lambda)$}
        & \textcolor{jnSUDigitalGreen}{$t < \frac13 n$}
        \\[1pt]\midrule
        This work: \comm \textsuperscript{H-CR}
        & \textcolor{jnSUDigitalGreen}{\cmark{}}
        & \textcolor{jnSUDigitalGreen}{$O(n^2)$}
        & \textcolor{jnSUDigitalGreen}{$O(n \ell + n^2 \kappa \log n )$}
        & \textcolor{jnSUDigitalGreen}{$O(1)$}
        & \textcolor{jnSUAccentPoppyLight}{$t < \frac14 n$}
        \\[2pt]
        This work: \commplus \textsuperscript{H-RO}
        & \textcolor{jnSUDigitalGreen}{\cmark{}}
        & \textcolor{jnSUDigitalGreen}{$O(n^2)$}
        & \textcolor{jnSUDigitalGreen}{$O (n \ell + n^2 \kappa \log n)$}
        & \textcolor{jnSUDigitalGreen}{$O(1)$}
        & \textcolor{jnSUDigitalGreen}{$t < (\frac13 - \epsilon)n$}
        \\
        \bottomrule
    \end{tabular}
    \end{center}

    \smallskip
    Here, $\ell$ denotes the bit-length of the considered values, while $\kappa$ represents the length of a hash value or a signature.
    Notations ``\textsuperscript{H-CR}'', ``\textsuperscript{H-RO}'', and ``\textsuperscript{TS}'' indicate the use of collision-resistant hashes, hashes modeled as a random oracle, and threshold signatures, respectively, each of length $\kappa$;
    ``\textsuperscript{NO}'' indicates use of no cryptography (assuming authenticated channels).
    Notation ``/HS''
    refers to a hash-based variant in which the threshold signatures of the original protocol are substituted with lists of hash-based signatures from $n-t$ processes; here, we ignore the fact that some protocols crucially rely on an unforgeable unique threshold signature to generate a common coin.
    In theory, utilizing STARKs for threshold signatures would allow each entry with ``/HS'' to go from a factor $n^3\kappa$ to $n^2\log(n)^2\kappa$ in bits; in practice, however, STARKs would introduce significant performance penalties (see, for example,~\cite[Table~5]{preonpdf}).
    The $\lambda$ parameter in the complexity of FLT24-MVBA is a statistical security parameter.
    \commplus's complexity exhibits a multiplicative constant factor of $C^2$, where $C  = \lceil 1 + \frac{2}{\epsilon} \rceil^2$, and an additional $1/\epsilon$ factor in the payload term of the bit complexity.
    (We discuss in \Cref{subsection:reducing_constants} whether the multiplicative constant factor can be reduced to $O(1 / \epsilon^3)$.)
    As per convention in the asynchronous-agreement literature, the reported message, bit, and time complexities pertain to the protocol core, assuming the availability of common coins and consistently omitting their associated costs.
    Other asynchronous agreement protocols (including those with a weak common coin) are covered in \Cref{section:related_work}.
    \textsuperscript{\textdagger} Preprints announced shortly after preprint of this work.
    \label{tab:mvba-constructions}
\end{table}

\smallskip
\noindent\textbf{Scope \& state-of-the-art.}
In summary, the broad subject of this paper is \emph{adaptively-secure hash-based asynchronous MVBA}.
Within this scope, as is usual,
we want to minimize a protocol's time, message, and bit complexity,
and to maximize its resilience (the number $t$ of faulty processes it can tolerate relative to the number $n$ of all processes).
The best known MVBA protocols 
for this standard setting
are HMVBA~\cite{hmvba} (\emph{DSN'25}) and SQ~\cite{sq} (\emph{S\&P'25}) (see \Cref{tab:mvba-constructions}).
HMVBA 
relies solely on collision-resistant hash functions 
and achieves optimal $O(n^2)$ expected message complexity,
(near-)optimal $O(n\ell + n^2 \kappa \log n)$ expected bit complexity (where $\kappa$ denotes the size of a hash value) 
and optimal $O(1)$ expected time complexity.
A key drawback of HMVBA is its sub-optimal $t < n/5$ resilience.
SQ, on the other hand, tolerates up to $t < n/3$ faults, exchanges $O(n^2)$ messages and terminates in $O(1)$ time,
while also relying exclusively on collision-resistant hash functions.
However, this comes at the cost of increased communication. 
Specifically, SQ incurs a bit complexity of $O(n^2 \ell + n^3 \kappa)$.

\smallskip
\noindent\textbf{Contributions.}
It is therefore natural to ask for a protocol
that combines the strengths of HMVBA and SQ.
To this end, we first present \comm, an adaptively-secure hash-based asynchronous MVBA protocol that matches the complexity of the HMVBA protocol while improving its resilience to $t < n/4$ (see \Cref{tab:mvba-constructions}).
Like HMVBA and SQ, \comm relies solely on collision-resistant hash functions.
A novel aspect of \comm's design is its reliance on \emph{strong multi-valued Byzantine agreement} (SMBA), a new variant of the Byzantine agreement primitive, which we introduce and construct, and that ensures that correct processes reach agreement on the proposal of a correct process.
In particular, during ``good'' iterations---where the common coin produces favorable outputs---\comm utilizes the SMBA primitive as a crucial mechanism to guarantee the protocol's termination.
To achieve \comm's optimal complexity, an optimal SMBA algorithm was required.
To this end, we introduce a cryptography-free SMBA algorithm (assuming a common coin) with optimal complexity.
Both SMBA as a primitive and our SMBA algorithm may be of independent interest.

To approach the optimal one-third resilience found in SQ, we then propose \commplus, a hash-based MVBA protocol
with further improved $t < (1/3 - \epsilon)n$ resilience, for any fixed constant $\epsilon > 0$ independent of $n$.
Expected time, message, and bit complexity of \commplus remain constant, quadratic, and quasi-quadratic, respectively, with the multiplicative constant factor of $C^2$, where $C  = \lceil 1 + \frac{2}{\epsilon} \rceil^2$, and an additional $1/\epsilon$ factor in the payload term of the bit complexity.
(As discussed in \Cref{subsection:reducing_constants}, whether the constant multiplicative factor can be reduced to $O(1 / \epsilon^3)$ remains open.)
\commplus, however, necessitates an additional (but standard) assumption: the underlying hash function must be modeled as a random oracle.
This assumption plays a crucial role in ensuring the protocol's termination.

\smallskip
\noindent \textbf{Instantiations of common coins.}
As noted earlier, we assume an idealized common-coin object, whose implementation typically requires either a trusted dealer~\cite{Rabin83} or threshold pseudorandom functions~\cite{CachinKS05}.
Crucially, however, an \emph{idealized} common coin is assumed 
\emph{purely} to simplify the exposition: our algorithms (can be easily modified to) retain their asymptotic complexity 
guarantees 
even when the coin is \emph{weak}, allowing a constant probability of disagreement.
(We discuss these modifications in \Cref{section:preliminaries}.)
This distinction matters, as weak coins can be implemented efficiently using only hash functions.
The best known adaptively-secure hash-based implementation of a weak common coin is HashRand~\cite{DBLP:journals/iacr/BandarupalliBBKR23}, which incurs a communication cost of $O(\kappa n^3 \log n)$ bits and time complexity of $O(1)$.
Therefore, to be fully precise, our results---like those of~\cite{hmvba,flt24mvba,chen2024ociormvbanearoptimalerrorfreeasynchronous}---should be viewed as a hash-only reduction from efficient MVBA to an efficient common coin, while preserving adaptive security.
Importantly, however, HashRand has amortized bit complexity of $O(\kappa n^2 \log n)$, which means it does \emph{not} increase the overall complexity of our protocols, at least in the amortized sense.\footnote{HashRand relies on both authenticated and private channels, while our algorithms require only authenticated channels.
In practice, both assumptions can be met using a post-quantum symmetric encryption scheme such as AES-GCM.
Designing weak common coins---and agreement primitives more broadly---using \emph{only} authenticated channels remains a significant challenge, as evidenced by the fact that the best known solution incurs a communication cost of $\Omega(n^6)$ bits while terminating in $\Tilde{O}(n^{12})$ time~\cite{huang2023byzantine}.
See \Cref{section:related_work} for more details.}
From a practical perspective, our protocols' randomness (i.e., our weak common coins) could be derived from natural shared entropy sources (e.g., proof-of-work blockchains). 
Since only the liveness---not the safety---of our protocols depends on correct processes receiving the \emph{same} random value from the common-coin abstraction, expensive adversarial manipulation of the randomness source, leading to disagreement over a random value, can at most affect liveness of our solutions, and only while the manipulation is ongoing.
This approach may be sufficient for real-world applications.
Having reviewed possible implementations, 
we now abstract away from these details, and, going forward, assume an ideal common coin.

\smallskip
\noindent\textbf{Concurrent work.}
In a preprint announced shortly after the first version of this work, FLT24-MVBA~\cite{flt24mvba} studies the same problem as we do, and 
achieves results that are incomparable and complementary to ours.
On the positive side, FLT24-MVBA attains optimal resilience, tolerating up to one-third faulty processes, which improves upon the resilience of
\comm and \commplus.
However, our algorithms offer two significant advantages along two other dimensions:
(1)~\comm and \commplus achieve better complexity, as FLT24-MVBA exchanges $O(n^2 \lambda$) messages, $O(n \ell + n^2 \kappa \log n + n^2 \kappa \lambda)$ bits, and terminates in $O(\log \lambda)$ time, where $\lambda$ denotes a statistical security parameter.
Note that $\lambda$ is a security parameter~\cite[Sec.\ 3.1]{DBLP:books/crc/KatzLindell2007}, not a constant (see also the detailed discussion in \Cref{subsection:techniques_comparison}):
for constant $\lambda$, any FLT24-MVBA instance stalls indefinitely with constant probability.
This would prohibit the composition of any more than a constant number of FLT24-MVBA instances and would thus severely limit its applicability,
for instance
in the standard and widespread ``repeated MVBA'' structure, which 
constructs
atomic broadcast or state-machine replication from MVBA (e.g., in Tendermint~\cite{BKM19}).
(2)~Our protocols satisfy the \emph{quality} property~\cite{vaba}, which ensures that, with constant probability, the decided value was proposed by a correct process. In contrast, FLT24-MVBA does not satisfy this property. 
This, too, for instance, renders FLT24-MVBA 
ill-suited
for use 
in ``repeated MVBA'',
since the absence of the quality property allows the adversary to indefinitely censor transactions from correct participants.
We highlight that FLT24-MVBA, like other performant MVBA protocols (including ours), is sufficient to efficiently solve the asynchronous common subset (ACS) problem, although this requires the use of (hash-)signatures.
A detailed discussion of FLT24-MVBA, its techniques, and how they relate to ours is in \Cref{subsection:techniques_comparison}.

\smallskip
\noindent\textbf{Roadmap.}
In \Cref{section:technical_overview}, we provide a technical overview of our algorithms.
We introduce the system model and formally define the MVBA problem in \Cref{section:model_problem_definition}.
We outline the building blocks of our algorithms in \Cref{section:preliminaries}.
In \Cref{section:reducer}, we present the \comm algorithm and provide an informal proof of its correctness and complexity.
Then, \Cref{section:reducer_plus} introduces \commplus, including a sketch of its correctness and complexity, as well as the open question of reducing the constant multiplicative factor in its complexity from $O(1 / \epsilon^4)$ to $O(1 / \epsilon^3)$.
A discussion of related work is relegated to \Cref{section:related_work}.
Full definitions of the preliminaries are provided in \Cref{section:preliminaries_full}, while the resilience limits of \comm and \commplus are discussed in \Cref{subsection:discussion}.
The rest of the optional appendix contains all omitted algorithms and proofs.

\section{Technical Overview} \label{section:technical_overview}

Given that both \comm and \commplus are inspired by the HMVBA algorithm~\cite{hmvba}, which achieves an expected message complexity of $O(n^2)$, an expected bit complexity of $O(n\ell + n^2 \kappa \log n)$ and terminates in $O(1)$ expected time (see \Cref{tab:mvba-constructions}), we start by revisiting HMVBA (\Cref{subsection:revisitin_hmvba}).
This review provides valuable insights into the design principles underlying our algorithms.
We then proceed to outline the key mechanisms of \comm (\Cref{subsection:reducer_overview}) and \commplus (\Cref{subsection:commplus_overview}).
We defer the discussion of why \comm's resilience is limited to $t < \frac{1}{4}n$ and \commplus's to $t < (\frac{1}{3} - \epsilon)n$ to \Cref{subsection:discussion}.

\subsection{Revisiting HMVBA} \label{subsection:revisitin_hmvba}

At a high level, the HMVBA algorithm follows the ``Disseminate-Elect-Agree'' paradigm, which we describe below.
\Cref{fig:hmvba-recap} depicts the structure of HMVBA.
For simplicity, let $n = 5t + 1$.

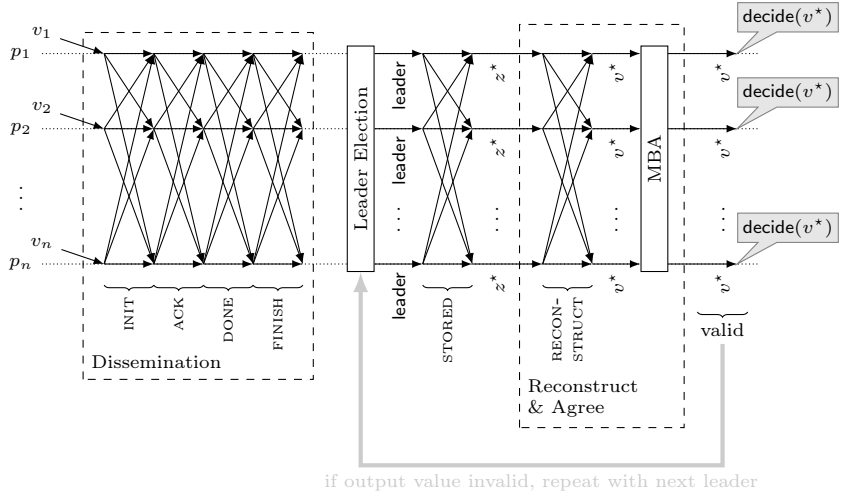
\begin{figure}[tbp]
    \centering
    \begin{tikzpicture}[
        x=1.1cm,
        y=0.9cm,
        subprotocol/.style = {
            rotate=90,
            minimum width=3cm,
            draw,
            fill=white,
        },
        transparentsubprotocol/.style = {
            draw,
            dashed,
        },
        msg/.style = {
            -latex,
        },
        msglabel/.style = {
            below,
            rotate=90,
            anchor=east,
        },
        col/.style = {
            draw,
            inner sep=0,
            minimum width=1.2em,
            minimum height=1.2em,
            fill=white,
        },
        colfirst/.style = {
            col,
            anchor=south west,
            yshift=2pt,
            xshift=1pt,
        },
        colrest/.style = {
            col,
            anchor=west,
            xshift=-0.4pt,   %
        },
        colhighlight/.style = {
            fill=NavyBlue!20,
        },
    ]
        \scriptsize

        \coordinate (p1) at (0.5,1.5);
        \coordinate (p2) at (0.5,0.5);
        \coordinate (pdots) at (0.5,-0.6);
        \coordinate (pn) at (0.5,-1.3);
        
        \coordinate (center) at (0,0.1);

        \coordinate (s1Input) at (0.5,0.0 |- center);
        \coordinate (s2InputOut) at ($(s1Input) + (0.75,0)$);
        \coordinate (s3InitOut) at ($(s2InputOut) + (0.6,0)$);
        \coordinate (s4AckOut) at ($(s3InitOut) + (0.6,0)$);
        \coordinate (s5DoneOut) at ($(s4AckOut) + (0.6,0)$);
        \coordinate (s6FinishOut) at ($(s5DoneOut) + (0.6,0)$);
        \coordinate (s7LePos) at ($(s6FinishOut) + (0.7,0)$);
        \coordinate (s8LeOut) at ($(s7LePos) + (0.75,0)$);
        \coordinate (s9StoredOut) at ($(s8LeOut) + (0.6,0)$);
        \coordinate (s11SmbaOneOut) at ($(s9StoredOut) + (0.85,0)$);
        \coordinate (s12ReconstructOut) at ($(s11SmbaOneOut) + (0.6,0)$);
        \coordinate (s13SmbaTwoPos) at ($(s12ReconstructOut) + (0.75,0)$);
        \coordinate (s14SmbaTwoOut) at ($(s13SmbaTwoPos) + (1,0)$);

        \node at ([xshift=-1em] p1) {$p_1$};
        \node at ([xshift=-1em] p2) {$p_2$};
        \node at ([xshift=-1em] pn) {$p_n$};
        
        \draw [densely dotted] (p1) -- ([xshift=1em] s14SmbaTwoOut |- p1);
        \draw [densely dotted] (p2) -- ([xshift=1em] s14SmbaTwoOut |- p2);
        \draw [densely dotted] (pn) -- ([xshift=1em] s14SmbaTwoOut |- pn);

        \node (inp1) at ([yshift=1em] s1Input |- p1) {$v_{1}$};
        \node (inp2) at ([yshift=1em] s1Input |- p2) {$v_{2}$};
        \node (inpn) at ([yshift=1em] s1Input |- pn) {$v_{n}$};

        \node at ([yshift=1em,xshift=-1em] s1Input |- pdots) {$\vdots$};

        \draw [msg] (inp1) -- (s2InputOut |- p1);
        \draw [msg] (inp2) -- (s2InputOut |- p2);
        \draw [msg] (inpn) -- (s2InputOut |- pn);

        \draw [msg] (s2InputOut |- p1) -- (s3InitOut |- p1);
        \draw [msg] (s2InputOut |- p1) -- (s3InitOut |- p2);
        \draw [msg] (s2InputOut |- p1) -- (s3InitOut |- pn);
        \draw [msg] (s2InputOut |- p2) -- (s3InitOut |- p1);
        \draw [msg] (s2InputOut |- p2) -- (s3InitOut |- p2);
        \draw [msg] (s2InputOut |- p2) -- (s3InitOut |- pn);
        \draw [msg] (s2InputOut |- pn) -- (s3InitOut |- p1);
        \draw [msg] (s2InputOut |- pn) -- (s3InitOut |- p2);
        \draw [msg] (s2InputOut |- pn) -- (s3InitOut |- pn);

        \draw [decorate,decoration={brace,amplitude=2pt,mirror,raise=1em}] (s2InputOut |- pn) -- (s3InitOut |- pn) node [midway,yshift=-1.2em,rotate=90,anchor=east] {$\textsc{init}$};

        \draw [msg] (s3InitOut |- p1) -- (s4AckOut |- p1);
        \draw [msg] (s3InitOut |- p1) -- (s4AckOut |- p2);
        \draw [msg] (s3InitOut |- p1) -- (s4AckOut |- pn);
        \draw [msg] (s3InitOut |- p2) -- (s4AckOut |- p1);
        \draw [msg] (s3InitOut |- p2) -- (s4AckOut |- p2);
        \draw [msg] (s3InitOut |- p2) -- (s4AckOut |- pn);
        \draw [msg] (s3InitOut |- pn) -- (s4AckOut |- p1);
        \draw [msg] (s3InitOut |- pn) -- (s4AckOut |- p2);
        \draw [msg] (s3InitOut |- pn) -- (s4AckOut |- pn);

        \draw [decorate,decoration={brace,amplitude=2pt,mirror,raise=1em}] (s3InitOut |- pn) -- (s4AckOut |- pn) node [midway,yshift=-1.2em,rotate=90,anchor=east] {$\textsc{ack}$};

        \draw [msg] (s4AckOut |- p1) -- (s5DoneOut |- p1);
        \draw [msg] (s4AckOut |- p1) -- (s5DoneOut |- p2);
        \draw [msg] (s4AckOut |- p1) -- (s5DoneOut |- pn);
        \draw [msg] (s4AckOut |- p2) -- (s5DoneOut |- p1);
        \draw [msg] (s4AckOut |- p2) -- (s5DoneOut |- p2);
        \draw [msg] (s4AckOut |- p2) -- (s5DoneOut |- pn);
        \draw [msg] (s4AckOut |- pn) -- (s5DoneOut |- p1);
        \draw [msg] (s4AckOut |- pn) -- (s5DoneOut |- p2);
        \draw [msg] (s4AckOut |- pn) -- (s5DoneOut |- pn);

        \draw [decorate,decoration={brace,amplitude=2pt,mirror,raise=1em}] (s4AckOut |- pn) -- (s5DoneOut |- pn) node [midway,yshift=-1.2em,rotate=90,anchor=east] {$\textsc{done}$};

        \draw [msg] (s5DoneOut |- p1) -- (s6FinishOut |- p1);
        \draw [msg] (s5DoneOut |- p1) -- (s6FinishOut |- p2);
        \draw [msg] (s5DoneOut |- p1) -- (s6FinishOut |- pn);
        \draw [msg] (s5DoneOut |- p2) -- (s6FinishOut |- p1);
        \draw [msg] (s5DoneOut |- p2) -- (s6FinishOut |- p2);
        \draw [msg] (s5DoneOut |- p2) -- (s6FinishOut |- pn);
        \draw [msg] (s5DoneOut |- pn) -- (s6FinishOut |- p1);
        \draw [msg] (s5DoneOut |- pn) -- (s6FinishOut |- p2);
        \draw [msg] (s5DoneOut |- pn) -- (s6FinishOut |- pn);

        \draw [decorate,decoration={brace,amplitude=2pt,mirror,raise=1em}] (s5DoneOut |- pn) -- (s6FinishOut |- pn) node [midway,yshift=-1.2em,rotate=90,anchor=east] (tmpFinishLabel) {$\textsc{finish}$};

        \draw [transparentsubprotocol] ([xshift=-1em,yshift=-1em] s2InputOut |- pn |- tmpFinishLabel.north west) rectangle ([xshift=0.5em,yshift=1em] s6FinishOut |- p1) node [pos=0,anchor=south west,inner sep=0,xshift=3pt,yshift=4pt] {Dissemination};

        \node [subprotocol] (s7Le) at (s7LePos) {Leader Election};

        \draw [msg] (s7Le.south |- p1) -- (s8LeOut |- p1) node [midway,msglabel] {$\mathsf{leader}$};
        \draw [msg] (s7Le.south |- p2) -- (s8LeOut |- p2) node [midway,msglabel] {$\mathsf{leader}$};
        \draw [msg] (s7Le.south |- pn) -- (s8LeOut |- pn) node [midway,msglabel] {$\mathsf{leader}$};
        \node at ($(s7Le.south |- pdots)!0.5!(s8LeOut |- pdots)$) {$\vdots$};

        \draw [msg] (s8LeOut |- p1) -- (s9StoredOut |- p1);
        \draw [msg] (s8LeOut |- p1) -- (s9StoredOut |- p2);
        \draw [msg] (s8LeOut |- p1) -- (s9StoredOut |- pn);
        \draw [msg] (s8LeOut |- p2) -- (s9StoredOut |- p1);
        \draw [msg] (s8LeOut |- p2) -- (s9StoredOut |- p2);
        \draw [msg] (s8LeOut |- p2) -- (s9StoredOut |- pn);
        \draw [msg] (s8LeOut |- pn) -- (s9StoredOut |- p1);
        \draw [msg] (s8LeOut |- pn) -- (s9StoredOut |- p2);
        \draw [msg] (s8LeOut |- pn) -- (s9StoredOut |- pn);

        \draw [decorate,decoration={brace,amplitude=2pt,mirror,raise=1em}] (s8LeOut |- pn) -- (s9StoredOut |- pn) node [midway,yshift=-1.2em,rotate=90,anchor=east] {$\textsc{stored}$};

        \draw [msg] (s9StoredOut |- p1) -- (s11SmbaOneOut |- p1) node [pos=0.15,msglabel,anchor=north east] {$z^\star$};
        \draw [msg] (s9StoredOut |- p2) -- (s11SmbaOneOut |- p2) node [pos=0.15,msglabel,anchor=north east] {$z^\star$};
        \draw [msg] (s9StoredOut |- pn) -- (s11SmbaOneOut |- pn) node [pos=0.15,msglabel,anchor=north east] {$z^\star$};
        \node at ($(s9StoredOut |- pdots)!0.35!(s11SmbaOneOut |- pdots)$) {$\vdots$};

        \draw [msg] (s11SmbaOneOut |- p1) -- (s12ReconstructOut |- p1);
        \draw [msg] (s11SmbaOneOut |- p1) -- (s12ReconstructOut |- p2);
        \draw [msg] (s11SmbaOneOut |- p1) -- (s12ReconstructOut |- pn);
        \draw [msg] (s11SmbaOneOut |- p2) -- (s12ReconstructOut |- p1);
        \draw [msg] (s11SmbaOneOut |- p2) -- (s12ReconstructOut |- p2);
        \draw [msg] (s11SmbaOneOut |- p2) -- (s12ReconstructOut |- pn);
        \draw [msg] (s11SmbaOneOut |- pn) -- (s12ReconstructOut |- p1);
        \draw [msg] (s11SmbaOneOut |- pn) -- (s12ReconstructOut |- p2);
        \draw [msg] (s11SmbaOneOut |- pn) -- (s12ReconstructOut |- pn);

        \draw [decorate,decoration={brace,amplitude=2pt,mirror,raise=1em}] (s11SmbaOneOut |- pn) -- (s12ReconstructOut |- pn) node (tmpReconstructLabel) [midway,yshift=-1.2em,rotate=90,anchor=east,align=right] {$\textsc{recon-}$\\$\textsc{struct}$};

        \node [subprotocol] (s13SmbaTwo) at (s13SmbaTwoPos) {MBA};
        \draw [msg] (s12ReconstructOut |- p1) -- (s13SmbaTwo.north |- p1) node [midway,msglabel] {$v^\star$};
        \draw [msg] (s12ReconstructOut |- p2) -- (s13SmbaTwo.north |- p2) node [midway,msglabel] {$v^\star$};
        \draw [msg] (s12ReconstructOut |- pn) -- (s13SmbaTwo.north |- pn) node [midway,msglabel] {$v^\star$};
        \node at ($(s12ReconstructOut |- pdots)!0.5!(s13SmbaTwo.north |- pdots)$) {$\vdots$};

        \draw [transparentsubprotocol] ([xshift=-1.1em,yshift=-2.75em] s11SmbaOneOut |- pn |- tmpReconstructLabel.north west) rectangle ([xshift=0.75em,yshift=1em] s13SmbaTwo.south east) node [pos=0,anchor=south west,inner sep=0,xshift=3pt,yshift=4pt,align=left] {Reconstruct\\\& Agree};

        \draw [msg] (s13SmbaTwo.south |- p1) -- (s14SmbaTwoOut |- p1) node [pos=1,msglabel,anchor=south east] {$v^\star$};
        \draw [msg] (s13SmbaTwo.south |- p2) -- (s14SmbaTwoOut |- p2) node [pos=1,msglabel,anchor=south east] {$v^\star$};
        \draw [msg] (s13SmbaTwo.south |- pn) -- (s14SmbaTwoOut |- pn) node (tmpLabelValidityValue) [pos=1,msglabel,anchor=south east] {$v^\star$};
        \node at ($(s13SmbaTwo.south |- pdots)!0.8!(s14SmbaTwoOut |- pdots)$) {$\vdots$};
        
        \draw [decorate,decoration={brace,amplitude=2pt,mirror,raise=0.5em}] ([xshift=-0.5em] tmpLabelValidityValue.north west) -- ([xshift=0.5em] tmpLabelValidityValue.south west) node (tmpLabelValidity) [midway,yshift=-0.7em,anchor=north,align=center] {valid};
        
        \node [rectangle callout, callout absolute pointer={(s14SmbaTwoOut |- p1)}, draw=black!50, anchor=south west, inner sep=2pt, fill=black!10] at ([yshift=1em,xshift=0em] s14SmbaTwoOut |- p1) {$\mathsf{decide}(v^\star)$};
        \node [rectangle callout, callout absolute pointer={(s14SmbaTwoOut |- p2)}, draw=black!50, anchor=south west, inner sep=2pt, fill=black!10] at ([yshift=1em,xshift=0em] s14SmbaTwoOut |- p2) {$\mathsf{decide}(v^\star)$};
        \node [rectangle callout, callout absolute pointer={(s14SmbaTwoOut |- pn)}, draw=black!50, anchor=south west, inner sep=2pt, fill=black!10] at ([yshift=1em,xshift=0em] s14SmbaTwoOut |- pn) {$\mathsf{decide}(v^\star)$};

        \draw [ultra thick,black!20,-latex] (tmpLabelValidity) -- ++(0,-1.8) -| (s7Le.west) node [pos=0.25,below] {if output value invalid, repeat with next leader};
        
    \end{tikzpicture}
    \caption{%
        Depiction of HMVBA's structure.
        The depiction focuses on a good iteration $k$, where  $\mathsf{leader}(k)$ has disseminated its valid proposal $v^{\star}(k)$ and the corresponding digest $z^{\star}(k)$.
        We abridge
        $\mathsf{leader} \triangleq \mathsf{leader}(k)$,
        $z^\star \triangleq z^\star(k)$,
        $v^\star \triangleq v^\star(k)$.%
    }
    \label{fig:hmvba-recap}
\end{figure}

\smallskip
\noindent\textbf{Dissemination phase.}
Each process first disseminates its proposal.
Specifically, when a correct process $p_i$ proposes a valid value $v_i$, process $p_i$ computes $n$ Reed-Solomon (RS) symbols $[ m_1, m_2, ..., m_n ]$ of value $v_i$, where $v_i$ is treated as a polynomial of degree $t$.
Process $p_i$ then utilizes Merkle-tree-based~\cite{merkle-tree-crypto87} cryptographic accumulators (see \Cref{section:preliminaries}) in the following manner:
\begin{compactenum}
    \item Process $p_i$ computes the accumulation value (i.e., the Merkle root) $z_i$ for the $[ m_1, m_2, ..., m_n ]$ set.
    We refer to this accumulation value $z_i$ as the digest of $p_i$'s proposal $v_i$.

    \item For each RS symbol $m_j$, process $p_i$ computes the witness (i.e., the Merkle proof of inclusion) $w_j$ proving that $m_j$ belongs to the $[ m_1, m_2, ..., m_n ]$ set.
\end{compactenum}
Subsequently, process $p_i$ sends each RS symbol $m_j$ along with the digest $z_i$ and witness $w_j$ via an \textsc{init} message to process $p_j$.
Once a process $p_j$ receives a valid RS symbol $m_j$ from process $p_i$, i.e., an RS symbol corresponding to the received digest $z_i$ and the received witness $w_j$, process $p_j$ replies back via an \textsc{ack} message confirming that it has received and stored $[m_j, z_i, w_j]$.
When process $p_i$ receives $n - t = 4t + 1$ \textsc{ack} messages, process $p_i$ knows that at least $(n - t) - t = 3t + 1$ valid RS symbols are stored at as many correct processes.
Then, process $p_i$ broadcasts a \textsc{done} message.\footnote{In the paper, when a process ``broadcasts'' a message, it
unicasts it to all processes.}
Once process $p_i$ receives $n - t = 4t + 1$ \textsc{done} messages, which implies that at least $(n - t) - t = 3t + 1$ correct processes successfully disseminated their values, process $p_i$ broadcasts a \textsc{finish} message.
Finally, upon receiving $n - t = 4t + 1$ \textsc{finish} messages, process $p_i$ completes the dissemination phase.

\smallskip
\noindent\emph{Key takeaways from the dissemination phase.}
First, if a process $p_i$ successfully disseminates its valid proposal $v_i$,
it is guaranteed that at least $(n - t) - t = 3t + 1$ correct processes have stored (1) the digest $z_i$ of value $v_i$, and (2) RS symbols of the value $v_i$ (one per process).
Hence, even if process $p_i$ later gets corrupted, the original value $v_i$ can be reconstructed using the material held only by those $3t + 1$ correct processes.

Second, if a correct process completes the dissemination phase, at least $(n - t) - t = 3t + 1$ processes have already successfully disseminated their valid proposals (as at most $t$ processes can be corrupted).
Thus, there are $3t + 1$ processes whose proposals can be reconstructed, even if the adversary corrupts (some of) them after dissemination.
The HMVBA algorithm leverages this insight in the subsequent phases.

\smallskip
\noindent \textbf{Election \& agreement phases.}
After the dissemination phase is concluded, processes start executing HMVBA through \emph{iterations}.
Each iteration utilizes the multi-valued Byzantine agreement (MBA) primitive~\cite{book-cachin-guerraoui-rodrigues,BE03,DBLP:journals/acta/MostefaouiR17} ensuring strong unanimity: if all correct processes propose the same value $v$, then $v$ is decided.
(We formally define the MBA primitive in \Cref{section:preliminaries}.)

At the beginning of each iteration $k$, processes go through the election phase: by utilizing an idealized common coin, processes elect the leader of iteration $k$, denoted by $\mathsf{leader}(k)$.
Then, the agreement phase starts.
The goal of this phase is for processes to agree on the $\mathsf{leader}(k)$'s proposal.
Specifically, once the leader is elected, each process $p_i$ broadcasts via a \textsc{stored} message the digest received from $\mathsf{leader}(k)$ during the dissemination phase.
If no such digest was received, process $p_i$ broadcasts a \textsc{stored} message with $\bot$.
Once process $p_i$ receives $n - t = 4t + 1$ \textsc{stored} messages, it executes the following logic:
\begin{compactitem}
    \item If there exists a digest $z$ received in a majority ($\geq 2t + 1$) of \textsc{stored} messages, process $p_i$ adopts $z$.
    If such digest $z$ exists, it is guaranteed that at least $(2t + 1) - t = t + 1$ correct processes have valid RS symbols that correspond to $z$.
    (Recall that a correct process accepts an \textsc{init} message with the digest $z$ during
    dissemination
    only if the received RS symbol and witness match $z$.)

    \item Otherwise, process $p_i$ adopts $\bot$.
\end{compactitem}
Then, process $p_i$ disseminates via a \textsc{reconstruct} message the RS symbol received from the leader during the dissemination phase.
If process $p_i$ adopted a digest $z$ ($\neq \bot$), $p_i$ waits until it receives via the aforementioned \textsc{reconstruct} messages $t + 1$ RS symbols corresponding to digest $z$, uses the received symbols to rebuild some value $r_i$, and proposes $r_i$ to the MBA primitive.\footnote{HMVBA~\cite{hmvba} actually combines \textsc{stored} and \textsc{reconstruct} messages: each \textsc{stored} message contains both the digest and the RS symbol received from the leader. For clarity, we separate them as they serve distinct functions.}
Otherwise, process $p_i$ proposes its proposal $v_i$ to the MBA primitive.
We refer to the combination of the reconstruction step and the MBA primitive as the Reconstruct \& Agree (R\&A) mechanism.
If value $v$ decided from the R\&A mechanism (i.e., from the underlying MBA primitive) is valid, processes decide $v$ from HMVBA and terminate.
If not, processes continue to the next iteration.

\smallskip
\noindent \textbf{Correctness analysis.}
We
briefly explain how HMVBA's design guarantees
correctness.
Recall that if a correct process completes the dissemination phase, it is guaranteed that at least $(n - t) - t = 3t + 1$ processes have already successfully disseminated their valid proposals.
HMVBA ensures that all correct processes decide (and terminate) in an iteration $k$ whose leader is one of the aforementioned $3t + 1$ processes.
In the rest of the section, we refer to such iterations as \emph{good}.
Let us analyze how any such good iteration $k$ of HMVBA unfolds.

The dissemination phase ensures that at least $(n - t) - t = 3t + 1$ correct processes have stored (1) the digest $z^{\star}(k)$ of the $\mathsf{leader}(k)$'s valid proposal $v^{\star}(k)$, and (2) RS symbols of value $v^{\star}(k)$.
Hence, when each correct process $p_i$ receives $n - t = 4t + 1$ \textsc{stored} messages, it is guaranteed (due to quorum intersection) that $p_i$ receives the ``good'' digest $z^{\star}(k)$ from at least $(n - t) + (n - 2t) - n = 2t + 1$ processes.
Therefore, it is ensured that \emph{all} correct processes adopt digest $z^{\star}(k)$.
We emphasize that $t < \frac{1}{5} n$ is critical in this step.
Namely, the HMVBA algorithm cannot guarantee that all correct processes adopt $z^{\star}(k)$ if $t \geq \frac{1}{5}n$.
As we argue in \Cref{subsection:reducer_overview}, this ``adoption issue'' is the primary challenge one must overcome to achieve better resilience while preserving HMVBA's complexity.
After all correct processes have adopted digest $z^{\star}(k)$, making them agree on the $\mathsf{leader}(k)$'s valid value $v^{\star}(k)$ does not represent a significant challenge.
Using the fact that at least $t + 1$ correct processes have valid RS symbols corresponding to the digest $z^{\star}(k)$, all correct processes manage to (1) reconstruct $v^{\star}(k)$ after receiving $t + 1$ such RS symbols via \textsc{reconstruct} messages, and (2) propose $v^{\star}(k)$ to the MBA primitive.
Then, all correct processes decide $v^{\star}(k)$ from the MBA primitive due to its strong unanimity property.
In other words, all correct processes output valid value $v^{\star}(k)$ from the R\&A mechanism, which ensures that all correct processes decide $v^{\star}(k)$ from HMVBA in iteration $k$ and terminate.

In summary, ensuring that all correct processes adopt the digest $z^{\star}(k)$ of the $\mathsf{leader}(k)$'s successfully disseminated valid proposal $v^{\star}(k)$ is a critical step in efficiently solving the MVBA problem against an adaptive adversary.
Our algorithms adhere to this principle as well.

\subsection{Overview of \comm} \label{subsection:reducer_overview}

To improve the resilience of HMVBA while maintaining its message, bit and time complexity, we propose \comm that tolerates up to $t < \frac{1}{4} n$ failures.
For simplicity, let $n = 4t + 1$.
(We remind the reader that the discussion of why \comm's resilience is limited to $t < \frac{1}{4} n$ can be found in \Cref{subsection:discussion}.)

\begin{figure}[tbp]
    \centering
    \begin{tikzpicture}[
        x=1.2cm,
        y=1.25cm,
        subprotocol/.style = {
            rotate=90,
            minimum width=4cm,
            draw,
            fill=white,
        },
        msg/.style = {
            -latex,
        },
        msglabel/.style = {
            below,
            rotate=90,
            anchor=east,
        },
        col/.style = {
            draw,
            inner sep=0,
            minimum width=1.2em,
            minimum height=1.2em,
            fill=white,
        },
        colfirst/.style = {
            col,
            anchor=south west,
            yshift=2pt,
            xshift=3pt,
        },
        colrest/.style = {
            col,
            anchor=west,
            xshift=-0.4pt,   %
        },
        colhighlight/.style = {
            fill=jnSUDigitalBlue!30,
        },
    ]
        \scriptsize

        \coordinate (p1) at (-0,1.5);
        \coordinate (p2) at (-0,0.5);
        \coordinate (pdots) at (-0,-0.3);
        \coordinate (pn) at (-0,-1.3);
        
        \coordinate (center) at (0,0.1);

        \coordinate (s1LePos) at (1,0.1);
        \coordinate (s2LeOutPos) at ($(s1LePos) + (0.6,0)$);
        \coordinate (s3StoredOutPos) at ($(s2LeOutPos) + (0.4,0)$);
        \coordinate (s4SuggestOutPos) at ($(s3StoredOutPos) + (0.4,0)$);
        \coordinate (s5SmbaOnePos) at ($(s4SuggestOutPos) + (1.2,0)$);
        \coordinate (s7ArOnePos) at ($(s5SmbaOnePos) + (0.8,0)$);
        \coordinate (s8ArOneOutPos) at ($(s7ArOnePos) + (0.75,0)$);
        \coordinate (s9SmbaTwoPos) at ($(s8ArOneOutPos) + (1.2,0)$);
        \coordinate (s11ArTwoPos) at ($(s9SmbaTwoPos) + (0.8,0)$);
        \coordinate (s12ArTwoOutPos) at ($(s11ArTwoPos) + (0.75,0)$);
        \coordinate (s13SmbaThreePos) at ($(s12ArTwoOutPos) + (1.2,0)$);
        \coordinate (s15ArThreePos) at ($(s13SmbaThreePos) + (0.8,0)$);
        \coordinate (s16ArThreeOutPos) at ($(s15ArThreePos) + (0.75,0)$);

        \node at ([xshift=-1em] p1) {$p_1$};
        \node at ([xshift=-1em] p2) {$p_2$};
        \node at ([xshift=-1em] pdots) {$\vdots$};
        \node at ([xshift=-1em] pn) {$p_n$};
        
        \draw [densely dotted] (p1) -- ([xshift=1em] s12ArTwoOutPos |- p1);
        \draw [densely dotted] (p2) -- ([xshift=1em] s12ArTwoOutPos |- p2);
        \draw [densely dotted] (pn) -- ([xshift=1em] s12ArTwoOutPos |- pn);
        
        \coordinate (s0DisseminationPos) at ($(s1LePos) - (0.5,0)$);
        \coordinate (sm1Input) at ($(s0DisseminationPos) - (0.65,0)$);
        \node [subprotocol] (s0Dissemination) at (s0DisseminationPos) {Dissemination};
        \node (inp1) at ([yshift=1em] sm1Input |- p1) {$v_{1}$};
        \node (inp2) at ([yshift=1em] sm1Input |- p2) {$v_{2}$};
        \node (inpn) at ([yshift=1em] sm1Input |- pn) {$v_{n}$};
        \draw [msg] (inp1) -- (s0Dissemination.north |- p1);
        \draw [msg] (inp2) -- (s0Dissemination.north |- p2);
        \draw [msg] (inpn) -- (s0Dissemination.north |- pn);

        \node [subprotocol] (s1Le) at (s1LePos) {Leader Election};

        \draw [msg] (s1Le.south |- p1) -- (s2LeOutPos |- p1) node [midway,msglabel] {$\mathsf{leader}$};
        \draw [msg] (s1Le.south |- p2) -- (s2LeOutPos |- p2) node [midway,msglabel] {$\mathsf{leader}$};
        \draw [msg] (s1Le.south |- pn) -- (s2LeOutPos |- pn) node [midway,msglabel] {$\mathsf{leader}$};
        \node at ($(s1Le.south |- pdots)!0.5!(s2LeOutPos |- pdots)$) {$\vdots$};

        \draw [msg] (s2LeOutPos |- p1) -- (s3StoredOutPos |- p1);
        \draw [msg] (s2LeOutPos |- p1) -- (s3StoredOutPos |- p2);
        \draw [msg] (s2LeOutPos |- p1) -- (s3StoredOutPos |- pn);
        \draw [msg] (s2LeOutPos |- p2) -- (s3StoredOutPos |- p1);
        \draw [msg] (s2LeOutPos |- p2) -- (s3StoredOutPos |- p2);
        \draw [msg] (s2LeOutPos |- p2) -- (s3StoredOutPos |- pn);
        \draw [msg] (s2LeOutPos |- pn) -- (s3StoredOutPos |- p1);
        \draw [msg] (s2LeOutPos |- pn) -- (s3StoredOutPos |- p2);
        \draw [msg] (s2LeOutPos |- pn) -- (s3StoredOutPos |- pn);

        \draw [decorate,decoration={brace,amplitude=2pt,mirror,raise=1em}] (s2LeOutPos |- pn) -- (s3StoredOutPos |- pn) node [midway,yshift=-1.2em,rotate=90,anchor=east] {$\textsc{stored}$};
        \draw [decorate,decoration={brace,amplitude=2pt,mirror,raise=1em}] (s3StoredOutPos |- pn) -- (s4SuggestOutPos |- pn) node [midway,yshift=-1.2em,rotate=90,anchor=east] {$\textsc{suggest}$};

        \draw [msg] (s3StoredOutPos |- p1) -- (s4SuggestOutPos |- p1);
        \draw [msg] (s3StoredOutPos |- p1) -- (s4SuggestOutPos |- p2);
        \draw [msg] (s3StoredOutPos |- p1) -- (s4SuggestOutPos |- pn);
        \draw [msg] (s3StoredOutPos |- p2) -- (s4SuggestOutPos |- p1);
        \draw [msg] (s3StoredOutPos |- p2) -- (s4SuggestOutPos |- p2);
        \draw [msg] (s3StoredOutPos |- p2) -- (s4SuggestOutPos |- pn);
        \draw [msg] (s3StoredOutPos |- pn) -- (s4SuggestOutPos |- p1);
        \draw [msg] (s3StoredOutPos |- pn) -- (s4SuggestOutPos |- p2);
        \draw [msg] (s3StoredOutPos |- pn) -- (s4SuggestOutPos |- pn);

        \node [colfirst,colhighlight] (colFirst) at (s4SuggestOutPos |- p1) {$z^{\phantom{\star}}_1$};
        \node [colrest] (colSecond) at (colFirst.east) {$z^\star_{\phantom{1}}$};
        \node [colfirst,colhighlight] (colFirst) at (s4SuggestOutPos |- p2) {$z^\star_{\phantom{1}}$};
        \node [colrest] (colSecond) at (colFirst.east) {$z^{\phantom{\star}}_2$};
        \node [colfirst,colhighlight] (colFirst) at (s4SuggestOutPos |- pn) {$z^{\phantom{\star}}_1$};
        \node [colrest] (colSecond) at (colFirst.east) {$z^\star_{\phantom{1}}$};

        \node [subprotocol] (s5SmbaOne) at (s5SmbaOnePos) {SMBA};
        \draw [msg] (s4SuggestOutPos |- p1) -- (s5SmbaOne.north |- p1) node [pos=0.7,msglabel] {$z^{\phantom{\star}}_1$};
        \draw [msg] (s4SuggestOutPos |- p2) -- (s5SmbaOne.north |- p2) node [pos=0.7,msglabel] {$z^\star_{\phantom{1}}$};
        \draw [msg] (s4SuggestOutPos |- pn) -- (s5SmbaOne.north |- pn) node [pos=0.7,msglabel] {$z^{\phantom{\star}}_1$};
        \node at ($(s4SuggestOutPos |- pdots)!0.5!(s5SmbaOne.north |- pdots)$) {$\vdots$};

        \node [subprotocol] (s7ArOne) at (s7ArOnePos) {Reconstruct \& Agree};
        \draw [msg] (s5SmbaOne.south |- p1) -- (s7ArOne.north |- p1) node [midway,msglabel] {$z^{\phantom{\star}}_1$};
        \draw [msg] (s5SmbaOne.south |- p2) -- (s7ArOne.north |- p2) node [midway,msglabel] {$z^{\phantom{\star}}_1$};
        \draw [msg] (s5SmbaOne.south |- pn) -- (s7ArOne.north |- pn) node [midway,msglabel] {$z^{\phantom{\star}}_1$};
        \node at ($(s5SmbaOne.south |- pdots)!0.5!(s7ArOne.north |- pdots)$) {$\vdots$};

        \draw [msg] (s7ArOne.south |- p1) -- (s8ArOneOutPos |- p1) node [midway,msglabel] {$v_1$};
        \draw [msg] (s7ArOne.south |- p2) -- (s8ArOneOutPos |- p2) node [midway,msglabel] {$v_1$};
        \draw [msg] (s7ArOne.south |- pn) -- (s8ArOneOutPos |- pn) node [midway,msglabel] {$v_1$};
        \node at ($(s7ArOne.south |- pdots)!0.5!(s8ArOneOutPos |- pdots)$) {$\vdots$};
        
        \draw [decorate,decoration={brace,amplitude=2pt,mirror,raise=2em}] (s7ArOne.south |- pn) -- (s8ArOneOutPos |- pn) node [midway,yshift=-2.2em,anchor=north] {$v_1$ invalid};

        \node [colfirst] (colFirst) at (s8ArOneOutPos |- p1) {$z^{\phantom{\star}}_1$};
        \node [colrest,colhighlight] (colSecond) at (colFirst.east) {$z^\star_{\phantom{1}}$};
        \node [colfirst] (colFirst) at (s8ArOneOutPos |- p2) {$z^\star_{\phantom{1}}$};
        \node [colrest,colhighlight] (colSecond) at (colFirst.east) {$z^{\phantom{\star}}_2$};
        \node [colfirst] (colFirst) at (s8ArOneOutPos |- pn) {$z^{\phantom{\star}}_1$};
        \node [colrest,colhighlight] (colSecond) at (colFirst.east) {$z^\star_{\phantom{1}}$};

        \node [subprotocol] (s9SmbaTwo) at (s9SmbaTwoPos) {SMBA};
        \draw [msg] (s8ArOneOutPos |- p1) -- (s9SmbaTwo.north |- p1) node [pos=0.7,msglabel] {$z^\star_{\phantom{1}}$};
        \draw [msg] (s8ArOneOutPos |- p2) -- (s9SmbaTwo.north |- p2) node [pos=0.7,msglabel] {$z^{\phantom{\star}}_2$};
        \draw [msg] (s8ArOneOutPos |- pn) -- (s9SmbaTwo.north |- pn) node [pos=0.7,msglabel] {$z^\star_{\phantom{1}}$};
        \node at ($(s8ArOneOutPos |- pdots)!0.5!(s9SmbaTwo.north |- pdots)$) {$\vdots$};

        \node [subprotocol] (s11ArTwo) at (s11ArTwoPos) {Reconstruct \& Agree};
        \draw [msg] (s9SmbaTwo.south |- p1) -- (s11ArTwo.north |- p1) node [midway,msglabel] {$z^{\phantom{\star}}_2$};
        \draw [msg] (s9SmbaTwo.south |- p2) -- (s11ArTwo.north |- p2) node [midway,msglabel] {$z^{\phantom{\star}}_2$};
        \draw [msg] (s9SmbaTwo.south |- pn) -- (s11ArTwo.north |- pn) node [midway,msglabel] {$z^{\phantom{\star}}_2$};
        \node at ($(s9SmbaTwo.south |- pdots)!0.5!(s11ArTwo.north |- pdots)$) {$\vdots$};

        \draw [msg] (s11ArTwo.south |- p1) -- (s12ArTwoOutPos |- p1) node [midway,msglabel] {$v_2$};
        \draw [msg] (s11ArTwo.south |- p2) -- (s12ArTwoOutPos |- p2) node [midway,msglabel] {$v_2$};
        \draw [msg] (s11ArTwo.south |- pn) -- (s12ArTwoOutPos |- pn) node [midway,msglabel] {$v_2$};
        \node at ($(s11ArTwo.south |- pdots)!0.5!(s12ArTwoOutPos |- pdots)$) {$\vdots$};
        
        \draw [decorate,decoration={brace,amplitude=2pt,mirror,raise=2em}] (s11ArTwo.south |- pn) -- (s12ArTwoOutPos |- pn) node [midway,yshift=-2.2em,anchor=north] {$v_2$ invalid};

        \node [rectangle callout, callout absolute pointer={(s12ArTwoOutPos |- p1)}, callout pointer width=2em, draw=jnSUDigitalRed!50, anchor=south west, inner sep=2pt, fill=jnSUDigitalRed!20] at ([yshift=1.7em,xshift=-1.6em] s12ArTwoOutPos |- p1) {Switch};
        \node [colfirst] (colFirst) at (s12ArTwoOutPos |- p1) {$z^{\phantom{\star}}_1$};
        \node [colrest,colhighlight] (colSecond) at (colFirst.east) {$z^\star_{\phantom{1}}$};
        \node [rectangle callout, callout absolute pointer={(s12ArTwoOutPos |- p2)}, callout pointer width=3em, draw=jnSUDigitalGreen!50, anchor=south west, inner sep=2pt, fill=jnSUDigitalGreen!20] at ([yshift=1.7em,xshift=-1.6em] s12ArTwoOutPos |- p2) {No switch};
        \node [colfirst,colhighlight] (colFirst) at (s12ArTwoOutPos |- p2) {$z^\star_{\phantom{1}}$};
        \node [colrest] (colSecond) at (colFirst.east) {$z^{\phantom{\star}}_2$};
        \node [rectangle callout, callout absolute pointer={(s12ArTwoOutPos |- pn)}, callout pointer width=2em, draw=jnSUDigitalRed!50, anchor=south west, inner sep=2pt, fill=jnSUDigitalRed!20] at ([yshift=1.7em,xshift=-1.6em] s12ArTwoOutPos |- pn) {Switch};
        \node [colfirst] (colFirst) at (s12ArTwoOutPos |- pn) {$z^{\phantom{\star}}_1$};
        \node [colrest,colhighlight] (colSecond) at (colFirst.east) {$z^\star_{\phantom{1}}$};

        \node [subprotocol] (s13SmbaThree) at (s13SmbaThreePos) {SMBA};
        \draw [msg] (s12ArTwoOutPos |- p1) -- (s13SmbaThree.north |- p1) node [pos=0.7,msglabel] {$z^\star_{\phantom{1}}$};
        \draw [msg] (s12ArTwoOutPos |- p2) -- (s13SmbaThree.north |- p2) node [pos=0.7,msglabel] {$z^\star_{\phantom{1}}$};
        \draw [msg] (s12ArTwoOutPos |- pn) -- (s13SmbaThree.north |- pn) node [pos=0.7,msglabel] {$z^\star_{\phantom{1}}$};
        \node at ($(s12ArTwoOutPos |- pdots)!0.5!(s13SmbaThree.north |- pdots)$) {$\vdots$};

        \node [subprotocol] (s15ArThree) at (s15ArThreePos) {Reconstruct \& Agree};
        \draw [msg] (s13SmbaThree.south |- p1) -- (s15ArThree.north |- p1) node [midway,msglabel] {$z^\star_{\phantom{1}}$};
        \draw [msg] (s13SmbaThree.south |- p2) -- (s15ArThree.north |- p2) node [midway,msglabel] {$z^\star_{\phantom{1}}$};
        \draw [msg] (s13SmbaThree.south |- pn) -- (s15ArThree.north |- pn) node [midway,msglabel] {$z^\star_{\phantom{1}}$};
        \node at ($(s13SmbaThree.south |- pdots)!0.5!(s15ArThree.north |- pdots)$) {$\vdots$};

        \draw [msg] (s15ArThree.south |- p1) -- (s16ArThreeOutPos |- p1) node [midway,msglabel] {$v^\star$};
        \draw [msg] (s15ArThree.south |- p2) -- (s16ArThreeOutPos |- p2) node [midway,msglabel] {$v^\star$};
        \draw [msg] (s15ArThree.south |- pn) -- (s16ArThreeOutPos |- pn) node [midway,msglabel] {$v^\star$};
        \node at ($(s15ArThree.south |- pdots)!0.5!(s16ArThreeOutPos |- pdots)$) {$\vdots$};
        
        \draw [decorate,decoration={brace,amplitude=2pt,mirror,raise=2em}] (s15ArThree.south |- pn) -- (s16ArThreeOutPos |- pn) node (tmpLabelValidity) [midway,yshift=-2.2em,anchor=north] {$v^\star$ valid};

        \draw [ultra thick,black!20,-latex] (tmpLabelValidity) -- ++(0,-0.7) -| (s1Le.west) node [pos=0.25,below] {if all output values invalid, repeat with next leader};
        
        \node [rectangle callout, callout absolute pointer={(s16ArThreeOutPos |- p1)}, draw=black!50, anchor=south west, inner sep=2pt, fill=black!10] at ([yshift=1em,xshift=-1.6em] s16ArThreeOutPos |- p1) {$\mathsf{decide}(v^\star)$};
        \node [rectangle callout, callout absolute pointer={(s16ArThreeOutPos |- p2)}, draw=black!50, anchor=south west, inner sep=2pt, fill=black!10] at ([yshift=1em,xshift=-1.6em] s16ArThreeOutPos |- p2) {$\mathsf{decide}(v^\star)$};
        \node [rectangle callout, callout absolute pointer={(s16ArThreeOutPos |- pn)}, draw=black!50, anchor=south west, inner sep=2pt, fill=black!10] at ([yshift=1em,xshift=-1.6em] s16ArThreeOutPos |- pn) {$\mathsf{decide}(v^\star)$};

    \end{tikzpicture}
    \caption{%
        Depiction of \comm's structure.
       It focuses on a good iteration $k$ in which the first two SMBA invocations decide adversarial digests $z_1$ and $z_2$, respectively.
        Finally, in the third invocation, processes that proposed $z_1$ during the first invocation (e.g., $p_1$ and $p_n$) switch their proposals to the SMBA primitive, while those that did not (e.g., $p_2$) do not switch their proposals from the first invocation.
This selective ``proposal switching'' enables the third SMBA invocation to decide on the ``good'' digest $z^{\star}(k)$ corresponding to $\mathsf{leader}(k)$'s valid proposal $v^{\star}(k)$.
        See \Cref{fig:hmvba-recap} for ``Dissemination'' and ``Reconstruct \& Agree'' sub-protocols.
        We abridge
        $\mathsf{leader} \triangleq \mathsf{leader}(k)$,
        $z^\star \triangleq z^\star(k)$,
        $v^\star \triangleq v^\star(k)$.
    }
    \label{fig:reducer-good-iteration}
\end{figure}

\smallskip
\noindent \textbf{Key concepts behind \comm.}
In \comm (\Cref{fig:reducer-good-iteration}, \Cref{algorithm:reducer}), processes first disseminate their proposals using a dissemination phase identical to that of HMVBA.
After
dissemination,
processes start executing \comm through iterations.
Each \comm's iteration $k$ starts in the same way as HMVBA's iterations: (1) correct processes elect the leader of the iteration $k$ using a common coin, (2) correct processes broadcast their \textsc{stored} messages containing the digest received from the leader during the dissemination phase, and (3) each correct process waits for $n - t = 3t + 1$ such \textsc{stored} messages.
To motivate our design choices, we explain how \comm ensures termination in a good iteration.
Hence, for the remainder of this subsection, we fix a good iteration $k$.

\smallskip
\noindent\emph{Resolving the adoption issue.}
As mentioned in \Cref{subsection:revisitin_hmvba}, \comm (and HMVBA with $n = 4t + 1$) cannot ensure that any correct process receives the ``good'' digest $z^{\star}(k)$ of the $\mathsf{leader}(k)$'s valid proposal $v^{\star}(k)$ in a majority ($> \frac{3t + 1}{2}$) of the received \textsc{stored} messages.
Instead, if the adversary corrupts the leader and makes it disseminate an additional adversarial digest, that adversarial digest might appear in a majority.
Crucially, however, it is guaranteed that every correct process $p_i$ receives $z^{\star}(k)$ in (at least) $t + 1$ received \textsc{stored} messages.
Indeed, as at least $(n - t) - t = 2t + 1$ correct processes have stored $z^{\star}(k)$, any set of $n - t = 3t + 1$ \textsc{stored} messages must include at least $(n - t) + (n - 2t) - n = t + 1$ messages for $z^{\star}(k)$.
Therefore, in \comm, once a correct process $p_i$ receives $n- t = 3t + 1$ \textsc{stored} messages, $p_i$ marks any digest received in at least $t + 1$ such messages as a \emph{candidate} digest.
Concretely, each correct process $p_i$ has its local list $\mathit{candidates}_i$ that initially contains all digests received in $t + 1$ \textsc{stored} messages.
Importantly, given that $k$ is a good iteration, $z^{\star}(k)$ belongs to the $\mathit{candidates}_i$ list at every correct process $p_i$.
Moreover, since $\frac{3t + 1}{t + 1} < 3$, each correct process has at most two candidates.
However, if an adaptive adversary corrupts $\mathsf{leader}(k)$ after its election and forces it to disseminate adversarial digests, the adversary can ensure the existence of linearly many different candidates \emph{across all} correct processes:
\begin{equation*}
    |\{ \text{candidate of \emph{some} correct process} \}| \in O(n).
\end{equation*}
The (first) core idea of our \comm algorithm is to \emph{reduce} the number of different candidates across all correct processes in a good iteration to a constant.\footnote{A ``reducing'' technique similar to ours was previously employed for asynchronous Byzantine agreement by Mostéfaoui \emph{et al.}~\cite{DBLP:journals/acta/MostefaouiR17}. However, their technique is insufficient for our algorithms, as detailed in \Cref{subsection:techniques_comparison}.}
Thus, the algorithm's name.

\smallskip
\noindent\emph{Reducing the number of different candidates.}
We achieve the reduction in a good iteration $k$ using an additional ``all-to-all'' communication step.
Specifically, when a correct process $p_i$ has determined its list of candidates, it disseminates them via a \textsc{suggest} message.
If process $p_i$ includes a digest $z$ in the aforementioned \textsc{suggest} message, we say that $p_i$ \emph{suggests} $z$ in iteration $k$.
Recall that each correct process suggests at most two digests out of which one is $z^{\star}(k)$.

Once process $p_i$ receives $n - t = 3t + 1$ \textsc{suggest} messages, process $p_i$ refines its $\mathit{candidates}_i$ list.
Concretely, for every digest $z$ suggested by $p_i$ in iteration $k$, process $p_i$ removes $z$ from the $\mathit{candidates}_i$ list unless it receives $z$ in at least $(n - t) - t  = 2t + 1$ \textsc{suggest} messages.
First, this design ensures that $z^{\star}(k)$ ``survives'' this communication step at $p_i$ as (1) every correct process suggests $z^{\star}(k)$, and (2) $p_i$ hears suggestions of at least $(n - t) - t = 2t + 1$ correct processes.
Second, this design ensures that at most $3 \in O(1)$ digests ``survive'' this communication step \emph{across all correct processes}.
Let us elaborate.
Given that each correct process suggests at most two digests out of which one is $z^{\star}(k)$, there are (at most) $n - t = 3t + 1$ suggestions coming from correct processes for adversarial non-$z^{\star}(k)$ digests (assuming there are $t$ faulty processes).
Moreover, each adversarial non-$z^{\star}(k)$ digest that ``survives'' this step receives (at least) $(2t + 1 )- t = t + 1$ suggestions coming from correct processes.
Hence, as $\frac{3t + 1}{t + 1} < 3$, at most two adversarial non-$z^{\star}(k)$ candidates get through the suggestion step.
Consequently, a maximum of three candidates, including $z^{\star}(k)$, remain.

\smallskip
 \noindent\emph{Establishing order in the chaos of candidates.}
At this point, each correct process has up to two candidate digests (one of which is $z^{\star}(k)$), and across all correct processes there are only up to three different candidate digests.
As we will see below, we isolate and construct a special agreement primitive---strong multi-valued Byzantine agreement (SMBA), defined in \Cref{section:preliminaries}---which ensures that if up to two different digests are proposed by correct processes, then the decided digest is among those proposed by a correct process.
The high-level idea now is to invoke this primitive multiple times, and in each invocation correct processes pick the digest to propose from their local candidates in a ``smart'' way so that: (1)~For each invocation, correct processes propose at most two different digests.
(2)~As a result, correct processes learn from the decided digests about the candidate digests of other correct processes, and can adjust their proposals so that (3)~in one of the invocations, all correct processes will inevitably propose and decide $z^{\star}(k)$.

Specifically, suppose each correct process $p_i$ proceeds by sorting its $\mathit{candidates}_i$ list lexicographically.
If $p_i$ has only one candidate, which must be $z^{\star}(k)$, $p_i$ duplicates $z^{\star}(k)$, resulting in $\mathit{candidates}_i = [z^{\star}(k), z^{\star}(k)]$.
We say that a correct process $p_i$ \emph{1-commits} (resp., \emph{2-commits}) a digest $z$ in iteration $k$ if $\mathit{candidates}_i[1] = z$ (resp., $\mathit{candidates}_i[2] = z$) after the sorting and (potentially) duplicating steps.
For any $c \in \{ 1, 2 \}$, let us define the $\mathsf{committed}(k, c)$ set:
\begin{equation*}
    \mathsf{committed}(k, c) = \{ z \,|\, \text{$z$ is $c$-committed by any correct process in iteration $k$} \}.
\end{equation*}
We also say that a correct process $p_i$ \emph{commits} a digest $z$ in iteration $k$ if $p_i$ $1$-commits or $2$-commits $z$ in iteration $k$.
Let us define the $\mathsf{committed}(k)$ set:
\begin{equation*}
    \mathsf{committed}(k) = \{ z \,|\, \text{$z$ is committed by any correct process in iteration $k$} \}.
\end{equation*}
Note that $\mathsf{committed}(k) = \mathsf{committed}(k, 1) \cup \mathsf{committed}(k, 2)$.

First, suppose $|\mathsf{committed}(k)|=2$.
Let $\mathsf{committed}(k) = \{ z, z^{\star}(k) \}$, and assume that, without loss of generality, $z^{\star}(k)$ is lexicographically smaller than $z$.
Each correct process $p_i$ has its $\mathit{candidates}_i$ list as either $[z^{\star}(k), z^{\star}(k)]$ or $[z^{\star}(k), z]$.
(Recall that each correct process necessarily commits $z^{\star}(k)$.)
In this case, if the correct processes invoke the SMBA primitive twice---first proposing their first committed candidate, and then proposing their second committed candidate---they agree on $z^{\star}(k)$ during the first invocation.
Agreement on $z^{\star}(k)$ would be sufficient for the processes to reconstruct $v^{\star}(k)$, agree on it, and thus terminate.

Now, suppose $|\mathsf{committed}(k)|=3$,
with $\mathsf{committed}(k) = \{ z_1, z_2, z^\star(k) \}$.
Assume that $z^\star(k)$ is lexicographically smaller than $z_1$ and $z_2$.
Then, each correct process $p_i$ has its $\mathit{candidates}_i$ list as either $[z^\star(k), z^\star(k)]$, $[z^\star(k), z_1]$, or $[z^\star(k), z_2]$.
(We again stress that every correct process is guaranteed to commit $z^{\star}(k)$.)
The same procedure as in the case above leads to termination: during the first invocation of the SMBA primitive, correct processes agree on $z^{\star}(k)$, then agree on $v^{\star}(k)$, and terminate.
The above argument naturally carries over to the case where $z^\star(k)$ is lexicographically greater than $z_1$ and $z_2$: the processes agree on $z^{\star}(k)$ during the second invocation of the SMBA primitive, ensuring termination.

Finally, consider the case where $\mathsf{committed}(k) = \{ z_1, z_2, z^{\star}(k) \}$ and $z^\star(k)$ is lexicographically between $z_1$ and $z_2$.
This is where the situation becomes more intricate.
Each correct process $p_i$ now has its $\mathit{candidates}_i$ list as
either $[z^\star(k), z^\star(k)]$, $[z_1, z^\star(k)]$, or $[z^\star(k), z_2]$.
We employ the same procedure as outlined above: correct processes invoke the SMBA primitive twice, initially proposing their first committed candidate, followed by proposing their second committed candidate.
However, in this case, correct processes are not guaranteed to agree on $z^{\star}(k)$ in any of the two invocations.
The only assurance for the correct processes is that the decided digest from the first invocation is either $z_1$ or $z^{\star}(k)$, and the decided digest from the second invocation is either $z^{\star}(k)$ or $z_2$.
If either of these two invocations decides $z^\star(k)$, we are in a favorable position, following the same reasoning as previously discussed.

It is left to deal with the case where the digests decided by the two invocations are $z_1$ and $z_2$, respectively.
Our approach is to ``retry'' the first invocation.
Let us elaborate on this.
After correct processes agree on $z_1$ and $z_2$ during the first and second invocation, respectively, they will initiate the third invocation in the following manner:
All correct processes that proposed the digest $z_1$ decided in the first invocation now propose their other committed candidate, which must be $z^{\star}(k)$.
All correct processes that proposed a digest other than $z_1$ to the first invocation stick with the same committed candidate for the third invocation---this digest must also be $z^{\star}(k)$.
As a result, all correct processes propose $z^{\star}(k)$ in the third invocation (which represents a ``repetition'' of the first invocation).
Therefore, agreement on $z^{\star}(k)$ is ensured in the third invocation, which, following the earlier arguments, implies \comm's termination.

\smallskip
\noindent \emph{Complexity analysis.}
As \comm is guaranteed to terminate in the first good iteration, and the good-leader set is frozen before the election randomness is disclosed, every reached iteration is good with conditional probability $\geq \frac{n - 2t}{n} = \frac{2t + 1}{4t + 1} \geq \frac{1}{2}$.
Therefore, \comm terminates in $O(1)$ expected time.
As correct processes send $O(n^2)$ messages and $O(n \ell + n^2 \kappa \log n)$ bits during
dissemination
and during each iteration, with \comm terminating in constantly many iterations, the expected message complexity is $O(n^2$) and the expected bit complexity is $O(n \ell + n^2 \kappa \log n)$.

\subsection{Overview of \commplus}
\label{subsection:commplus_overview}

As already noted, \commplus improves \comm's resilience to $t < (\frac{1}{3} - \epsilon) n$ Byzantine failures, for any fixed constant $\epsilon > 0$, while maintaining its complexity.
Let $n = (3 + \epsilon)t + 1$ in this subsection.
(We remind the reader that the reason \commplus's resilience is bounded by $t < (\frac{1}{3} - \epsilon) n$ is discussed in \Cref{subsection:discussion}.)

\begin{figure}[tbp]
    \centering
    \begin{tikzpicture}[
        x=1cm,
        y=0.9cm,
        subprotocol/.style = {
            rotate=90,
            minimum width=3cm,
            draw,
            fill=white,
        },
        transparentsubprotocol/.style = {
            draw,
            dashed,
        },
        msg/.style = {
            -latex,
        },
        msglabel/.style = {
            below,
            rotate=90,
            anchor=east,
        },
        col/.style = {
            draw,
            inner sep=0,
            minimum width=0.6em,
            minimum height=1.2em,
            fill=white,
        },
        colfirst/.style = {
            col,
            anchor=south west,
            yshift=2pt,
            xshift=2pt,
        },
        colrest/.style = {
            col,
            anchor=west,
            xshift=-0.4pt,   %
        },
        colzstar/.style = {
            minimum width=1.2em,
        },
    ]
        \scriptsize

        \coordinate (p1) at (-0.4,1.5);
        \coordinate (p2) at (-0.4,0.5);
        \coordinate (pdots) at (-0.4,-0.6);
        \coordinate (pn) at (-0.4,-1.3);
        
        \coordinate (center) at (0,0.1);

        \coordinate (s1LePos) at (1,0.0 |- center);
        \coordinate (s2LeOutPos) at ($(s1LePos) + (0.8,0)$);
        \coordinate (s3StoredOutPos) at ($(s2LeOutPos) + (0.75,0)$);
        \coordinate (s4SuggestOutPos) at ($(s3StoredOutPos) + (0.75,0)$);
        \coordinate (s5GameOnePos) at ($(s4SuggestOutPos) + (1.5,0)$);
        \coordinate (s6GameDotsPos) at ($(s5GameOnePos) + (0.85,0)$);
        \coordinate (s7GameNPos) at ($(s6GameDotsPos) + (0.85,0)$);
        \coordinate (s8GameNOutPos) at ($(s7GameNPos) + (0.75,0)$);

        \node at ([xshift=-1em] p1) {$p_1$};
        \node at ([xshift=-1em] p2) {$p_2$};
        \node at ([xshift=-1em,yshift=1em] pdots) {$\vdots$};
        \node at ([xshift=-1em] pn) {$p_n$};
        
        \draw [densely dotted] (p1) -- ([xshift=1.5em] s8GameNOutPos |- p1);
        \draw [densely dotted] (p2) -- ([xshift=1.5em] s8GameNOutPos |- p2);
        \draw [densely dotted] (pn) -- ([xshift=1.5em] s8GameNOutPos |- pn);
        
        \coordinate (s0DisseminationPos) at ($(s1LePos) - (0.7,0)$);
        \coordinate (sm1Input) at ($(s0DisseminationPos) - (0.8,0)$);
        \node [subprotocol] (s0Dissemination) at (s0DisseminationPos) {Dissemination};
        \node (inp1) at ([yshift=1em] sm1Input |- p1) {$v_{1}$};
        \node (inp2) at ([yshift=1em] sm1Input |- p2) {$v_{2}$};
        \node (inpn) at ([yshift=1em] sm1Input |- pn) {$v_{n}$};
        \draw [msg] (inp1) -- (s0Dissemination.north |- p1);
        \draw [msg] (inp2) -- (s0Dissemination.north |- p2);
        \draw [msg] (inpn) -- (s0Dissemination.north |- pn);

        \node [subprotocol] (s1Le) at (s1LePos) {Leader Election};

        \draw [msg] (s1Le.south |- p1) -- (s2LeOutPos |- p1) node [midway,msglabel] {$\mathsf{leader}$};
        \draw [msg] (s1Le.south |- p2) -- (s2LeOutPos |- p2) node [midway,msglabel] {$\mathsf{leader}$};
        \draw [msg] (s1Le.south |- pn) -- (s2LeOutPos |- pn) node [midway,msglabel] {$\mathsf{leader}$};
        \node at ($(s1Le.south |- pdots)!0.5!(s2LeOutPos |- pdots)$) {$\vdots$};

        \draw [msg] (s2LeOutPos |- p1) -- (s3StoredOutPos |- p1);
        \draw [msg] (s2LeOutPos |- p1) -- (s3StoredOutPos |- p2);
        \draw [msg] (s2LeOutPos |- p1) -- (s3StoredOutPos |- pn);
        \draw [msg] (s2LeOutPos |- p2) -- (s3StoredOutPos |- p1);
        \draw [msg] (s2LeOutPos |- p2) -- (s3StoredOutPos |- p2);
        \draw [msg] (s2LeOutPos |- p2) -- (s3StoredOutPos |- pn);
        \draw [msg] (s2LeOutPos |- pn) -- (s3StoredOutPos |- p1);
        \draw [msg] (s2LeOutPos |- pn) -- (s3StoredOutPos |- p2);
        \draw [msg] (s2LeOutPos |- pn) -- (s3StoredOutPos |- pn);

        \draw [decorate,decoration={brace,amplitude=2pt,mirror,raise=1em}] (s2LeOutPos |- pn) -- (s3StoredOutPos |- pn) node [midway,yshift=-1.2em,rotate=90,anchor=east] {$\textsc{stored}$};

        \draw [msg] (s3StoredOutPos |- p1) -- (s4SuggestOutPos |- p1);
        \draw [msg] (s3StoredOutPos |- p1) -- (s4SuggestOutPos |- p2);
        \draw [msg] (s3StoredOutPos |- p1) -- (s4SuggestOutPos |- pn);
        \draw [msg] (s3StoredOutPos |- p2) -- (s4SuggestOutPos |- p1);
        \draw [msg] (s3StoredOutPos |- p2) -- (s4SuggestOutPos |- p2);
        \draw [msg] (s3StoredOutPos |- p2) -- (s4SuggestOutPos |- pn);
        \draw [msg] (s3StoredOutPos |- pn) -- (s4SuggestOutPos |- p1);
        \draw [msg] (s3StoredOutPos |- pn) -- (s4SuggestOutPos |- p2);
        \draw [msg] (s3StoredOutPos |- pn) -- (s4SuggestOutPos |- pn);
        
        \draw [decorate,decoration={brace,amplitude=2pt,mirror,raise=1em}] (s3StoredOutPos |- pn) -- (s4SuggestOutPos |- pn) node [midway,yshift=-1.2em,rotate=90,anchor=east] {$\textsc{suggest}$};

        \node [colfirst] (colOne11) at (s4SuggestOutPos |- p1) {};
        \node [colrest,colzstar] (colTwo11) at (colOne11.east) {$z^\star_{\phantom{1}}$};
        \node [colrest] (colThree11) at (colTwo11.east) {};
        \node [colrest] (colFour11) at (colThree11.east) {};
        \node [colrest] (colFive11) at (colFour11.east) {};
        \node [colfirst] (colOne12) at (s4SuggestOutPos |- p2) {};
        \node [colrest] (colTwo12) at (colOne12.east) {};
        \node [colrest] (colThree12) at (colTwo12.east) {};
        \node [colrest] (colFour12) at (colThree12.east) {};
        \node [colrest,colzstar] (colFive12) at (colFour12.east) {$z^\star_{\phantom{1}}$};
        \node [colfirst] (colOne13) at (s4SuggestOutPos |- pn) {};
        \node [colrest] (colTwo13) at (colOne13.east) {};
        \node [colrest] (colThree13) at (colTwo13.east) {};
        \node [colrest,colzstar] (colFour13) at (colThree13.east) {$z^\star_{\phantom{1}}$};
        \node [colrest] (colFive13) at (colFour13.east) {};

        \node [subprotocol] (s5GameOne) at (s5GameOnePos) {Trial $1$};
        \draw [msg] (s4SuggestOutPos |- p1) -- (s5GameOne.north |- p1); %
        \draw [msg] (s4SuggestOutPos |- p2) -- (s5GameOne.north |- p2); %
        \draw [msg] (s4SuggestOutPos |- pn) -- (s5GameOne.north |- pn); %
        \node at ([yshift=1.5em] $(s4SuggestOutPos |- pdots)!0.5!(s5GameOne.north |- pdots)$) {$\vdots$};

        \node [fill=white] (tmpDots1) at (s6GameDotsPos |- p1) {$...$};
        \node [fill=white] (tmpDots2) at (s6GameDotsPos |- p2) {$...$};
        \node [fill=white] (tmpDotsn) at (s6GameDotsPos |- pn) {$...$};
        \draw [msg] (s5GameOne.south |- p1) -- (tmpDots1);
        \draw [msg] (s5GameOne.south |- p2) -- (tmpDots2);
        \draw [msg] (s5GameOne.south |- pn) -- (tmpDotsn);

        \node [subprotocol] (s7GameN) at (s7GameNPos) {Trial $4C^2$};
        \draw [msg] (tmpDots1) -- (s7GameN.north |- p1); %
        \draw [msg] (tmpDots2) -- (s7GameN.north |- p2); %
        \draw [msg] (tmpDotsn) -- (s7GameN.north |- pn); %
        \node at ([yshift=1.5em] $(s4SuggestOutPos |- pdots)!0.5!(s5GameOne.north |- pdots)$) {$\vdots$};

        \draw [msg] (s7GameN.south |- p1) -- (s8GameNOutPos |- p1) node [pos=1,msglabel,anchor=south east] {$v^\star$};
        \draw [msg] (s7GameN.south |- p2) -- (s8GameNOutPos |- p2) node [pos=1,msglabel,anchor=south east] {$v^\star$};
        \draw [msg] (s7GameN.south |- pn) -- (s8GameNOutPos |- pn) node (tmpLabelValidityValue) [pos=1,msglabel,anchor=south east] {$v^\star$};
        \node at ([yshift=1.5em] $(s7GameN.south |- pdots)!0.8!(s8GameNOutPos |- pdots)$) {$\vdots$};
        
        \draw [decorate,decoration={brace,amplitude=2pt,mirror,raise=0.5em}] ([xshift=-0.5em] tmpLabelValidityValue.north west) -- ([xshift=0.5em] tmpLabelValidityValue.south west) node (tmpLabelValidity) [midway,yshift=-0.7em,anchor=north,align=center] {some trial decided};

        \draw [ultra thick,black!20,-latex] (tmpLabelValidity) -- ++(0,-0.85) -| (s1Le.west) node [pos=0.25,below] {if no trial decides, repeat with next leader};
        
        \node [rectangle callout, callout absolute pointer={(s8GameNOutPos |- p1)}, draw=black!50, anchor=south west, inner sep=2pt, fill=black!10] at ([yshift=1em,xshift=0em] s8GameNOutPos |- p1) {$\mathsf{decide}(v^\star)$};%
        \node [rectangle callout, callout absolute pointer={(s8GameNOutPos |- p2)}, draw=black!50, anchor=south west, inner sep=2pt, fill=black!10] at ([yshift=1em,xshift=0em] s8GameNOutPos |- p2) {$\mathsf{decide}(v^\star)$};%
        \node [rectangle callout, callout absolute pointer={(s8GameNOutPos |- pn)}, draw=black!50, anchor=south west, inner sep=2pt, fill=black!10] at ([yshift=1em,xshift=0em] s8GameNOutPos |- pn) {$\mathsf{decide}(v^\star)$};%
        
    \end{tikzpicture}
    \caption{%
        Depiction of \commplus's structure.
        The depiction focuses on a good iteration $k$ where correct processes decide on the $\mathsf{leader}(k)$'s valid proposal $v^{\star}(k)$ whose digest is $z^{\star}(k)$.
        See \Cref{fig:hmvba-recap,fig:reducerpp-game} for ``Dissemination''
        and ``Trial'' sub-protocols, respectively.
        We abridge
        $\mathsf{leader} \triangleq \mathsf{leader}(k)$,
        $z^\star \triangleq z^\star(k)$,
        $v^\star \triangleq v^\star(k)$.
    }
    \label{fig:reducerpp-overview}
\end{figure}
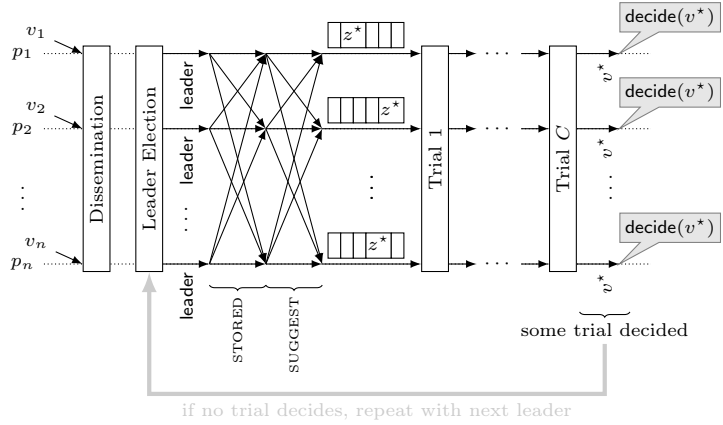

\smallskip
\noindent \textbf{Key concepts behind \commplus.}
\commplus (\Cref{fig:reducerpp-overview,fig:reducerpp-game}, \Cref{algorithm:reducer_plus}) starts in the same way as \comm.
First, processes engage in dissemination identical to that of \comm (and HMVBA).
The only difference is that, when encoding its proposal $v_i$ into $n$ RS symbols, each correct process $p_i$ treats $v_i$ as a polynomial of degree $\epsilon t$ (and not $t$, like in \comm and HMVBA).
Then, \commplus proceeds in iterations.
Each iteration $k$ of \commplus begins in the same manner as \comm's iterations:
(1) The leader of iteration $k$ is elected using a common coin.
(2) Each correct process determines its candidates upon receiving $n - t = (2 + \epsilon)t + 1$ \textsc{stored} messages.
A correct process marks a digest $z$ as its candidate if it receives $z$ in (at least) $n - 3t = \epsilon t + 1$ \textsc{stored} messages.
(3) Each correct process commits some of its candidates upon receiving $n - t = (2 + \epsilon)t + 1$ \textsc{suggest} messages.
Concretely, a correct process commits its candidate digest $z$ if it receives $z$ in at least $(n - t) - t = (1 + \epsilon)t + 1$ \textsc{suggest} messages.
This ensures that each correct process commits the ``good'' digest $z^{\star}(k)$ in a good iteration $k$.

\begin{figure}[tbp]
    \centering
    \begin{tikzpicture}[
        x=1cm,
        y=0.9cm,
        subprotocol/.style = {
            rotate=90,
            minimum width=3cm,
            draw,
            fill=white,
        },
        transparentsubprotocol/.style = {
            draw,
            dashed,
        },
        msg/.style = {
            -latex,
        },
        msglabel/.style = {
            below,
            rotate=90,
            anchor=east,
        },
        col/.style = {
            draw,
            inner sep=0,
            minimum width=0.6em,
            minimum height=1.2em,
            fill=white,
        },
        colfirst/.style = {
            col,
            anchor=south west,
            yshift=2pt,
            xshift=1pt,
        },
        colrest/.style = {
            col,
            anchor=west,
            xshift=-0.4pt,   %
        },
        colzstar/.style = {
            minimum width=1.2em,
        },
    ]
        \scriptsize

        \coordinate (p1) at (0.5,1.5);
        \coordinate (p2) at (0.5,0.5);
        \coordinate (pdots) at (0.5,-0.3);
        \coordinate (pn) at (0.5,-1.3);
        
        \coordinate (center) at (0,0.1);

        \coordinate (s1Input) at (0.75,0.0 |- center);
        \coordinate (s2NoisePos) at ($(s1Input) + (1.75,0)$);
        \coordinate (s3NoiseOut) at ($(s2NoisePos) + (0.75,0)$);
        \coordinate (s4RoOut) at ($(s3NoiseOut) + (2.5,0)$);
        \coordinate (s5RaPos) at ($(s4RoOut) + (1.8,0)$);
        \coordinate (s6RaOut) at ($(s5RaPos) + (1.2,0)$);

        \node at ([xshift=-1em] p1) {$p_1$};
        \node at ([xshift=-1em] p2) {$p_2$};
        \node at ([xshift=-1em] pdots) {$\vdots$};
        \node at ([xshift=-1em] pn) {$p_n$};
        
        \draw [densely dotted] (p1) -- ([xshift=3em] s6RaOut |- p1);
        \draw [densely dotted] (p2) -- ([xshift=3em] s6RaOut |- p2);
        \draw [densely dotted] (pn) -- ([xshift=3em] s6RaOut |- pn);

        \node [colfirst] (colOne) at (s1Input |- p1) {};
        \node [colrest,colzstar] (colTwo) at (colOne.east) {$z^\star_{\phantom{1}}$};
        \node [colrest] (colThree) at (colTwo.east) {};
        \node [colrest] (colFour) at (colThree.east) {};
        \node [colrest] (colFive) at (colFour.east) {};
        \node [colfirst] (colOne) at (s1Input |- p2) {};
        \node [colrest] (colTwo) at (colOne.east) {};
        \node [colrest] (colThree) at (colTwo.east) {};
        \node [colrest] (colFour) at (colThree.east) {};
        \node [colrest,colzstar] (colFive) at (colFour.east) {$z^\star_{\phantom{1}}$};
        \node [colfirst] (colOne) at (s1Input |- pn) {};
        \node [colrest] (colTwo) at (colOne.east) {};
        \node [colrest] (colThree) at (colTwo.east) {};
        \node [colrest,colzstar] (colFour) at (colThree.east) {$z^\star_{\phantom{1}}$};
        \node [colrest] (colFive) at (colFour.east) {};

        \node [subprotocol] (s2Noise) at (s2NoisePos) {Noise};

        \draw [msg] (s2Noise.south |- p1) -- (s3NoiseOut |- p1) node [midway,msglabel] {$\phi$};
        \draw [msg] (s2Noise.south |- p2) -- (s3NoiseOut |- p2) node [midway,msglabel] {$\phi$};
        \draw [msg] (s2Noise.south |- pn) -- (s3NoiseOut |- pn) node [midway,msglabel] {$\phi$};
        \node at ($(s2Noise.south |- pdots)!0.5!(s3NoiseOut |- pdots)$) {$\vdots$};

        \node [colfirst] (colOne11) at (s3NoiseOut |- p1) {};
        \node [colrest,colzstar] (colTwo11) at (colOne11.east) {$z^\star_{\phantom{1}}$};
        \node [colrest] (colThree11) at (colTwo11.east) {};
        \node [colrest] (colFour11) at (colThree11.east) {};
        \node [colrest] (colFive11) at (colFour11.east) {};
        \node [colfirst] (colOne12) at (s3NoiseOut |- p2) {};
        \node [colrest] (colTwo12) at (colOne12.east) {};
        \node [colrest] (colThree12) at (colTwo12.east) {};
        \node [colrest] (colFour12) at (colThree12.east) {};
        \node [colrest,colzstar] (colFive12) at (colFour12.east) {$z^\star_{\phantom{1}}$};
        \node [colfirst] (colOne13) at (s3NoiseOut |- pn) {};
        \node [colrest] (colTwo13) at (colOne13.east) {};
        \node [colrest] (colThree13) at (colTwo13.east) {};
        \node [colrest,colzstar] (colFour13) at (colThree13.east) {$z^\star_{\phantom{1}}$};
        \node [colrest] (colFive13) at (colFour13.east) {};

        \node [colfirst,colzstar] (colOne21) at (s4RoOut |- p1) {$z^\star_{\phantom{1}}$};
        \node [colrest] (colTwo21) at (colOne21.east) {};
        \node [colrest] (colThree21) at (colTwo21.east) {};
        \node [colrest] (colFour21) at (colThree21.east) {};
        \node [colrest] (colFive21) at (colFour21.east) {};
        \node [colfirst,colzstar] (colOne22) at (s4RoOut |- p2) {$z^\star_{\phantom{1}}$};
        \node [colrest] (colTwo22) at (colOne22.east) {};
        \node [colrest] (colThree22) at (colTwo22.east) {};
        \node [colrest] (colFour22) at (colThree22.east) {};
        \node [colrest] (colFive22) at (colFour22.east) {};
        \node [colfirst,colzstar] (colOne23) at (s4RoOut |- pn) {$z^\star_{\phantom{1}}$};
        \node [colrest] (colTwo23) at (colOne23.east) {};
        \node [colrest] (colThree23) at (colTwo23.east) {};
        \node [colrest] (colFour23) at (colThree23.east) {};
        \node [colrest] (colFive23) at (colFour23.east) {};

        \draw [-latex,dashed] (colFive11) -- (colOne21) node [midway,above] {$\mathsf{hash}(\cdot, \phi)$};
        \draw [-latex,dashed] (colFive12) -- (colOne22) node [midway,above] {$\mathsf{hash}(\cdot, \phi)$};
        \draw [-latex,dashed] (colFive13) -- (colOne23) node [midway,above] {$\mathsf{hash}(\cdot, \phi)$};

        \node [subprotocol] (s5Ra) at (s5RaPos) {Reconstruct \& Agree};
        \draw [msg] (s4RoOut |- p1) -- (s5Ra.north |- p1) node [pos=1,msglabel,anchor=south east] {$z^\star_{\phantom{1}}$};
        \draw [msg] (s4RoOut |- p2) -- (s5Ra.north |- p2) node [pos=1,msglabel,anchor=south east] {$z^\star_{\phantom{1}}$};
        \draw [msg] (s4RoOut |- pn) -- (s5Ra.north |- pn) node [pos=1,msglabel,anchor=south east] {$z^\star_{\phantom{1}}$};
        \node at ($(s4RoOut |- pdots)!0.8!(s5Ra.north |- pdots)$) {$\vdots$};

        \draw [transparentsubprotocol] ([xshift=-1.1em,yshift=-2.75em] s2Noise.north west |- pn) rectangle ([xshift=1em,yshift=2.3em] s5Ra.south east) node [pos=0,anchor=south west,inner sep=0,xshift=3pt,yshift=4pt,align=left] {Trial};

        \draw [msg] (s5Ra.south |- p1) -- (s6RaOut |- p1) node [pos=1,msglabel,anchor=south east] {$v^\star$};
        \draw [msg] (s5Ra.south |- p2) -- (s6RaOut |- p2) node [pos=1,msglabel,anchor=south east] {$v^\star$};
        \draw [msg] (s5Ra.south |- pn) -- (s6RaOut |- pn) node (tmpLabelValidityValue) [pos=1,msglabel,anchor=south east] {$v^\star$};
        \node at ($(s5Ra.south |- pdots)!0.8!(s6RaOut |- pdots)$) {$\vdots$};
        
        \draw [decorate,decoration={brace,amplitude=2pt,mirror,raise=0.5em}] ([xshift=-0.5em] tmpLabelValidityValue.north west) -- ([xshift=0.5em] tmpLabelValidityValue.south west) node (tmpLabelValidity) [midway,yshift=-0.7em,anchor=north,align=center] {valid};
        
        \node [rectangle callout, callout absolute pointer={(s6RaOut |- p1)}, draw=black!50, anchor=south west, inner sep=2pt, fill=black!10] at ([yshift=1em,xshift=0em] s6RaOut |- p1) {$\mathsf{decide}(v^\star)$};%
        \node [rectangle callout, callout absolute pointer={(s6RaOut |- p2)}, draw=black!50, anchor=south west, inner sep=2pt, fill=black!10] at ([yshift=1em,xshift=0em] s6RaOut |- p2) {$\mathsf{decide}(v^\star)$};%
        \node [rectangle callout, callout absolute pointer={(s6RaOut |- pn)}, draw=black!50, anchor=south west, inner sep=2pt, fill=black!10] at ([yshift=1em,xshift=0em] s6RaOut |- pn) {$\mathsf{decide}(v^\star)$};%

    \end{tikzpicture}
    \caption{%
        Depiction of \commplus's adoption procedure.
        The depiction focuses on a case where $\phi$ happens to be such that the ``good'' digest $z^{\star}(k)$
        is smallest according to $\mathsf{hash}(\cdot, \phi)$ and is thus adopted by all correct processes.
        See \Cref{fig:hmvba-recap} for ``Reconstruct \& Agree'' sub-protocol.
        We abridge
        $z^\star \triangleq z^\star(k)$,
        $v^\star \triangleq v^\star(k)$.
    }
    \label{fig:reducerpp-game}
\end{figure}
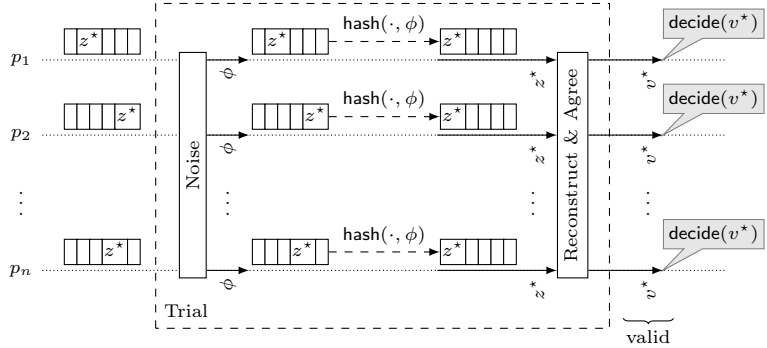

From this point forward, iterations of \commplus differ in design from iterations of \comm.
To justify our design choices, we now explain how \commplus guarantees termination with \emph{constant probability} in a good iteration.
Recall that, in contrast to \commplus, \comm deterministically terminates in a good iteration.
For the rest of the subsection, we focus on a fixed good iteration $k$.

\smallskip
\noindent\emph{Establishing only constantly many different candidates across correct processes.}
\commplus guarantees that $|\mathsf{committed}(k)| \leq C$, where $C = \lceil 1 + \frac{2}{\epsilon} \rceil^2$.
Let us explain.
First, each correct process has at most $\frac{(2 + \epsilon)t + 1}{\epsilon t + 1} \leq \lceil 1 + \frac{2}{\epsilon} \rceil$ candidates after receiving $n - t = (2 + \epsilon)t + 1$ \textsc{stored} messages.
Second, as each digest committed by a correct process is suggested by (at least) $(n - t) - t - t = \epsilon t + 1$ correct processes, $|\mathsf{committed}(k)| \leq C$.
(We prove this inequality in \Cref{section:reducer_plus_proof}.)

\smallskip
\noindent\emph{Unsuccessfully ensuring termination with constant probability.}
The remainder of iteration $k$ unfolds as follows.
Each correct process adopts (in some way) one of its committed digests.
If all correct processes adopt the ``good'' digest $z^{\star}(k)$, correct processes decide the $\mathsf{leader}(k)$'s valid proposal $v^{\star}(k)$ from \commplus by relying on the R\&A mechanism of \comm (and HMVBA) and terminate.
Therefore, for \commplus to terminate with constant probability in iteration $k$, it is crucial to ensure that all correct processes adopt $z^{\star}(k)$ with constant probability.
To accomplish this, each process follows the adoption procedure outlined below.

For the adoption procedure, \commplus employs a common coin, denoted by $\mathsf{Noise}$, that returns some random $\kappa$-bit value $\phi$.
Intuitively, each correct process $p_i$ (1) concatenates each committed digest $z$ with the obtained random value $\phi$, and (2) hashes the concatenation using the hash function $\mathsf{hash}(\cdot)$ modeled as a random oracle.
Concretely, once $\phi$ is obtained, each correct process $p_i$ constructs its local set $H_i$ in the following way:
\begin{equation*}
    H_i = \{ h \,|\, h = \mathsf{hash}(z, \phi) \land \text{$z$ is committed by $p_i$} \}.
\end{equation*}
Finally, process $p_i$ adopts the committed digest $z'$ that produced  the lexicographically smallest hash:
\begin{equation*}
    \forall h \in H_i: \mathsf{hash}(z', \phi) \leq h.
\end{equation*}

Let us first analyze the probability that all correct processes adopt $z^{\star}(k)$, assuming that $\mathsf{committed}(k)$ is fixed before $\phi$ is revealed.
Given that only polynomially many (in $\kappa$) random oracle queries can be made and the common coin outputs a $\kappa$-bit random value, the procedure described above emulates a random permutation of the committed digests, except with negligible probability.
Concretely, for every $z \in \mathsf{committed}(k)$, $\mathsf{hash}(z, \phi)$ is uniformly random with all but negligible probability.
Hence, the probability that all correct processes adopt $z^{\star}(k)$ is equal to the probability that, given the obtained random value $\phi$, $\mathsf{hash}(z^{\star}(k), \phi)$ is lexicographically smallest in the $\mathcal{H} = \{ h \,|\, h = \mathsf{hash}(z, \phi) \land z \in \mathsf{committed}(k) \}$ set.
As all members of $\mathcal{H}$ are uniformly random (except with negligible probability),
this probability is at least $\frac{1}{|\mathsf{committed}(k)|} - \mathsf{negl}(\kappa) \geq \frac{1}{C} - \mathsf{negl}(\kappa)$ given that $|\mathsf{committed}(k)| \leq C$.\footnote{One could obtain a random permutation of digests via a common-coin object, thus eliminating the need for the random oracle assumption. However, this would necessitate the coin to disseminate $2^{\kappa} \cdot \kappa \gg \kappa$ bits.}

\smallskip
\noindent\emph{The problem.}
Unfortunately, the approach above has a clear problem.
Namely, an adaptive adversary can rig the described probabilistic trial, thus ensuring that not \emph{all} correct processes adopt $z^{\star}(k)$.
To illustrate why this is the case, we now showcase a simple adversarial attack.

Note that, once the random value $\phi$ gets revealed, the adversary learns the value $\mathsf{hash}(z^{\star}(k), \phi)$ it needs to ``beat''.
At this point, the adversary can find an adversarial digest $z_A$ such that $\mathsf{hash}(z_A, \phi) < \mathsf{hash}(z^{\star}(k), \phi)$.
Then, by corrupting $\mathsf{leader}(k)$ and making it disseminate $z_A$, the adversary can introduce digest $z_A$ to correct processes.
Specifically, by delaying some correct processes and carefully controlling the scheduling of the \textsc{stored} and \textsc{suggest} messages, the adversary can force some slow correct processes to commit $z_A$.
As $\mathsf{hash}(z_A, \phi) < \mathsf{hash}(z^{\star}(k), \phi)$, these correct processes would adopt $z_A$, thus preventing termination in good iteration $k$.
In brief, the adversary is capable of rigging the trial as the set of all digests committed by correct processes is not fixed once the randomness is revealed: upon observing $\phi$, the adversary gains the ability to manipulate the trial to its advantage.

\smallskip
\noindent\emph{The solution.}
Luckily, we can prevent the adversary from manipulating trials ``too many'' times.
The key insight is this: for the adversary to rig a trial, it needs to force correct processes to commit an adversarial digest.
Hence, whenever the adversary ``cheats'', the number of different digests committed across correct processes increases.
However, recall that $|\mathsf{committed}(k)| \leq C$.
Therefore, given that $z^{\star}(k)$ is committed by each correct process, the adversary can rig at most $C - 1$ trials of iteration $k$.
Immediately before each trial's randomness is revealed, we freeze the set of digests already committed by so-far-uncorrupted processes.
This frozen pool contains $z^{\star}(k)$ and has at most $C$ elements (see \Cref{section:reducer_plus_proof}).
A trial is \emph{won} if $z^{\star}(k)$ has the strictly smallest hash over this pool; conditioned on the overwhelmingly likely random-oracle event established in \Cref{section:reducer_plus_proof}, every trial is won with conditional probability at least $1/(2C)$ for sufficiently large $\kappa$, irrespective of the preceding trials.
A won trial ensures that correct processes quasi-decide $v^{\star}(k)$ unless a correct process commits a new digest during that trial.
Hence, with $4C^2$ sequential trials, correct processes terminate unless fewer than $C$ trials are won.
An exponential-moment bound shows that this happens with probability at most $1/2$, without assuming independence between trials.
Thus, correct processes decide and terminate in good iteration $k$ with probability at least $1/2$.

\smallskip
\noindent \emph{Complexity analysis.}
\commplus terminates in a good iteration with probability at least $1/2$.
Given that every reached iteration is good with conditional probability $\geq \frac{n - 2t}{n} = \frac{(1 + \epsilon)t + 1}{(3 + \epsilon)t + 1} \geq \frac{1}{3}$, \commplus terminates in each reached iteration with conditional probability at least $1/6$, and hence in at most six iterations in expectation.
As each iteration takes $O(C^2)$ time in expectation (since there are $4C^2$ trials), the expected time complexity is $O(C^2)$.
As for the exchanged information, correct processes send $O(n^2)$ messages and $O(\frac{n \ell}{\epsilon} + n^2 \kappa \log n)$ bits in the dissemination phase.
Additionally, each iteration exchanges $O(C^2n^2)$ messages and $O\big( C^2 (\frac{n \ell}{\epsilon} + n^2 \kappa \log n) \big)$ bits in expectation.
As \commplus terminates in expected $O(1)$ iterations, \commplus yields an expected message complexity of $O(C^2n^2)$ and an expected bit complexity of $O\big( C^2 (\frac{n \ell}{\epsilon} + n^2 \kappa \log n) \big)$.
Recall that in \Cref{subsection:reducing_constants}, we discuss whether the constant multiplicative factor can be reduced to $O(1 / \epsilon^3)$.

\section{System Model \& Problem Definition} \label{section:model_problem_definition}

\noindent\textbf{System model.}
We consider 
$n$ processes
$p_1, p_2, ..., p_n$
connected through pairwise authenticated channels.
(A post-quantum symmetric encryption scheme such as AES can fulfill this assumption.)
This work considers a computationally bounded and \emph{adaptive} adversary capable of corrupting up to $t > 0$ processes throughout (and not only at the beginning of) the protocol execution.
For \comm, we consider $t < \frac{1}{4}n$, while for \commplus, we assume $t < (\frac{1}{3} - \epsilon)n$, where $\epsilon > 0$ is any fixed constant.
Processes not corrupted by the adversary at a certain stage are said to be \emph{so-far-uncorrupted}.
Once the adversary corrupts a process, the process falls under the adversary's control and may behave maliciously.
A process is 
so called
\emph{correct} if it is never corrupted; a non-correct process is 
\emph{faulty}.

We focus 
on an \emph{asynchronous} communication network where message delays are unbounded (but finite).
Thus, we assume 
the adversary controls the network
and
can delay messages arbitrarily, but every message exchanged between correct processes must be delivered eventually.
The adversary possesses 
the 
\emph{after-the-fact-removal} capabilities: if a so-far-uncorrupted process $p_i$ sent a message and then got corrupted 
before the message was delivered, the adversary is capable of retracting the message, thus preventing its delivery.

\smallskip
\noindent\textbf{Multi-valued validated Byzantine agreement (MVBA).}
We aim to design an MVBA~\cite{vaba,dumbomvba,hmvba,CT05,CL02,Kotla2009} protocol that operates in the model described above.
Informally, MVBA requires correct processes to agree on a \emph{valid} $\ell$-bit value.
The formal specification is given in \Cref{mod:mvba}.

\begin{module}[tbp]
\caption{MVBA}
\label{mod:mvba}
\scriptsize
\begin{algorithmic}[1]

\Statex \textbf{Associated values:}
\begin{compactitem} [-]
    \item set $\valuemvba$ of $\ell$-bit values
\end{compactitem}

\smallskip
\Statex \textbf{Associated functions:}
\begin{compactitem} [-]
    \item function $\mathsf{valid}: \valuemvba \to \{ \mathit{true}, \mathit{false} \}$; a value $v \in \valuemvba$ is said to be valid if and only if $\mathsf{valid}(v) = \mathit{true}$
\end{compactitem}

\smallskip
\Statex \textbf{Events:}
\begin{compactitem}[-]
    \item \emph{input} $\mathsf{propose}(v \in \valuemvba)$: a process proposes value $v$.

    \item \emph{output} $\mathsf{decide}(v' \in \valuemvba)$: a process decides value $v'$.
\end{compactitem}

\smallskip 
\Statex \textbf{Assumed behavior:} 
\begin{compactitem}[-]
    \item Every correct process proposes exactly once and it does so with a valid value.

    \item Every process that proposes while so-far-uncorrupted proposes a fixed valid value and follows the protocol until it is corrupted.
\end{compactitem}

\smallskip 
\Statex \textbf{Properties:} \BlueComment{ensured only if correct processes follow the behavior stated above}
\begin{compactitem}[-]
    \item \emph{External validity:} No correct process decides an invalid value.

    \item \emph{Weak validity:} If all processes are correct and a correct process decides a value $v \in \valuemvba$, then $v$ was proposed by a correct process.
    
    \item \emph{Agreement:} No two correct processes decide different values. 

    \item \emph{Integrity:} No correct process decides more than once.

    \item \emph{Termination:} All correct processes eventually decide.

   \item \emph{Quality:} If a correct process decides a value $v \in \valuemvba$, then the probability that $v$ is a value determined by the adversary is at most $q < 1$.
\end{compactitem}
\end{algorithmic}
\end{module}

\section{Preliminaries}\label{section:preliminaries}

This section recapitulates the building blocks employed in our algorithms.

\smallskip
\noindent \textbf{Broadcasting operation.}
When a process \emph{broadcasts} a message in our algorithms, it
simply unicasts the message to all processes individually.
This broadcasting operation is thus \emph{unreliable}: if the sender is faulty, correct processes may receive different messages, or some may not receive any message at all.

\smallskip
\noindent\textbf{Reed--Solomon codes.}
Our algorithms rely on Reed--Solomon (RS) codes~\cite{reed1960}.
\comm and \commplus use RS codes as erasure codes; no (substitution-)error correction is required.
We use $\mathsf{encode}(\cdot)$ and $\mathsf{decode}(\cdot)$ to denote RS' encoding and decoding algorithms.
Namely, in a nutshell, $\mathsf{encode}(v)$ takes a value $v$, chunks it into the coefficients of a polynomial of degree $t$ (for \comm) or degree $\epsilon t$ (for \commplus), and outputs evaluations of the polynomial (i.e., the RS symbols) at $n$ (the total number of processes) distinct locations.
Similarly, $\mathsf{decode}(S)$ takes a set $S$ of $t + 1$ (for \comm) or $\epsilon t + 1$ (for \commplus) RS symbols (from distinct locations) and interpolates them into a polynomial of degree $t$ (for \comm) or degree $\epsilon t$ (for \commplus), whose coefficients are concatenated and output
as decoded value $v$.
The bit-size of an RS symbol obtained by the $\mathsf{encode}(v)$ algorithm is $O(\frac{|v|}{n} + \log n)$ for \comm and $O(\frac{|v|}{\epsilon n} + \log n)$ for \commplus, where $|v|$ denotes the size of value $v$.

\smallskip
\noindent\textbf{Hash functions.}
For \comm, we assume a collision-resistant hash function $\mathsf{hash}(\cdot)$ guaranteeing that a computationally bounded adversary cannot find two different inputs resulting in the same hash value (except with negligible probability).
In contrast, \commplus requires hash functions modeled as a random oracle
with independent and uniformly distributed hash values.
Each hash value is of size $\kappa$ bits; we assume $\kappa > \log n$.\footnote{Else, $t \in O(n)$ faulty processes would have computational power exponential in $\kappa$.}

\smallskip
\noindent\textbf{Cryptographic accumulators.}
Our algorithms use cryptographic accumulators.
A full definition is provided in \Cref{section:preliminaries_full}.
We only give
a brief summary here.

In a nutshell, following~\cite{bhat2021randpiper,Nayak0SVX20}, a cryptographic accumulator scheme constructs an accumulation value for a set of values and produces a witness for each value in the set.
Given the accumulation value and a witness, any process can verify if a value is indeed in the set.
More formally, given a set $\mathcal{D}$ of $n$ values $d_1, ..., d_n$, an accumulator has the following syntax:
\begin{compactitem}
    \item $\mathsf{Eval}(\mathcal{D})$ computes the accumulation value $z$ for the set $\mathcal{D}$.

    \item $\mathsf{CreateWit}(z, d_i, \mathcal{D})$, where $z$ is the accumulation value for $\mathcal{D}$, computes a witness $w_i$ if $d_i \in \mathcal{D}$, and returns $\bot$ otherwise.

    \item $\mathsf{Verify}(z, w_i, d_i)$, where $z$ is the accumulation value for $\mathcal{D}$, returns $\mathit{true}$ if $w_i$ is the witness for $d_i \in \mathcal{D}$, and $\mathit{false}$ otherwise.
\end{compactitem}
In our algorithms, we use Merkle trees~\cite{merkle-tree-crypto87} as our
accumulators given they are
hash-based.
Elements of $\mathcal{D}$ form the leaves of a Merkle tree, the accumulator key is a specific hash function (see \Cref{section:preliminaries_full}), an accumulation value is the Merkle tree root, and a witness is a Merkle tree proof.
Importantly, the size of an accumulation value is $O(\kappa)$ bits, and the size of a witness is $O(\kappa \log n)$ bits, where $\kappa$ denotes the size of a hash value.
Our algorithms instruct each process $p_i$ to construct a Merkle tree over the RS symbols of its proposal $v_i$, with the resulting Merkle root serving as a digest of $v_i$.
Therefore, throughout the remainder of this paper, we use the terms ``accumulation value'' and ``digest'' interchangeably.
Since our algorithms require tracking the position of each RS symbol, the Merkle trees used in our algorithms effectively function as vector commitments~\cite{catalano2013vector}.

\smallskip
\noindent\textbf{Common coin.}
We follow the approach of prior works~\cite{BKR94,DBLP:journals/corr/abs-2002-08765,MostefaouiMR15,pace,finmvba,hmvba} and assume the existence of an idealized common coin, an object introduced by Rabin~\cite{Rabin83}, that delivers the same sequence of random coins to all processes.
To ensure that the adversary cannot anticipate the coin values in advance, we establish the condition that the value of a coin instance is disclosed only after $t + 1$ distinct processes (thus, at least one correct process) have queried that instance.
Concretely, both of our algorithms use the common-coin objects for (1) choosing a process uniformly at random from the $n$ processes, denoted by $\mathsf{Election}()$, and (2) obtaining a uniformly random integer in a specified constant range, denoted by $\mathsf{Index}()$.
\commplus utilizes an additional common-coin object, denoted by $\mathsf{Noise}()$, that generates a uniformly random $\kappa$-bit value.
All coin instances are independent of one another and of the random oracle.
Our protocols can be easily adapted to rely only on a weak common coin, which allows a constant probability of disagreement among correct processes.
This can be achieved by following the $\mathsf{Index}()$ common coin with an MBA instance (see \Cref{section:reducer,section:reducer_plus}) to ensure agreement on the coin's random output.\footnote{We underline that the SMBA and MBA algorithms utilized in our protocols require only a weak common coin.}

\smallskip
\noindent\textbf{Multi-valued Byzantine agreement (MBA).}
Our algorithms internally utilize the well-known MBA primitive~\cite{book-cachin-guerraoui-rodrigues,BE03,DBLP:journals/acta/MostefaouiR17}.
MBA is similar to the MVBA primitive: processes propose their values and decide on a common value.
The formal specification of the MBA primitive is given in \Cref{mod:async_mba}.
In contrast to the MVBA primitive, MBA ensures justification (sometimes also called ``non-intrusion''), but it does not guarantee external validity.
Note that, whenever correct processes propose different values, MBA might decide futile $\botmba$.
An MVBA algorithm must decide a valid (non-$\botmba$) value in this case, which represents a crucial difference between these two primitives.
We treat the special value $\botmba$ as an invalid value.
Specifically, the function $\mathsf{valid}(\cdot)$ (see \Cref{mod:mvba}) is defined for $\botmba$, with $\mathsf{valid}(\botmba) = \mathit{false}$.
We utilize this in our algorithms.

\begin{module}[tb]
\caption{MBA}
\label{mod:async_mba}
\scriptsize
\begin{algorithmic}[1]

\Statex \textbf{Associated values:}
\begin{compactitem} [-]
    \item set $\valuemba$

    \item special value $\botmba \notin \valuemba$; we assume $\mathsf{valid}(\botmba) = \mathit{false}$ \BlueComment{$\botmba$ is an invalid value}
\end{compactitem}

\smallskip
\Statex \textbf{Events:}
\begin{compactitem}[-]
    \item \emph{input} $\mathsf{propose}(v \in \valuemba)$: a process proposes value $v$.

    \item \emph{output} $\mathsf{decide}(v' \in \valuemba \cup \{ \botmba \})$: a process decides value $v'$.
\end{compactitem}

\smallskip
\Statex \textbf{Assumed behavior:}
\begin{compactitem}[-]
    \item Every correct process proposes exactly once.
\end{compactitem}

\smallskip
\Statex \textbf{Properties:} \BlueComment{ensured only if correct processes follow the behavior stated above}
\begin{compactitem}[-]
    \item \emph{Strong unanimity:} If all correct processes propose the same value $v \in \valuemba$ and a correct process decides a value $v' \in \valuemba \cup \{ \botmba \}$, then $v' = v$.

    \item \emph{Agreement:} No two correct processes decide different values.

    \item \emph{Justification:} If any correct process decides a value $v' \in \valuemba$ ($v' \neq \botmba$), then $v'$ is proposed by a correct process.

    \item \emph{Integrity:} No correct process decides more than once.

    \item \emph{Termination:} All correct processes eventually decide.
\end{compactitem}
\end{algorithmic}
\end{module}

Our algorithms utilize MBA to agree on (1) $\ell$-bit values (i.e., $\valuemba \equiv \valuemvba$), and (2) on $O(\kappa)$-bit digests (i.e., $\valuemba \equiv \valuedigest$, where $\valuedigest$ denotes the set of all digests).
For agreeing on $\ell$-bit values, our algorithms rely on our MBA implementation $\mbasimple$ (relegated to \Cref{section:regular_mba}) with $O(n^2)$ expected message complexity, $O(n\ell + n^2 \kappa \log n)$ expected bit complexity and $O(1)$ expected time complexity.
For agreeing on digests (utilized in our implementation of the SMBA primitive; see \Cref{section:strong_mba}), we rely on the cryptography-free MBA implementation proposed in \cite{DBLP:journals/acta/MostefaouiR17,Abraham2022crusader} with $O(n^2)$ expected message complexity, $O(n^2 \kappa)$ expected bit complexity and $O(1)$ expected time complexity.
Both utilized MBA algorithms tolerate up to $t < \frac{1}{3}n$ failures.

\smallskip
\noindent\textbf{Strong multi-valued Byzantine agreement (SMBA).}
Finally, the \comm algorithm relies on
\emph{strong multi-valued Byzantine agreement} (SMBA),
a new variant of Byzantine agreement, which we isolate as a separate primitive and construct, and which may be of independent interest.
Concretely, \comm utilizes SMBA to enable correct processes to agree on a digest.
The formal specification of the SMBA primitive can be found in \Cref{mod:async_smba}.
We stress that the specification of the SMBA primitive assumes that correct processes propose only a constant number of different digests; this assumption is introduced solely for complexity reasons (see \Cref{section:strong_mba} for more details).
Crucially, our \comm algorithm---specifically its reducing technique---enforces this condition: only $O(1)$ different digests are proposed by correct processes in any instance of the SMBA primitive utilized in \comm.

Intuitively, the SMBA primitive ensures that all correct processes eventually agree on the same digest $z$.
Moreover, if no more than two different digests are proposed by correct processes, the primitive ensures that $z$ was proposed by a correct process.
If correct processes propose three (or more) different digests, a non-proposed digest can be decided.
In \comm, we utilize our cryptography-free SMBA algorithm \smba (relegated to \Cref{section:strong_mba}), which has an expected message complexity of $(n^2)$, an expected bit complexity of $O(n^2 \kappa)$ and an expected time complexity of $O(1)$; \smba tolerates up to $t < \frac{1}{4} n$ adaptive corruptions.
In \Cref{section:preliminaries_full}, we discuss the differences between the SMBA primitive and the well-known strong consensus primitive~\cite{fitzi2003efficient}, which also enables processes to agree on the proposal of a correct process.

\begin{module}[tb]
\caption{SMBA}
\label{mod:async_smba}
\scriptsize
\begin{algorithmic}[1]

\Statex \textbf{Associated values:}
\begin{compactitem} [-]
    \item set $\valuedigest$ of $O(\kappa)$-bit digests \BlueComment{SMBA is exclusively utilized for agreement on digests}
\end{compactitem}

\smallskip
\Statex \textbf{Events:}
\begin{compactitem}[-]
    \item \emph{input} $\mathsf{propose}(z \in \valuedigest)$: a process proposes digest $z$.

    \item \emph{output} $\mathsf{decide}(z' \in \valuedigest)$: a process decides digest $z'$.
\end{compactitem}

\smallskip
\Statex \textbf{Assumed behavior:}
\begin{compactitem}[-]
    \item Every correct process proposes exactly once.

    \item Only $O(1)$ different digests are proposed by correct processes.
\end{compactitem}

\smallskip
\Statex \textbf{Properties:} \BlueComment{ensured only if correct processes follow the behavior stated above}
\begin{compactitem}[-]
    \item \emph{Strong validity:} If up to two different digests are proposed by correct processes and a correct process decides a digest $z' \in \valuedigest$, then $z'$ is proposed by a correct process.

    \item \emph{Agreement:} No two correct processes decide different values.

    \item \emph{Integrity:} No correct process decides more than once.

    \item \emph{Termination:} All correct processes eventually decide.
\end{compactitem}
\end{algorithmic}
\end{module}

\section{\comm: Pseudocode \& Proof Sketch} \label{section:reducer}

This section presents our \comm algorithm.
Recall that \comm exchanges $O(n^2)$ messages and $O(n \ell + n^2 \kappa \log n)$ bits, and terminates in $O(1)$ time.
\comm tolerates up to $t < \frac{1}{4}n$ failures.
We start by introducing the pseudocode of the \comm algorithm (\Cref{subsection:reducer_implementation}).
Then, we provide an informal analysis of \comm's correctness and complexity (\Cref{subsection:reducer_analysis}).
A formal proof can be found in \Cref{section:reducer_proof}.

\begin{algorithm}[tb]
\caption{\comm: Pseudocode (for process $p_i$) [part 1 of 2]}
\label{algorithm:reducer}
\begin{algorithmic}[1]
\scriptsize

\State \textbf{Uses:} \label{line:reducer_uses}
\State \hskip2em \textcolor{jnSUDigitalRedLight}{\(\triangleright\) \smba exchanges $O(n^2)$ messages and $O(n^2 \kappa)$ bits and terminates in $O(1)$ time}
\State \hskip2em SMBA algorithm \smba, \textbf{instances} $\mathcal{SMBA}[k][x]$, $\forall k \in \mathbb{N}, \forall x \in \{ 1, 2, 3 \}$

\smallskip
\State \hskip2em \textcolor{jnSUDigitalRedLight}{\(\triangleright\) $\mbasimple$ exchanges $O(n^2)$ messages and $O(n \ell + n^2 \kappa \log n)$ bits and terminates in $O(1)$ time}
\State \hskip2em MBA algorithm $\mbasimple$ with $\valuemba = \valuemvba$, \textbf{instances} $\mathcal{MBA}[k][x]$, $\forall k \in \mathbb{N}, \forall x \in \hphantom{|||||L}\{ 1, 2, 3 \}$

\medskip
\State \textbf{Rules:}
    \State \hskip2em - Any message with an invalid witness is ignored.
    \State \hskip2em - Only one \textsc{init} message is processed per process.

\medskip
\State \textbf{Constants:}
\State \hskip2em $\mathsf{Digest}$ $\mathit{default}$ \BlueComment{default digest} \label{line:reducer_fixed_default}

\medskip
\State \textbf{Local variables:}
\State \hskip2em $\valuemvba$ $v_i \gets p_i$'s proposal
\State \hskip2em $\mathsf{Boolean}$ $\mathit{dissemination\_completed}_i \gets \mathit{false}$
\State \hskip2em $\mathsf{Map}(\mathsf{Process} \to [\mathsf{RS}, \mathsf{Digest}, \mathsf{Witness}])$ $\mathit{symbols}_i \gets \text{empty map}$
\State \hskip2em $\mathsf{List}(\mathsf{Digest})$ $\mathit{candidates}_i \gets \text{empty list}$ \BlueComment{will be reset every iteration}
\State \hskip2em $\mathsf{Digest}$ $\mathit{adopted\_digest}_i \gets \bot$
\State \hskip2em $\mathsf{Digest}$ $\mathit{first\_digest}_i \gets \bot$
\State \hskip2em $\mathsf{List}(\valuemvba)$ $\mathit{quasi\_decisions}_i \gets $ empty list \label{line:reducer_last_line}

\medskip
\State \textbf{upon} $\mathsf{propose}(\valuemvba \text{ } v_i)$ \label{line:reducer_propose} \BlueComment{start of the algorithm}
\State \hskip2em \textcolor{jnSUDigitalRedLight}{\(\triangleright\) dissemination phase starts}
\State \hskip2em $\mathsf{List}(\mathsf{RS})$ $[m_1, m_2, ..., m_n] \gets \mathsf{encode}(v_i)$ \label{line:compute_rs_symbols} \BlueComment{$\mathsf{encode}(v_i)$ treats $v_i$ as a polynomial of degree $t$}
\State \hskip2em $\mathsf{Digest}$ $z_i \gets \mathsf{Eval}\Big( \big[ (1, m_1), (2, m_2), ..., (n, m_n) \big] \Big)$ \label{line:compute_accumulation_value} \BlueComment{compute the digest}
\State \hskip2em \textbf{for each} $\mathsf{Process} \text{ } p_j$:
\State \hskip4em $\mathsf{Witness}$ $w_j \gets \mathsf{CreateWit}\Big( z_i, (j, m_j), \big[ (1, m_1), (2, m_2), ..., (n, m_n) \big] \Big)$ \label{line:compute_witness} \BlueComment{obtain witness}
\State \hskip4em \textbf{send} $\langle \textsc{init}, m_j, z_i, w_j \rangle$ to process $p_j$ \label{line:send_init}%

\medskip
\State \textbf{upon} receiving $\langle \textsc{init}, \mathsf{RS} \text{ } m_i, \mathsf{Digest} \text{ } z_j, \mathsf{Witness} \text{ } w_i  \rangle$ from a process $p_j$: \label{line:receive_init}
\State \hskip2em \textbf{if} $\mathit{dissemination\_completed}_i = \mathit{false}$:
\State \hskip4em $\mathit{symbols}_i[p_j] \gets [m_i, z_j, w_i]$ \label{line:store_received_symbol} \BlueComment{store the received RS symbol}%
\State \hskip4em \textbf{send} $\langle \textsc{ack} \rangle$ to process $p_j$ \label{line:send_ack}

\medskip
\State \textbf{upon} receiving $\langle \textsc{ack} \rangle$ from $n - t$ processes (for the first time): \label{line:receive_ack}
\State \hskip2em \textbf{broadcast} $\langle \textsc{done} \rangle$ \label{line:broadcast_done}

\medskip
\State \textbf{upon} receiving $\langle \textsc{done} \rangle$ from $n - t$ processes (for the first time): \label{line:receive_quorum_done}
\State \hskip2em \textbf{broadcast} $\langle \textsc{finish} \rangle$ if $p_i$ has not broadcast $\langle \textsc{finish} \rangle$ before \label{line:first_finish}

\medskip
\State \textbf{upon} receiving $\langle \textsc{finish} \rangle$ from $t + 1$ processes (for the first time): \label{line:receive_plurality_finish}
\State \hskip2em \textbf{broadcast} $\langle \textsc{finish} \rangle$ if $p_i$ has not broadcast $\langle \textsc{finish} \rangle$ before \label{line:second_finish}

\algstore{reducer}
\end{algorithmic}
\end{algorithm}

\setcounter{algorithm}{0}

\begin{algorithm}[tb]
\caption{\comm: Pseudocode (for process $p_i$) [part 2 of 2]}
\begin{algorithmic}[1]
\scriptsize
\algrestore{reducer}

\State \textbf{upon} receiving $\langle \textsc{finish} \rangle$ from $n - t$ processes (for the first time)\label{line:receive_quorum_finish}
\State \hskip2em $\mathit{dissemination\_completed}_i \gets \mathit{true}$\label{line:dissemination_complete} \BlueComment{dissemination phase completes}

\smallskip
\State \hskip2em \textbf{for each} $k = 1, 2, ...$: \label{line:start_iteration} \BlueComment{iteration $k$ starts}
\State \hskip4em $\mathit{candidates}_i \gets \text{empty list}$ \BlueComment{reset the list of candidates} \label{line:reset_candidates}

\State \hskip4em $\mathsf{Process}$ $\mathsf{leader}(k) \gets \mathsf{Election()}$\BlueComment{elect a random leader}\label{line:random_election}
\State \hskip4em \textbf{broadcast} $\langle \textsc{stored}, k, \mathit{symbols}_i[\mathsf{leader}(k)].\mathsf{digest()} \rangle$\label{line:broadcast_stored} \BlueComment{disseminate the leader's digest}
\State \hskip4em \textbf{wait for} $n - t = 3t + 1$ \textsc{stored} messages for iteration $k$\label{line:wait_for_stored_messages}

\State \hskip4em \textbf{for each} $\mathsf{Digest}$ $z$ included in $n - 3t = t + 1$ received \textsc{stored} messages:\label{line:check_stored}
\State \hskip6em $\mathit{candidates}_i.\mathsf{append}(z)$\label{line:reducer_candidates_append}

\State \hskip4em \textbf{broadcast} $\langle \textsc{suggest}, k, \mathit{candidates}_i \rangle$\label{line:broadcast_suggest} \BlueComment{disseminate $p_i$'s candidates}
\State \hskip4em \textbf{wait for} $n - t = 3t + 1$ \textsc{suggest} messages for iteration $k$:\label{line:wait_for_suggest_messages}

\State \hskip4em \textbf{for each} $\mathsf{Digest} \text{ } z \in \mathit{candidates}_i$:
\State \hskip6em \textbf{if} $z$ is not included in $n - 2t = 2t + 1$ received \textsc{suggest} messages:\label{line:rule_precommit}
\State \hskip8em $\mathit{candidates}_i.\mathsf{remove}(z)$ \label{line:candidates_remove}

\State \hskip4em \textbf{if} $\mathit{candidates}_i.\mathsf{size} = 0$:\label{line:reducer_if_size_0} \BlueComment{if no candidate ``survives'' the suggestion step}
\State \hskip6em $\mathit{candidates}_i[1] \gets \mathit{default}$; $\mathit{candidates}_i[2] \gets \mathit{default}$\label{line:reducer_set_default} \BlueComment{commit the default digest}
\State \hskip4em \textbf{else if} $\mathit{candidates}_i.\mathsf{size} = 1$: \BlueComment{if exactly one candidate ``survives'' the suggestion step}
\State \hskip6em $\mathit{candidates}_i[2] \gets \mathit{candidates}_i[1]$\label{line:reducer_set_copy} \BlueComment{then, copy the candidate}
\State \hskip4em Sort $\mathit{candidates}_i$ in the lexicographic order\label{line:reducer_sort} \BlueComment{these digests are committed}

\smallskip
\State \hskip4em \textbf{for each} $x = 1, 2, 3$:\label{line:agreeing_phase}
\State \hskip6em \textbf{if} $x \neq 3$: \label{line:sub-iteration_starts}
\State \hskip8em $\mathit{adopted\_digest}_i \gets \mathit{candidates}_i[x]$\label{line:acc_proposal_1} \BlueComment{adopt the $x$-th committed digest}
\State \hskip6em \textcolor{jnSUDigitalRedLight}{\(\triangleright\) is the first committed digest decided from $\mathcal{SMBA}[k][1]$?}
\State \hskip6em \textbf{else if} \label{} $\mathit{first\_digest}_i = \mathit{candidates}_i[1]$: \label{line:proposal_switching_start}
\State \hskip8em $\mathit{adopted\_digest}_i \gets \mathit{candidates}_i[2]$  \label{line:acc_proposal_2}\BlueComment{if so, adopt the other committed digest}
\State \hskip6em \textbf{else:}
\State \hskip8em $\mathit{adopted\_digest}_i \gets \mathit{candidates}_i[1]$ \label{line:acc_proposal_3} \BlueComment{if not, adopt the first digest} \label{line:proposal_switching_end}

\smallskip
\State \hskip6em $\mathsf{Digest}$ $z \gets \mathcal{SMBA}[k][x].\mathsf{propose}(\mathit{adopted\_digest}_i)$\label{line:smba}
\State \hskip6em \textcolor{jnSUDigitalRedLight}{\(\triangleright\) store the first decided digest as it is important for third sub-iteration}
\State \hskip6em \textbf{if} $x = 1$: $\mathit{first\_digest}_i \gets z$

\smallskip
\State \hskip6em \textcolor{jnSUDigitalRedLight}{\(\triangleright\) Reconstruct \& Agree}
\State \hskip6em \textbf{broadcast} $\langle \textsc{reconstruct}, k, x, \mathit{symbols}_i[\mathsf{leader}(k)] \rangle$ \label{line:reducer_reconstruct}

\State \hskip6em \textbf{wait for} $n - t = 3t + 1$ \textsc{reconstruct} messages for sub-iteration $(k, x)$\label{line:wait_for_reconstruct_messages}

\State \hskip6em $\mathsf{Set}(\mathsf{RS})$ $S_i \gets$ received RS symbols with valid witnesses for digest $z$

\State \hskip6em \textbf{if} $|S_i| \geq t + 1$: \BlueComment{check if a value can be decoded using $S_i$}
\State \hskip8em $\valuemvba$ $r_i \gets \mathsf{decode}(S_i)$ \BlueComment{if yes, set $r_i$ to the decoded value} \label{line:reducer_decode}
\State \hskip6em \textbf{else:}
\State \hskip8em $\valuemvba$ $r_i \gets v_i$ \BlueComment{if not, set $r_i$ to $p_i$'s proposal} \label{line:reducer_own_proposal}

\State \hskip6em $\valuemvba \cup \{ \botmba \}$ $v \gets \mathcal{MBA}[k][x].\mathsf{propose}(r_i)$\label{line:lmba} \BlueComment{propose $r_i$}

\State \hskip6em \textbf{if} $\mathsf{valid}(v) = \mathit{true}$:\label{line:final_check}
\State \hskip8em $\mathit{quasi\_decisions}_i.\mathsf{append}(v)$ \BlueComment{quasi-decide $v$} \label{line:reducer_quasi_decide}

\smallskip
\State \hskip4em \textbf{if} $\mathit{quasi\_decisions}_i.\mathsf{size} > 0$ and $p_i$ has not previously decided: \label{line:check_quasi_decisions}
\State \hskip6em $\mathsf{Integer}$ $I \gets \mathsf{Index}()$ \label{line:random_index} \BlueComment{obtain a random integer $I$ in the $[1, 3]$ range}
\State \hskip6em \textcolor{jnSUDigitalRedLight}{\(\triangleright\) for quality}
\State \hskip6em \textbf{trigger} $\mathsf{decide}\big( \mathit{quasi\_decisions}_i[(I \text{ mod } \mathit{quasi\_decisions}_i.\mathsf{size}) + 1] \big)$ \label{line:reducer_decide}

\end{algorithmic}
\end{algorithm}

\subsection{Pseudocode} \label{subsection:reducer_implementation}

The pseudocode of \comm is given in \Cref{algorithm:reducer}.
Here, $\mathsf{decode}(S)$ interpolates any $t + 1$ symbols from distinct verified positions in $S$.
An absent entry of $\mathit{symbols}_i$ has the fixed $\kappa$-bit digest $\mathit{default}$ and no RS symbol or witness.
Its empty \textsc{reconstruct} report counts toward the waiting threshold but contributes no symbol to $S_i$.
Given that the correctness of our solution crucially depends on the exact number of ``reduced'' digests held by correct processes in an iteration, the presented solution assumes $n = 4t + 1$.
Note that it is trivial to adapt the solution for any $n > 4t + 1$ in the following way:
(1) First $4t + 1$ processes (i.e., the $4t + 1$ processes with smallest identifiers) execute the \comm algorithm among $4t + 1$ processes (as explained in the rest of the section).
(2) The aforementioned $4t + 1$ processes then utilize the cryptography-free asynchronous data dissemination (ADD) primitive~\cite{das2021asynchronous} that efficiently disseminates the decided value to all $n > 4t + 1$ processes.

\smallskip
\noindent \textbf{Pseudocode description.}
Lines~\ref{line:reducer_uses} to~\ref{line:reducer_last_line} define the employed primitives, the rules governing the behavior of correct processes, as well as the constants and local variables.
When processes start executing the \comm algorithm, they first disseminate their proposals in the \emph{dissemination phase} (lines~\ref{line:reducer_propose}-\ref{line:dissemination_complete}).
The dissemination phase of \comm is identical to that of HMVBA and has already been covered in \Cref{section:technical_overview}.
Briefly, processes disseminate their proposals via \textsc{init} messages: each \textsc{init} message contains an RS symbol, a digest, and a witness proving the validity of the RS symbol against the digest.
When a process receives a valid \textsc{init} message, the process stores the content of the message and acknowledges the reception by sending an \textsc{ack} message back.
Once a process receives $n - t$ \textsc{ack} messages, it informs all other processes that its proposal is disseminated by broadcasting a \textsc{done} message.
Upon receiving $n - t$ \textsc{done} messages, a process broadcasts a \textsc{finish} message; a process may also broadcast a \textsc{finish} message upon receiving $t + 1$ \textsc{finish} messages.
Lastly, when a process receives $n - t$ \textsc{finish} messages, the process completes the dissemination phase.

After completing the dissemination phase, processes start executing \comm through \emph{iterations}.
Each iteration $k \in \mathbb{N}$ unfolds as follows (lines~\ref{line:start_iteration}-\ref{line:reducer_decide}):
\begin{compactenum}
    \item Processes randomly elect the iteration's leader, denoted by $\mathsf{leader}(k)$ (line~\ref{line:random_election}).

    \item Processes establish their candidate digests through \textsc{stored} and \textsc{suggest} messages (lines~\ref{line:broadcast_stored}-\ref{line:reducer_sort}), as previously discussed in \Cref{subsection:reducer_overview}.
    Concretely, each process $p_i$ \emph{commits} up to two candidate digests.
    Formally, we say that a correct process $p_i$ \emph{$c$-commits} a digest $z$ in iteration $k$, for any $c \in \{ 1, 2 \}$, if and only if $\mathit{candidates}_i[c] = z$ when process $p_i$ reaches line~\ref{line:agreeing_phase}.
    Moreover, we define the $\mathsf{committed}(k, c)$ set, for any $c \in \{ 1, 2\}$, as
    \begin{equation*}
        \mathsf{committed}(k, c) = \{ z \,|\, z \text{ is $c$-committed by a correct process in iteration $k$} \}.
    \end{equation*}
    The default digest (line~\ref{line:reducer_fixed_default}) is introduced to ensure that processes have always two committed candidates (see line~\ref{line:reducer_set_default}), even if they do not commit any candidate through the standard \textsc{stored} and \textsc{suggest} steps.

    \item Processes aim to agree on a valid value (lines~\ref{line:agreeing_phase}-\ref{line:reducer_quasi_decide}).
    To achieve this, in every good iteration $k$---$\mathsf{leader}(k)$ has disseminated proposal $v^{\star}(k)$ with digest $z^{\star}(k)$---the following holds:
    \begin{equation}\label{equation:crucial}
        \big( z^{\star}(k) \in \mathsf{committed}(k, 1) \land |\mathsf{committed}(k, 1)| \leq 2 \big) \lor \big( \{ z^{\star}(k) \} = \mathsf{committed}(k, 2) \big).
        \tag{$\circledcirc$}
    \end{equation}
    As noted in \Cref{subsection:reducer_overview}, if (only) the first disjunct holds true, it may require two repetitions for processes to agree on $z^{\star}(k)$ when correct processes propose two different digests from the $\mathsf{committed}(k, 1)$ set to the SMBA primitive.
    In the first repetition, they may decide on a non-$z^{\star}(k)$ digest (the other digest from the $\mathsf{committed}(k, 1)$ set), and only in the second repetition do they succeed in agreeing on $z^{\star}(k)$.
    For this reason, each iteration $k$ is divided into three sub-iterations $(k, 1)$, $(k, 2)$,  and $(k, 3)$.
    In sub-iterations $(k, 1)$ and $(k, 2)$, processes aim to agree on their first and second committed digests, respectively.
    The sub-iteration $(k, 3)$ serves as a repetition of the first sub-iteration.

    Each sub-iteration $(k, x)$ unfolds as follows (lines~\ref{line:sub-iteration_starts}-\ref{line:reducer_quasi_decide}).
    Each process $p_i$ begins by adopting a digest.
    If $x = 1$ or $x = 2$, $ p_i $ adopts its $x$-committed digest.
    Otherwise, $p_i$ engages its ``proposal-switching'' logic: it checks whether the value it proposed to the first invocation of the SMBA primitive was decided.
    If so, $p_i$ adopts its $2$-committed digest; otherwise, $p_i$ adopts its $1$-committed digest.
    Process $p_i$ then proposes its adopted digest (i.e., $\mathit{adopted\_digest}_i$) to the SMBA primitive.
    Once processes agree on a digest $z$ via SMBA, processes start the R\&A mechanism.
    Specifically, processes disseminate the RS symbols received from $\mathsf{leader}(k)$ during the dissemination phase.
    Using the received RS symbols, each correct process $p_i$ decodes some value $r_i$ if possible; otherwise, $p_i$ sets $r_i$ to its proposal $v_i$.
    Then, each correct process $p_i$ proposes value $r_i$ to the MBA primitive.
    If the decided value $v$ is valid, each process $p_i$ quasi-decides $v$ by appending it to its $\mathit{quasi\_decision}_i$ list.

    \item The final phase (lines~\ref{line:check_quasi_decisions}-\ref{line:reducer_decide}) of the iteration is designed to ensure the quality property.
    After completing all three sub-iterations, processes check if any value was quasi-decided.
    If so, processes obtain a random integer $I \in [1, 3]$, which is then used to select one of the previously quasi-decided values for the decision.
\end{compactenum}

\subsection{Proof Sketch} \label{subsection:reducer_analysis}

This subsection provides a proof sketch of \comm's correctness and complexity.
Recall that a formal proof can be found in \Cref{section:reducer_proof}.

\smallskip
\noindent \textbf{Correctness.}
We first give a proof sketch of the following theorem.

\begin{theorem} [\comm is correct] \label{theorem:reducer_correct}
Given $n = 4t + 1$ and the existence of a collision-resistant hash function, \comm (see \Cref{algorithm:reducer}) is a correct implementation of the MVBA primitive in the presence of a computationally bounded adversary.
\end{theorem}

We now analyze each property separately.

\smallskip
\noindent \emph{Agreement.} The agreement property is ensured by (1) the agreement property of the MBA primitive employed in each iteration (line~\ref{line:lmba}), and (2) the fact that any $\mathsf{Index}()$ request (line~\ref{line:random_index}) returns the same integer to all correct processes.

\smallskip
\noindent \emph{Weak validity.} Suppose all processes are correct and a correct process $p_i$ decides some value $v$ in an iteration $k$.
Due to the justification property of the MBA primitive, some correct process $p_j$ proposed $v$ to the MBA primitive in iteration $k$.
If $p_j$ decoded $v$ at line~\ref{line:reducer_decode}, $v$ is the proposal of $\mathsf{leader}(k)$.
Otherwise, $v$ is $p_j$'s proposal (line~\ref{line:reducer_own_proposal}).
In any case, $v$ is proposed by a correct process.

\smallskip
\noindent \emph{External validity.}
The property is trivially satisfied due to the check at line~\ref{line:final_check}.

\smallskip
\noindent \emph{Integrity.}
The property is trivially satisfied due to the check at line~\ref{line:check_quasi_decisions}.

\smallskip
\noindent \emph{Termination.}
As discussed in \Cref{subsection:reducer_overview}, \comm is guaranteed to terminate in a good iteration.
We now formally define what constitutes a good iteration.
Let $\tau_{\mathsf{first}}$ be the first event at which a so-far-uncorrupted process broadcasts a \textsc{finish} message, and let $\first$ denote that process.
Process $\first$ necessarily broadcasts at line~\ref{line:first_finish}: otherwise, the $t + 1$ \textsc{finish} messages that activate line~\ref{line:second_finish} would contain a message previously broadcast by a so-far-uncorrupted process, contradicting the choice of $\tau_{\mathsf{first}}$.
Hence, $\first$ receives \textsc{done} messages from $n - t$ processes before $\tau_{\mathsf{first}}$ (for \comm, $n - t = 3t + 1$).
Let $\dfirst$ denote the set of those senders that are so-far-uncorrupted immediately before $\tau_{\mathsf{first}}$; subsequent corruptions do not alter $\dfirst$.
As at most $t$ senders are corrupted at $\tau_{\mathsf{first}}$, $|\dfirst| \geq n - 2t$ (for \comm, $n - 2t = 2t + 1$).
Moreover, $\tau_{\mathsf{first}}$ precedes every election disclosure: before an election can be disclosed, at least one so-far-uncorrupted process must complete the dissemination phase, whereas $\tau_{\mathsf{first}}$ is the first \textsc{finish} broadcast by any such process.

\begin{definition} [Good iterations] \label{definition:good_iteration}
An iteration $k \in \mathbb{N}$ is \emph{good} if and only if $\mathsf{leader}(k) \in \dfirst$.
\end{definition}

We are now ready to show that \comm terminates in the first good iteration $k$.
Let $v^{\star}(k)$ denote the valid proposal disseminated by $\mathsf{leader}(k)$ while it is uncorrupted, and let $z^{\star}(k)$ denote the digest of $v^{\star}(k)$.
As $k$ is a good iteration, $\mathsf{leader}(k) \in \dfirst$ broadcasts a \textsc{done} message while uncorrupted and, before doing so, receives \textsc{ack} messages from $n - t$ processes.
Consequently, $\mathsf{leader}(k)$ has stored valid RS symbols of its proposal $v^{\star}(k)$ at $(n - t) - t = 2t + 1$ correct processes.
As discussed in \Cref{subsection:reducer_implementation} (and proven in \Cref{section:reducer_proof}), \comm ensures that \Cref{equation:crucial} holds:
\begin{equation*}
    \big( z^{\star}(k) \in \mathsf{committed}(k, 1) \land |\mathsf{committed}(k, 1)| \leq 2 \big) \lor \big( \{ z^{\star}(k) \} = \mathsf{committed}(k, 2) \big).
\end{equation*}
We now separate two cases:
\begin{compactitem}
    \item Let $\mathsf{committed}(k, 2) = \{ z^{\star}(k) \}$.
    Consider sub-iteration $(k, 2)$.
    All correct processes decide $z^{\star}(k)$ from the SMBA primitive (line~\ref{line:smba}).
    Then, as $\mathsf{leader}(k)$ has disseminated RS symbols of its valid proposal $v^{\star}(k)$ to (at least) $n - 2t = 2t + 1$ correct processes, each correct process receives $(n - t) + (n - 2t) - n = n - 3t = t + 1$ RS symbols that correspond to digest $z^{\star}(k)$ (line~\ref{line:wait_for_reconstruct_messages}).
    Hence, all correct processes decode $v^{\star}(k)$ (line~\ref{line:reducer_decode}) and propose $v^{\star}(k)$ to the MBA primitive (line~\ref{line:lmba}).
    The strong unanimity property of the MBA primitive ensures that all correct processes decide $v^{\star}(k)$ from it and, thus, quasi-decide $v^{\star}(k)$ (line~\ref{line:reducer_quasi_decide}).
    Hence, termination is ensured.

    \item Let $z^{\star}(k) \in \mathsf{committed}(k, 1)$ and $|\mathsf{committed}(k, 1)| \leq 2$.
    Consider sub-iteration $(k, 1)$.
    As $|\mathsf{committed}(k, 1)| \leq 2$, correct processes propose at most two different digests to the SMBA primitive (line~\ref{line:smba}).
    The strong validity property guarantees that the decided digest $z$ is proposed by a correct process.
    We further investigate two possibilities:
    \begin{compactitem}
        \item Let $z = z^{\star}(k)$.
        Following the same reasoning as in the previous scenario, all correct processes decode $v^{\star}(k)$ (line~\ref{line:reducer_decode}) and propose $v^{\star}(k)$ to the MBA primitive (line~\ref{line:lmba}).
        The strong unanimity property of the MBA primitive ensures that all correct processes decide $v^{\star}(k)$ from it and, thus, quasi-decide $v^{\star}(k)$ (line~\ref{line:reducer_quasi_decide}).
        Hence, termination is guaranteed in this case.

        \item Let $z \neq z^{\star}(k)$.
        In this case, the ``proposal-switching'' logic (lines~\ref{line:proposal_switching_start}-\ref{line:proposal_switching_end}) ensures that all correct processes propose $z^{\star}(k)$ to the SMBA primitive in sub-iteration $(k, 3)$.
        (Recall that all correct processes commit $z^{\star}(k)$.)
        Therefore, $z^{\star}(k)$ is decided from the SMBA primitive (line~\ref{line:smba}) in sub-iteration $(k, 3)$.
        This implies that all correct processes decode $v^{\star}(k)$ (line~\ref{line:reducer_decode}) and propose $v^{\star}(k)$ to the MBA primitive (line~\ref{line:lmba}).
        Hence, all correct processes quasi-decide $v^{\star}(k)$ (line~\ref{line:reducer_quasi_decide}), showing that \comm terminates even in this case.
    \end{compactitem}
\end{compactitem}
Finally, it is important to note that while correct processes are not guaranteed to decide in a non-good iteration, each such iteration will complete because: (1) each correct process waits for messages from at most $n - t = 3t + 1$ processes, (2) there are at least $n - t = 3t + 1$ correct processes, and (3) the distributed primitives used internally are guaranteed to terminate.

\smallskip
\noindent \emph{Quality.}
Let $P_1$ denote the probability that the first iteration is good; note that $P_1 \geq \frac{n - 2t}{n} = \frac{2t + 1}{4t + 1} \geq \frac{1}{2}$.
As seen in the analysis of termination, correct processes are guaranteed to quasi-decide a non-adversarial value (i.e., the leader's valid proposal) in a good iteration $k$.
Let $P_2$ denote the probability that the $\mathsf{Index}()$ request (line~\ref{line:random_index}) invoked in a good iteration $k$ that quasi-decides a non-adversarial value indeed selects a non-adversarial quasi-decided value; $P_2 \geq \frac{1}{3}$ as at most three values (some of which might be the same) can be quasi-decided.
Therefore, the probability that the decided value is non-adversarial is (at least) the probability that (1) the first iteration is good, and (2) the $\mathsf{Index}()$ request chooses a non-adversarial value.
Therefore, the probability that an adversarial value is decided is at most $1 - P_1 \cdot P_2 \leq \frac{5}{6} < 1$.

We emphasize that a non-adversarial value is \emph{guaranteed} to be quasi-decided in a good iteration, although at most two additional adversarial values may also be quasi-decided.
Consequently, the probability that a non-adversarial value is quasi-decided (potentially alongside up to two adversarial values) is at least the probability that the first iteration is good, and is therefore at least $\frac{2t + 1}{4t + 1} > \frac{1}{2}$.
To clarify, if processes are restricted to deciding a \emph{single} value from \comm, an adversarial value may indeed be decided with probability up to $\frac{5}{6}$.
However, if we allow processes to decide a vector containing at most three distinct values, the probability that this vector includes a non-adversarial value is at least $\frac{1}{2}$.

\smallskip
\noindent \textbf{Complexity.}
Next, we provide a proof sketch of the following theorem.

\begin{theorem} [\comm's expected complexity] \label{theorem:reducer_complexity}
Given $n = 4t + 1$ and the existence of a collision-resistant hash function, the following holds for \comm (see \Cref{algorithm:reducer}) in the presence of a computationally bounded adversary:
\begin{compactitem}
    \item The expected message complexity is $O(n^2)$.

    \item The expected bit complexity is $O(n\ell + n^2 \kappa \log n)$.

    \item The expected time complexity is $O(1)$.
\end{compactitem}
\end{theorem}

Let us informally prove the theorem.

\smallskip
\noindent \emph{Expected time complexity.}
As previously argued, \comm terminates in a good iteration.
Given that $\dfirst$ is frozen before every election disclosure, each reached iteration is good with conditional probability $P \geq \frac{1}{2}$; hence, \comm terminates in at most two iterations in expectation.
As (1) each iteration takes $O(1)$ time to complete in expectation, and (2) the dissemination phase consists of $O(1)$ rounds of communication, the expected time complexity is $O(1)$.

\smallskip
\noindent \emph{Expected message complexity.}
Given that (1) each iteration exchanges $O(n^2)$ messages in expectation, and (2) \comm terminates in $O(1)$ iterations in expectation, correct processes exchange $O(n^2)$ messages in expectation through the iterations of \comm.
Finally, correct processes exchange $O(n^2)$ messages in expectation throughout the entire \comm protocol as the dissemination phase exchanges $O(n^2)$ messages as well.

\smallskip
\noindent \emph{Expected bit complexity.}
The dissemination phase exchanges $O(n\ell + n^2 \kappa \log n)$ bits.
Moreover, each iteration exchanges $O(n\ell + n^2 \kappa \log n)$ bits in expectation.
Given that there are $O(1)$ iterations in expectation, the expected bit complexity of \comm is
\begin{equation*}
    \underbrace{O(n\ell + n^2 \kappa \log n)}_{\text{dissemination phase}}
    {}+{}
    O(1) \cdot \underbrace{O(n\ell + n^2 \kappa \log n)}_{\text{each iteration}}
    \subseteq O(n\ell + n^2 \kappa \log n).
\end{equation*}

\section{\commplus: Pseudocode \& Proof Sketch} \label{section:reducer_plus}

This section introduces \commplus, our MVBA algorithm that comes $\epsilon$-close to optimal one-third resilience, for any fixed constant $\epsilon > 0$.
\commplus exchanges $O(n^2C^2)$ messages and $O\big( C^2 (\frac{n \ell}{\epsilon} + n^2 \kappa \log n) \big)$ bits, and terminates in $O(C^2)$ time, where $C = \lceil 1 + \frac{2}{\epsilon} \rceil^2$.
As a downside compared to \comm, \commplus requires hash functions modeled as a random oracle that uniformly distributes its outputs.
We start by presenting \commplus's pseudocode (\Cref{subsection:reducer_plus_implementation}).
Then, we give an informal analysis of \commplus's correctness and complexity (\Cref{subsection:reducer_plus_analysis}).
We conclude the section with an open question about reducing the constant multiplicative factor in the complexity of \commplus (\Cref{subsection:reducing_constants}).
A formal proof of \commplus's correctness and complexity is relegated to \Cref{section:reducer_plus_proof}.

\subsection{Pseudocode} \label{subsection:reducer_plus_implementation}

The pseudocode of \commplus is given in \Cref{algorithm:reducer_plus}.
Here, $\mathsf{decode}(S)$ interpolates any $\epsilon t + 1$ symbols from distinct verified positions in $S$.
An absent entry of $\mathit{symbols}_i$ has the fixed $\kappa$-bit digest $\mathit{default}$ and no RS symbol or witness.
Its empty \textsc{reconstruct} report counts toward the waiting threshold but contributes no symbol to $S_i$.

\smallskip
\noindent\textbf{Pseudocode description.}
Lines~\ref{line:reducer_plus_uses} to~\ref{line:reducer_plus_last} define the employed primitives, the rules governing the behavior of correct processes, as well as the constants and local variables.
When processes start executing the \commplus algorithm (line~\ref{line:reducert_plus_start}), they engage in the dissemination phase (line~\ref{line:dissemination_phase_plus}) that is identical to the dissemination phase of \comm (and HMVBA) except that values are treated as polynomials of degree $\epsilon t$.
Once the dissemination phase is completed, processes execute \commplus through \emph{iterations}.
Each iteration $k \in \mathbb{N}$ proceeds as follows (lines~\ref{line:iteration_start_plus}-\ref{line:reducer_decide_plus}):
\begin{compactenum}
    \item Processes elect the leader---$\mathsf{leader}(k)$---using a common coin (line~\ref{line:random_election_plus}).

    \item Processes determine their candidate digests through \textsc{stored} and \textsc{suggest} messages (lines~\ref{line:broadcast_stored_plus}-\ref{line:candidates_remove_plus}), as explained in \Cref{subsection:commplus_overview}.
    Formally, we say that a correct process $p_i$ \emph{commits} a digest $z$ in iteration $k$ if and only if $z$ belongs to the $\mathit{candidates}_i$ list when $p_i$ reaches line~\ref{line:for_plus} in iteration $k$.
    As argued in \Cref{subsection:commplus_overview} (and proven in \Cref{section:reducer_plus_proof}), given a good iteration $k$, (1) each correct process commits up to $\lceil 1 + \frac{2}{\epsilon} \rceil$ digests, and (2) there are at most $C = \lceil 1 + \frac{2}{\epsilon} \rceil^2$ different digests committed across all correct processes (i.e., $|\mathsf{committed}(k)| \leq C$).

    \item Processes aim to agree on a valid value through $4C^2$ sequential probabilistic trials (lines~\ref{line:for_plus}-\ref{line:quasi_decide_plus}).
    Concretely, we divide iteration $k$ into $4C^2$ sub-iterations $(k, 1)$, $(k, 2)$, ..., $(k, 4C^2)$.
    Each sub-iteration $(k, x \in [1, 4C^2])$ represents the $x$-th probabilistic trial within iteration $k$, and unfolds as follows.
    First, correct processes obtain a random value $\phi$ by utilizing a common coin.
    Then, each correct process $p_i$ adopts a digest by updating its $\mathit{adopted\_digest}_i$ variable:
    \begin{compactitem}
        \item If no digest is committed by $p_i$, process $p_i$ adopts the fixed default digest.

        \item Otherwise, process $p_i$ adopts a digest as follows: 
        (1) For each committed digest $z$, $p_i$ appends $\mathsf{hash}(z, \phi)$ to the $\mathit{hashed\_candidates}_i$ list.
        (2) Process $p_i$ sorts $\mathit{hashed\_candidates}_i$ in the lexicographic order.
        (3) Finally, process $p_i$ adopts the committed digest $z'$ that, along with $\phi$, produced the smallest element of the sorted $\mathit{hashed\_candidates}_i$ list.
    \end{compactitem}
    We underline that correct processes might adopt different digests.
    
    Once processes adopt their digests, they start the R\&A mechanism.
    Specifically, processes disseminate the RS symbols received from $\mathsf{leader}(k)$ during the dissemination phase.
    Each correct process $p_i$, if possible, decodes some value $r_i$ using the received RS symbols that correspond to its adopted digest; if it is impossible to decode any value, process $p_i$ sets $r_i$ to its proposal $v_i$.
    Finally, each correct process $p_i$ proposes value $r_i$ to the MBA primitive.
    If the decided value $v$ is valid, process $p_i$ quasi-decides $v$.

    \item The final phase (lines~\ref{line:check_quasi_decisions_plus}-\ref{line:reducer_decide_plus}) of the iteration ensures the quality property.
    After finishing all $4C^2$ sub-iterations, processes check if any value was quasi-decided.
    If so, processes obtain a random integer $I \in [1, 4C^2]$, which is then used to select one of the previously quasi-decided values for the decision.
\end{compactenum}

\begin{algorithm}[!t]
\caption{\commplus: Pseudocode (for process $p_i$)}
\label{algorithm:reducer_plus}
\begin{algorithmic}[1]
\scriptsize

\State \textbf{Uses:} \label{line:reducer_plus_uses}
\State \hskip2em \textcolor{jnSUDigitalRedLight}{\(\triangleright\) $\mbasimple$ exchanges $O(n^2)$ messages and $O(n\ell + n^2 \kappa \log n)$ bits and terminates in $O(1)$ time}
\State \hskip2em MBA algorithm $\mbasimple$ with $\valuemba = \valuemvba$, \textbf{instances} $\mathcal{MBA}[k][x]$, $\forall k, x \in \mathbb{N}$

\medskip
\State \textbf{Rules:}
\State \hskip2em - Any message with an invalid witness is ignored.

\State \hskip2em - Only one \textsc{init} message is processed per process.

\State \hskip2em - After deciding from \commplus, start no new iteration; continue serving already invoked primitives.

\medskip
\State \textbf{Constants:}
\State \hskip2em $\mathsf{Digest}$ $\mathit{default}$%
        \BlueComment{default digest} \label{line:reducer_fixed_default_plus}

\State \hskip2em $\mathsf{Integer}$ $C = \lceil 1 + 2/\epsilon \rceil^2$ \BlueComment{constant $C$ is used by all processes} \label{line:constant_c}

\medskip
\State \textbf{Local variables:}
\State \hskip2em $\valuemvba$ $v_i \gets p_i$'s proposal
\State \hskip2em $\mathsf{Boolean}$ $\mathit{dissemination\_completed}_i \gets \mathit{false}$
\State \hskip2em $\mathsf{Map}(\mathsf{Process} \to [\mathsf{RS}, \mathsf{Digest}, \mathsf{Witness}])$ $\mathit{symbols}_i \gets \text{empty map}$
\State \hskip2em $\mathsf{List}(\mathsf{Digest})$ $\mathit{candidates}_i \gets \text{empty list}$\BlueComment{will be reset every iteration}
\State \hskip2em $\mathsf{List}(\mathsf{Hash})$ $\mathit{hashed\_candidates}_i \gets \text{empty list}$\BlueComment{will be reset every iteration}
\State \hskip2em $\mathsf{Digest}$ $\mathit{adopted\_digest}_i \gets \bot$
\State \hskip2em $\mathsf{List}(\valuemvba)$ $\mathit{quasi\_decisions}_i \gets$ empty list \label{line:reducer_plus_last}

\medskip
\State \textbf{upon} $\mathsf{propose}(\valuemvba \text{ } v_i)$: \BlueComment{start of the algorithm} \label{line:reducert_plus_start}
\State \hskip2em Execute the dissemination phase identical to that of \comm (lines~\ref{line:reducer_propose}-\ref{line:dissemination_complete} of \Cref{algorithm:reducer}) except that values are treated as polynomials \hphantom{|||||L}of degree $\epsilon t$ \label{line:dissemination_phase_plus}

\medskip
\State \textbf{upon} receiving $\langle \textsc{finish} \rangle$ from $n - t$ processes (for the first time):\label{line:receive_quorum_finish_plus}
\State \hskip2em $\mathit{dissemination\_completed}_i \gets \mathit{true}$\label{line:dissemination_complete_plus} \BlueComment{dissemination phase completes}

\State \hskip2em \textbf{for each} $k = 1, 2, ...$: \label{line:iteration_start_plus}
\State \hskip4em $\mathit{candidates}_i \gets \text{empty list}$; $\mathit{hashed\_candidates}_i \gets \text{empty list}$\BlueComment{reset the candidates} \label{line:reset_candidates_plus}
\State \hskip4em $\mathsf{Process}$ $\mathsf{leader}(k) \gets \mathsf{Election()}$\BlueComment{elect a random leader}\label{line:random_election_plus}
\State \hskip4em \textbf{broadcast} $\langle \textsc{stored}, k, \mathit{symbols}_i[\mathsf{leader}(k)].\mathsf{digest()} \rangle$ \BlueComment{disseminate the leader's digest}\label{line:broadcast_stored_plus}
\State \hskip4em \textbf{wait for} $n - t = (2 + \epsilon) t + 1$ \textsc{stored} messages for iteration $k$\label{line:wait_for_stored_messages_plus}

\State \hskip4em \textbf{for each} $\mathsf{Digest}$ $z$ included in $n - 3t = \epsilon t + 1$ received \textsc{stored} messages:\label{line:check_stored_plus}
            \State \hskip6em $\mathit{candidates}_i.\mathsf{append}(z)$ \label{line:reducer_candidates_append_plus}

        \State \hskip4em \textbf{broadcast} $\langle \textsc{suggest}, k, \mathit{candidates}_i \rangle$\label{line:broadcast_suggest_plus} \BlueComment{disseminate $p_i$'s candidates}
        \State \hskip4em \textbf{wait for} $n - t = (2 + \epsilon)t + 1$ \textsc{suggest} messages for iteration $k$\label{line:wait_for_suggest_messages_plus}

        \State \hskip4em \textbf{for each} $\mathsf{Digest} \text{ } z \in \mathit{candidates}_i$:
            \State \hskip6em \textbf{if} $z$ is not included in $n - 2t = (1 + \epsilon)t + 1$ received \textsc{suggest} messages:\label{line:rule_precommit_plus}
                \State \hskip8em $\mathit{candidates}_i.\mathsf{remove}(z)$ \label{line:candidates_remove_plus}
        \State \hskip4em \textbf{for each} $x = 1, 2, ..., 4C^2$: \label{line:for_plus}
            \State \hskip6em $\mathsf{Integer}$ $\phi \gets \mathsf{Noise()}$\BlueComment{obtain a random $\kappa$-bit value $\phi$}\label{line:random_noise_plus}
            \State \hskip6em \textcolor{jnSUDigitalRedLight}{\(\triangleright\) if no digest is committed, then adopt $\mathit{default}$}
            \State \hskip6em \textbf{if} $\mathit{candidates}_i.\mathsf{size} = 0$:  $\mathit{adopted\_digest}_i \gets \mathit{default}$\label{line:acc_proposal_bot_plus}
            \State \hskip6em \textbf{else:}
                \State \hskip8em \textbf{for each} $j = 1, 2, ..., \mathit{candidates}_i.\mathsf{size}$: \label{line:start_hash_finding}
                    \State \hskip10em $\mathit{hashed\_candidates}_i[j] = \mathsf{hash}(\mathit{candidates}_i[j], \phi)$ \label{line:hash_candidates_plus}
                \State \hskip8em \textcolor{jnSUDigitalRedLight}{$\triangleright$ find the smallest hashed candidate and adopt the associated digest}
                \State \hskip8em Sort $\mathit{hashed\_candidates}_i$ in the lexicographic order \label{line:sort_hashed}
                \State \hskip8em $\mathsf{Hash}$ $\mathit{smallest}_i \gets \mathit{hashed\_candidates}_i[1]$ \label{line:find_smallest_plus}
                \State \hskip8em Let $z' \in \mathit{candidates}_i$ be the digest such that $\mathit{smallest}_i = \mathsf{hash}(z', \phi)$ \label{line:find_C_plus}
                \State \hskip8em $\mathit{adopted\_digest}_i \gets z'$ \label{line:update_acc_proposal_2_plus}
                    \smallskip
                    \State \hskip6em \textcolor{jnSUDigitalRedLight}{$\triangleright$ Reconstruct \& Agree} \

                    \State \hskip6em \textbf{broadcast} $\langle \textsc{reconstruct}, k, x, \mathit{symbols}_i[\mathsf{leader}(k)] \rangle$ \label{line:reconstruct_broadcast_reducer_plus}
                \State \hskip6em \textbf{wait for} $n - t = (2 + \epsilon)t + 1$ \textsc{reconstruct} messages for sub-iteration $(k, x)$\label{line:wait_for_reconstruct_messages_plus}

                \State \hskip6em $\mathsf{Set}(\mathsf{RS})$ $S_i \gets$ received RS symbols with valid witnesses for $\mathit{adopted\_digest}_i$
                \State \hskip6em \textcolor{jnSUDigitalRedLight}{\(\triangleright\) if a value can be decoded, set $r_i$ to the decoded value}
                \State \hskip6em \textbf{if} $|S_i| \geq \epsilon t + 1$: $\valuemvba$ $r_i \gets \mathsf{decode}(S_i)$ \label{line:reducer_plus_decode}
                \State \hskip6em \textbf{else:} $\valuemvba$ $r_i \gets v_i$ \BlueComment{if not, set $r_i$ to $p_i$'s proposal} \label{line:reducer_plus_own_proposal}

                \State \hskip6em $\valuemvba \cup \{ \bot_{\mathsf{MBA}} \}$ $v \gets \mathcal{MBA}[k][x].\mathsf{propose}(r_i)$\label{line:lmba_plus} \BlueComment{propose $r_i$}
                \State \hskip6em \textbf{if} $\mathsf{valid}(v) = \mathit{true}$: $\mathit{quasi\_decisions}_i.\mathsf{append}(v)$ \BlueComment{quasi-decide $v$} \label{line:quasi_decide_plus}
        \smallskip
        \State \hskip4em \textbf{if} $\mathit{quasi\_decisions}_i.\mathsf{size} > 0$ and $p_i$ has not previously decided: \label{line:check_quasi_decisions_plus}
            \State \hskip6em $\mathsf{Integer}$ $I \gets \mathsf{Index}()$ \label{line:random_index_plus} \BlueComment{obtain a random integer $I$ in the $[1, 4C^2]$ range}
            \State \hskip6em \textcolor{jnSUDigitalRedLight}{\(\triangleright\) for quality}
            \State \hskip6em \textbf{trigger} $\mathsf{decide}\big( \mathit{quasi\_decisions}_i[(I \text{ mod } \mathit{quasi\_decisions}_i.\mathsf{size}) + 1] \big)$ \label{line:reducer_decide_plus}

\end{algorithmic}
\end{algorithm}

\subsection{Proof Sketch} \label{subsection:reducer_plus_analysis}

This subsection gives a proof sketch of \commplus's correctness and complexity.
Recall that a formal proof is relegated to \Cref{section:reducer_plus_proof}.

\smallskip
\noindent\textbf{Correctness.}
We start by providing a proof sketch of the following theorem.

\begin{theorem} [\commplus is correct] \label{theorem:reducer_plus_correct}
Given $n = (3 + \epsilon)t + 1$, for any fixed constant $\epsilon > 0$, and the existence of a hash function modeled as a random oracle, \commplus (see \Cref{algorithm:reducer_plus}) is a correct implementation of the MVBA primitive in the presence of a computationally bounded adversary.
\end{theorem}

Since the analysis of \commplus's agreement, weak validity, external validity, and integrity follows the arguments made in the analysis of \comm (see \Cref{subsection:reducer_analysis}), we primarily focus on \commplus's termination and quality.

\smallskip
\noindent\emph{Termination.}
For \commplus, let $\first$ be the first so-far-uncorrupted process to broadcast \textsc{finish}.
Let $\dfirst$ contain the senders of its $n-t$ \textsc{done} messages that are still uncorrupted at that event; freeze this set at that moment.
As shown in \Cref{section:reducer_plus_proof}, $|\dfirst| \geq n-2t$, and this event precedes every election disclosure.
\begin{definition}[Good iterations of \commplus]\label{definition:good_iteration_plus}
An iteration $k$ of \commplus is \emph{good} if and only if $\mathsf{leader}(k) \in \dfirst$.
\end{definition}
\commplus ensures termination with probability at least $\frac{1}{2}$ in any good iteration.
We now provide an informal justification for this statement.
As in the formal proof, the following probability bounds hold for sufficiently large $\kappa$, conditioned on the overwhelmingly likely random-oracle event established in \Cref{section:reducer_plus_proof}.
Let $k$ be any good iteration.
Let $v^{\star}(k)$ denote the valid proposal disseminated by $\mathsf{leader}(k)$ while uncorrupted, and let $z^{\star}(k)$ denote its digest.
For each trial $(k,x)$, define $\mathsf{start}(k,x)$ and $\mathsf{end}(k,x)$ as the digests committed by correct processes before the first correct process starts and ends that trial, respectively.
The inclusion chain
\begin{equation*}
    \mathsf{start}(k,1) \subseteq \mathsf{end}(k,1) \subseteq \cdots \subseteq \mathsf{start}(k,4C^2) \subseteq \mathsf{end}(k,4C^2) \subseteq \mathsf{committed}(k)
\end{equation*}
and the facts that $z^{\star}(k) \in \mathsf{start}(k,1)$ and $|\mathsf{committed}(k)| \leq C$ imply that at most $C-1$ trials have $\mathsf{start}(k,x) \neq \mathsf{end}(k,x)$.

Immediately before the noise $\phi$ of trial $(k,x)$ is disclosed, let $u$ denote the number of processes that are still uncorrupted.
For the analysis of this trial, take a snapshot of the lists fixed for iteration $k$ by processes that (1) are still uncorrupted and (2) have already completed candidate selection (line~\ref{line:for_plus}).
Let $\mathsf{pool}(k,x)$ be the union of these lists; this set is fixed before $\phi$ is revealed, even if some of these processes are corrupted later.
To bound its size, recall that a process retains a digest only after receiving $n-2t$ \textsc{suggest} messages containing it.
At most $n-u$ of these senders are already corrupted, so every digest in the pool was suggested by at least $(n-2t)-(n-u)=u-2t$ of the currently uncorrupted processes.
Each of these $u$ processes has suggested at most $d=\lceil1+2/\epsilon\rceil$ distinct digests in iteration $k$, giving at most $ud$ distinct process--digest pairs in total.
Since every pool digest accounts for at least $u-2t$ such pairs, $|\mathsf{pool}(k,x)|(u-2t)\leq ud$.
As at most $t$ processes are corrupted, $u\geq n-t$ and hence $u-2t\geq n-3t>0$.
Therefore,
\[
    \frac{u}{u-2t}=1+\frac{2t}{u-2t}
    \leq1+\frac{2t}{n-3t}=\frac{n-t}{n-3t}.
\]
Moreover, substituting $n=(3+\epsilon)t+1$ gives
\[
    \frac{n-t}{n-3t}=1+\frac{2t}{\epsilon t+1}
    \leq1+\frac{2}{\epsilon}
    \leq\left\lceil1+\frac{2}{\epsilon}\right\rceil=d.
\]
Multiplying these bounds by $d$, we obtain
\[
    |\mathsf{pool}(k,x)|\leq\frac{ud}{u-2t}
    \leq d\frac{n-t}{n-3t}\leq d^2=C.
\]
Moreover, $z^{\star}(k) \in \mathsf{start}(k,x) \subseteq \mathsf{pool}(k,x)$.
Call a trial \emph{won} if $\mathsf{hash}(z^{\star}(k),\phi)$ is strictly smaller than the hashes of the other digests in this frozen pool.
The random-oracle argument in \Cref{section:reducer_plus_proof} gives conditional win probability at least $1/(2C)$, irrespective of the preceding history.
If a won trial has $\mathsf{start}(k,x)=\mathsf{end}(k,x)$, every correct process proposing before the first correct MBA decision adopts $z^{\star}(k)$ and reconstructs $v^{\star}(k)$ from at least $n-3t=\epsilon t+1$ symbols.
The strong unanimity and agreement properties of MBA then ensure that all correct processes quasi-decide $v^{\star}(k)$ in that trial.

Let $W$ count won trials.
Failure to terminate requires $W<C$, since at most $C-1$ won trials can have unequal start and end sets.
Successive conditioning, without assuming independent trial outcomes, gives
\begin{equation*}
    \Pr[W<C] \leq e^C\mathbb{E}[e^{-W}]
    \leq e^C\left(1-\frac{1-e^{-1}}{2C}\right)^{4C^2}
    \leq e^{-(1-2/e)C}<\frac12,
\end{equation*}
where the last inequality uses $C\geq4$.
As the frozen set $\dfirst$ contains at least $n-2t$ processes, each reached iteration succeeds with conditional probability at least $\frac13\cdot\frac12=\frac16$.
Consequently, the probability of requiring more than $k$ iterations is at most $(5/6)^k$, and the expected number of iterations is at most $6$.
A formal argument for \commplus's termination can be found in \Cref{section:reducer_plus_proof}.

\smallskip
\noindent \emph{Quality.}
If (1) the first iteration is good, (2) correct processes quasi-decide the valid non-adversarial value $v^{\star}(1)$, and (3) the $\mathsf{Index}()$ request (line~\ref{line:random_index_plus}) selects $v^{\star}(1)$ as the final decision, a non-adversarial value is indeed decided from \commplus.
Let $M$ count occurrences of $v^{\star}(1)$ in the common quasi-decision list of a good first iteration.
As every won trial except at most $C-1$ contributes such an occurrence, $M\geq W-(C-1)$ and $\mathbb{E}[M]\geq4C^2/(2C)-(C-1)=C+1$.
The list has at most $4C^2$ entries and is fixed before the independent $\mathsf{Index}()$ output is disclosed.
Every entry has at least one preimage under the modulo selection rule, so each occurrence is selected with probability at least $1/(4C^2)$.
Hence, conditioned on a good first iteration, the probability of deciding $v^{\star}(1)$ is at least $(C+1)/(4C^2)\geq1/C^2$.
Thus, the probability $P$ of deciding a non-adversarial value satisfies
\begin{equation*}
    P \geq \frac{(1 + \epsilon)t + 1}{(3 + \epsilon)t + 1} \cdot \frac{1}{C} \cdot \frac{1}{C} \geq \frac{1}{3C^2}.
\end{equation*}
As a result, the probability that an adversarial value is decided is $1-P\leq1-1/(3C^2)<1$.

A remark analogous to the one regarding the quality of \comm\ applies here.
A non-adversarial value is quasi-decided with probability at least $\frac13\cdot\frac12=\frac16$: the first iteration is good with probability at least $\frac13$, and its good value is quasi-decided with conditional probability at least $\frac12$.
Hence, if we allow processes to decide a vector (rather than a single value) from \commplus, the probability that this vector includes a non-adversarial value is at least $\frac{1}{3C} \gg \frac{1}{3C^2}$.

\smallskip
\noindent\textbf{Complexity.}
Next, we provide a proof sketch of the following theorem.

\begin{theorem} [\commplus's expected complexity] \label{theorem:reducer_plus_complexity}
Given $n = (3 + \epsilon)t + 1$, for any fixed constant $\epsilon > 0$, and the existence of a hash function modeled as a random oracle, the following holds for \commplus (see \Cref{algorithm:reducer_plus}) in the presence of a computationally bounded adversary:
\begin{compactitem}
    \item The expected message complexity is $O(n^2 C^2)$.

    \item The expected bit complexity is $O\big( C^2 (\frac{n\ell}{\epsilon} + n^2 \kappa \log n) \big)$.

    \item The expected time complexity is $O(C^2)$.
\end{compactitem}
\end{theorem}

Let us now informally prove the theorem.

\smallskip
\noindent \emph{Expected time complexity.}
As previously discussed, each reached iteration succeeds with conditional probability at least $1/6$, so \commplus terminates within at most six iterations in expectation.
As (1) each iteration takes $O(C^2)$ time to complete in expectation, and (2) the dissemination phase consists of $O(1)$ rounds of communication, the expected time complexity of \commplus is $O(C^2)$.

\smallskip
\noindent \emph{Expected message complexity.}
As (1) each iteration exchanges $O(n^2 C^2)$ messages in expectation, and (2) \commplus terminates in $O(1)$ iterations in expectation, correct processes exchange $O(n^2 C^2)$ messages in expectation through the iterations of \commplus.
Finally, correct processes exchange $O(n^2 C^2)$ messages in expectation throughout the entire \commplus protocol as the dissemination phase exchanges only $O(n^2)$ messages.
The formal proof aggregates the per-iteration costs using the preceding conditional tail bound; no independence between costs and the number of iterations is assumed.

\smallskip
\noindent \emph{Expected bit complexity.}
An RS symbol has size $O(\ell/(\epsilon n)+\log n)$, so the dissemination phase exchanges $O(n\ell/\epsilon+n^2\kappa\log n)$ bits.
Moreover, each iteration exchanges $O\big(C^2(n\ell/\epsilon+n^2\kappa\log n)\big)$ bits in expectation.
Given that there are $O(1)$ iterations in expectation, the expected bit complexity of \commplus is
\begin{equation*}
    \underbrace{O(\frac{n\ell}{\epsilon}+n^2\kappa\log n)}_{\text{dissemination phase}}
    {}+{}
    O(1)\cdot\underbrace{O\big(C^2(\frac{n\ell}{\epsilon}+n^2\kappa\log n)\big)}_{\text{each iteration}}
    \subseteq O\big(C^2(\frac{n\ell}{\epsilon}+n^2\kappa\log n)\big).
\end{equation*}

\subsection{Reducing the Constant Multiplicative Overhead} \label{subsection:reducing_constants}

The trial overhead established above is $O(C^2)\subseteq O(1/\epsilon^4)$.
Whether it can be reduced to $O(1/\epsilon^3)$ remains open.

\section{Related Work}\label{section:related_work}

First, we examine earlier results on Byzantine agreement in asynchronous settings (\Cref{subsection:existing-results}).
Second, we compare our techniques with those of closely-related works (\Cref{subsection:techniques_comparison}).

\subsection{Earlier Results}\label{subsection:existing-results}

\noindent \textbf{Byzantine agreement in the full information model with an adaptive adversary.}
Ben-Or~\cite{BenOr83} introduced the first solution to the asynchronous Byzantine agreement problem.
In Ben-Or's algorithm, each process relies on a local coin. 
This algorithm has the notable advantage of operating in the \emph{full information model} (where the adversary is aware of the internal states of all processes) and tolerating an adaptive adversary. 
However, since each process uses a local coin---ensuring an agreement probability of $2^{-n}$---the algorithm suffers from exponential latency. 
Solving asynchronous Byzantine agreement in the full information model with an adaptive adversary, and without assuming a common coin, remains a significant challenge, as highlighted in several works \cite{huang2022byzantine,huang2023byzantine,kimmett2020improvement,king2016byzantine,king2018correction,melnyk2020byzantine}.
A major breakthrough occurred in 2018 with the first polynomial-time algorithm offering linear resilience \cite{king2018correction}, which corrected a technical flaw made in an earlier result \cite{king2016byzantine}. Despite this progress, the resilience of the solution in \cite{king2018correction} was limited to $1.14 \cdot 10^{-9} \cdot n$. More recently, the first polynomial-time algorithm achieving optimal resilience for this model was presented in \cite{huang2023byzantine}, building on a near-optimal-resilience result from the same authors \cite{huang2022byzantine}. However, this achievement comes with a significant drawback: an expected time complexity of $\tilde{O}(n^{12})$, which is prohibitively high for practical applications.

\smallskip
\noindent \textbf{Byzantine agreement with private channels.}
Faced with the significant challenges of solving Byzantine agreement in the full information model, the research community shifted its focus to a model incorporating \emph{private channels}. 
In this alternative model, algorithms typically depend on a \emph{weak common coin}, which permits a constant probability of disagreement on a random value. 
The use of a weak common coin for designing Byzantine agreement protocols was pioneered by Canetti and Rabin~\cite{canetti1993fast}, who proposed an information-theoretically and adaptively secure asynchronous binary agreement algorithm with $O(n^7)$ bit complexity and $O(1)$ time complexity.
Building on this approach, Abraham \emph{et al.}~\cite{AbrahamJMMST21} developed an adaptively secure asynchronous common subset (ACS) algorithm with $O(n^3 \kappa)$ bit complexity and constant time complexity, relying on public-key cryptography (where $\kappa$ represents the signature size).
More recently, Abraham \emph{et al.}~\cite{AAPS23} introduced a statistically secure ACS protocol with $O(n^5)$ bit complexity and $O(1)$ time complexity, designed for $t < \frac{1}{4}n$.
To achieve optimal one-third resilience, this protocol incorporates asynchronous verifiable secret sharing (AVSS).
While it is formally proven to be secure against a static adversary, the authors conjecture that their protocol can be extended to withstand an adaptive adversary as well.
Similarly, the hash-based ACS protocol proposed by Das \emph{et al.}~\cite{DBLP:conf/ccs/DasDLM0S24} employs a weak common coin, achieving $O(1)$ time complexity and $O(n^3 \kappa)$ bit complexity, but its security is limited to static adversaries.

\smallskip
\noindent \textbf{Byzantine agreement with an idealized common coin.}
It has become standard practice~\cite{hmvba,finmvba,sq,dumbomvba,flt24mvba} when constructing asynchronous Byzantine agreement protocols to do so in two parts:
an \emph{idealized common-coin abstraction},
and an otherwise (possibly) deterministic \emph{protocol core}.
The common coin
encapsulates the randomness,
and upon invocation by sufficiently many processes provides the same unpredictable and unbiasable random sequence to all processes.
The rest of the protocol is the actual deterministic distributed-computing ``core mechanism''.
A common coin can be implemented without a trusted dealer (assumed in the pioneering work of Rabin~\cite{Rabin83}) by utilizing threshold cryptography~\cite{CachinKS05}.
Moreover, the literature contains some dedicated common coin protocols~\cite{DasYXMK022,DBLP:conf/icdcs/GaoLLTXZ22,DBLP:journals/iacr/BandarupalliBBKR23,Bacho2023,Souza22}.
As discussed in \Cref{section:preliminaries},
\comm and \commplus can also be made to work with only weak common coins.

\smallskip
\noindent \textbf{MVBA algorithms.}
The MVBA problem was introduced in~\cite{CKPS01} alongside a protocol that achieves $O(1)$ time complexity, $O(n^2 \ell + n^2 \kappa + n^3)$ bit complexity, and optimal one-third resilience (see \Cref{tab:mvba-constructions}). VABA~\cite{vaba} improves bit complexity to $O(n^2 \ell + n^2 \kappa)$ while maintaining optimal resilience and time complexity, as does sMVBA~\cite{smvba}. 
Dumbo-MVBA~\cite{dumbomvba} further reduces bit complexity to $O(n \ell + n^2 \kappa)$ while retaining $O(1)$ time complexity and one-third resilience. 
All these protocols are secure against an adaptive adversary, but rely on threshold cryptography which requires trusted setup (or expensive key generation protocols), is not post-quantum secure, and tends to be slow. 

These limitations have sparked growing interest in designing adaptively-secure MVBA protocols that are hash-based from the ground up (assuming a common-coin object),
the state-of-the-art of which are HMVBA~\cite{hmvba}, FIN-MVBA~\cite{finmvba}, and FLT24-MVBA~\cite{flt24mvba} (see \Cref{tab:mvba-constructions}).
HMVBA achieves $O(n\ell + n^2 \kappa \log n)$ bit complexity and $O(1)$ time complexity, but only sub-optimal one-fifth resilience.
FIN-MVBA tolerates up to $t < n / 3$ faults, while also achieving $O(1)$ time complexity, but only sub-optimal $O(n^2 \ell + n^3 \kappa)$ or $O(n^2 \ell + n^2 \kappa + n^3 \log n)$ bit complexity.
FIN-MVBA itself is an improvement over the hash-based MVBA protocol implied by the distributed key generation protocol of~\cite{DasYXMK022,eprint_2023_1196},
which suffers from $O(\log n)$ time complexity, and is only secure against static adversaries.
Both HMVBA and FIN-MVBA satisfy the quality property.
FLT24-MVBA~\cite{flt24mvba} achieves optimal resilience, tolerating up to one-third faulty processes.
Moreover, FLT24-MVBA exchanges $O(n \ell + n^2 \kappa \log n + n^2 \kappa \lambda)$ bits and terminates in $O(\log \lambda)$ time, where $\lambda$ denotes a statistical security parameter (that cannot be treated as a constant, as explained in \Cref{subsection:techniques_comparison}).
It is also important to mention that FLT24-MVBA does not satisfy the quality property (see \Cref{section:introduction}).

\smallskip
\noindent \textbf{Other Byzantine agreement solutions.}
Beyond MVBA, Byzantine agreement primitives such as multi-valued Byzantine agreement (MBA), asynchronous common subset (ACS), and atomic broadcast (ABC) are well-studied in distributed systems. 
Mostéfaoui \emph{et al.}~\cite{DBLP:journals/acta/MostefaouiR17,MostefaouiMR15} introduced a cryptography-free asynchronous MBA protocol with optimal resilience, $O(n^2 \ell)$ bits and $O(1)$ time.
In~\cite{Nayak0SVX20}, a general way of building MBA protocols for long values by ``extending'' MBA protocols for short values is proposed (both in synchrony and asynchrony).

PACE~\cite{pace} solves the ACS problem by relying on $n$ parallel instances of asynchronous binary agreement, thus obtaining $O(\log n)$ time complexity.
Building on FIN-MVBA, \cite{finmvba} presents a hash-based ACS protocol with $O(1)$ time complexity and $O(n^2 \ell + n^3 \kappa)$ bit complexity.
Similarly, \cite{AAPS23} achieves ACS with $O(1)$ time complexity without assuming a common coin, but only with $1/4$ resilience and $O(n^4 \log n)$ bit complexity; it is worth noting that~\cite{AAPS23} is safe against an adaptive adversary.
Moreover, it is worth mentioning that~\cite{AAPS23} proposes a similar result with statistical security for $t < n / 3$.
Constructing ABC from MVBA is possible, as shown in~\cite{CKPS01}.
ABC constructed from FIN-MVBA has $O(n^3)$ message complexity, which~\cite{sq} reduces to $O(n^2)$ but without reducing the $O(n^2 \ell + n^3 \kappa)$ bit complexity any further. 
All of these protocols have optimal resilience.

\subsection{Our Techniques vs.\ Techniques of Closely-Related Works}\label{subsection:techniques_comparison}

\noindent \textbf{FLT24-MVBA~\cite{flt24mvba}.}
The FLT24-MVBA algorithm starts with the dissemination phase in which each process disseminates its value using a method similar (but not identical) to the approach employed by HMVBA and our algorithms.
The dissemination phase ensures that if a process $p_i$ successfully disseminates its proposal, then at least $n - 2t$ correct processes can reconstruct it even if process $p_i$ later gets corrupted.
Then, processes elect $\lambda$ leaders $L_1, L_2, ..., L_{\lambda}$ whose values they try to reconstruct; $\lambda$ denotes a statistical security parameter.

The crux of the FLT24-MVBA algorithm is a primitive called synchronized multi-valued broadcast (SMB); this primitive is inspired by the MV broadcast primitive introduced by Mostéfaoui \emph{et al.}~\cite{DBLP:journals/acta/MostefaouiR17}.
The primitive ensures that if $n - 2t$ correct processes broadcast the same input $v$, then all correct processes output a set containing at most two values. 
Furthermore, the outputs of any two correct processes are guaranteed to overlap: if one process outputs a single value, that value must appear in the other process's output set. 
If both processes output two values, their sets are identical. 
The SMB primitive acts as a robust filtering mechanism, ensuring that processes consider at most two different valid values per each elected leader.

Let us focus on a specific leader $L_j$.
To select one of the two values as the final output of leader $L_j$, FLT24-MVBA employs an asynchronous reliable consensus (ARC) protocol. 
ARC ensures agreement among processes and guarantees termination if all correct processes share the same input value. 
Since SMB may leave two possible values, all correct processes participate in two parallel ARC protocols, one for each value. 
However, a termination issue may arise if one of the two values is not held by all correct processes. 
To address this, the algorithm utilizes a standard technique~\cite{CGL18} based on the asynchronous binary agreement (ABA) primitive: an ABA protocol follows each ARC instance in order to decide whether that value (out of the two) should be selected.
If a process outputs 1 in one ABA instance, it proposes 0s to all other ABA instances whose corresponding ARC counterparts have not yet terminated.
Finally, as there exists at least one ``good'' leader $L$ (except with negligible probability in $\lambda$), $L$'s SMB instance will terminate, allowing all correct processes to reach consensus.

\smallskip
\noindent \emph{Comparison with FLT24-MVBA.}
The filtering SMB step and our reducing step serve the same fundamental purpose, making our algorithms similar in this regard.
However, a key difference lies in how termination is handled.
In FLT24-MVBA, the SMB-ARC-ABA sequence fails to terminate if it is tied to a ``bad'' leader.
This constraint prevents FLT24-MVBA from adopting our ``one-leader-per-iteration'' structure, as electing a bad leader---which happens with constant probability---would then cause the entire algorithm to stall indefinitely.

To address this problem, FLT24-MVBA incorporates a statistical security parameter $\lambda$.
It is essential to understand why $\lambda$ cannot be treated as a constant.
First, note that, if no ``good'' leader is elected, FLT24-MVBA stalls indefinitely, i.e., fails to terminate.
Consequently, the protocol solves the MVBA problem with probability $1 - c^{\lambda}$, where $c$ represents the fraction of ``bad'' leaders.
Only treating $\lambda$ as a security parameter (rather than as a constant) results in the desired negligible (in $\lambda$) error probability, and allows for instance the algorithm to be used in a polynomial (in $\lambda$) composition while maintaining this negligible error probability (e.g., in the ``repeated MVBA'' construction of atomic broadcast or state-machine replication; see \Cref{section:introduction}).
In contrast, treating $\lambda$ as a constant would result in a constant error probability, making the algorithm ill-suited for composition.

The current design of FLT24-MVBA, where ``the fastest leader wins'', prevents the protocol from ensuring quality. 
For example, a single adaptive corruption can lead to the decision of an adversarial value. 
If a fast, corrupted leader propagates an adversarial value and quickly moves through the SMB, ARC, and ABA phases, correct processes are forced to follow it. 
To these processes, the leader may appear correct and be seen as their only path to termination. As this leader is significantly faster than others, only the ABA instance tied to its adversarial value will decide on $1$, while all other instances will decide $0$, ultimately forcing correct processes to decide on the adversarial value.
In contrast, our algorithms guarantee that the original valid proposal of a ``good'' leader is eventually decided (with at least some constant probability), thereby ensuring quality---the importance of which is discussed in \Cref{section:introduction}.

\smallskip
\noindent \textbf{Comparison with~\cite{DBLP:journals/acta/MostefaouiR17}.}
While the ideas of~\cite{DBLP:journals/acta/MostefaouiR17} share the same spirit as ours, they are insufficient for our algorithms.
To elaborate, \cite{DBLP:journals/acta/MostefaouiR17} introduces two broadcasting primitives: 
(1)~the reducing (RD) broadcast primitive that reduces the number of values to a constant, and 
(2)~the multi-valued validated (MV) broadcast primitive that outputs a set of values such that, if the output set of a correct process contains a single value $v$, then the output set of any other correct process contains $v$. 
(Recall that the MV broadcast primitive served as inspiration for the SMB primitive of 
FLT24-MVBA~\cite{flt24mvba}.)
Crucially for our discussion, both the RD and MV broadcast primitives guarantee a delivery of a given value $v$ only if \emph{all} correct processes broadcast $v$.
In other words, the RD and MV broadcast primitives ''preserve'' a value only if all correct processes hold it.
If a majority---but not all---of the correct processes broadcast $v$, these primitives do not guarantee that any correct process will deliver $v$, meaning $v$ might not be preserved.
This design aligns with the focus of~\cite{DBLP:journals/acta/MostefaouiR17} on multi-valued Byzantine agreement (MBA; see \Cref{mod:async_mba}), which satisfies (only) strong unanimity:
if all correct processes propose the same value, that value must be decided; otherwise, any value can be decided.
In this context, the RD and MV broadcast primitives are sufficient, as the MBA problem only requires preserving a value when all correct processes propose it.

Our algorithms, on the other hand, necessitate the preservation of a value even when it is not held by all correct processes, which represents the primary reason why the techniques of~\cite{DBLP:journals/acta/MostefaouiR17} cannot be directly applied to our algorithms.
To clarify, the main sub-problem in our algorithms can be seen as follows: 
If $n - 2t$ correct processes hold a ``good'' digest $z$ in a good iteration, each correct process must obtain $z$ and rebuild the corresponding ``good'' value $v$ (which occurs only with constant probability in \commplus).
This must be satisfied even if other correct processes hold adversarial non-$z$ digests, as such an attack can occur in any good iteration. 
As a result, our algorithms require stronger techniques (SMBA in \comm and hash-based adoption procedure in \commplus) than those of~\cite{DBLP:journals/acta/MostefaouiR17}: as only $n - 2t$ (and not all!) correct processes hold a ``good'' digest, RD and MV broadcasts may never deliver $z$, thus preventing termination.
Finally, we design our own MBA algorithm instead of using the one from~\cite{DBLP:journals/acta/MostefaouiR17} as that one requires $O(n^2\ell)$ bits.

\smallskip
\noindent \textbf{Comparison with SQ~\cite{sq}.}
First, it is important to emphasize that SQ, strictly speaking, addresses the problem of \emph{atomic broadcast}, which involves reaching agreement not on a single value (or message), but on an ordered sequence of values (or messages).
However, since their solution to the atomic broadcast problem determines a value for each position in the agreed sequence, it directly provides a solution to the MVBA problem with the same complexity.

In terms of techniques, SQ, like our approach and many others~\cite{finmvba,flt24mvba}, employs a common coin for leader election.
The overarching strategy, as in our protocols, is that once the leader election succeeds---what we refer to as a ``good iteration''---the protocol is guaranteed to terminate using the MBA primitive.
The key distinction between our protocols and SQ lies in the dissemination phase.
Specifically, SQ employs the \emph{parallel consistent broadcast with weak agreed set (PCBW)} primitive, which, as the name implies, essentially consists of executing $n$ consistent broadcasts in parallel, one per process.
PCBW guarantees that every correct process eventually delivers the inputs from all other correct processes.
This, however, is precisely why SQ incurs a bit complexity of $O(n^2\ell)$: since each correct process must deliver $O(n)$ messages of $\ell$ bits, the total communication amounts to $O(n^2\ell)$ bits.
As a result, any approach based on PCBW inherently leads to $O(n^2\ell)$ bit complexity, which stands in contrast to the communication efficiency we aim for in our work.

Unlike SQ, processes in our protocols do not deliver (i.e., reconstruct) a linear number of values during the dissemination phase.
Instead, they reconstruct values only ``when necessary''---specifically, one value per iteration, corresponding to the proposal of that iteration's leader.
Given that (1) the protocol terminates as soon as a ``good'' leader is elected, and (2) this happens in expectation within $O(1)$ iterations, only $O(1)$ values need to be reconstructed (whereas SQ reconstructs $O(n)$ values).
This is why our approach avoids the $O(n^2\ell)$ bit complexity and instead achieves a total complexity of $O(1) \cdot O(n\ell) \subseteq O(n\ell)$.

\section*{Acknowledgment}
We thank Pierre Civit, Daniel Collins, Sourav Das, Jason Milionis, and Manuel Vidigueira for fruitful discussions.
The work of Jovan Komatovic was conducted in part while at a16z Crypto Research and while at EPFL.
The research of Tim Roughgarden at Columbia University was supported in part by NSF awards CCF-2006737 and CNS-2212745, and research awards from the Briger Family Digital Finance Lab and the Center for Digital Finance and Technologies.

\nocite{fullversion}
\bibliographystyle{splncs04}
\bibliography{references}

\begin{thebibliography}{10}
\providecommand{\url}[1]{\texttt{#1}}
\providecommand{\urlprefix}{URL }
\providecommand{\doi}[1]{https://doi.org/#1}

\bibitem{abd2005fault}
Abd{-}El{-}Malek, M., Ganger, G.R., Goodson, G.R., Reiter, M.K., Wylie, J.J.:
  Fault-scalable {Byzantine} fault-tolerant services. In: {SOSP}. pp. 59--74.
  {ACM} (2005)

\bibitem{AAPS23}
Abraham, I., Asharov, G., Patra, A., Stern, G.: Asynchronous agreement on a
  core set in constant expected time and more efficient asynchronous {VSS} and
  {MPC}. In: {TCC} {(4)}. pp. 451--482. Lecture Notes in Computer Science,
  Springer (2024)

\bibitem{Abraham2022crusader}
Abraham, I., Ben{-}David, N., Yandamuri, S.: Efficient and adaptively secure
  asynchronous binary agreement via binding crusader agreement. In: {PODC}. pp.
  381--391. {ACM} (2022)

\bibitem{AbrahamJMMST21}
Abraham, I., Jovanovic, P., Maller, M., Meiklejohn, S., Stern, G., Tomescu, A.:
  Reaching consensus for asynchronous distributed key generation. In: {PODC}.
  pp. 363--373. {ACM} (2021)

\bibitem{vaba}
Abraham, I., Malkhi, D., Spiegelman, A.: Asymptotically optimal validated
  asynchronous {Byzantine} agreement. In: {PODC}. pp. 337--346. {ACM} (2019)

\bibitem{adya2002farsite}
Adya, A., Bolosky, W.J., Castro, M., Cermak, G., Chaiken, R., Douceur, J.R.,
  Howell, J., Lorch, J.R., Theimer, M., Wattenhofer, R.: {FARSITE}: Federated,
  available, and reliable storage for an incompletely trusted environment. In:
  {OSDI}. {USENIX} Association (2002)

\bibitem{amir2006scaling}
Amir, Y., Danilov, C., Kirsch, J., Lane, J., Dolev, D., Nita{-}Rotaru, C.,
  Olsen, J., Zage, D.J.: Scaling {Byzantine} fault-tolerant replication to wide
  area networks. In: {DSN}. pp. 105--114. {IEEE} Computer Society (2006)

\bibitem{Bacho2023}
Bacho, R., Lenzen, C., Loss, J., Ochsenreither, S., Papachristoudis, D.:
  {GRandLine}: Adaptively secure {DKG} and randomness beacon with
  (log-)quadratic communication complexity. In: {CCS}. pp. 941--955. {ACM}
  (2024)

\bibitem{DBLP:journals/iacr/BandarupalliBBKR23}
Bandarupalli, A., Bhat, A., Bagchi, S., Kate, A., Reiter, M.K.: Random beacons
  in {Monte Carlo}: Efficient asynchronous random beacon \emph{without}
  threshold cryptography. In: {CCS}. pp. 2621--2635. {ACM} (2024)

\bibitem{BenOr83}
Ben{-}Or, M.: Another advantage of free choice: Completely asynchronous
  agreement protocols (extended abstract). In: {PODC}. pp. 27--30. {ACM} (1983)

\bibitem{BE03}
Ben{-}Or, M., El{-}Yaniv, R.: Resilient-optimal interactive consistency in
  constant time. Distributed Comput.  \textbf{16}(4),  249--262 (2003)

\bibitem{BKR94}
Ben{-}Or, M., Kelmer, B., Rabin, T.: Asynchronous secure computations with
  optimal resilience (extended abstract). In: {PODC}. pp. 183--192. {ACM}
  (1994)

\bibitem{bhat2021randpiper}
Bhat, A., Shrestha, N., Luo, Z., Kate, A., Nayak, K.: {RandPiper} -
  reconfiguration-friendly random beacons with quadratic communication. In:
  {CCS}. pp. 3502--3524. {ACM} (2021)

\bibitem{Bracha87}
Bracha, G.: Asynchronous {Byzantine} agreement protocols. Inf. Comput.
  \textbf{75}(2),  130--143 (1987)

\bibitem{BKM19}
Buchman, E., Kwon, J., Milosevic, Z.: The latest gossip on {BFT} consensus.
  arXiv:1807.04938v3 [cs.DC] (2019), \url{https://arxiv.org/abs/1807.04938v3}

\bibitem{book-cachin-guerraoui-rodrigues}
Cachin, C., Guerraoui, R., Rodrigues, L.E.T.: Introduction to Reliable and
  Secure Distributed Programming {(2.} ed.). Springer (2011)

\bibitem{CKPS01}
Cachin, C., Kursawe, K., Petzold, F., Shoup, V.: Secure and efficient
  asynchronous broadcast protocols. In: {CRYPTO}. pp. 524--541. Lecture Notes
  in Computer Science, Springer (2001)

\bibitem{CachinKS05}
Cachin, C., Kursawe, K., Shoup, V.: Random oracles in {Constantinople}:
  Practical asynchronous {Byzantine} agreement using cryptography. J. Cryptol.
  \textbf{18}(3),  219--246 (2005)

\bibitem{CT05}
Cachin, C., Tessaro, S.: Asynchronous verifiable information dispersal. In:
  {SRDS}. pp. 191--202. {IEEE} Computer Society (2005)

\bibitem{canetti1993fast}
Canetti, R., Rabin, T.: Fast asynchronous {Byzantine} agreement with optimal
  resilience. In: {STOC}. pp. 42--51. {ACM} (1993)

\bibitem{CL02}
Castro, M., Liskov, B.: Practical {Byzantine} fault tolerance and proactive
  recovery. {ACM} Trans. Comput. Syst.  \textbf{20}(4),  398--461 (2002)

\bibitem{catalano2013vector}
Catalano, D., Fiore, D.: Vector commitments and their applications. In: Public
  Key Cryptography. pp. 55--72. Lecture Notes in Computer Science, Springer
  (2013)

\bibitem{chen2024ociormvbanearoptimalerrorfreeasynchronous}
Chen, J.: {OciorMVBA}: Near-optimal error-free asynchronous {MVBA}.
  arXiv:2501.00214v1 [cs.CR] (2024), \url{https://arxiv.org/abs/2501.00214v1}

\bibitem{preonpdf}
Chen, M.S., Chen, Y.S., Cheng, C.M., Fu, S., Hong, W.C., Hsiang, J.H., Hu,
  S.T., Kuo, P.C., Lee, W.B., Liu, F.H., Thaler, J.: Preon: {zk-SNARK} based
  signature scheme (2023),
  \url{https://csrc.nist.gov/csrc/media/Projects/pqc-dig-sig/documents/round-1/spec-files/Preon-spec-web.pdf}

\bibitem{free_partial_sync}
Civit, P., Dzulfikar, M.A., Gilbert, S., Guerraoui, R., Komatovic, J.,
  Vidigueira, M., Zablotchi, I.: Partial synchrony for free: New upper bounds
  for {Byzantine} agreement. In: {SODA}. pp. 4227--4291. {SIAM} (2025)

\bibitem{DBLP:journals/corr/abs-2002-08765}
Crain, T.: Two more algorithms for randomized signature-free asynchronous
  binary {Byzantine} consensus with {$t < n/3$} and {$O(n^2)$} messages and
  {$O(1)$} round expected termination. arXiv:2002.08765v1 [cs.DC] (2020),
  \url{https://arxiv.org/abs/2002.08765v1}

\bibitem{CGL18}
Crain, T., Gramoli, V., Larrea, M., Raynal, M.: {DBFT}: Efficient leaderless
  {Byzantine} consensus and its application to blockchains. In: {NCA}.
  pp.~1--8. {IEEE} (2018)

\bibitem{DBLP:conf/ccs/DasDLM0S24}
Das, S., Duan, S., Liu, S., Momose, A., Ren, L., Shoup, V.: Asynchronous
  consensus without trusted setup or public-key cryptography. In: {CCS}. pp.
  3242--3256. {ACM} (2024)

\bibitem{das2021asynchronous}
Das, S., Xiang, Z., Ren, L.: Asynchronous data dissemination and its
  applications. In: {CCS}. pp. 2705--2721. {ACM} (2021)

\bibitem{eprint_2023_1196}
Das, S., Xiang, Z., Tomescu, A., Spiegelman, A., Pinkas, B., Ren, L.:
  Verifiable secret sharing simplified. In: {SP}. pp. 633--651. {IEEE} (2025)

\bibitem{DasYXMK022}
Das, S., Yurek, T., Xiang, Z., Miller, A., Kokoris{-}Kogias, L., Ren, L.:
  Practical asynchronous distributed key generation. In: {SP}. pp. 2518--2534.
  {IEEE} (2022)

\bibitem{DBLP:conf/tcc/DeligiosHL21}
Deligios, G., Hirt, M., Liu{-}Zhang, C.: Round-efficient {Byzantine} agreement
  and multi-party computation with asynchronous fallback. In: {TCC} {(1)}. pp.
  623--653. Lecture Notes in Computer Science, Springer (2021)

\bibitem{finmvba}
Duan, S., Wang, X., Zhang, H.: {FIN}: Practical signature-free asynchronous
  common subset in constant time. In: {CCS}. pp. 815--829. {ACM} (2023)

\bibitem{hmvba}
Feng, H., Lu, Z., Mai, T., Tang, Q.: Faster hash-based multi-valued validated
  asynchronous {Byzantine} agreement. Cryptology {ePrint} Archive, Paper
  2024/479 (2024), \url{https://eprint.iacr.org/2024/479}

\bibitem{flt24mvba}
Feng, H., Lu, Z., Tang, Q.: {$\widetilde{\mbox{O}}$ptimal} adaptively secure
  hash-based asynchronous common subset. Cryptology {ePrint} Archive, Paper
  2024/1710 (2024), \url{https://eprint.iacr.org/2024/1710}

\bibitem{fischer1985impossibility}
Fischer, M.J., Lynch, N.A., Paterson, M.: Impossibility of distributed
  consensus with one faulty process. J. {ACM}  \textbf{32}(2),  374--382 (1985)

\bibitem{fitzi2003efficient}
Fitzi, M., Garay, J.A.: Efficient player-optimal protocols for strong and
  differential consensus. In: {PODC}. pp. 211--220. {ACM} (2003)

\bibitem{DBLP:conf/eurocrypt/FitziGMR02}
Fitzi, M., Gisin, N., Maurer, U.M., von Rotz, O.: Unconditional {Byzantine}
  agreement and multi-party computation secure against dishonest minorities
  from scratch. In: {EUROCRYPT}. pp. 482--501. Lecture Notes in Computer
  Science, Springer (2002)

\bibitem{DBLP:conf/icdcs/GaoLLTXZ22}
Gao, Y., Lu, Y., Lu, Z., Tang, Q., Xu, J., Zhang, Z.: Efficient asynchronous
  {Byzantine} agreement without private setups. In: {ICDCS}. pp. 246--257.
  {IEEE} (2022)

\bibitem{DBLP:conf/crypto/GennaroIKR02}
Gennaro, R., Ishai, Y., Kushilevitz, E., Rabin, T.: On 2-round secure
  multiparty computation. In: {CRYPTO}. pp. 178--193. Lecture Notes in Computer
  Science, Springer (2002)

\bibitem{smvba}
Guo, B., Lu, Y., Lu, Z., Tang, Q., Xu, J., Zhang, Z.: Speeding {Dumbo}: Pushing
  asynchronous {BFT} closer to practice. In: {NDSS}. The Internet Society
  (2022)

\bibitem{huang2022byzantine}
Huang, S., Pettie, S., Zhu, L.: {Byzantine} agreement in polynomial time with
  near-optimal resilience. In: {STOC}. pp. 502--514. {ACM} (2022)

\bibitem{huang2023byzantine}
Huang, S., Pettie, S., Zhu, L.: {Byzantine} agreement with optimal resilience
  via statistical fraud detection. In: {SODA}. pp. 4335--4353. {SIAM} (2023)

\bibitem{DBLP:books/crc/KatzLindell2007}
Katz, J., Lindell, Y.: Introduction to Modern Cryptography. Chapman and
  Hall/CRC Press (2007)

\bibitem{kimmett2020improvement}
Kimmett, B.: Improvement and partial simulation of {King} \& {Saia}'s
  expected-polynomial-time Byzantine agreement algorithm. Ph.D. thesis,
  University of {Victoria}, {Canada} (2020)

\bibitem{king2016byzantine}
King, V., Saia, J.: {Byzantine} agreement in expected polynomial time. J. {ACM}
   \textbf{63}(2),  13:1--13:21 (2016)

\bibitem{king2018correction}
King, V., Saia, J.: Correction to {Byzantine} agreement in expected polynomial
  time, {JACM} 2016. arXiv:1812.10169v2 [cs.DC] (2019),
  \url{https://arxiv.org/abs/1812.10169v2}

\bibitem{Kokoris-KogiasM20}
Kokoris{-}Kogias, E., Malkhi, D., Spiegelman, A.: Asynchronous distributed key
  generation for computationally-secure randomness, consensus, and threshold
  signatures. In: {CCS}. pp. 1751--1767. {ACM} (2020)

\bibitem{fullversion}
Komatovic, J., Neu, J., Roughgarden, T.: Toward optimal-complexity hash-based
  asynchronous {MVBA} with optimal resilience. Cryptology {ePrint} Archive,
  Paper 2024/1682 (2024), \url{https://eprint.iacr.org/2024/1682}

\bibitem{Kotla2009}
Kotla, R., Alvisi, L., Dahlin, M., Clement, A., Wong, E.L.: {Zyzzyva}:
  Speculative {Byzantine} fault tolerance. {ACM} Trans. Comput. Syst.
  \textbf{27}(4),  7:1--7:39 (2009)

\bibitem{kotla2004high}
Kotla, R., Dahlin, M.: High throughput {Byzantine} fault tolerance. In: {DSN}.
  pp. 575--584. {IEEE} Computer Society (2004)

\bibitem{dumbomvba}
Lu, Y., Lu, Z., Tang, Q., Wang, G.: {Dumbo-MVBA}: Optimal multi-valued
  validated asynchronous {Byzantine} agreement, revisited. In: {PODC}. pp.
  129--138. {ACM} (2020)

\bibitem{melnyk2020byzantine}
Melnyk, D.: {Byzantine} Agreement on Representative Input Values Over Public
  Channels. Ph.D. thesis, {ETH} Zurich, Z{\"{u}}rich, Switzerland (2020)

\bibitem{merkle-tree-crypto87}
Merkle, R.C.: A digital signature based on a conventional encryption function.
  In: {CRYPTO}. pp. 369--378. Lecture Notes in Computer Science, Springer
  (1987)

\bibitem{MostefaouiMR15}
Most{\'{e}}faoui, A., Moumen, H., Raynal, M.: Signature-free asynchronous
  binary {Byzantine} consensus with {$t < n/3$}, {$O(n^2)$} messages, and
  {$O(1)$} expected time. J. {ACM}  \textbf{62}(4),  31:1--31:21 (2015)

\bibitem{DBLP:journals/acta/MostefaouiR17}
Most{\'{e}}faoui, A., Raynal, M.: Signature-free asynchronous {Byzantine}
  systems: from multivalued to binary consensus with {$t < n/3$}, {$O(n^2)$}
  messages, and constant time. Acta Informatica  \textbf{54}(5),  501--520
  (2017)

\bibitem{Nayak0SVX20}
Nayak, K., Ren, L., Shi, E., Vaidya, N.H., Xiang, Z.: Improved extension
  protocols for {Byzantine} broadcast and agreement. In: {DISC}. pp.
  28:1--28:17. LIPIcs, Schloss Dagstuhl - Leibniz-Zentrum f{\"{u}}r Informatik
  (2020)

\bibitem{Rabin83}
Rabin, M.O.: Randomized {Byzantine} generals. In: {FOCS}. pp. 403--409. {IEEE}
  Computer Society (1983)

\bibitem{reed1960}
Reed, I.S., Solomon, G.: Polynomial codes over certain finite fields. Journal
  of the Society for Industrial and Applied Mathematics  \textbf{8}(2),
  300--304 (1960). \doi{10.1137/0108018}

\bibitem{Souza22}
de~Souza, L.F., Kuznetsov, P., Tonkikh, A.: Distributed randomness from
  approximate agreement. In: {DISC}. pp. 24:1--24:21. LIPIcs, Schloss Dagstuhl
  - Leibniz-Zentrum f{\"{u}}r Informatik (2022)

\bibitem{sq}
Sui, X., Wang, X., Duan, S.: Signature-free atomic broadcast with optimal
  {$O(n^2)$} messages and {$O(1)$} expected time. In: {SP}. pp. 1547--1565.
  {IEEE} (2025)

\bibitem{pace}
Zhang, H., Duan, S.: {PACE}: Fully parallelizable {BFT} from reproposable
  {Byzantine} agreement. In: {CCS}. pp. 3151--3164. {ACM} (2022)

\end{thebibliography}

\clearpage

\section*{Organization of the Appendix}

We provide a full definition of cryptographic accumulators and compare the SMBA primitive with strong consensus in \Cref{section:preliminaries_full}.
We discuss the resilience of \comm and \commplus in \Cref{subsection:discussion}.
In \Cref{section:regular_mba}, we present the pseudocode and proof of our MBA algorithm, $\mbasimple$.
The pseudocode and proof of our SMBA algorithm, \smba, are provided in \Cref{section:strong_mba}.
Finally, we formally prove the correctness and complexity of \comm and \commplus in \Cref{section:reducer_proof} and \Cref{section:reducer_plus_proof}, respectively.

\appendix
\renewcommand{\theHsection}{\Alph{section}}

\section{Preliminaries: Omitted Definitions \& Discussions}
\label{section:preliminaries_full}

\noindent \textbf{Cryptographic accumulators.}
We adopt the definition of 
cryptographic accumulators 
from earlier works~\cite{bhat2021randpiper,Nayak0SVX20}.
A cryptographic accumulator scheme constructs an accumulation value for a set of values and produces a witness for each value in the set.
Given the accumulation value and a witness, any process can verify if a value is indeed in the set.
More formally, given a security parameter $\kappa$ and a set $\mathcal{D}$ of $n$ values $d_1, ..., d_n$, an accumulator has the following syntax:
\begin{compactitem}
    \item $\mathsf{Gen}(1^\kappa, n)$: Takes a parameter $\kappa$ in unary representation $1^\kappa$ and an accumulation threshold $n$ (an upper bound on the number of values that can be accumulated securely); returns a public accumulator key $ak$.

    \item $\mathsf{Eval}(ak, \mathcal{D})$: Takes an accumulator key $ak$ and a set of values $\mathcal{D}$ to be accumulated; returns an accumulation value $z$ for the set $\mathcal{D}$.

    \item $\mathsf{CreateWit}(ak, z, d_i, \mathcal{D})$: Takes an accumulator key $ak$, an accumulation value $z$ for $\mathcal{D}$, a value $d_i$ and a set of values $\mathcal{D}$; returns $\bot$ if $d_i \notin \mathcal{D}$, and a witness $w_i$ if $d_i \in \mathcal{D}$.

    \item $\mathsf{Verify}(ak, z, w_i, d_i)$: Takes an accumulator key $ak$, an accumulation value $z$ for $\mathcal{D}$, a witness $w_i$, and a value $d_i$; returns $\mathit{true}$ if $w_i$ is the witness for $d_i \in \mathcal{D}$, and $\mathit{false}$ otherwise.
\end{compactitem}
An accumulator scheme is secure if the adversary cannot produce a valid witness for a value $d_i$ that was not in the set $\mathcal{D}$ used to produce the accumulation value $z$,
i.e., for any accumulator key $ak \gets \mathsf{Gen}(1^{\kappa}, n)$, it is computationally infeasible to obtain $(\{d_1, ..., d_n\}, d', w')$ such that (1)~$d' \notin \{d_1, ..., d_n\}$, (2)~$z \leftarrow \mathsf{Eval}(ak, \{d_1, ..., d_n\})$, and (3)~$\mathsf{Verify}(ak, z, w', d') = \mathit{true}$. 
Throughout the paper, we refrain from explicitly mentioning the accumulator key $ak$ as we assume that the associated hash function is fixed.

\smallskip
\noindent \textbf{SMBA vs. strong consensus.}
Here, we point out subtle differences between the SMBA primitive (defined in \Cref{section:preliminaries}) and the well-known strong consensus primitive~\cite{fitzi2003efficient}.
Like SMBA, strong consensus requires agreement on the proposal of a correct process.
Importantly, the strong consensus primitive provides strong validity, agreement, integrity, and termination only if correct processes propose values from a domain of possible inputs
where
(1)~the domain is known to the protocol,
and (2)~the domain contains at most two values.
Fitzi and Garay~\cite{fitzi2003efficient} prove that strong consensus (with two different values) can be solved in asynchrony if and only if $n > 3t$.
In contrast, our SMBA primitive does not require knowledge of the domain of correct processes' proposals,
and provides agreement, integrity, and termination irrespective of the size of the set comprising all correct processes' proposals,
while it provides 
strong validity if it so happens that the set comprising all correct processes' proposals contains at most two values.
In other words, SMBA ``knows'' that the \emph{size} of the domain of correct inputs is at most two, but does not ``know'' the \emph{values} in said domain, and only its validity is conditional on this knowledge, not its agreement, integrity, or termination.

Given that the SMBA problem is harder than the strong consensus problem, our \smba algorithm naturally solves the strong consensus problem.
Moreover, we underline that our \smba algorithm can trivially be adapted to solve the strong consensus problem even with \emph{optimal} $t < \frac{1}{3} n$ resilience (see \Cref{section:strong_mba} for more details).\footnote{Recall that the resilience of the \smba algorithm when solving the SMBA problem is $t < \frac{1}{4} n$.}
Additionally, \smba can be easily modified to allow us to solve the strong consensus problem with $x$ different proposals, for any $x \in O(1)$, with optimal resilience of $n > (x + 1)t$~\cite{fitzi2003efficient}, no employed cryptography, and optimal complexity.
\section{Discussion on \comm's and \commplus's Resilience} \label{subsection:discussion}

This section discusses the resilience limits of our algorithms.

\smallskip
\noindent \textbf{\comm below $n = 4t + 1$.}
Recall that one of the crucial ingredients of the \comm algorithm is the use of the SMBA primitive that guarantees agreement on a digest proposed by a correct process, as long as at most two different digests are proposed by correct processes.
To ensure that the ``two-different-proposals'' precondition is met in a good iteration $k$, \comm enforces two key constraints: (1) any correct process commits at most two digests in iteration $k$, and (2) there are at most three different digests committed across all correct processes (i.e., $|\mathsf{committed}(k)| \leq 3$).
The two aforementioned constraints guarantee that correct processes will eventually reach an invocation of the SMBA primitive where they all propose and decide $z^{\star}(k)$. 
In brief, \comm cannot improve upon the $t < \frac{1}{4} n$ resilience threshold because such a ``$z^{\star}(k)$-agreeing'' invocation of the SMBA primitive may never be reached in a good iteration $k$ assuming $n = 4t$.

\smallskip
\noindent\emph{Observation 1: Each correct process can have more than two candidates.}
First, we note that when $n = 4t$, each correct process may have up to three different candidates after the exchange of \textsc{stored} messages in good iteration $k$.
Indeed, as $n = 4t$, $\mathsf{leader}(k)$ has successfully disseminated the digest $z^{\star}(k)$ of its valid proposal $v^{\star}(k)$ to at least $(n - t) - t = 2t$ correct processes.
Hence, each correct process is guaranteed to hear $z^{\star}(k)$ only from $(n - t) + (n - 2t) - n = n - 3t =  t$ processes once it receives $n - t = 3t$ \textsc{stored} messages in iteration $k$.
Thus, to ensure that every correct process marks the ``good'' digest $z^{\star}(k)$ as a candidate, the ``candidate-threshold'' is $n - 3t = t$.
As each correct process waits for $n - t = 3t$ \textsc{stored} messages before determining its candidates, each correct process may end up with up to $\frac{3t}{t} = 3$ different candidates.

\smallskip
\noindent\emph{Observation 2: The number of different candidates across correct processes may be greater than three.}
To ensure that every correct process commits the ``good'' digest $z^{\star}(k)$, each correct process commits its candidate $z$ if it hears $(n - t) - t = 2t$ \textsc{suggest} messages for $z$.
The fact that each correct process can suggest two different adversarial (non-$z^{\star}(k)$) digests implies that the number of different candidates ``surviving'' the suggestion phase across all correct processes might be greater than three: $|\mathsf{committed}(k)| > 3$.
Concretely, it can be shown that $|\mathsf{committed}(k)| \leq 7$ with $n = 4t$.

\smallskip
\noindent\emph{Observation 3: \comm among $n = 4t$ would require an impossible variant of the SMBA primitive.}
Let us demonstrate a problematic scenario that may arise given $|\mathsf{committed}(k)| \leq 7$.
Suppose $\mathsf{committed}(k) = \{ z_1, z_2, z^{\star}(k), z_3, z_4 \}$ with $z_1 < z_2 < z^{\star}(k) < z_3 < z_4$ according to the lexicographic order.
Let us partition all correct processes into three non-empty and disjoint sets $\mathcal{S}_1$, $\mathcal{S}_2$, and $\mathcal{S}_3$.
The following ``spread'' of the committed digests is possible:
\begin{compactitem}
    \item Correct processes in the $\mathcal{S}_1$ set commit $[z_1, z_2, z^{\star}(k)]$.

    \item Correct processes in the $\mathcal{S}_2$ set commit $[z_2, z^{\star}(k), z_3]$.

    \item Correct processes in the $\mathcal{S}_3$ set commit $[z^{\star}(k), z_3, z_4]$.
\end{compactitem}
Thus, for any $c \in \{ 1, 2, 3 \}$, there exist $3 > 2$ different $c$-committed digests.

For \comm to deal with the scenario above (when $n = 4t$), we would need to develop the SMBA primitive with strictly stronger properties than those required for $n = 4t + 1$.
Specifically, for correct processes to learn about the proposals of other correct processes---the core principle underlying \comm{}---the primitive must satisfy the following guarantee: if up to $d > 2$ different digests are proposed by correct processes, then the decided digest was proposed by a correct process.
(Without this guarantee, correct processes might keep deciding ``useless'' digests not held by any correct process.)
Unfortunately, Fitzi and Garay~\cite{fitzi2003efficient} prove that, among $n = 4t$ processes, such a primitive cannot be implemented in asynchrony even for $d = 3$, indicating that \comm's structure is not suitable for $n < 4t + 1$.

\smallskip
\noindent \textbf{\commplus with optimal $n = 3t + 1$.}
We conclude the section by explaining why \commplus cannot achieve optimal resilience.
One reason is that, when $t < \frac{1}{3} n$, \commplus cannot maintain its quasi-quadratic expected bit complexity.
To ensure that the ``good'' digest $z^{\star}(k)$ is identified as a candidate by each correct process after receiving $n - t = 2t + 1$ \textsc{stored} messages in a good iteration $k$, the ``candidate threshold'' must be set at $(n - t) + (n - 2t) - n = 1$.
Therefore, each correct process could have \emph{linearly} many candidates after collecting the \textsc{stored} messages.
As a result, disseminating these candidates via the \textsc{suggest} messages would require $O(n^3 \kappa)$ exchanged bits (digests are of size $\kappa$ bits), thereby violating the desired quasi-quadratic upper bound.
Incorporating the $\epsilon \cdot t$ gap into the resilience of \commplus allows each correct process to have only \emph{constantly} many candidates, thus ensuring that the exchange of \textsc{suggest} messages incurs $O(n^2\kappa)$ bits.

\section{$\mbasimple$: Pseudocode \& Proof} \label{section:regular_mba}

This section presents $\mbasimple$, our adaptively-secure asynchronous MBA protocol employed by both \comm and \commplus.
$\mbasimple$ exchanges $O(n^2)$ messages and $O(n \ell + n^2 \kappa \log n)$ bits, and terminates in $O(1)$ time.
Moreover, $\mbasimple$ tolerates up to $t < \frac{1}{3}n$ failures.
Recall that the specification of the MBA primitive is given in \Cref{mod:async_mba}.

\subsection{Pseudocode}

The pseudocode of the $\mbasimple$ algorithm is given in \Cref{algorithm:mbasimple}.
Before discussing $\mbasimple$'s pseudocode, we formally introduce the graded consensus and binary Byzantine agreement primitives employed in $\mbasimple$.

\smallskip
\noindent \textbf{Graded consensus.}
The formal specification of the primitive is given in \Cref{mod:graded_consensus}.

\begin{module}[tbp]
\caption{Graded Consensus}
\label{mod:graded_consensus}
\scriptsize
\begin{algorithmic}[1]

\Statex \textbf{Associated values:}
\begin{compactitem} [-]
    \item set $\valuemba$ of $\ell$-bit values
\end{compactitem}

\smallskip
\Statex \textbf{Events:}
\begin{compactitem}[-]
    \item \emph{input} $\mathsf{propose}(v \in \valuemba)$: a process proposes value $v$.

    \item \emph{output} $\mathsf{decide}(v' \in \valuemba, g' \in \{0, 1\})$: a process decides value $v'$ with grade $g'$.
\end{compactitem}

\smallskip 
\Statex \textbf{Assumed behavior:} 
\begin{compactitem}[-]
    \item Every correct process proposes exactly once.
\end{compactitem}

\smallskip 
\Statex \textbf{Properties:} \BlueComment{ensured only if correct processes follow the behavior stated above}
\begin{compactitem}[-]
    \item \emph{Strong unanimity:} If all correct processes propose the same value $v \in \valuemba$ and a correct process decides a pair $(v' \in \valuemba, g' \in \{ 0, 1\})$, then $v' = v$ and $g' = 1$.

    \item \emph{Consistency:} If any correct process decides a pair $(v \in \valuemba, 1)$, then no correct process decides any pair $(v' \in \valuemba, \cdot)$ with $v' \neq v$.

    \item \emph{Justification:} If any correct process decides a pair $(v' \in \valuemba, \cdot)$, then $v'$ is proposed by a correct process.

    \item \emph{Integrity:} No correct process decides more than once.

    \item \emph{Termination:} All correct processes eventually decide.
\end{compactitem}
\end{algorithmic}
\end{module}

In our $\mbasimple$ algorithm, we rely on a \emph{deterministic} implementation~\cite{free_partial_sync} of the graded consensus primitive that exchanges $O(n^2)$ messages and $O(n \ell + n^2 \kappa \log n)$ bits, and terminates in $O(1)$ time.
This implementation tolerates up to $t < \frac{1}{3} n$ faulty processes and relies solely on a collision-resistant hash function.
Note that, as the implementation is deterministic, it is inherently secure against an adaptive adversary.

\smallskip
\noindent \textbf{Binary Byzantine agreement.}
The binary agreement primitive is similar to the MBA primitive (see \Cref{mod:async_mba}).
However, there are two major differences: (1) correct processes only propose $0$ and $1$, and (2) correct processes cannot decide a special value, i.e., only $0$ and $1$ can be decided from BA.
The formal specification is given in \Cref{mod:async_ba}.

\begin{module}[tbp]
\caption{Binary Byzantine Agreement}
\label{mod:async_ba}
\scriptsize
\begin{algorithmic}[1]

\Statex \textbf{Events:}
\begin{compactitem}[-]
    \item \emph{input} $\mathsf{propose}(v \in \{ 0, 1 \})$: a process proposes binary value $v$.

    \item \emph{output} $\mathsf{decide}(v' \in \{ 0, 1 \})$: a process decides binary value $v'$.
\end{compactitem}

\smallskip 
\Statex \textbf{Assumed behavior:} 
\begin{compactitem}[-]
    \item Every correct process proposes exactly once.
\end{compactitem}

\smallskip 
\Statex \textbf{Properties:} \BlueComment{ensured only if correct processes follow the behavior stated above}
\begin{compactitem}[-]
    \item \emph{Strong unanimity:} If all correct processes propose the same value $v \in \{ 0, 1 \}$ and a correct process decides a value $v' \in \{ 0, 1 \}$, then $v' = v$.
    
    \item \emph{Agreement:} No two correct processes decide different values. 

    \item \emph{Integrity:} No correct process decides more than once.

    \item \emph{Termination:} All correct processes eventually decide.
\end{compactitem}
\end{algorithmic}
\end{module}

In our $\mbasimple$ algorithm, we utilize a binary agreement algorithm proposed by Mostéfaoui and Raynal~\cite{MostefaouiMR15}; this algorithm is safe and live against an adaptive adversary, exchanges $O(n^2)$ messages and bits in expectation, terminates in $O(1)$ time in expectation, and relies on no cryptography.

\smallskip
\noindent \textbf{Pseudocode description.}
When a correct process $p_i$ proposes its value $v_i$ (line~\ref{line:regular_mba_start}), it forwards the value to the $\mathcal{GC}$ graded consensus algorithm (line~\ref{line:regular_mba_gc}).
Then, process $p_i$ proposes the grade decided from $\mathcal{GC}$ to the $\mathcal{BA}$ binary Byzantine agreement algorithm (line~\ref{line:regular_mba_ba_1}).
If the bit decided from $\mathcal{BA}$ is $1$, then process $p_i$ decides the value decided from $\mathcal{GC}$ (line~\ref{line:regular_mba_decide_1}).
Otherwise, process $p_i$ decides $\botmba$ (line~\ref{line:regular_mba_decide_bot}).

\begin{algorithm}[tbp]
\caption{$\mbasimple$: Pseudocode (for process $p_i$)}
\label{algorithm:mbasimple}
\begin{algorithmic}[1]
\scriptsize

\State \textbf{Uses:}
    \State \hskip2em \textcolor{jnSUDigitalRedLight}{\(\triangleright\) \cite{free_partial_sync} exchanges $O(n^2)$ messages and $O(n\ell + n^2 \kappa \log n)$ bits and terminates in $O(1)$ time}
    \State \hskip2em Graded consensus algorithm~\cite{free_partial_sync} on the $\valuemba$ set, \textbf{instance} $\mathcal{GC}$

    \smallskip
    \State \hskip2em \textcolor{jnSUDigitalRedLight}{\(\triangleright\) \cite{MostefaouiMR15} exchanges $O(n^2)$ messages and $O(n^2)$ bits and terminates in $O(1)$ time}
    \State \hskip2em Binary Byzantine agreement algorithm~\cite{MostefaouiMR15}, \textbf{instance} $\mathcal{BA}$

\medskip
\State \textbf{upon} $\mathsf{propose}(\valuemba \text{ } v_i)$:\label{line:regular_mba_start} \BlueComment{start of the algorithm}
    \State \hskip2em $(\valuemba, \{ 0, 1 \})$ $(\mathit{adopted\_value}, g) \gets \mathcal{GC}.\mathsf{propose}(v_i)$ \label{line:regular_mba_gc}

    \State \hskip2em $\{0, 1\}$ $g' \gets \mathcal{BA}.\mathsf{propose}(g)$ \label{line:regular_mba_ba_1}

    \State \hskip2em \textbf{if} $g' = 1$:
        \State \hskip4em \textbf{trigger} $\mathsf{decide}(\mathit{adopted\_value})$ \label{line:regular_mba_decide_1}
    \State \hskip2em \textbf{else:}
        \State \hskip4em \textbf{trigger} $\mathsf{decide}(\botmba)$ \label{line:regular_mba_decide_bot}

\end{algorithmic}
\end{algorithm}

\subsection{Proof of Correctness \& Complexity}

This section proves $\mbasimple$'s correctness and complexity.

\smallskip
\noindent \textbf{Proof of correctness.}
To prove the correctness of $\mbasimple$, we prove the following lemma.

\begin{lemma} [$\mbasimple$ is correct]
Given $t < \frac{1}{3} n$ and the existence of a collision-resistant hash function, $\mbasimple$ (see \Cref{algorithm:mbasimple}) is a correct implementation of the MBA primitive in the presence of a computationally bounded adversary.
\end{lemma}

In the rest of the section, we say that a correct process $p_i$ \emph{adopts} a value $v$ if and only if $\mathit{adopted\_value} = v$ at process $p_i$ when process $p_i$ reaches line~\ref{line:regular_mba_ba_1}.
First, we prove that if any correct process proposes $1$ to the $\mathcal{BA}$ instance of the BA primitive, then no two correct processes adopt different values.

\begin{myclaim} \label{claim:mbasimple_different_adopts}
If any correct process proposes $1$ to $\mathcal{BA}$, then no two correct processes adopt different values.
\end{myclaim}
\begin{proof}
Let $p_i$ be any correct process that proposes $1$ to $\mathcal{BA}$.
Hence, $p_i$ decides with grade $1$ from $\mathcal{GC}$.
The statement of the claim then holds due to the consistency property of $\mathcal{GC}$.
\end{proof}

We are ready to prove $\mbasimple$'s strong unanimity.

\begin{proposition} [$\mbasimple$ satisfies strong unanimity]
Given $t < \frac{1}{3} n$ and the existence of a collision-resistant hash function, $\mbasimple$ (see \Cref{algorithm:mbasimple}) satisfies strong unanimity in the presence of a computationally bounded adversary.
\end{proposition}
\begin{proof}
Suppose all correct processes propose the same value $v$.
Hence, each correct process $p_i$ decides $(v, 1)$ from $\mathcal{GC}$ (due to its strong unanimity property) and adopts $v$.
Furthermore, each correct process proposes $1$ to $\mathcal{BA}$, which ensures that all correct processes decide $1$ from $\mathcal{BA}$ (due to its strong unanimity property).
Finally, each correct process decides $v$.
\end{proof}

Next, we prove $\mbasimple$'s agreement.

\begin{proposition} [$\mbasimple$ satisfies agreement]
Given $t < \frac{1}{3} n$ and the existence of a collision-resistant hash function, $\mbasimple$ (see \Cref{algorithm:mbasimple}) satisfies agreement in the presence of a computationally bounded adversary.
\end{proposition}
\begin{proof}
Let $g^{\star}$ denote the binary value decided from $\mathcal{BA}$.
We distinguish two cases:
\begin{compactitem}
    \item Let $g^{\star} = 0$.
    In this case, all correct processes that decide do so with $\botmba$.
    The agreement property is satisfied in this case.

    \item Let $g^{\star} = 1$.
    Therefore, every correct process decides its adopted value.
    Moreover, the strong unanimity property of $\mathcal{BA}$ ensures that a correct process proposes $1$ to $\mathcal{BA}$.
    Therefore, the agreement is satisfied due to \Cref{claim:mbasimple_different_adopts}.
\end{compactitem}
As agreement is ensured in any possible case, the proof is concluded.
\end{proof}

We proceed by proving $\mbasimple$'s justification.

\begin{proposition} [$\mbasimple$ satisfies justification]
Given $t < \frac{1}{3} n$ and the existence of a collision-resistant hash function, $\mbasimple$ (see \Cref{algorithm:mbasimple}) satisfies justification in the presence of a computationally bounded adversary.
\end{proposition}
\begin{proof}
Suppose a correct process $p_i$ decides a non-$\botmba$ value $v$.
Hence, process $p_i$ adopts $v$.
Finally, the justification property of $\mathcal{GC}$ guarantees that $v$ is proposed to $\mbasimple$ by a correct process, which concludes the proof.
\end{proof}

Next, we prove the integrity property.

\begin{proposition} [$\mbasimple$ satisfies integrity]
Given $t < \frac{1}{3} n$ and the existence of a collision-resistant hash function, $\mbasimple$ (see \Cref{algorithm:mbasimple}) satisfies integrity in the presence of a computationally bounded adversary.
\end{proposition}
\begin{proof}
The integrity property is satisfied due to the integrity property of $\mathcal{BA}$.
\end{proof}

Finally, we prove $\mbasimple$'s termination.

\begin{proposition} [$\mbasimple$ satisfies termination]
Given $t < \frac{1}{3} n$ and the existence of a collision-resistant hash function, $\mbasimple$ (see \Cref{algorithm:mbasimple}) satisfies termination in the presence of a computationally bounded adversary.
\end{proposition}
\begin{proof}
The termination property is satisfied as both $\mathcal{GC}$ and $\mathcal{BA}$ terminate.
\end{proof}

\smallskip
\noindent \textbf{Proof of complexity.}
We now proceed to prove $\mbasimple$'s complexity.

\begin{lemma} [$\mbasimple$'s expected complexity] 
Given $t < \frac{1}{3}n$ and the existence of a collision-resistant hash function, the following holds for $\mbasimple$ (see \Cref{algorithm:mbasimple}) in the presence of a computationally bounded adversary:
\begin{compactitem}
    \item The expected message complexity is $O(n^2)$.

    \item The expected bit complexity is $O(n\ell + n^2 \kappa \log n)$.

    \item The expected time complexity is $O(1)$.
\end{compactitem} 
\end{lemma}
\begin{proof}
As the worst-case message complexity of $\mathcal{GC}$ is $O(n^2)$ and the expected message complexity of $\mathcal{BA}$ is $O(n^2)$, the expected message complexity of $\mbasimple$ is $O(n^2)$.
Given that correct processes send $O(n\ell + n^2 \kappa \log n)$ bits in $\mathcal{GC}$ and $O(n^2)$ bits in $\mathcal{BA}$, the expected bit complexity of $\mbasimple$ is indeed $O(n\ell + n^2 \kappa \log n)$.
Finally, as both $\mathcal{GC}$ and $\mathcal{BA}$ terminate in $O(1)$ time in expectation, the expected time complexity of $\mbasimple$ is $O(1)$.
\end{proof}

\section{\smba: Pseudocode \& Proof} \label{section:strong_mba}

This section presents \smba, our adaptively-secure asynchronous SMBA protocol employed in \comm.
\smba exchanges $O(n^2)$ messages and $O(n^2 \kappa)$ bits, and terminates in $O(1)$ time.
The specification of the SMBA primitive can be found in \Cref{mod:async_smba}.

\subsection{Collective Reliable Broadcast (CRB): Pseudocode \& Proof} \label{section:crb}

First, we introduce the CRB primitive that plays a major role in our \smba algorithm.

\smallskip
\noindent \textbf{Primitive definition.}
The CRB primitive closely resembles the reliable broadcast primitive described in \cite{book-cachin-guerraoui-rodrigues}.  
Its formal specification is provided in \Cref{mod:async_crb}.

\begin{module}[tbp]
\caption{CRB}
\label{mod:async_crb}
\scriptsize
\begin{algorithmic}[1]

\Statex \textbf{Associated values:}
\begin{compactitem} [-]
    \item set $\valuedigest$ of $O(\kappa)$-bit digests

    \item special value $\botcrb \notin \valuedigest$
\end{compactitem}

\smallskip
\Statex \textbf{Events:}
\begin{compactitem}[-]
    \item \emph{input} $\mathsf{broadcast}(z \in \valuedigest)$: a process broadcasts digest $z$.

    \item \emph{output} $\mathsf{deliver}(z' \in \valuedigest \cup \{ \botcrb \})$: a process delivers digest $z'$ or the special value $\botcrb$.
\end{compactitem}

\smallskip 
\Statex \textbf{Assumed behavior:} 
\begin{compactitem}[-]
    \item Every correct process broadcasts exactly once.

    \item Only $O(1)$ different digests are broadcast by correct processes.
\end{compactitem}

\smallskip 
\Statex \textbf{Properties:} \BlueComment{ensured only if correct processes follow the behavior stated above}
\begin{compactitem}[-]
    \item \emph{Validity:} If up to two different digests are broadcast by correct processes, no correct process delivers the special value $\botcrb$.

    \item \emph{Justification:} If any correct process delivers a digest $z' \in \valuedigest$ ($z' \neq \botcrb$), then $z'$ is broadcast by a correct process.
    
    \item \emph{Integrity:} No correct process delivers unless it has previously broadcast.

    \item \emph{Termination:} All correct processes deliver at least once.

    \item \emph{Totality:} If any correct process delivers $z \in \valuedigest \cup \{ \botcrb \}$, then all correct processes eventually deliver $z$.
\end{compactitem}
\end{algorithmic}
\end{module}

We emphasize that, per \Cref{mod:async_crb}, any correct process may deliver multiple times.
Moreover, the termination property ensures that each correct process delivers at least once.

\subsubsection{Pseudocode}
The pseudocode of \crb, our CRB algorithm, can be found in \Cref{algorithm:crb}.
Our \crb algorithm exchanges $O(n^2)$ messages and $O(n^2 \kappa)$ bits, and terminates in $O(1)$ time.
Moreover, the \crb algorithm tolerates up to $t < \frac{1}{4} n$ faulty processes.
We underline that the \crb algorithm is heavily inspired by Bracha's reliable broadcast algorithm~\cite{Bracha87}.

\smallskip
\noindent \textbf{Pseudocode description.}
We describe the pseudocode of the \crb algorithm from the perspective of a correct process $p_i$.
Upon broadcasting its digest $z_i$ (line~\ref{line:crb-broadcast}), process $p_i$ broadcasts an \textsc{init} message for $z_i$ (line~\ref{line:crb-init-broadcast}).
Once process $p_i$ receives the same digest $z$ in (at least) $t + 1$ \textsc{init} messages (line~\ref{line:receive_init_rule}), which proves that $z$ is broadcast by a correct process, process $p_i$ broadcasts an \textsc{echo} message for $z$ (line~\ref{line:crb_broadcast_echo}).
Similarly, once process $p_i$ receives the same digest $z$ in (at least) $2t + 1$ \textsc{echo} (line~\ref{line:crb_quorum_echo_rule}) or $t + 1$ \textsc{ready} messages (line~\ref{line:crb_plurality_ready}), process $p_i$ broadcasts a \textsc{ready} message for $z$ (line~\ref{line:crb_broadcast_ready_1} or line~\ref{line:crb_broadcast_ready_2}).
Finally, once process $p_i$ receives $2t + 1$ \textsc{ready} messages for the same digest (line~\ref{line:crb_receive_quorum_ready}), it delivers the digest (line~\ref{line:crb_deliver}).

Upon realizing that there are at least three different digests broadcast by correct processes (line~\ref{line:send_broken_rule}), process $p_i$ broadcasts a \textsc{broken} message (line~\ref{line:crb_broadcast_broken_2}).
Similarly, upon receiving $t + 1$ \textsc{broken} messages (line~\ref{line:crb-broken-plurality}), process $p_i$ rebroadcasts the message (line~\ref{line:crb_broadcast_broken_1}).
Once process $p_i$ receives $2t + 1$ \textsc{broken} messages (line~\ref{line:crb_deliver_def_rule}), it delivers the special value $\botcrb$ (line~\ref{line:crb_deliver_def}).
 
\begin{algorithm}[tbp]
\caption{\crb: Pseudocode (for process $p_i$)}
\label{algorithm:crb}
\begin{algorithmic}[1]
\scriptsize

\State \textbf{Rules:}
    \State \hskip2em - Only one \textsc{init} message is processed per process. \label{line:crb_rule_1}
    \State \hskip2em - At most one \textsc{echo} message is broadcast per digest.
    \State \hskip2em - At most one \textsc{ready} message is broadcast per digest.
    \State \hskip2em - At most one \textsc{broken} message is broadcast.

    \State \hskip2em - No computational steps are taken unless a digest was previously broadcast. \label{line:crb_rule_integrity}

\medskip
\State \textbf{Local variables:}
    \State \hskip2em $\valuedigest$ $z_i \gets p_i$'s broadcast digest
    \State \hskip2em $\mathsf{Map}(\valuedigest \to \mathsf{Integer})$ $\mathit{num}_i \gets \{0, \text{ for every $z \in \valuedigest$} \}$

\medskip
\State \textbf{Local functions:}
    \State \hskip2em - $\mathsf{distinct()} = |\{ z \in \valuedigest \,|\, \mathit{num}_i[z] > 0 \} |$.
    \State \hskip2em - $\mathsf{eliminated()}$:
    \State \hskip4em (1) Let $Z = \{ z \in \valuedigest \,|\, \mathit{num}_i[z] > 0 \}$.
    \State \hskip4em (2) Sort $Z$ in the ascending order according to the $\mathit{num}_i$ map.\label{line:crb_sorted_list}
    \State \hskip4em (3) Return the greatest integer $x \in \mathbb{N}_{\geq 0}$ such that $\mathit{num}_i[Z[1]] + ... + \mathit{num}_i[Z[x]] \leq t$.

\medskip
\State \textbf{upon} $\mathsf{broadcast}(\valuedigest \text{ } z_i)$: \BlueComment{start of the algorithm} \label{line:crb-broadcast}
    \State \hskip2em \textbf{broadcast} $\langle \textsc{init}, z_i \rangle$ \label{line:crb-init-broadcast}

\medskip
\State \textbf{upon} exists $\valuedigest$ $z$ such that $\langle \textsc{init}, z \rangle$ is received from $t + 1$ processes:\label{line:receive_init_rule}
    \State \hskip2em \textbf{broadcast} $\langle \textsc{echo}, z \rangle$\label{line:crb_broadcast_echo}

\medskip
\State \textbf{upon} exists $\valuedigest$ $z$ such that $\langle \textsc{echo}, z \rangle$ is received from $2t + 1$ processes:\label{line:crb_quorum_echo_rule}
    \State \hskip2em \textbf{broadcast} $\langle \textsc{ready}, z \rangle$\label{line:crb_broadcast_ready_1}

\medskip
\State \textbf{upon} exists $\valuedigest$ $z$ such that $\langle \textsc{ready}, z \rangle$ is received from $t + 1$ processes:\label{line:crb_plurality_ready}
    \State \hskip2em Broadcast $\langle \textsc{ready}, z \rangle$\label{line:crb_broadcast_ready_2}

\medskip
\State \textbf{upon} exists $\valuedigest$ $z$ such that $\langle \textsc{ready}, z \rangle$ is received from $2t + 1$ processes:\label{line:crb_receive_quorum_ready}
    \State \hskip2em \textbf{trigger} $\mathsf{deliver}(z)$\label{line:crb_deliver}

\medskip
\State \textbf{upon} $\langle \textsc{broken} \rangle$ is received from $t + 1$ processes: \label{line:crb-broken-plurality}
    \State \hskip2em \textbf{broadcast} $\langle \textsc{broken} \rangle$ \label{line:crb_broadcast_broken_1}

\medskip
\State \textbf{upon} $\langle \textsc{broken} \rangle$ is received from $2t + 1$ processes:\label{line:crb_deliver_def_rule}
    \State \hskip2em \textbf{trigger} $\mathsf{deliver}(\botcrb)$\label{line:crb_deliver_def}

\medskip
\State \textbf{upon} receiving $\langle \textsc{init}, \valuedigest \text{ } z \rangle$:
    \State \hskip2em $\mathit{num}_i[z] \gets \mathit{num}_i[z] + 1$

\medskip
\State \textbf{upon} (1) $\mathsf{distinct}() - \mathsf{eliminated()} \geq 3$, and (2) $\geq n - t \geq 3t + 1$ \textsc{init} messages are received:\label{line:send_broken_rule}
    \State \hskip2em \textbf{broadcast} $\langle \textsc{broken} \rangle$\label{line:crb_broadcast_broken_2}

\end{algorithmic}
\end{algorithm}

\subsubsection{Proof of Correctness \& Complexity}

This subsection proves \crb's correctness and complexity.

\smallskip
\noindent \textbf{Proof of correctness.}
To prove \crb's correctness, we prove the following lemma.

\begin{lemma} [\crb is correct]
Given $t < \frac{1}{4}n$, \crb (see \Cref{algorithm:crb}) is a correct implementation of the CRB primitive in the presence of a computationally unbounded adversary.
\end{lemma}

In the rest of the proof, let $\mathcal{B}_C$ denote the set of all digests broadcasts by correct processes.
We start by proving that \crb satisfies the validity property.
To this end, we show that if there exists a correct process that broadcasts a \textsc{broken} message at line~\ref{line:crb_broadcast_broken_2}, then $|\mathcal{B}_C| \geq 3$.

\begin{myclaim} \label{claim:crb_validity_crucial}
If there exists a correct process that broadcasts a \textsc{broken} message at line~\ref{line:crb_broadcast_broken_2}, then $|\mathcal{B}_C| \geq 3$.
\end{myclaim}
\begin{proof}
Let $p_i$ be any correct process that broadcasts a \textsc{broken} message at line~\ref{line:crb_broadcast_broken_2}.
By contradiction, suppose $|\mathcal{B}_C| \leq 2$.
Let $\mathit{init}$ denote the set of all \textsc{init} messages received by $p_i$ prior to broadcasting the \textsc{broken} message at line~\ref{line:crb_broadcast_broken_2}.
Let us define a set $\mathcal{I}$:
\begin{equation*}
    \mathcal{I} = \{ (p_j, z) \,|\, \text{$ \exists m = \langle \textsc{init}, z \rangle: m \in \mathit{init}$ $\land$ $p_j$ is the sender of $m$}\}.
\end{equation*}
Note that $|\mathcal{I}| \geq n - t \geq 3t + 1$ (due to the second condition of the rule at line~\ref{line:send_broken_rule}).
We now define a set of digests $\mathit{correct\_digs}$ in the following way:
\begin{equation*}
    \mathit{correct\_digs} = \{ z \,|\, \exists (p_j, z) \in \mathcal{I}: \text{$p_j$ is a correct process} \}.
\end{equation*}
Let $X = |\mathit{correct\_digs}|$.
Note that $X \in \{1, 2\}$ (as $|\mathcal{B}_C| \leq 2$, $|\mathcal{I}| \geq 3t + 1$ and there are at most $t$ faulty processes).
We also define a set of digests $\mathit{faulty\_digs}$:
\begin{equation*}
    \mathit{faulty\_digs} = \{ z \,|\, \exists(p_j, z) \in \mathcal{I}: \text{$p_j$ is a faulty process}\}.
\end{equation*}
Note that $\mathsf{distinct()} = |\mathit{correct\_digs} \cup \mathit{faulty\_digs}|$.
Let $F$ be defined as 
\begin{equation*}
    F = |\{ (p_j, z) \,|\, (p_j, z) \in \mathcal{I} \text{ $\land$ $p_j$ is a faulty process} \}|.
\end{equation*}
Note that $F \leq t$ as process $p_i$ accepts at most one \textsc{init} message per process (line~\ref{line:crb_rule_1}) and there are up to $t$ faulty processes.

Let $Z$ denote the sorted list of digests constructed by the $\mathsf{eliminated()}$ function (line~\ref{line:crb_sorted_list}).
We say that a digest $z$ is \emph{eliminated} if and only if (1) $z = Z[i]$, and (2) $i \in [1, \mathsf{eliminated()}]$.
(If $\mathsf{eliminated()} = 0$, no digest is eliminated.)

As process $p_i$ sends the \textsc{broken} message at line~\ref{line:crb_broadcast_broken_2}, $\mathsf{distinct()} - \mathsf{eliminated()} \geq 3$ at process $p_i$ (line~\ref{line:send_broken_rule}).
Therefore, there are at least three non-eliminated digests.
We distinguish two possible scenarios:
\begin{compactitem}
    \item Let $X = 1$.
    Let $\mathit{correct\_digs} = \{ z \}$, for some digest $z$.
    Note that $\mathit{num}_i[z] \geq 2t + 1$ (as there are at least $2t + 1$ messages in $\mathit{init}$ that are received from correct processes).
    Therefore, $z$ cannot be eliminated.
    Moreover, note that there does not exist a digest $z' \in \mathit{correct\_digs} \cup \mathit{faulty\_digs}$ such that $z$ precedes $z'$ in $Z$.
    If such $z'$ existed, $\mathit{num}_i[z'] \geq 2t + 1$, which further implies that $z' \in \mathit{correct\_digs}$.
    As $\mathit{correct\_digs} = \{ z \}$, this is impossible.
    Therefore, $z$ must be the last digest in the list $Z$.

    For every digest $z' \in \mathit{faulty\_digs} \setminus { \{z \}}$, all \textsc{init} messages for $z'$ received by $p_i$ are sent by faulty processes.
    Hence, every digest $z' \in \mathit{faulty\_digs} \setminus{ \{z \} }$ is eliminated.
    Thus, $\mathsf{distinct()} - \mathsf{eliminated()} = |\{ z \} \cup \mathit{faulty\_digs}| - |\mathit{faulty\_digs} \setminus{ \{ z \}}| = 1$, which represents a contradiction with the fact that $\mathsf{distinct()} - \mathsf{eliminated()} \geq 3$.
    This means that this case is impossible.

    \item Let $X = 2$.
    Let $\mathit{correct\_digs} = \{ z_1, z_2 \}$, for some digests $z_1$ and $z_2$.
    Note that $\mathit{num}_i[z_1] \geq t + 1$ or $\mathit{num}_i[z_2] \geq t + 1$ (as there are at least $2t + 1$ \textsc{init} messages received by $p_i$ from correct processes).
    Without loss of generality, let $\mathit{num}_i[z_1] \geq t + 1$.
    Therefore, the digest $z_1$ cannot be eliminated.
    Observe also that there cannot exist a digest $z_3 \in \mathit{faulty\_digs} \setminus{ \mathit{correct\_digs} }$ such that $z_1$ precedes $z_3$ in $Z$.
    Indeed, for such digest $z_3$ to exist, $\mathit{num\_i}[z_3] \geq t + 1$, which then implies that $z_3$ must belong to $\mathit{correct\_digs}$.
    As this is not the case, the only digest that can succeed $z_1$ in the list $Z$ is $z_2$.
    We distinguish two cases:
    \begin{compactitem}
        \item Let $z_2$ not be eliminated.
        We further study two cases:
        \begin{compactitem}
            \item Let $z_1$ precede $z_2$ in the sorted list of digests $Z$.
            In this case, $z_2$ is the last digest in the list $Z$ and $z_1$ is the penultimate digest in the list $Z$.
            (This holds as only $z_2$ can succeed $z_1$ in $Z$.)
            Crucially, all other digests in the list are eliminated as (1) \textsc{init} messages for them are sent only by faulty processes, and (2) there are at most $t$ faulty processes.
            Therefore, $\mathsf{distinct()} - \mathsf{eliminated()} = 2$, which contradicts the fact that $\mathsf{distinct()} - \mathsf{eliminated()} \geq 3$.
            Thus, this case is impossible.

            \item Let $z_2$ precede $z_1$ in the sorted list of digests $Z$.
            In this case, $z_1$ is the last digest in the list.
            (This holds as only $z_2$ can succeed $z_1$ in $Z$.)
            Note that all values that precede $z_2$ in the list $Z$ are eliminated (as these are values held by faulty processes only).
            Let $\mathit{eliminated}$ denote the set of eliminated digests.
            Let $B_{\mathit{eliminated}} = \sum\limits_{z \in \mathit{eliminated}} \mathit{num}_i[z]$.
            Note that $B_{\mathit{eliminated}} \leq t$ due to the definition of the eliminated digests.
            Importantly, as $z_2$ is not eliminated, we have $B_{\mathit{eliminated}} + \mathit{num}_i[z_2] \geq t + 1$, which further implies $\mathit{num}_i[z_2] \geq t + 1 - B_{\mathit{eliminated}}$.

            By contradiction, suppose there exists a digest $z_3$ such that (1) $z_2$ precedes $z_3$ in $Z$, and (2) $z_3$ precedes $z_1$ in $Z$.
            As $z_3 \in \mathit{faulty\_digs} \setminus{ \mathit{correct\_digs} }$, $B_{\mathit{eliminated}}$ messages are used on the values preceding $z_2$ and there are $F$ messages issued by faulty processes, $\mathit{num}_i[z_3] \leq F - B_{\mathit{eliminated}}$.
            Moreover, $\mathit{num}_i[z_3] \geq \mathit{num}_i[z_2]$ (as $z_2$ precedes $z_3$ in $Z$).
            This implies  $\mathit{num}_i[z_2] \leq F - B_{\mathit{eliminated}}$.
            Hence, $t + 1 - B_{\mathit{eliminated}} \leq \mathit{num}_i[z_2] \leq F - B_{\mathit{eliminated}}$, which further implies $t + 1 \leq F$.
            This is impossible as $F \leq t$.
        \end{compactitem}

        \item Let $z_2$ be eliminated.
        Note that, as $z_2$ is eliminated and $z_1$ is not eliminated, $z_2$ precedes $z_1$ in list $Z$.
        Given that only $z_2$ can succeed $z_1$ in list $Z$, $z_1$ is the last digest in $Z$.
        As $\mathsf{distinct()} - \mathsf{eliminated()} \geq 3$ and $z_1$ is not eliminated, there are at least two digests in $Z$ that are (1) not eliminated, and (2) not broadcast by correct processes.
        Let $z' = Z[\mathsf{eliminated}() + 1]$ and $z'' = Z[\mathsf{eliminated}() + 2]$. 
        Note that $z_2$ precedes both $z'$ and $z''$ as $z_2$ is eliminated and $z'$ and $z''$ are not.
        Moreover, note that all \textsc{init} messages sent for $z'$ or $z''$ are sent by faulty processes (as $z', z'' \in \mathit{faulty\_digs} \setminus{ \mathit{correct\_digs} } $).

        Let $\mathit{eliminated}$ denote the set of eliminated digests.
        For each digest $z \in \mathit{eliminated}$, let $C_{z}$ (resp., $B_{z}$) denote the number of $\langle \textsc{init}, z \rangle$ messages received by $p_i$ from correct (resp., faulty) processes.
        (Hence, for each digest $z \in \mathit{eliminated}$, $\mathit{num}_i[z] = C_{z} + B_{z}$.)
        Note that $C_{z_2} > 0$ (as $z_2 \in \mathit{correct\_digs}$); for every digest $z \in \mathit{eliminated} \setminus { \{z_2\} }$, $C_{z} = 0$.
        Let $B_{\mathit{eliminated}} = \sum\limits_{z \in \mathit{eliminated}} B_{z}$.
        Observe that $C_{z_2}+ B_{\mathit{eliminated}} + \mathit{num}_i[z'] \geq t + 1$ as, otherwise, $z'$ would also be eliminated.
        Therefore, $\mathit{num}_i[z'] \geq t + 1 - C_{z_2} - B_{\mathit{eliminated}}$.
        As $\mathit{num}_i[z''] \geq \mathit{num}_i[z']$, $\mathit{num}_i[z''] \geq t + 1 - C_{z_2} - B_{\mathit{eliminated}}$.
        Thus, $\mathit{num}_i[z'] + \mathit{num}_i[z''] \geq 2t + 2 - 2C_{z_2} - 2B_{\mathit{eliminated}}$.
        Moreover, $\mathit{num}_i[z'] + \mathit{num}_i[z''] \leq F - B_{\mathit{eliminated}}$ (as all \textsc{init} messages for $z'$ or $z''$ are sent by faulty processes and there are already $B_{\mathit{eliminated}}$ faulty processes that sent \textsc{init} messages for digests different from $z'$ and $z''$).
        Hence, we have:
        \begin{equation*}
            2t + 2 - 2C_{z_2} - 2B_{\mathit{eliminated}} \leq \mathit{num}_i[z'] + \mathit{num}_i[z''] \leq F - B_{\mathit{eliminated}}.
        \end{equation*}
        This implies $2t + 2 - 2C_{z_2} - 2B_{\mathit{eliminated}} \leq F - B_{\mathit{eliminated}}$, which implies $2C_{z_2} \geq 2t + 2 -B_{\mathit{eliminated}} - F$.

        Observe that $\mathit{num}_i[z''] \geq \mathit{num}_i[z'] \geq C_{z_2}$ (as neither $z'$ nor $z''$ are eliminated, $\mathit{num}_i[z_2] \geq C_{z_2}$ and $z_2$ is eliminated). 
        Therefore, $\mathit{num}_i[z'] + \mathit{num}_i[z''] \geq 2C_{z_2}$.
        Therefore, $F - B_{\mathit{eliminated}} \geq 2C_{z_2}$.
        
        Thus, we have $2t + 2 - B_{\mathit{eliminated}} - F \leq 2C_{z_2} \leq F - B_{\mathit{eliminated}}$.
        Therefore, $2t + 2 - B_{\mathit{eliminated}} - F \leq F - B_{\mathit{eliminated}}$, which implies $2t + 2 \leq 2F$.
        As $F \leq t$, this is impossible.
    \end{compactitem}
\end{compactitem}
The claim holds as its statement holds in all possible cases.
\end{proof}

We are now ready to prove the \crb satisfies validity.

\begin{proposition} [\crb satisfies validity]
Given $t < \frac{1}{4}n$, \crb (see \Cref{algorithm:crb}) satisfies validity in the presence of a computationally unbounded adversary.
\end{proposition}
\begin{proof}
Let $|\mathcal{B}_C| \leq 2$.
By contradiction, suppose there exists a correct process $p_i$ that delivers $\botcrb$.
Therefore, there exists a correct process that broadcasts a \textsc{broken} message.
The first correct process that broadcasts a \textsc{broken} message does so at line~\ref{line:crb_broadcast_broken_2}.
Hence, \Cref{claim:crb_validity_crucial} proves that $|\mathcal{B}_C| \geq 3$, which implies contradiction with $|\mathcal{B}_C| \leq 2$.
Thus, \crb satisfies validity.
\end{proof}

Next, we prove the termination property of \crb.
To do so, we first show that if there exists a digest $z$ broadcast by at least $t + 1$ correct processes, then every correct process delivers $z$ within $O(1)$ asynchronous rounds.

\begin{myclaim} \label{claim:crb_termination_crucial_1}
Let there exist a digest $z$ broadcast by at least $t + 1$ correct processes.
Then, every correct process delivers $z$ within $O(1)$ asynchronous rounds.
\end{myclaim}
\begin{proof}
As $t + 1$ correct processes broadcast $z$, every correct process eventually receives $t + 1$ \textsc{init} messages for $z$ (line~\ref{line:receive_init_rule}) and broadcasts an \textsc{echo} message for $z$ (line~\ref{line:crb_broadcast_echo}).
(This occurs within a single message delay.)
Therefore, every correct process eventually broadcasts a \textsc{ready} message for $z$ (line~\ref{line:crb_broadcast_ready_1}) upon receiving an \textsc{echo} message for $z$ from (at least) $2t + 1$ processes (line~\ref{line:crb_quorum_echo_rule}).
(This incurs another message delay.)
Finally, every correct process eventually receives a \textsc{ready} message for $z$ from (at least) $2t + 1$ processes (line~\ref{line:crb_receive_quorum_ready}) and delivers $z$ (line~\ref{line:crb_deliver}) in $3 \in O(1)$ message delays.
\end{proof}

The following claim proves that all correct processes deliver the special value $\botcrb$ within $O(1)$ asynchronous rounds given that (1) $|\mathcal{B}_C| > 3$, and (2) no digest is broadcast by $t + 1$ (or more) correct processes.

\begin{myclaim} \label{claim:crb_termination_crucial_2}
If (1) $|\mathcal{B}_C| > 3$, and (2) no digest is broadcast by $t + 1$ correct processes, then every correct process delivers $\botcrb$ within $O(1)$ asynchronous rounds.
\end{myclaim}
\begin{proof}
To prove the claim, we prove that all correct processes broadcast a \textsc{broken} message within $O(1)$ message delays, which then implies that all correct processes receive $2t + 1$ \textsc{broken} messages (line~\ref{line:crb_deliver_def_rule}) and deliver $\botcrb$ (line~\ref{line:crb_deliver_def}) within $O(1)$ message delays. 
Let $|\mathcal{B}_C| = x > 3$.
Without loss of generality, let $\mathcal{B}_C = \{ z_1, z_2, ..., z_{x - 1}, z_x \}$, for some digests $z_1, z_2, ..., z_x$.
Consider any correct process $p_i$ and time $\tau$ at which process $p_i$ receives \textsc{init} messages from all correct processes.
Note that this occurs within a single message delay.
Let $\mathcal{I}$ denote the set of \textsc{init} messages received by $p_i$ at time $\tau$; note that $|\mathcal{I}| \geq n - t \geq 3t + 1$ as there are at least $n - t \geq 3t + 1$ correct processes.
Next, we define a set of digests $\mathit{correct\_digs}$ in the following way:
\begin{equation*}
    \mathit{correct\_digs} = \{ z \,|\, z \text{ is received in a message $m \in \mathcal{I}$ whose sender is correct} \}.
\end{equation*}
Note that $\mathit{correct\_digs} = \{ z_1, z_2, ..., z_{x - 1}, z_x \}$.
For each digest $z \in \mathit{correct\_digs}$, let $C_{z}$ (resp., $B_{z}$) denote the number of $\langle \textsc{init}, z \rangle$ messages received by $p_i$ from correct (resp., faulty) processes at time $\tau$.
Observe that, for every $z \in \mathit{correct\_digs}$, $\mathit{num}_i[z] = C_{z} + B_{z}$.
Moreover, $C_{z_1} + C_{z_2} + ... + C_{z_{x - 1}} + C_{z_x} \geq n - t \geq 3t + 1$ as there are at least $n - t \geq 3t + 1$ correct processes.

Let $Z$ denote the sorted list of digests constructed by the $\mathsf{eliminated()}$ function (line~\ref{line:crb_sorted_list}).
We say that a digest $z$ is \emph{eliminated} if and only if (1) $z = Z[i]$, and (2) $i \in [1, \mathsf{eliminated()}]$.
(If $\mathsf{eliminated()} = 0$, no digest is eliminated.)
To prove the claim, it suffices to show that at most $x - 3$ digests from the $\mathit{correct\_digs}$ set are eliminated at time $\tau$ as this (along with the fact that $|\mathcal{I}| \geq n - t \geq 3t + 1$) guarantees that the rule at line~\ref{line:send_broken_rule} activates at process $p_i$.
By contradiction, suppose that $x - 2$ (or more) digests from the $\mathit{correct\_digs}$ set are eliminated.
We distinguish three cases:
\begin{compactitem}
    \item Suppose exactly $x - 2$ digests from the $\mathit{correct\_digs}$ set are eliminated.
    Without loss of generality, let $z_1, z_2, ..., z_{x - 3}, z_{x - 2}$ be eliminated at time $\tau$.
    This implies that $\mathit{num}_i[z_1] + \mathit{num}_i[z_2] + ... + \mathit{num}_i[z_{x - 2}] = (C_{z_1} + B_{z_1}) + (C_{z_2} + B_{z_2}) + ... + (C_{z_{x - 2}} + B_{z_{x - 2}}) \leq t$, which further implies that $C_{z_1} + C_{z_2} + ... + C_{z_{x - 2}} \leq t$.
    As $C_{z_1} + C_{z_2} + ... + C_{z_{x - 1}} + C_{z_x} \geq 3t + 1$ and $C_{z_1} + C_{z_2} + ... + C_{z_{x - 2}} \leq t$, $C_{z_{x - 1}} + C_{z_x} \geq 2t + 1$.
    Therefore, $C_{z_{x - 1}} \geq t + 1$ or $C_{z_x} \geq t + 1$, which contradicts the fact that no digest is broadcast by $t + 1$ correct processes.

    \item Suppose exactly $x - 1$ digests from the $\mathit{correct\_digs}$ set are eliminated.
    Without loss of generality, let $z_1, z_2, ..., z_{x - 3}, z_{x - 2}, z_{x - 1}$ be eliminated at time $\tau$.
    This implies that $\mathit{num}_i[z_1] + \mathit{num}_i[z_2] + ... + \mathit{num}_i[z_{x - 2}] + \mathit{num}_i[z_{x - 1}] = (C_{z_1} + B_{z_1}) + (C_{z_2} + B_{z_2}) + ... + (C_{z_{x - 2}} + B_{z_{x - 2}}) + (C_{z_{x - 1}} + B_{z_{x - 1}}) \leq t$, which further implies that $C_{z_1} + C_{z_2} + ... + C_{z_{x - 2}} + C_{z_{x - 1}} \leq t$.
    As $C_{z_1} + C_{z_2} + ... + C_{z_{x - 1}} + C_{z_x} \geq 3t + 1$ and $C_{z_1} + C_{z_2} + ... + C_{z_{x - 2}} + C_{z_{x - 1}} \leq t$, $C_{z_x} \geq 2t + 1$.
    This contradicts the fact that no digest is broadcast by  $t + 1$ correct processes.

    \item Suppose exactly $x$ digests from the $\mathit{correct\_digs}$ set are eliminated.
    This is impossible as $\mathit{num}[z_1] + \mathit{num}[z_2] + ... + \mathit{num}[z_x] \geq 3t + 1$ due to the fact that $C_{z_1} + C_{z_2} + ... + C_{z_{x - 1}} + C_{z_x} \geq 3t + 1$.
\end{compactitem}
As neither of the aforementioned cases can occur, the claim holds.
\end{proof}

We are finally ready to prove the termination property of \Cref{algorithm:crb}.

\begin{proposition} [\crb satisfies termination] \label{proposition:crb_termination}
Given $t < \frac{1}{4} n$, \crb (see \Cref{algorithm:crb}) satisfies termination in the presence of a computationally unbounded adversary.
Precisely, every correct process delivers within $O(1)$ asynchronous rounds.
\end{proposition}
\begin{proof}
To prove the proposition, we study two possible scenarios:
\begin{compactitem}
    \item Let $|\mathcal{B}_C| \leq 3$.
    In this case, there exists a digest $z$ broadcast by at least $t + 1$ correct processes.
    Therefore, every correct process delivers $z$ within $O(1)$ message delays (by \Cref{claim:crb_termination_crucial_1}).
    Hence, the termination property is satisfied in this case.

    \item Let $|\mathcal{B}_C| > 3$.
    We further distinguish two cases:
    \begin{compactitem}
        \item Let there exist a digest $z$ broadcast by $t + 1$ correct processes.
        In this case, the termination property is ensured due to \Cref{claim:crb_termination_crucial_1}.

        \item Let there be no digest broadcast by $t + 1$ correct processes.
        The termination property holds in this case due to \Cref{claim:crb_termination_crucial_2}.
    \end{compactitem}
\end{compactitem}
As termination is satisfied in every possible case, the proposition holds.
\end{proof}

Note that \Cref{proposition:crb_termination} proves that the worst-case time complexity of \crb is $O(1)$: correct processes deliver the first digest in $O(1)$ time.
Next, we prove \crb's integrity.

\begin{proposition} [\crb satisfies integrity]
Given $t < \frac{1}{4} n$, \crb (see \Cref{algorithm:crb}) satisfies integrity in the presence of a computationally unbounded adversary.
\end{proposition}
\begin{proof}
The integrity property is satisfied due to the rule at line~\ref{line:crb_rule_integrity}.
\end{proof}

The following proposition proves \crb's justification.

\begin{proposition} [\crb satisfies justification]
Given $t < \frac{1}{4} n$, \crb (see \Cref{algorithm:crb}) satisfies justification in the presence of a computationally unbounded adversary.
\end{proposition}
\begin{proof}
Suppose a correct process $p_i$ delivers a digest $z \in \valuedigest$.
Therefore, there exists a \textsc{ready} message for $z$ broadcast by a correct process.
Note that the first correct process to broadcast a \textsc{ready} message for $z$ does so at line~\ref{line:crb_broadcast_ready_1}.
Hence, there exists a correct process that broadcasts an \textsc{echo} message for $z$, which implies $z \in \mathcal{B}_C$ (as at least $t + 1$ \textsc{init} messages for $z$ are received).
\end{proof}

Finally, we prove that \crb satisfies totality.

\begin{proposition} [\crb satisfies totality] \label{proposition:crb_totality}
Given $t < \frac{1}{4} n$, \crb (see \Cref{algorithm:crb}) satisfies totality in the presence of a computationally unbounded adversary.
Precisely, if any correct process delivers $z \in \valuedigest \cup \{ \botcrb \}$, then every correct process delivers $z$ within $O(1)$ asynchronous rounds from the aforementioned delivery.
\end{proposition}
\begin{proof}
Suppose a correct process $p_i$ delivers $z \in \valuedigest \cup \{ \botcrb \}$.
We distinguish two possibilities:
\begin{compactitem}
    \item Let $z \in \valuedigest$.
    (This implies that $z \neq \botcrb$.)
    Hence, $p_i$ received a $\langle \textsc{ready}, z \rangle$ message from (at least) $2t + 1$ processes, which implies that all correct processes eventually receive (at least) $t + 1$ \textsc{ready} message for $z$ and broadcast a \textsc{ready} message for $z$.
    Thus, every correct process receives $2t + 1$ \textsc{ready} messages for $z$ within $O(1)$ message delays and delivers $z$.

    \item Let $z = \botcrb$.
    Following the same argument as above (just applied to \textsc{broken} messages), we conclude that every correct process delivers $\botcrb$ within $O(1)$ message delays.
\end{compactitem}
As the statement of the proposition holds in both cases, the proof is concluded.
\end{proof}

Observe that \Cref{proposition:crb_totality} proves that the totality property is satisfied within $O(1)$ asynchronous rounds from the first delivery.

\smallskip
\noindent \textbf{Proof of complexity.}
We now prove \crb's complexity.

\begin{lemma} [\crb's worst-case complexity] 
Given $t < \frac{1}{4}n$, the following holds for \crb (see \Cref{algorithm:crb}) in the presence of a computationally unbounded adversary:
\begin{compactitem}
    \item The worst-case message complexity is $O(n^2)$.

    \item The worst-case bit complexity is $O(n^2 \kappa)$.

\end{compactitem} 
\end{lemma}
\begin{proof}
Consider any correct process $p_i$.
If $p_i$ broadcasts an \textsc{echo} or a \textsc{ready} message for a digest $z$, then $z \in \mathcal{B}_C$.
Moreover, process $p_i$ broadcasts at most one \textsc{init} and at most one \textsc{broken} message.
Therefore, process $p_i$ sends
\begin{equation*}
    \underbrace{O(n)}_\text{\textsc{init}} 
    {}+{}
    \underbrace{O(n)}_\text{\textsc{broken}} 
    {}+{}
    \underbrace{|\mathcal{B}_C| \cdot O(n)}_\text{\textsc{echo} \& \textsc{ready}} \text{ messages}.
\end{equation*}
Moreover, process $p_i$ sends 
\begin{equation*}
    \underbrace{O(n \kappa)}_\text{\textsc{init}}
    {}+{}
    \underbrace{O(n)}_\text{\textsc{broken}} 
    {}+{}
    \underbrace{|\mathcal{B}_C| \cdot O(n \kappa)}_\text{\textsc{echo} \& \textsc{ready}} \text{ bits}.
\end{equation*}
Given that $|\mathcal{B}_C| \in O(1)$, process $p_i$ sends $O(n)$ messages and $O(n \kappa)$ bits, which implies that all correct processes send $O(n^2)$ messages and $O(n^2\kappa)$ bits.
\end{proof}

\subsection{Pseudocode}

The pseudocode of \smba is given in \Cref{algorithm:async_smba}.
Internally, \smba relies on (1) instances $\mathcal{MBA}_1$ and $\mathcal{MBA}_2$ of the MBA primitive run on digests (lines~\ref{line:strong_mba_init} and~\ref{line:strong_mba_init_2}), and (2) an instance $\mathcal{CRB}$ of the CRB algorithm \crb (line~\ref{line:strong_crb_init}).
Concretely, for $\mathcal{MBA}_1$ and $\mathcal{MBA}_2$, we rely on the MBA algorithm introduced in~\cite{DBLP:journals/acta/MostefaouiR17} that (1) exchanges $O(n^2)$ messages and $O(n^2 \kappa)$ bits in expectation, and (2) terminates in $O(1)$ time in expectation.

\begin{algorithm}[tbp]
\caption{\smba: Pseudocode (for process $p_i$)}
\label{algorithm:async_smba}
\begin{algorithmic}[1]
\scriptsize

\State \textbf{Uses:}
    \State \hskip2em \textcolor{jnSUDigitalRedLight}{\(\triangleright\) \cite{DBLP:journals/acta/MostefaouiR17} exchanges $O(n^2)$ messages and $O(n^2 \kappa)$ bits and terminates in $O(1)$ time}
    \State \hskip2em MBA algorithm~\cite{DBLP:journals/acta/MostefaouiR17} with $\valuemba = \valuedigest \cup \{ \botcrb \}$, \textbf{instance} $\mathcal{MBA}_1$  \label{line:strong_mba_init} 

    \smallskip
    \State \hskip2em \textcolor{jnSUDigitalRedLight}{\(\triangleright\) \cite{DBLP:journals/acta/MostefaouiR17} exchanges $O(n^2)$ messages and $O(n^2 \kappa)$ bits and terminates in $O(1)$ time}
    \State \hskip2em MBA algorithm~\cite{DBLP:journals/acta/MostefaouiR17} with $\valuemba = \valuedigest$, \textbf{instance} $\mathcal{MBA}_2$  \label{line:strong_mba_init_2} 

    \smallskip
    \State \hskip2em \textcolor{jnSUDigitalRedLight}{\(\triangleright\) \crb exchanges $O(n^2)$ messages and $O(n^2 \kappa)$ bits}
    \State \hskip2em CRB algorithm \crb, \textbf{instance} $\mathcal{CRB}$  \label{line:strong_crb_init}

\medskip
\State \textbf{Constants:}
    \State \hskip2em $\valuedigest$ $\mathit{default}$ \BlueComment{default digest}

\medskip
\State \textbf{Local variables:}
    \State \hskip2em $\valuedigest$ $z_i \gets p_i$'s proposal
    \State \hskip2em $\mathsf{Set}(\valuedigest)$ $\mathit{delivered}_i \gets \emptyset$

\medskip
\State \textbf{upon} $\mathsf{propose}(\valuedigest \text{ } z_i)$:\label{line:strong_start} \BlueComment{start of the algorithm}
    \State \hskip2em \textbf{invoke} $\mathcal{CRB}.\mathsf{broadcast}(z_i)$\label{line:strong_crb_broadcast}

\medskip
\State \textbf{upon} $\mathcal{CRB}.\mathsf{deliver}(\valuedigest \cup \{ \botcrb \} \text{ } z)$:\label{line:upon_deliver_crb} 
    \State \hskip2em $\mathit{delivered}_i \gets \mathit{delivered}_i \cup \{ z \}$
    \State \hskip2em \textbf{if} $|\mathit{delivered}_i| = 1$:
        \State \hskip4em \textbf{invoke} $\mathcal{MBA}_1.\mathsf{propose}(z)$\label{line:collective_propose_mba}

\medskip
\State \textbf{upon} $\mathcal{MBA}_1.\mathsf{decide}(\valuedigest \cup \{ \botcrb, \botmba \} \text{ } z')$:\label{line:collective_upon_mba} \BlueComment{we assume $\botcrb \neq \botmba$}
    \State \hskip2em \textbf{if} $z' \neq \botmba$:
        \State \hskip4em $\valuedigest \cup \{ \botcrb \}$ $z^{\star} \gets z'$ \label{line:strong_set_decided_value}
        \State \hskip4em \textbf{if} $z^{\star} = \botcrb$:
            \State \hskip6em $z^{\star} \gets \mathit{default}$\label{line:strong_set_default_acc_value}
        \State \hskip4em \textbf{invoke} $\mathcal{MBA}_2.\mathsf{propose}(z^{\star})$\label{line:strong_mba_2_propose_1}
    \State \hskip2em \textbf{else:}
        \State \hskip4em \textbf{wait for} $|\mathit{delivered}_i| = 2$%
            \label{line:strong_wait_for_delivered}
        \State \hskip4em Let $z^{\star}$ be the lexicographically smallest digest ($\neq \botcrb$) in $\mathit{delivered}_i$%
            \label{line:collective_choose_smallest} 
        \State \hskip4em \textbf{invoke} $\mathcal{MBA}_2.\mathsf{propose}(z^{\star})$\label{line:strong_mba_2_propose_2}

\medskip
\State \textbf{upon} $\mathcal{MBA}_2.\mathsf{decide}(\valuedigest \cup \{ \botmba \} \text{ } z'')$:\label{line:strong_decide_rule}
    \State \hskip2em \textbf{if} $z'' = \botmba$: \label{line:strong_mba_check_bot}
        \State \hskip4em  $z'' \gets \mathit{default}$\label{line:strong_update_bot_decide_2}
    \State \hskip2em \textbf{trigger} $\mathsf{decide}(z'')$\label{line:strong_decide_2}

\end{algorithmic}
\end{algorithm}

\smallskip
\noindent \textbf{Pseudocode description.}
We explain the pseudocode of our SMBA protocol \smba from the perspective of a correct process $p_i$.
Once $p_i$ proposes its digest $z_i$ (line~\ref{line:strong_start}), process $p_i$ broadcasts $z_i$ via $\mathcal{CRB}$ (line~\ref{line:strong_crb_broadcast}).
When $p_i$ delivers the first digest (or the special value $\botcrb$) from $\mathcal{CRB}$ (line~\ref{line:upon_deliver_crb}), $p_i$ proposes it to $\mathcal{MBA}_1$ (line~\ref{line:collective_propose_mba}).
Let $p_i$ decide $z' \in \valuedigest \cup \{ \botcrb, \botmba \}$ from $\mathcal{MBA}_1$ (line~\ref{line:collective_upon_mba}); we assume $\botcrb \neq \botmba$.
Now, we distinguish two cases:
\begin{compactitem}
    \item Let $z' \neq \botmba$.
    Note that the justification property of $\mathcal{MBA}_1$ ensures that $z'$ was proposed by a correct process.
    We further investigate two scenarios:
    \begin{compactitem}
        \item Let $z' = \botcrb$.
        In this case, process $p_i$ proposes the default digest $\mathit{default}$ to $\mathcal{MBA}_2$ (lines~\ref{line:strong_set_default_acc_value} and~\ref{line:strong_mba_2_propose_1}).
        Note that, as $\botcrb$ is delivered from $\mathcal{CRB}$, there must exist more than two different digests proposed by correct processes (due to $\mathcal{CRB}$'s validity property).
        Therefore, this case does \emph{not} require the decision to be proposed by a correct process.

        \item Let $z' \neq \botcrb$.
        Then, process $p_i$ proposes $z'$ to $\mathcal{MBA}_2$ (lines~\ref{line:strong_set_decided_value} and~\ref{line:strong_mba_2_propose_1}).
        Observe that the justification property of $\mathcal{CRB}$ guarantees that $z'$ is proposed by correct process.
    \end{compactitem}

    \item Let $z' = \botmba$.
    The strong unanimity property of $\mathcal{MBA}_1$ guarantees that not all correct processes proposed the same  value $z \in \valuedigest \cup \{ \botcrb \}$ to $\mathcal{MBA}_1$.
    (Otherwise, $\botmba$ could not have been decided.)
    Thus, there are at least two different values $z_1, z_2 \in \valuedigest \cup \{ \botcrb \}$ delivered from $\mathcal{CRB}$ by correct processes.
    In this case, process $p_i$ waits to deliver the second digest (or $\botcrb$) from $\mathcal{CRB}$ (line~\ref{line:strong_wait_for_delivered}).
    Once that happens (and it will due to $\mathcal{CRB}$'s totality property), process $p_i$ proposes to $\mathcal{MBA}_2$ the lexicographically smallest digest (thus, not $\botcrb$!) it has delivered from $\mathcal{CRB}$ (lines~\ref{line:collective_choose_smallest} and~\ref{line:strong_mba_2_propose_2}).
\end{compactitem}
Finally, once process $p_i$ decides $z'' \in \valuedigest \cup \{ \botmba \}$ from $\mathcal{MBA}_2$ (line~\ref{line:strong_decide_rule}), $p_i$ decides (1) the default digest $\mathit{default}$ if $z'' = \botmba$ (lines~\ref{line:strong_update_bot_decide_2} and~\ref{line:strong_decide_2}), or (2) $z''$ otherwise (line~\ref{line:strong_decide_2}).

\subsection{Proof of Correctness \& Complexity}

This subsection formally proves the correctness and complexity of \smba.

\smallskip
\noindent \textbf{Proof of correctness.}
We first prove the following lemma.

\begin{lemma} [\smba is correct]
Given $t < \frac{1}{4} n$, \smba (see \Cref{algorithm:async_smba}) is a correct implementation of the SMBA primitive in the presence of a computationally unbounded adversary.
\end{lemma}

We start by proving the agreement property of \smba.

\begin{proposition} [\smba satisfies agreement]
Given $t < \frac{1}{4} n$, \smba (see \Cref{algorithm:async_smba}) satisfies agreement in the presence of a computationally unbounded adversary.
\end{proposition}
\begin{proof}
The agreement property follows directly from the agreement property of $\mathcal{MBA}_2$.
\end{proof}

Next, we prove that \smba satisfies strong validity.
In the rest of the proof, let $\mathcal{Z}_C$ denote the set of digests proposed by correct processes.
We start by proving that only constantly many different digests are broadcast by correct processes via $\mathcal{CRB}$.
(Recall that this assumption is required by $\mathcal{CRB}$; see \Cref{mod:async_crb}.)

\begin{myclaim} \label{claim:crb_works}
Only $O(1)$ different digests are broadcast by correct processes via $\mathcal{CRB}$.
\end{myclaim}
\begin{proof}
The claim follows directly from the assumption that only $O(1)$ different digests are proposed by correct processes (see \Cref{mod:async_smba}).
\end{proof}

\Cref{claim:crb_works} proves that $\mathcal{CRB}$ behaves according to its specification.
(To not pollute the presentation of the proof, we might not explicitly rely on \Cref{claim:crb_works} in the rest of the proof.)
Next, we prove that if $|\mathcal{Z}_C| \leq 2$, then all correct processes propose the same digest $z$ to $\mathcal{MBA}_2$ with $z \in \mathcal{Z}_C$.

\begin{myclaim} \label{claim:strong_all_propose_same_mba_2}
Let $|\mathcal{Z}_C| \leq 2$.
Then, there exists a digest $z$ such that (1) $z \in \mathcal{Z}_C$, and (2) every correct process proposes $z$ to $\mathcal{MBA}_2$.
\end{myclaim}
\begin{proof}
As $|\mathcal{Z}_C| \leq 2$, at most two different values are broadcast via $\mathcal{CRB}$ by correct processes.
Therefore, the validity property of $\mathcal{CRB}$ guarantees that no correct process delivers the special value $\botcrb$ from $\mathcal{CRB}$.
Moreover, the termination property of $\mathcal{CRB}$ ensures that all correct processes eventually deliver from $\mathcal{CRB}$ and, thus, propose to $\mathcal{MBA}_1$.
Hence, the termination and agreement properties of $\mathcal{MBA}_1$ ensure that all correct processes eventually decide $z' \in \valuedigest \cup \{ \botcrb, \botmba \}$ from $\mathcal{MBA}_1$.
We distinguish two cases:
\begin{compactitem}
    \item Let $z' \neq \botmba$.
    The justification property of $\mathcal{MBA}_1$ guarantees that $z'$ was proposed to $\mathcal{MBA}_1$ by a correct process.
    Therefore, $z'$ was delivered by a correct process from $\mathcal{CRB}$, which further implies that (1) $z' \neq \botcrb$, and (2) $z' \in \mathcal{Z}_C$ due to $\mathcal{CRB}$'s justification property.
    This means that every correct process proposes $z'$ to $\mathcal{MBA}_2$.
    Hence, the statement of the claim holds in this case.

    \item Let $z' = \botmba$.
    The strong unanimity property of $\mathcal{MBA}_1$ proves that at least two different values from the $\valuedigest \cup \{ \botcrb \}$ set are proposed to $\mathcal{MBA}_1$ by correct processes.
    Let $z_1$ denote one such value and let $z_2 \neq z_1$ denote another.
    Both $z_1$ and $z_2$ are delivered from $\mathcal{CRB}$ by correct processes, which implies that $z_1 \neq \botcrb$ and $z_2 \neq \botcrb$.
    Moreover, the justification property of $\mathcal{CRB}$ guarantees that $z_1 \in \mathcal{Z}_C$ and $z_2 \in \mathcal{Z}_C$.
    As $z_1 \neq z_2$ and $|\mathcal{Z}_C| \leq 2$, $\mathcal{Z}_C = \{ z_1, z_2 \}$.

    Now, consider any correct process $p_i$.
    As both $z_1$ and $z_2$ are delivered by correct processes from $\mathcal{CRB}$, process $p_i$ eventually delivers these two digests (due to $\mathcal{CRB}$'s totality property).
    Moreover, process $p_i$ never delivers (1) any other digest $z_3 \neq \botcrb$ from $\mathcal{CRB}$ as $\mathcal{CRB}$'s justification would imply $z_3 \in \mathcal{Z}_C$ (recall that $\mathcal{Z}_C = \{ z_1, z_2 \}$), and (2) $\botcrb$ as no correct process ever delivers $\botcrb$ due to $\mathcal{CRB}$'s validity property (recall that $|\mathcal{Z}_C| \leq 2$).
    Hence, $\mathit{delivered}_i = \{ z_1, z_2 \}$ when $|\mathit{delivered}_i| = 2$ at process $p_i$.
    Thus, $p_i$ (and every other correct process) proposes the lexicographically smaller digest between $z_1$ and $z_2$, which implies that the statement of the claim holds even in this case.
\end{compactitem}
The claim holds as its statement is satisfied in both possible scenarios.
\end{proof}

We are ready to prove that \smba satisfies strong validity.

\begin{proposition} [\smba satisfies strong validity]
Given $t < \frac{1}{4} n$, \smba (see \Cref{algorithm:async_smba}) satisfies strong validity in the presence of a computationally unbounded adversary.
\end{proposition}
\begin{proof}
Let $|\mathcal{Z}_C| \leq 2$.
Let $p_i$ be any correct process that decides some digest $z'$.
\Cref{claim:strong_all_propose_same_mba_2} shows that there exists a digest $z$ such that (1) $z \in \mathcal{Z}_C$, and (2) all correct processes propose $z$ to $\mathcal{MBA}_2$.
Therefore, the strong unanimity property of $\mathcal{MBA}_2$ ensures that all correct processes decide $z \neq \bot_{\mathsf{MBA}}$ from $\mathcal{MBA}_2$.
Thus, $z' = z$, which implies $z' \in \mathcal{Z}_C$.
\end{proof}

Next, we prove the integrity property of \smba.

\begin{proposition} [\smba satisfies integrity]
Given $t < \frac{1}{4} n$, \smba (see \Cref{algorithm:async_smba}) satisfies integrity in the presence of a computationally unbounded adversary.
\end{proposition}
\begin{proof}
The proposition holds due to the integrity property of $\mathcal{MBA}_2$.
\end{proof}

Finally, we prove the termination property of \smba.

\begin{proposition} [\smba satisfies termination] \label{proposition:smba_termination}
Given $t < \frac{1}{4} n$, \smba (see \Cref{algorithm:async_smba}) satisfies termination in the presence of a computationally unbounded adversary.
Precisely, every correct process decides within $O(1)$ asynchronous rounds in expectation.
\end{proposition}
\begin{proof}
The $\mathcal{CRB}$ primitive guarantees that all correct processes propose to $\mathcal{MBA}_1$ within $O(1)$ time (by \Cref{proposition:crb_termination}).
Therefore, the termination and agreement properties of $\mathcal{MBA}_1$ guarantee that all correct processes eventually decide some $z' \in \valuedigest \cup \{ \botcrb, \botmba \}$.
Given that $\mathcal{MBA}_1$ terminates in $O(1)$ expected time, all correct processes decide from $\mathcal{MBA}_1$ in $O(1)$ expected time.
We now differentiate two cases:
\begin{compactitem}
    \item Let $z' \neq \botmba$.
    In this case, all correct processes propose to $\mathcal{MBA}_2$ within $O(1)$ time in expectation.

    \item Let $z' = \botmba$.
    The strong unanimity property of $\mathcal{MBA}_1$ proves that at least two different values from the $\valuedigest \cup \{ \botcrb \}$ set have been proposed to $\mathcal{MBA}_1$ by correct processes.
    Therefore, there are at least two different values delivered from $\mathcal{CRB}$ by correct processes.
    The totality property guarantees that all correct processes eventually deliver (at least) two values from $\mathcal{CRB}$, and they do so in additional $O(1)$ time (by \Cref{proposition:crb_totality}).
    Thus, all correct processes eventually propose to $\mathcal{MBA}_2$ even in this case.
    Concretely, all correct processes propose to $\mathcal{MBA}_2$ within $O(1)$ time in expectation.
\end{compactitem}
As all correct processes propose to $\mathcal{MBA}_2$ within $O(1)$ time in expectation (in any of the two cases), the termination property of $\mathcal{MBA}_2$ and its $O(1)$ expected time complexity ensure that all correct processes decide from $\mathcal{MBA}_2$ in additional $O(1)$ time in expectation.
Thus, \smba indeed terminates in $O(1)$ expected time, which concludes the proof.
\end{proof}

\smallskip
\noindent \textbf{Proof of complexity.}
To prove \smba's complexity, we prove the following lemma.

\begin{lemma} [\smba's expected complexity] 
Given $t < \frac{1}{4}n$, the following holds for \smba (see \Cref{algorithm:async_smba}) in the presence of a computationally unbounded adversary:
\begin{compactitem}
    \item The expected message complexity is $O(n^2)$.

    \item The expected bit complexity is $O(n^2 \kappa)$.

\end{compactitem} 
\end{lemma}
\begin{proof}
	As $\mathcal{CRB}$, $\mathcal{MBA}_1$, and $\mathcal{MBA}_2$ exchange $O(n^2)$ messages and $O(n^2 \kappa)$ bits in expectation, the expected message complexity of \smba is $O(n^2)$ and its expected bit complexity is $O(n^2 \kappa)$.
\end{proof}

\smallskip
\noindent \textbf{\smba for strong consensus.}
As noted earlier in \Cref{section:preliminaries}, \smba can be easily adapted to solve the strong consensus~\cite{fitzi2003efficient} problem with up to $x = |\mathcal{X}|$ predetermined proposals $\mathcal{X}$, for any $x \in O(1)$.
First, our \crb algorithm, when operating among $n = (x + 1)t + 1$ processes with up to $x$ predetermined values, should be modified to never send \textsc{broken} messages (see \Cref{algorithm:crb}).
The \smba algorithm modified for strong consensus operates as follows.
\begin{compactenum}
    \item Each correct process broadcasts its proposal using the \crb algorithm.

    \item Each correct process delivers a value $v \in \mathcal{X}$ from \crb.
    Due to the justification property of the \crb algorithm, $v$ was proposed by a correct process to \smba.

    \item For each $\mathit{index} = 1, 2, ..., x - 1$, execute the following logic:
    \begin{compactenum}
        \item Propose $v$ to the $\mathit{index}$-th instance of the MBA primitive.

        \item Decide some value $v' \in \mathcal{X} \cup \{ \botmba \}$.
        Now, we differentiate two cases:
        \begin{compactitem}
            \item If $v' \neq \botmba$, then decide $v'$ and terminate.
            Importantly, the justification property of the MBA primitive ensures that $v'$ was proposed by a correct process (to \smba).

            \item If $v' = \botmba$, then not all correct processes proposed the same value to the $\mathit{index}$-th instance of the MBA primitive.
            Hence, wait to deliver $\mathit{index} + 1$ different values from \crb.
            Once this happens (and it will due to the totality property of the \crb algorithm), update $v$ to the lexicographically smallest value delivered from \crb.
            The justification property of the \crb algorithm ensures that $v$ was proposed by a correct process to \smba.
        \end{compactitem}
    \end{compactenum}

    \item Decide the lexicographically smallest value from the $\mathcal{X}$ set.
    If this step is reached, it means that all values from the $\mathcal{X}$ set have been proposed by correct processes to \smba; otherwise, correct processes would have decided in one of the $x - 1$ iterations.
\end{compactenum}
As $|\mathcal{X}| = x \in O(1)$, the modified \smba algorithm exchanges $O(xn^2)$ messages and $O(x n^2 \ell)$ bits (where $\ell$ denotes the bit-size of the values from the $\mathcal{X}$ set), and terminates in $O(x)$ time. 
\section{\comm: Proof} \label{section:reducer_proof}

This section formally proves the correctness and complexity of our MVBA algorithm \comm.
Recall that \comm's pseudocode is given in \Cref{algorithm:reducer}.

\subsection{Proof of Correctness}

This subsection formally proves \Cref{theorem:reducer_correct}.
Recall that $n = 4t + 1$.

\smallskip
\noindent \textbf{External validity.}
We start by proving that \comm satisfies external validity.

\begin{lemma} [\comm satisfies external validity] \label{lemma:reducer_external_validity}
Given $n = 4t + 1$ and the existence of a collision-resistant hash function, \comm (see \Cref{algorithm:reducer}) satisfies external validity in the presence of a computationally bounded adversary.
\end{lemma}
\begin{proof}
The external validity property is satisfied as every value quasi-decided by a correct process is valid due to the check at line~\ref{line:final_check}.
\end{proof}

\smallskip
\noindent \textbf{Agreement.}
Next, we prove \comm's agreement.
We say that a correct process $p_i$ \emph{quasi-decides} a vector $\mathit{vec}$ in an iteration $k \in \mathbb{N}$ if and only if $\mathit{quasi\_decisions}_i = \mathit{vec}$ when process $p_i$ reaches line~\ref{line:check_quasi_decisions} in iteration $k$.
We now prove that, for every iteration $k \in \mathbb{N}$, different vectors cannot be quasi-decided by correct processes.

\begin{proposition} \label{proposition:quasi_decided_vectors}
Let $k \in \mathbb{N}$ be any iteration.
Suppose a correct process $p_i$ quasi-decides a vector $\mathit{vec}_i$ in iteration $k$ and another correct process $p_j$ quasi-decides a vector $\mathit{vec}_j$ in iteration $k$.
Then, $\mathit{vec}_i = \mathit{vec}_j$.
\end{proposition}
\begin{proof}
The proposition follows from the agreement property of the $\mathcal{MBA}[k][x]$ instance (of the MBA primitive), for every $x \in \{ 1, 2, 3 \}$.
\end{proof}

We are now ready to prove that \comm ensures agreement.

\begin{lemma} [\comm satisfies agreement]
Given $n = 4t + 1$ and the existence of a collision-resistant hash function, \comm (see \Cref{algorithm:reducer}) satisfies agreement in the presence of a computationally bounded adversary.
\end{lemma}
\begin{proof}
By contradiction, suppose (1) there exists a correct process $p_i$ that decides a value $v_i$, and (2) there exists a correct process $p_j$ that decides a value $v_j \neq v_i$.
Let $p_i$ (resp., $p_j$) decide $v_i$ (resp., $v_j$) in some iteration $k_i \in \mathbb{N}$ (resp., $k_j \in \mathbb{N}$).
Therefore, process $p_i$ (resp., $p_j$) quasi-decides $v_i$ (resp., $v_j$) in iteration $k_i$ (resp., $k_j$).
Without loss of generality, let $k_i \leq k_j$.

As process $p_i$ quasi-decides $v_i$ in iteration $k_i$, process $p_i$ quasi-decides a vector $\mathit{vec}_i$ in iteration $k_i$; note that $v_i$ belongs to $\mathit{vec}_i$.
By \Cref{proposition:quasi_decided_vectors}, process $p_j$ also quasi-decides the non-empty vector $\mathit{vec}_i$ in iteration $k_i$.
We separate two cases:
\begin{compactitem}
    \item Let $k_i = k_j$.
    Due to the fact that the $\mathsf{Index}()$ request invoked in iteration $k_i = k_j$ returns the same integer to all correct processes, we have that $v_j = v_i$.
    Thus, we reach a contradiction with $v_j \neq v_i$ in this case.

    \item Let $k_i < k_j$.
    As $p_j$ quasi-decides the non-empty vector $\mathit{vec}_i$ in iteration $k_i$, we reach a contradiction with the fact that $p_j$ decides in iteration $k_j > k_i$.
\end{compactitem}
As neither of the above cases can occur, the agreement property is satisfied.
\end{proof}

\smallskip
\noindent \textbf{Weak validity.}
To prove that \comm satisfies weak validity, we first show that if any correct process proposes a value $v $ to the $\mathcal{MBA}[k][x]$ instance, for any sub-iteration $(k, x)$, and all processes are correct, then $v$ is the proposal of a correct process.

\begin{proposition} \label{proposition:integrity_proposal_correct}
Let $(k \in \mathbb{N}, x \in \{ 1, 2, 3 \})$ be any sub-iteration and let all processes be correct.
If any correct process $p_i$ proposes a value $v$ to $\mathcal{MBA}[k][x]$, then $v$ is the proposal of a correct process.
\end{proposition}
\begin{proof}
Recall that $\mathsf{leader}(k)$ denotes the leader of iteration $k$.
We distinguish two cases:
\begin{compactitem}
    \item Let $p_i$ execute line~\ref{line:reducer_decode}.
    In this case, process $p_i$ has received (at least) $t + 1$ RS symbols.
    Given that all processes are correct, all these RS symbols are sent by $\mathsf{leader}(k)$ during the dissemination phase and they all correspond to $\mathsf{leader}(k)$'s proposal.
    Therefore, $v$ is the proposal of $\mathsf{leader}(k)$, which proves the statement of the proposition in this case.

    \item Let $p_i$ execute line~\ref{line:reducer_own_proposal}.
    The statement of the proposition trivially holds in this case as $v$ is $p_i$'s proposal.
\end{compactitem}
As the statement of the proposition holds in both cases, the proof is concluded.
\end{proof}

The following lemma proves that \comm satisfies weak validity.

\begin{lemma} [\comm satisfies weak validity]
Given $n = 4t + 1$ and the existence of a collision-resistant hash function, \comm (see \Cref{algorithm:reducer}) satisfies weak validity in the presence of a computationally bounded adversary.
\end{lemma}
\begin{proof}
Suppose all processes are correct.
Moreover, let a correct process $p_i$ decide some value $v$; note that value $v$ must be valid by \Cref{lemma:reducer_external_validity}.
Hence, process $p_i$ quasi-decides $v$ in some iteration $k \in \mathbb{N}$, which further implies that $v$ is decided from $\mathcal{MBA}[k][x]$ in some sub-iteration $(k, x \in \{1 , 2, 3\})$.
Given that $v$ is valid and $\botmba$ is invalid, $v \neq \botmba$.
Thus, the justification property of $\mathcal{MBA}[k][x]$ guarantees that $v$ is proposed to $\mathcal{MBA}[k][x]$ by a correct process.
\Cref{proposition:integrity_proposal_correct} then proves that $v$ is the proposal of a correct process, which concludes the proof.
\end{proof}

\smallskip
\noindent \textbf{Integrity.}
Next, we prove \comm's integrity.

\begin{lemma} [\comm satisfies integrity]
Given $n = 4t + 1$ and the existence of a collision-resistant hash function, \comm (see \Cref{algorithm:reducer}) satisfies integrity in the presence of a computationally bounded adversary.
\end{lemma}
\begin{proof}
The lemma trivially holds due to the check at line~\ref{line:check_quasi_decisions}.
\end{proof}

\smallskip
\noindent \textbf{Termination.}
Next, we prove that \comm satisfies termination.
We say that a correct process \emph{completes the dissemination phase} if and only if the process executes line~\ref{line:dissemination_complete}.
The following proposition proves that at least one correct process completes the dissemination phase (and thus starts the first iteration of \comm).

\begin{proposition} \label{lemma:one_process_completes_dispersion}
At least one correct process completes the dissemination phase.
\end{proposition}
\begin{proof}
By contradiction, suppose no correct process completes the dissemination phase.
Hence, no correct process stops responding with \textsc{ack} messages (line~\ref{line:send_ack}) upon receiving \textsc{init} messages (line~\ref{line:receive_init}).
As there are (at least) $n - t$ correct processes, every correct process eventually broadcasts a \textsc{done} message (line~\ref{line:broadcast_done}).
Similarly, every correct process eventually receives $n - t$ \textsc{done} messages (line~\ref{line:receive_quorum_done}) and broadcasts a \textsc{finish} message (line~\ref{line:first_finish}).
Thus, every correct process eventually receives a \textsc{finish} message from (at least) $n - t$ processes (line~\ref{line:receive_quorum_finish}) and completes the dissemination phase (line~\ref{line:dissemination_complete}), thus contradicting the fact that no correct process completes the dissemination phase.
\end{proof}

The following proposition proves that if a correct process completes the dissemination phase, then all correct processes eventually complete the dissemination phase.

\begin{proposition} \label{lemma:if_one_completes_dispersion_everyone_completes_dispersion}
If any correct process completes the dissemination phase, then every correct process eventually completes the dissemination phase.
\end{proposition}
\begin{proof}
Let $p_i$ be any correct process that completes the dissemination phase.
This implies that $p_i$ receives $n - t = 3t + 1$ \textsc{finish} messages (line~\ref{line:receive_quorum_finish}), out of which (at least) $2t + 1$ messages are sent by correct processes.
Therefore, every correct process eventually receives $2t + 1 \geq t + 1$ \textsc{finish} messages (line~\ref{line:receive_plurality_finish}) and broadcasts its \textsc{finish} message (line~\ref{line:second_finish}).
Given that there are (at least) $n - t$ correct processes, every correct process eventually receives $n - t$ \textsc{finish} messages (line~\ref{line:receive_quorum_finish}) and completes the dissemination phase (line~\ref{line:dissemination_complete}).
\end{proof}

We are ready to prove that every correct process eventually completes the dissemination phase.

\begin{proposition} \label{lemma:everyone_completes_dispersion}
Every correct process eventually completes the dissemination phase.
\end{proposition}
\begin{proof}
The proposition follows directly from \cref{lemma:one_process_completes_dispersion,lemma:if_one_completes_dispersion_everyone_completes_dispersion}.
\end{proof}

Recall that the specification of the SMBA primitive (see \Cref{mod:async_smba}) assumes that only $O(1)$ different proposals are input by correct processes.
Hence, to prove that the $\mathcal{SMBA}$ instances utilized in \comm operate according to their specification, we now prove that only $O(1)$ different proposals are input by correct processes to any $\mathcal{SMBA}$ instance.
Recall that we say that a correct process $p_i$ \emph{suggests} a digest $z$ in an iteration $k \in \mathbb{N}$ if and only if $p_i$ broadcasts a \textsc{suggest} message with digest $z$ in iteration $k$ (line~\ref{line:broadcast_suggest}).
Let $\mathsf{suggested}_i(k)$ denote the set of digests suggested by any correct process $p_i$ in any iteration $k \in \mathbb{N}$.
The following proposition proves that each correct process suggests at most two digests in every iteration.

\begin{proposition} \label{lemma:proposed_bound}
For every correct process $p_i$ and every iteration $k \in \mathbb{N}$, $|\mathsf{suggested}_i(k)| \leq 2$.
\end{proposition}
\begin{proof}
For every digest $z \in \mathsf{suggested}_i(k)$, process $p_i$ receives (at least) $n - 3t = t + 1$ \textsc{stored} messages in iteration $k$ (due to the check at line~\ref{line:check_stored}).
As $p_i$ receives $n - t = 3t + 1$ \textsc{stored} messages (line~\ref{line:wait_for_stored_messages}) before broadcasting its \textsc{suggest} message, there can be at most $\frac{3t + 1}{t + 1} < 3$ suggested digests, which concludes the proof.
\end{proof}

Recall that we say a correct process $p_i$ \emph{$1$-commits} (resp., \emph{$2$-commits}) a digest $z$ in an iteration $k \in \mathbb{N}$ if and only if $\mathit{candidates}_i[1] = z$ (resp., $\mathit{candidates}_i[2] = z$) when process $p_i$ reaches line~\ref{line:agreeing_phase} in iteration $k$.
Similarly, a correct process $p_i$ \emph{commits} a digest $z$ in an iteration $k \in \mathbb{N}$ if and only if $p_i$ $1$-commits or $2$-commits $z$ in iteration $k$.
We denote by $\mathsf{committed}_i(k)$ the set of digests a correct process $p_i$ commits in an iteration $k \in \mathbb{N}$.
The following proposition proves that any correct process $p_i$ commits only digests previously suggested by $p_i$ (in the same iteration) or the default digest $\mathit{default}$ (line~\ref{line:reducer_fixed_default}).

\begin{proposition} \label{lemma:precommitted_proposed}
For every correct process $p_i$ and every iteration $k \in \mathbb{N}$, the following holds: $$\mathsf{committed}_i(k) \subseteq (\mathsf{suggested}_i(k) \cup \{ \mathit{default} \}).$$
\end{proposition}
\begin{proof}
Consider any digest $z \in \mathsf{committed}_i(k)$.
To prove the proposition, it suffices to show that if $z \in \mathsf{committed}_i(k)$ and $z \neq \mathit{default}$, then $z \in \mathsf{suggested}_i(k)$.
By contradiction, suppose (1) $z \in \mathsf{committed}_i(k)$, (2) $z \neq \mathit{default}$, and (3) $z \notin \mathsf{suggested}_i(k)$.
Let us consider three possibilities:
\begin{compactitem}
    \item Let $\mathit{candidates}_i.\mathsf{size} = 0$ when $p_i$ reaches line~\ref{line:reducer_if_size_0}.
    This case is impossible as $p_i$ commits only $\mathit{default}$ in this case, i.e., $z \neq \mathit{default}$ is not committed.

    \item Let $\mathit{candidates}_i.\mathsf{size} = 1$ when $p_i$ reaches line~\ref{line:reducer_if_size_0}.
    As $z \notin \mathsf{suggested}_i(k)$, this case is also impossible as $z$ is not committed by $p_i$.

    \item Let $\mathit{candidates}_i.\mathsf{size} = 2$ when $p_i$ reaches line~\ref{line:reducer_if_size_0}.
    Again, this case cannot occur as $z$ cannot be committed given that $z \notin \mathsf{suggested}_i(k)$.
\end{compactitem}
As neither of the three cases can occur, the proposition holds.
\end{proof}

Next, we prove that if any correct process commits a digest $z \neq \mathit{default}$ in any iteration $k$, then (at least) $t + 1$ correct processes suggest $z$ in iteration $k$.

\begin{proposition} \label{lemma:precommitted_supported_by_t+1}
Consider any correct process $p_i$ and any iteration $k \in \mathbb{N}$.
If process $p_i$ commits a digest $z \neq \mathit{default}$, then (at least) $t + 1$ correct processes suggest $z$ in iteration $k$.
\end{proposition}
\begin{proof}
If process $p_i$ commits $z \neq \mathit{default}$, \Cref{lemma:precommitted_proposed} proves that $z \in \mathsf{suggested}_i(k)$.
Hence, for process $p_i$ to not remove $z$ from its $\mathit{candidates}_i$ list after receiving $n - t = 3t + 1$ \textsc{suggest} messages (line~\ref{line:wait_for_suggest_messages}), $p_i$ must have received at least $2t + 1$ \textsc{suggest} messages for $z$ (line~\ref{line:rule_precommit}).
Therefore, at least $2t + 1 - t = t + 1$ correct processes suggest $z$ in iteration $k$.
\end{proof}

Recall that, for any iteration $k \in \mathbb{N}$, we define the set $\mathsf{committed}(k)$:
\begin{equation*}
    \mathsf{committed}(k) = \{ z \,|\, z \text{ is committed by a correct process in iteration $k$}\}.
\end{equation*}
The following proposition proves that $|\mathsf{committed}(k)| \in O(1)$, for any iteration $k$.

\begin{proposition} \label{lemma:iteration_all_committed_bound}
For every iteration $k \in \mathbb{N}$, $|\mathsf{committed}(k)| \in O(1)$.
\end{proposition}
\begin{proof}
\Cref{lemma:precommitted_supported_by_t+1} proves that every committed digest $z \neq \mathit{default}$ is suggested by at least $t + 1$ correct processes.
Moreover, \Cref{lemma:proposed_bound} proves that each correct process suggests at most two digests in each iteration.
Therefore, there can be at most $\frac{2(4t + 1)}{t + 1} = \frac{8t + 2}{t + 1} < 8$ non-$\mathit{default}$ digests that are committed by correct processes in iteration $k$.
Finally, as the special default digest $\mathit{default}$ can also be committed, $|\mathsf{committed}(k)| < 9$, which concludes the proof.
\end{proof}

Finally, we are ready to prove that only $O(1)$ different digests are proposed by correct processes to any $\mathcal{SMBA}$ instance.

\begin{proposition} \label{lemma:smba_correct}
Let $(k \in \mathbb{N}, x \in \{1, 2, 3\})$ be any sub-iteration.
Then, only $O(1)$ different digests are proposed to $\mathcal{SMBA}[k][x]$ by correct processes.
\end{proposition}
\begin{proof}
If a correct process $p_i$ proposes a digest $z$ to $\mathcal{SMBA}[k][x]$ (line~\ref{line:smba}), then $p_i$ commits $z$ in iteration $k$ (lines~\ref{line:acc_proposal_1},~\ref{line:acc_proposal_2} and~\ref{line:acc_proposal_3}).
Hence, $z \in \mathsf{committed}(k)$.
As $|\mathsf{committed}(k)| \in O(1)$ (by \Cref{lemma:iteration_all_committed_bound}), the proof is concluded.
\end{proof}

We now prove that if all correct processes start any iteration $k$, all correct processes eventually start sub-iteration $(k, 1)$.

\begin{proposition} \label{lemma:first_sub_iteration}
Let $k \in \mathbb{N}$ be any iteration such that all correct processes start iteration $k$.
Then, all correct processes eventually start sub-iteration $(k, 1)$.
\end{proposition}
\begin{proof}
As all correct processes start iteration $k$ and there are at least $n - t = 3t + 1$ correct processes, all correct processes eventually receive $n - t = 3t + 1$ \textsc{stored} messages (line~\ref{line:wait_for_stored_messages}) and broadcast a \textsc{suggest} message (line~\ref{line:broadcast_suggest}).
Therefore, all correct processes eventually receive $n - t = 3t + 1$ \textsc{suggest} messages (line~\ref{line:wait_for_suggest_messages}) and start sub-iteration $(k, 1)$ at line~\ref{line:agreeing_phase}.
\end{proof}

Next, we prove that if all correct processes start any sub-iteration $(k, x)$, then all correct processes eventually complete sub-iteration $(k, x)$.

\begin{proposition} \label{lemma:sub_iteration_completion}
Let $(k \in \mathbb{N}, x \in \{1, 2, 3\})$ be any sub-iteration such that all correct processes start sub-iteration $(k, x)$.
Then, all correct processes eventually complete sub-iteration $(k, x)$.
\end{proposition}
\begin{proof}
Given that all correct processes start sub-iteration $(k, x)$, all correct processes propose to $\mathcal{SMBA}[k][x]$ (line~\ref{line:smba}).
By \Cref{lemma:smba_correct}, $\mathcal{SMBA}[k][x]$ behaves according to its specification (see \Cref{mod:async_smba}).
Thus, the termination property of $\mathcal{SMBA}[k][x]$ ensures that all correct processes eventually decide from $\mathcal{SMBA}[k][x]$.
Then, all correct processes broadcast a \textsc{reconstruct} message (line~\ref{line:reducer_reconstruct}), thus ensuring that every correct process eventually receives $n - t = 3t + 1$ \textsc{reconstruct} messages (as there are at least $n - t = 3t + 1$ correct processes).
Hence, all correct processes eventually propose to $\mathcal{MBA}[k][x]$ (line~\ref{line:lmba}).
The termination property of $\mathcal{MBA}[k][x]$ ensures that all correct processes complete sub-iteration $(k, x)$, thus concluding the proof.
\end{proof}

The following proposition proves that every iteration and sub-iteration are eventually started and completed by all correct processes.

\begin{proposition} \label{lemma:every_sub_iteration}
Every iteration and sub-iteration are eventually started and completed by all correct processes.
\end{proposition}
\begin{proof}
By \Cref{lemma:everyone_completes_dispersion}, all correct processes start iteration $1$.
Therefore, \Cref{lemma:first_sub_iteration} proves that all correct processes start sub-iteration $(1, 1)$.
By inductively applying \cref{lemma:sub_iteration_completion,lemma:first_sub_iteration}, we prove that every sub-iteration (and, thus, every iteration) is eventually started and completed by all correct processes.
\end{proof}

To not pollute the presentation, we might not explicitly rely on \Cref{lemma:every_sub_iteration} in the rest of the proof.
Recall that, by \Cref{definition:good_iteration}, an iteration $k \in \mathbb{N}$ is good if and only if $\mathsf{leader}(k) \in \dfirst$.
Recall that $\first$ is the first process that broadcasts a \textsc{finish} message while so-far-uncorrupted, and that $\dfirst$ contains the \textsc{done} senders that are so-far-uncorrupted immediately before that broadcast.
The set $\dfirst$ is frozen at that moment, before any election is disclosed, and $|\dfirst| \geq n - 2t = 2t + 1$.
Moreover, recall that, for every good iteration $k$, (1) $v^{\star}(k)$ denotes the valid proposal disseminated by $\mathsf{leader}(k)$ while it is uncorrupted, and (2) $z^{\star}(k)$ denotes the digest of $v^{\star}(k)$.
The proposition below proves that every correct process suggests $z^{\star}(k)$ in any good iteration $k$.

\begin{proposition} \label{lemma:good_iteration_all_correct_suggest}
Let $k \in \mathbb{N}$ be any good iteration.
Then, every correct process suggests $z^{\star}(k)$ in iteration $k$.
\end{proposition}
\begin{proof}
As $\mathsf{leader}(k) \in \dfirst$, $\mathsf{leader}(k)$ broadcasts \textsc{done} while uncorrupted and therefore has previously stored $z^{\star}(k)$ at $n - t = 3t + 1$ processes in the dissemination phase, out of which at most $t$ are faulty in the complete execution.
Thus, $\mathsf{leader}(k)$ stores $z^{\star}(k)$ at $\geq 2t + 1$ correct processes.
This further implies that each correct process receives $z^{\star}(k)$ in \textsc{stored} messages from at least $t + 1$ processes, which means that each correct process broadcasts a \textsc{suggest} message with $z^{\star}(k)$ (due to the check at line~\ref{line:check_stored}) and, thus, suggests $z^{\star}(k)$ in iteration $k$.
\end{proof}

Next, we prove that every correct process commits $z^{\star}(k)$ in any good iteration $k$.

\begin{proposition} \label{lemma:good_iteration_all_correct_commit}
Let $k \in \mathbb{N}$ be any good iteration.
Then, every correct process commits $z^{\star}(k)$ in iteration $k$.
\end{proposition}
\begin{proof}
By \Cref{lemma:good_iteration_all_correct_suggest}, all correct processes suggest $z^{\star}(k)$ in iteration $k$.
Hence, for each correct process $p_i$, $\mathit{candidates}_i[1] = z^{\star}(k)$ or $\mathit{candidates}_i[2] = z^{\star}(k)$ when process $p_i$ broadcasts a \textsc{suggest} message (line~\ref{line:broadcast_suggest}).
Furthermore, each correct process $p_i$ receives a \textsc{suggest} message with $z^{\star}(k)$ from at least $n - t - t = 2t + 1$ processes, which means that $p_i$ does not remove $z^{\star}(k)$ from $\mathit{candidates}_i$ at line~\ref{line:candidates_remove}.
Hence, the proposition holds.
\end{proof}

The following proposition proves that if any correct process commits any digest (including $\mathit{default}$) in any good iteration $k$, then (at least) $2t + 1 - f$ correct processes suggest $z$ in iteration $k$, where $f \leq t$ denotes the \emph{actual} number of faulty processes.

\begin{proposition} \label{lemma:good_iteration_committed_suggested}
Consider any correct process $p_i$ and any good iteration $k \in \mathbb{N}$.
If process $p_i$ commits a digest $z$ ($z$ might be equal to $\mathit{default}$), then (at least) $2t + 1 - f$ correct processes suggest $z$ in iteration $k$, where $f \leq t$ denotes the actual number of faulty processes.
\end{proposition}
\begin{proof}
By \Cref{lemma:good_iteration_all_correct_suggest}, all correct processes suggest $z^{\star}(k)$ in iteration $k$.
Thus, when process $p_i$ reaches line~\ref{line:reducer_if_size_0}, $\mathit{candidates}_i.\mathsf{size} > 0$ as $\mathit{candidates}_i[1] = z^{\star}(k)$ or $\mathit{candidates}_i[2] = z^{\star}(k)$.
Hence, process $p_i$ does not execute line~\ref{line:reducer_set_default}.
Therefore, for each digest $z$ that process $p_i$ commits, $z$ is suggested by at least $2t + 1 - f$ correct processes in iteration $k$ (due to the check at line~\ref{line:rule_precommit}).
\end{proof}

The following proposition proves that $|\mathsf{committed}(k)| \leq 3$ in any good iteration.

\begin{proposition} \label{lemma:good_iteration_bound_committed}
Let $k \in \mathbb{N}$ be any good iteration $k$.
Then, $|\mathsf{committed}(k)| \leq 3$.
\end{proposition}
\begin{proof}
\Cref{lemma:good_iteration_all_correct_commit} proves that all correct processes commit $z^{\star}(k)$ in iteration $k$.
Thus, $z^{\star}(k) \in \mathsf{committed}(k)$.
To prove the proposition, we analyze the cardinality of the $\mathsf{committed}(k) \setminus{ \{z^{\star}(k) \}}$ set; let $X = |\mathsf{committed}(k) \setminus{ \{z^{\star}(k) \}}|$.

For every digest $z \in \mathsf{committed}(k) \setminus { \{ z^{\star}(k) \} }$, \Cref{lemma:good_iteration_committed_suggested} proves that $z$ is suggested by (at least) $2t + 1 - f$ correct processes in iteration $k$.
By \Cref{lemma:good_iteration_all_correct_suggest}, all correct processes suggest $z^{\star}(k)$ in iteration $k$.
Given that each correct process suggests at most two digests (by \Cref{lemma:proposed_bound}), there are at most $n - f$ ``correct suggestions'' for non-$z^{\star}(k)$ digests.
By contradiction, suppose $X \geq 3$.
Therefore, there are at least $X (2t + 1 - f) \geq 3 (2t + 1 - f) = 6t + 3 - 3f$ ``correct suggestions'' for non-$z^{\star}(k)$ digests.
Thus, we have $n - f = 4t + 1 - f \geq 6t + 3 - 3f$, which implies $2f \geq 2t + 2$ and $f \geq t + 1$.
This is impossible as $f \leq t$, which means $X < 3$, thus concluding the proof.
\end{proof}

For every iteration $k \in \mathbb{N}$ and every $x \in \{ 1, 2 \}$, let us define the $\mathsf{committed}(k, x)$ set in the following way:
\begin{equation*}
    \mathsf{committed}(k, x) = \{ z \,|\, \text{$z$ is $x$-committed by a correct process in iteration $k$} \}.
\end{equation*}
Observe that the following holds: (1) $\mathsf{committed}(k, 1) \subseteq \mathsf{committed}(k)$, and (2) $\mathsf{committed}(k, 2) \subseteq \mathsf{committed}(k)$.
The following proposition proves that \Cref{equation:crucial} holds in any good iteration.

\begin{proposition} \label{lemma:good_iteration_bound_collumn}
Let $k \in \mathbb{N}$ be any good iteration.
Then, the following holds:
\begin{equation*}
	\big( z^{\star}(k) \in \mathsf{committed}(k, 1) \land |\mathsf{committed}(k, 1)| \leq 2 \big) \lor \big( \{ z^{\star}(k) \} = \mathsf{committed}(k, 2) \big).
\end{equation*}
\end{proposition}
\begin{proof}
Recall that \Cref{lemma:good_iteration_all_correct_commit} proves that every correct process commits $z^{\star}(k)$ in iteration $k$.
(Thus, $z^{\star}(k) \in \mathsf{committed}(k)$.)
To prove the proposition, we consider three possible cases:
\begin{compactitem}
    \item Let $z^{\star}(k)$ be the lexicographically smallest digest in $\mathsf{committed}(k)$.
    In this case, every correct process $1$-commits $z^{\star}(k)$ in iteration $k$.
    (Otherwise, $z^{\star}(k)$ could not be the smallest digest among $\mathsf{committed}(k)$.)
    Thus, $\mathsf{committed}(k, 1) = \{ z^{\star}(k) \}$.

    \item Let $z^{\star}(k)$ be the lexicographically greatest digest in $\mathsf{committed}(k)$.
    In this case, every correct process $2$-commits $z^{\star}(k)$ in iteration $k$.
    Hence, $\mathsf{committed}(k, 2) = \{ z^{\star}(k) \}$.

    \item Let $z^{\star}(k)$ not be the lexicographically smallest nor the lexicographically greatest digest in $\mathsf{committed}(k)$.
    Hence, there exist digests $z_1$ and $z_2$ in $\mathsf{committed}(k)$ such that (1) $z_1 < z^{\star}(k)$, and (2) $z^{\star}(k) < z_2$.
    (\Cref{lemma:good_iteration_bound_committed} then shows that $\mathsf{committed}(k) = \{ z_1, z^{\star}(k), z_2 \}$.)
    As (1) $z_2 \in \mathsf{committed}(k)$, (2) every correct process commits $z^{\star}(k) < z_2$ in iteration $k$ (by \Cref{lemma:good_iteration_all_correct_commit}), and (3) each correct process commits at most two different digests, there exists a correct process that $1$-commits $z^{\star}(k)$, i.e., $z^{\star}(k) \in \mathsf{committed}(k, 1)$.
    Moreover, since (1) every correct process commits $z^{\star}(k)$ in iteration $k$ (by \Cref{lemma:good_iteration_all_correct_commit}), and (2) $z^{\star}(k) < z_2$, no correct process $1$-commits $z_2$ in iteration $k$.
    Therefore, $z^{\star}(k) \in \mathsf{committed}(k, 1)$ and $|\mathsf{committed}(k, 1)| \leq 2$, which concludes the proof in this case.
\end{compactitem}
As the statement of the proposition holds in each scenario, the proof is concluded.
\end{proof}

We now prove that all correct processes quasi-decide $v^{\star}(k)$ in a good iteration.

\begin{proposition} \label{lemma:good_iteration_termination}
Let $k \in \mathbb{N}$ be any good iteration.
Then, all correct processes quasi-decide $v^{\star}(k)$ in iteration $k$.
\end{proposition}
\begin{proof}
By \Cref{lemma:good_iteration_bound_collumn}, the following holds:
\begin{equation*}
	\big( z^{\star}(k) \in \mathsf{committed}(k, 1) \land |\mathsf{committed}(k, 1)| \leq 2 \big) \lor \big( \{ z^{\star}(k) \} = \mathsf{committed}(k, 2) \big).
\end{equation*}
Moreover, recall that \Cref{lemma:good_iteration_all_correct_commit} proves that all correct processes commit $z^{\star}(k)$ in iteration $k$.

We distinguish two cases:
\begin{compactitem}
    \item Let $\mathsf{committed}(k, 2) = \{ z^{\star}(k) \}$.
    Hence, all correct processes propose $z^{\star}(k)$ to $\mathcal{SMBA}[k][2]$.
    The strong validity and termination properties of $\mathcal{SMBA}[k][2]$ ensure that all correct processes decide $z^{\star}(k)$ from $\mathcal{SMBA}[k][2]$.
    As $k$ is a good iteration, correctly encoded RS symbols (that correspond to value $v^{\star}(k)$) are stored at $\geq n - 2t = 2t + 1$ correct processes.
    Therefore, each correct process receives RS symbols with correct witnesses for $z^{\star}(k)$ via \textsc{reconstruct} messages from at least $t + 1$ processes and decodes value $v^{\star}(k)$ (line~\ref{line:reducer_decode}).
    Next, all correct processes propose $v^{\star}(k)$ to $\mathcal{MBA}[k][2]$, which ensures that all correct processes decide $v^{\star}(k)$ (due to its strong unanimity and termination properties).
    As $v^{\star}(k)$ is valid, all correct processes indeed quasi-decide $v^{\star}(k)$ in iteration $k$.

    \item Let $|\mathsf{committed}(k, 1)| \leq 2$ and $z^{\star}(k) \in \mathsf{committed}(k, 1)$.
    In this case, correct processes propose to $\mathcal{SMBA}[k][1]$ at most two different digests (as $|\mathsf{committed}(k, 1)| \leq 2$) out of which one is $z^{\star}(k)$ (as $z^{\star}(k) \in \mathsf{committed}(k, 1)$).
    The termination and strong validity properties of $\mathcal{SMBA}[k][1]$ prove that all correct processes decide from $\mathcal{SMBA}[k][1]$ a digest that belongs to $\mathsf{committed}(k, 1)$.
    We now distinguish two possibilities:
\begin{compactitem}
    \item Let the digest decided from $\mathcal{SMBA}[k][1]$ be $z^{\star}(k)$.
    As $k$ is a good iteration, correctly encoded RS symbols (that correspond to value $v^{\star}(k)$) are stored at $\geq n - 2t = 2t + 1$ correct processes.
    This means that each correct process receives RS symbols with correct witnesses for $z^{\star}(k)$ via \textsc{reconstruct} messages from at least $t + 1$ processes and decodes value $v^{\star}(k)$ (line~\ref{line:reducer_decode}).
    Furthermore, all correct processes propose $v^{\star}(k)$ to $\mathcal{MBA}[k][1]$, which ensures that all correct processes decide $v^{\star}(k)$ (due to its strong unanimity and termination properties).
    As $v^{\star}(k)$ is valid, all correct processes quasi-decide $v^{\star}(k)$ in iteration $k$.

    \item Let the digest decided from $\mathcal{SMBA}[k][1]$ be $z \neq z^{\star}(k)$.
    Due to the strong validity property of $\mathcal{SMBA}[k][1]$ (recall that only two different digests are proposed to $\mathcal{SMBA}[k][1]$ by correct processes), $z \in \mathsf{committed}(k, 1)$.
    Therefore, the following holds:
    \begin{compactitem}
        \item All correct processes that proposed $z$ to $\mathcal{SMBA}[k][1]$ switch their proposal for $\mathcal{SMBA}[k][3]$: all those correct processes propose $z^{\star}(k)$ to $\mathcal{SMBA}[k][3]$ (line~\ref{line:acc_proposal_2}).
        (Recall that all correct processes commit $z^{\star}(k)$ in iteration $k$.)

        \item All correct processes that proposed $z^{\star}(k)$ to $\mathcal{SMBA}[k][1]$ propose $z^{\star}(k)$ to $\mathcal{SMBA}[k][3]$ (as the check at line~\ref{line:proposal_switching_start} does not pass).
    \end{compactitem}
    Therefore, all correct processes propose $z^{\star}(k)$ to $\mathcal{SMBA}[k][3]$.
    The strong validity and termination properties of $\mathcal{SMBA}[k][3]$ ensure that all correct processes decide $z^{\star}(k)$ from $\mathcal{SMBA}[k][3]$.
    As in the previous cases, correctly encoded RS symbols (that correspond to value $v^{\star}(k)$) are stored at $\geq n - 2t = 2t + 1$ correct processes.
    Therefore, each correct process receives RS symbols with correct witnesses for $z^{\star}(k)$ via \textsc{reconstruct} messages from at least $t + 1$ processes and decodes value $v^{\star}(k)$ (line~\ref{line:reducer_decode}).
    This means that all correct processes propose $v^{\star}(k)$ to $\mathcal{MBA}[k][3]$, which ensures that all correct processes decide $v^{\star}(k)$ (due to its strong unanimity and termination properties).
    Given that $v^{\star}(k)$ is a valid value, all correct processes quasi-decide $v^{\star}(k)$.
\end{compactitem}
\end{compactitem}
As the statement of the proposition holds in both cases, the proof is concluded.
\end{proof}

Finally, we are ready to prove \comm's termination.

\begin{lemma} [\comm satisfies termination] \label{theorem:termination}
Given $n = 4t + 1$ and the existence of a collision-resistant hash function, \comm (see \Cref{algorithm:reducer}) satisfies termination in the presence of a computationally bounded adversary.
Precisely, every correct process decides within $O(1)$ iterations in expectation.
\end{lemma}
\begin{proof}
\Cref{lemma:good_iteration_termination} proves that all correct processes quasi-decide in a good iteration $k \in \mathbb{N}$.
Let $P = \frac{2t + 1}{4t + 1}$ and fix any execution prefix immediately before $\mathsf{Election}(k)$ is disclosed.
As $\dfirst$ is frozen before the election of iteration $k$, and as that election is uniform over the $n = 4t + 1$ processes and independent of the fixed prefix, iteration $k$ is good with conditional probability
\begin{equation*}
    \Pr[\mathsf{leader}(k) \in \dfirst \mid \text{the fixed prefix}]
    = \frac{|\dfirst|}{n}
    \geq P.
\end{equation*}
Let $E_k$ denote the event that \comm does not terminate after the $k$-th iteration, and let $E_0$ denote the entire probability space.
Thus, $\Pr[E_k] \leq (1 - P)^k \to 0$ as $k \to \infty$.
Moreover, let $K$ be the random variable that denotes the number of iterations required for \comm to terminate.
Then,
\begin{equation*}
    \mathbb{E}[K] = \sum_{k = 0}^{\infty} \Pr[K > k]
    \leq \sum_{k = 0}^{\infty} (1 - P)^k
    = \frac{1}{P} < 2,
\end{equation*}
which proves that \comm terminates within $O(1)$ iterations in expectation.
\end{proof}

\smallskip
\noindent \textbf{Quality.}
To conclude \comm's proof of correctness, we now show that \comm satisfies the quality property.
Recall that the quality property requires that the probability of deciding an adversarial value is strictly less than 1.

\begin{lemma} [\comm satisfies quality]
Given $n = 4t + 1$ and the existence of a collision-resistant hash function, \comm (see \Cref{algorithm:reducer}) satisfies quality in the presence of a computationally bounded adversary.
\end{lemma}
\begin{proof}
\Cref{lemma:good_iteration_termination} proves that all correct processes quasi-decide value $v^{\star}(k)$ in a good iteration $k$; recall that $v^{\star}(k)$ is a non-adversarial value.
Hence, if the first iteration of \comm is good, all correct processes quasi-decide $v^{\star}(1)$ in iteration $1$.
In that case, for \comm to decide $v^{\star}(1)$, it is required that the $\mathsf{Index}()$ request selects $v^{\star}(1)$.
The probability of that happening is at least $\frac{1}{3}$ as there can be at most three quasi-decided values.
Thus, the probability that \comm decides a non-adversarial value is at least $\frac{2t + 1}{4t + 1} \cdot \frac{1}{3} > \frac{1}{6}$.
Therefore, quality is ensured as the probability that an adversarial value is decided is strictly less than $\frac{5}{6} < 1$.
\end{proof}

\subsection{Proof of Complexity}

We now formally prove the complexity of \comm (see \Cref{theorem:reducer_complexity}).

\smallskip
\noindent \textbf{Message complexity.}
We start by proving that \comm sends $O(n^2)$ messages in expectation.

\begin{lemma} [\comm's expected message complexity]
Given $n = 4t + 1$ and the existence of a collision-resistant hash function, the expected message complexity of \comm (see \Cref{algorithm:reducer}) is $O(n^2)$ in the presence of a computationally bounded adversary.
\end{lemma}
\begin{proof}
Correct processes exchange $O(n^2)$ messages in the dissemination phase.
Moreover, for every iteration $k \in \mathbb{N}$ and every execution prefix in which iteration $k$ is started, correct processes exchange $O(n^2)$ messages in iteration $k$ in expectation.
Crucially, whether iteration $k$ is reached is determined by the prefix of the execution that precedes iteration $k$.
Hence, the law of total expectation shows that correct processes exchange at most $O(n^2) \cdot \mathbb{E}[K] \subseteq O(n^2)$ messages in expectation throughout the iterations, where $K$ denotes the number of iterations \comm takes to terminate and $\mathbb{E}[K] \in O(1)$ (by \Cref{theorem:termination}).
\end{proof}

\smallskip
\noindent \textbf{Bit complexity.}
Next, we prove \comm's expected bit complexity.
We start by proving that correct processes send $O(n\ell + n^2 \kappa \log n)$ bits in the dissemination phase.

\begin{proposition} \label{lemma:dispersion_cost}
Correct processes send $O(n\ell + n^2 \kappa \log n)$ bits in the dissemination phase.
\end{proposition}
\begin{proof}
Each RS symbol is of size $O(\frac{\ell}{n} + \log n)$ bits.
Let $p_i$ be any correct process.
Process $p_i$ sends $n \cdot  O(\frac{\ell}{n} + \log n + \kappa + \kappa \log n) \subseteq O(\ell + n \kappa \log n)$ bits via \textsc{init} messages.
Moreover, process $p_i$ sends $O(n)$ bits via \textsc{ack}, \textsc{done} and \textsc{finish} messages.
Therefore, process $p_i$ sends
\begin{equation*}
    \begin{aligned}
    & \underbrace{O(\ell + n \kappa \log n)}_\text{\textsc{init}}
    {}+{}
    \underbrace{O(n)}_\text{\textsc{ack}, \textsc{done} \& \textsc{finish}}
    \\
    & \subseteq O(\ell + n \kappa \log n) \text{ bits in the dissemination phase.}
    \end{aligned}
\end{equation*}
This implies that correct processes collectively send $O(n\ell + n^2 \kappa \log n)$ bits.
\end{proof}

Next, we prove that correct processes send $O(n\ell + n^2 \kappa \log n + n^2 \log k)$ bits in any iteration $k$.

\begin{proposition} \label{lemma:iteration_cost}
Correct processes send $O(n \ell + n^2 \kappa \log n + n^2\log k)$ bits in expectation in any iteration $k \in \mathbb{N}$.
\end{proposition}
\begin{proof}
Correct processes send $O(n^2 \kappa + n^2\log k)$ bits via \textsc{stored} and \textsc{suggest} messages.
Recall that \textsc{stored} messages include a single digest and \textsc{suggest} messages include at most two digests. \textsc{stored} and \textsc{suggest} messages also include the number of the iteration to which they refer.

Now, consider any sub-iteration $(k, x \in \{1 , 2, 3 \})$.
Correct processes send $O(n^2 \kappa + n^2 \log k)$ bits in expectation while executing $\mathcal{SMBA}[k][x]$.
Note that the $O(n^2\log k)$ term exists as (1) each message of $\mathcal{SMBA}[k][x]$ needs to be tagged with ``$k, x \in \{1, 2, 3\}$'', and (2) $\mathcal{SMBA}[k][x]$ exchanges $O(n^2)$ messages in expectation.
Next, correct processes send $O(n \ell + n^2 \kappa \log n + n^2 \log k)$ bits via \textsc{reconstruct} messages.
Finally, correct processes send $O(n \ell + n^2 \kappa \log n + n^2 \log k)$ bits in expectation while executing $\mathcal{MBA}[k][x]$.
Again, the $O(n^2\log k)$ term exists as (1) each message of $\mathcal{MBA}[k][x]$ needs to be tagged with ``$k, x \in \{1, 2, 3\}$'', and (2) $\mathcal{MBA}[k][x]$ exchanges $O(n^2)$ messages in expectation.
Therefore, correct processes collectively send
\begin{equation*}
    \begin{aligned}
        & \underbrace{O(n^2\kappa + n^2 \log k)}_\text{$\mathcal{SMBA}[k][x]$}
        {}+{}
        \underbrace{O(n\ell + n^2 \kappa \log n + n^2 \log k)}_\text{\textsc{reconstruct}}
        {}+{}
        \underbrace{O(n\ell + n^2 \kappa \log n + n^2 \log k)}_\text{$\mathcal{MBA}[k][x]$}
        \\
        & \subseteq O(n \ell + n^2 \kappa \log n + n^2 \log k) \text{ bits in sub-iteration $(k, x)$.}
    \end{aligned}
\end{equation*}
Given that iteration $k$ has three sub-iterations, correct processes collectively send
\begin{equation*}
    \begin{aligned}
        & \underbrace{O(n^2 \kappa + n^2 \log k)}_\text{\textsc{stored} \& \textsc{suggest}}
        {}+{}
        \underbrace{O(n \ell + n^2 \kappa \log n + n^2 \log k)}_\text{sub-iterations}
        \\
        & \subseteq O(n \ell + n^2 \kappa \log n + n^2 \log k) \text{ bits in iteration $k$.}
    \end{aligned}
\end{equation*}
Thus, the proposition holds.
\end{proof}

We are ready to prove \comm's expected bit complexity.

\begin{lemma} [\comm's expected bit complexity] \label{lemma:reducer_bit_complexity}
Given $n = 4t + 1$ and the existence of a collision-resistant hash function, the expected bit complexity of \comm (see \Cref{algorithm:reducer}) is $$O(n\ell + n^2 \kappa \log n)$$ in the presence of a computationally bounded adversary.
\end{lemma}
\begin{proof}
Let $B$ be the random variable that denotes the number of bits sent by correct processes in \comm, and let $K$ be the random variable that denotes the number of iterations it takes \comm to terminate.
For every iteration $k \in \mathbb{N}$, let $\mathcal{H}(k)$ denote the prefix of the execution that precedes the moment at which iteration $k$ starts, and let $B_k$ denote the number of bits correct processes send in iteration $k$.
\Cref{lemma:dispersion_cost} proves that correct processes send $O(n \ell + n^2 \kappa \log n)$ bits in the dissemination phase.
Moreover, \Cref{lemma:iteration_cost} proves that
\begin{equation*}
    \mathbb{E}\big[ B_k \,|\, \mathcal{H}(k) \big] \leq c_k = O(n \ell + n^2 \kappa \log n + n^2 \log k),
\end{equation*}
for every iteration $k \in \mathbb{N}$ and every execution prefix $\mathcal{H}(k)$ in which iteration $k$ is started.
We underline that this bound must hold \emph{uniformly} over all such prefixes, as $K$ and $B_k$ are in general correlated.
Crucially, the event $\{ K \geq k \}$ is determined by $\mathcal{H}(k)$: whether \comm has terminated before iteration $k$ starts is determined by $\mathcal{H}(k)$.
Therefore, the law of total expectation gives $\mathbb{E}\big[ B_k \cdot \mathbf{1}\{ K \geq k \} \big] \leq c_k \cdot \text{Pr}[K \geq k]$, for every $k \in \mathbb{N}$, and thus
\begin{equation*}
    \mathbb{E}[B] \leq O(n \ell + n^2 \kappa \log n) + \sum\limits_{k = 1}^{\infty} c_k \cdot \text{Pr}[K \geq k].
\end{equation*}
Recall that the proof of \Cref{theorem:termination} establishes $\text{Pr}[K > k] \leq (1 - P)^{k}$, for every $k \in \mathbb{N}$, where $P \geq \frac{2t + 1}{4t + 1} \approx \frac{1}{2}$; we emphasize that this is a bound on the \emph{tail} of $K$, and not a claim that $K$ is a geometric random variable.
Hence, $\text{Pr}[K \geq k] \leq (1 - P)^{k - 1}$, which yields
\begin{equation*}
    \sum\limits_{k = 1}^{\infty} \text{Pr}[K \geq k] \leq \frac{1}{P} \in O(1)
    \quad \text{and} \quad
    \sum\limits_{k = 1}^{\infty} \text{Pr}[K \geq k] \cdot \log k \in O(1),
\end{equation*}
where the second series converges as the geometric decay dominates the logarithmic growth.
Thus, $\mathbb{E}[B] \in O(n \ell + n^2 \kappa \log n)$.
\end{proof}

\smallskip
\noindent \textbf{Time complexity.}
We conclude the subsection by proving \comm's expected time complexity.

\begin{lemma} [\comm's expected time complexity]
Given $n = 4t + 1$ and the existence of a collision-resistant hash function, the expected time complexity of \comm (see \Cref{algorithm:reducer}) is $O(1)$ in the presence of a computationally bounded adversary.
\end{lemma}
\begin{proof}
For every iteration $k \in \mathbb{N}$ and every execution prefix in which iteration $k$ is started, iteration $k$ takes $O(1)$ time in expectation.
Moreover, the dissemination phase takes $O(1)$ time.
Aggregating over the iterations exactly as in the proof of \Cref{lemma:reducer_bit_complexity}---and using that $\sum_{k = 1}^{\infty} \text{Pr}[K \geq k] = \mathbb{E}[K] \in O(1)$ (by \Cref{theorem:termination})---we conclude that \comm indeed has $O(1)$ expected time complexity.
\end{proof}

\section{\commplus: Proof} \label{section:reducer_plus_proof}

This section formally proves the correctness and complexity of our MVBA algorithm \commplus.
Recall that \commplus's pseudocode is given in \Cref{algorithm:reducer_plus}.

\subsection{Proof of Correctness}

This subsection proves the correctness of \commplus, i.e., we formally prove \Cref{theorem:reducer_plus_correct}.
Recall that $n = (3 + \epsilon)t + 1$, for any fixed constant $\epsilon > 0$, and $C = \lceil 1 + 2/\epsilon \rceil^2$.

\smallskip
\noindent \textbf{External validity.}
We start by proving \commplus's external validity.

\begin{lemma} [\commplus satisfies external validity] \label{lemma:reducer_plus_external_validity}
Given $n = (3 + \epsilon)t + 1$, for any fixed constant $\epsilon > 0$, and the existence of a hash function modeled as a random oracle, \commplus (see \Cref{algorithm:reducer_plus}) satisfies external validity in the presence of a computationally bounded adversary.
\end{lemma}
\begin{proof}
The property is satisfied as every value quasi-decided by a correct process is valid as ensured by the check at line~\ref{line:quasi_decide_plus}.
\end{proof}

\smallskip
\noindent \textbf{Agreement.}
We proceed by proving \commplus's agreement.
We say that a correct process $p_i$ \emph{quasi-decides} a vector $\mathit{vec}$ in an iteration $k \in \mathbb{N}$ if and only if $\mathit{quasi\_decisions}_i = \mathit{vec}$ when process $p_i$ reaches line~\ref{line:check_quasi_decisions_plus}.
The following proposition proves that no two correct processes quasi-decide different vectors in any iteration $k$.

\begin{proposition} \label{lemma:quasi_committed_vectors_plus}
Let $k \in \mathbb{N}$ be any iteration.
Suppose a correct process $p_i$ quasi-decides a vector $\mathit{vec}_i$ in iteration $k$ and another correct process $p_j$ quasi-decides a vector $\mathit{vec}_j$ in iteration $k$.
Then, $\mathit{vec}_i = \mathit{vec}_j$.
\end{proposition}
\begin{proof}
The proposition follows from the agreement property of the $\mathcal{MBA}[k][x]$ instance (of the MBA primitive), for every $x \in [1, 4C^2]$.
\end{proof}

Completed vectors for the same iteration also coincide among processes still uncorrupted at a given prefix: a continuation with no further corruptions makes them correct, so MBA agreement applies.

We are ready to prove \commplus's agreement.

\begin{lemma} [\commplus satisfies agreement]
Given $n = (3 + \epsilon)t + 1$, for any fixed constant $\epsilon > 0$, and the existence of a hash function modeled as a random oracle, \commplus (see \Cref{algorithm:reducer_plus}) satisfies agreement in the presence of a computationally bounded adversary.
\end{lemma}
\begin{proof}
By contradiction, suppose (1) there exists a correct process $p_i$ that decides a value $v_i$, and (2) there exists a correct process $p_j$ that decides a value $v_j \neq v_i$.
Let $p_i$ (resp., $p_j$) decide $v_i$ (resp., $v_j$) in some iteration $k_i \in \mathbb{N}$ (resp., $k_j \in \mathbb{N}$).
Therefore, process $p_i$ (resp., $p_j$) quasi-decides $v_i$ (resp., $v_j$) in iteration $k_i$ (resp., $k_j$).
Without loss of generality, let $k_i \leq k_j$.

As process $p_i$ quasi-decides $v_i$ in iteration $k_i$, process $p_i$ quasi-decides a vector $\mathit{vec}_i$ in iteration $k_i$; note that $v_i$ belongs to $\mathit{vec}_i$.
By \Cref{lemma:quasi_committed_vectors_plus}, process $p_j$ also quasi-decides the non-empty vector $\mathit{vec}_i$ in iteration $k_i$.
We separate two cases:
\begin{compactitem}
    \item Let $k_i = k_j$.
    Due to the fact that the $\mathsf{Index}()$ request invoked in iteration $k_i = k_j$ returns the same integer to all correct processes, we have that $v_j = v_i$.
    Thus, we reach a contradiction with $v_j \neq v_i$ in this case.

    \item Let $k_i < k_j$.
    As $p_j$ quasi-decides the non-empty vector $\mathit{vec}_i$ in iteration $k_i$, $p_j$ decides in iteration $k_i$ (if it has not done so in an earlier iteration).
    Therefore, we reach a contradiction with the fact that $p_j$ decides in iteration $k_j > k_i$.
\end{compactitem}
As neither of the above cases can occur, the proof is concluded.
\end{proof}

\smallskip
\noindent \textbf{Weak validity.}
First, we show that if any correct process proposes a value $v$ to the $\mathcal{MBA}[k][x]$ instance, for any sub-iteration $(k, x)$, and all processes are correct, then $v$ is the proposal of a correct process.

\begin{proposition} \label{lemma:integrity_proposal_correct_plus}
Let $(k \in \mathbb{N}, x \in [1, 4C^2])$ be any sub-iteration and let all processes be correct.
If any correct process $p_i$ proposes a value $v$ to $\mathcal{MBA}[k][x]$, then $v$ is the proposal of a correct process.
\end{proposition}
\begin{proof}
Recall that $\mathsf{leader}(k)$ denotes the leader of iteration $k$.
We now separate two cases:
\begin{compactitem}
    \item Let $p_i$ execute line~\ref{line:reducer_plus_decode}.
    In this case, process $p_i$ has received (at least) $\epsilon t + 1$ RS symbols.
    Given that all processes are correct, all these RS symbols are sent by $\mathsf{leader}(k)$ during the dissemination phase and they all correspond to $\mathsf{leader}(k)$'s proposal.
    Therefore, $v$ is the proposal of $\mathsf{leader}(k)$, which proves the statement of the proposition in this case.

    \item Let $p_i$ execute line~\ref{line:reducer_plus_own_proposal}.
    The statement of the proposition trivially holds in this case as $v$ is $p_i$'s proposal to \commplus.
\end{compactitem}
As the statement of the proposition holds in both cases, the proof is concluded.
\end{proof}

We are now ready to prove \commplus's weak validity.

\begin{lemma} [\commplus satisfies weak validity]
Given $n = (3 + \epsilon)t + 1$, for any fixed constant $\epsilon > 0$, and the existence of a hash function modeled as a random oracle, \commplus (see \Cref{algorithm:reducer_plus}) satisfies weak validity in the presence of a computationally bounded adversary.
\end{lemma}
\begin{proof}
Suppose all processes are correct.
Moreover, let a correct process $p_i$ decide some value $v$; note that value $v$ must be valid by \Cref{lemma:reducer_plus_external_validity}.
Hence, process $p_i$ quasi-decides $v$ in some iteration $k \in \mathbb{N}$, which further implies that $v$ is decided from $\mathcal{MBA}[k][x]$ in some sub-iteration $(k, x \in [1, 4C^2])$.
Given that $v$ is valid and $\bot_{\mathsf{MBA}}$ is invalid, $v \neq \bot_{\mathsf{MBA}}$.
Thus, the justification property of $\mathcal{MBA}[k][x]$ guarantees that $v$ was proposed to $\mathcal{MBA}[k][x]$ by a correct process.
\Cref{lemma:integrity_proposal_correct_plus} then shows that $v$ is the proposal of a correct process, which concludes the proof.
\end{proof}

\smallskip
\noindent \textbf{Integrity.}
Next, we prove \commplus's integrity.

\begin{lemma} [\commplus satisfies integrity]
Given $n = (3 + \epsilon)t + 1$, for any fixed constant $\epsilon > 0$, and the existence of a hash function modeled as a random oracle, \commplus (see \Cref{algorithm:reducer_plus}) satisfies integrity in the presence of a computationally bounded adversary.
\end{lemma}
\begin{proof}
The lemma trivially holds due to the check at line~\ref{line:check_quasi_decisions_plus}.
\end{proof}

\smallskip
\noindent \textbf{Termination.}
We proceed to prove that \commplus satisfies termination.
As in the proof of \comm's correctness, we say that a correct process $p_i$ \emph{completes the dissemination phase} if and only if $p_i$ executes line~\ref{line:dissemination_complete_plus}.
Recall that the dissemination phase of \commplus follows the same structure as the dissemination phase of \comm, with the only difference being that, during encoding, values are considered as polynomials of degree $\epsilon t$.
For completeness, the following proposition proves that at least one correct process completes the dissemination phase.

\begin{proposition} \label{proposition:one_completes_dissemination_reducer_plus}
At least one correct process completes the dissemination phase.
\end{proposition}
\begin{proof}
By contradiction, suppose no correct process completes the dissemination phase.
Hence, no correct process stops responding with \textsc{ack} messages upon receiving \textsc{init} messages.
As there are (at least) $n - t$ correct processes, every correct process eventually broadcasts a \textsc{done} message.
Similarly, every correct process eventually receives $n - t$ \textsc{done} messages and broadcasts a \textsc{finish} message.
Thus, every correct process eventually receives a \textsc{finish} message from (at least) $n - t$ processes and completes the dissemination phase, thus contradicting the fact that no correct process completes the dissemination phase.
\end{proof}

Next, we prove that if any correct process completes the dissemination phase, every correct process completes the dissemination phase.

\begin{proposition} \label{proposition:if_one_completes_dispersion_everyone_completes_dispersion_reducer_plus}
If any correct process completes the dissemination phase, then every correct process eventually completes the dissemination phase.
\end{proposition}
\begin{proof}
Let $p_i$ be any correct process that completes the dissemination phase.
This implies that $p_i$ receives $n - t = (2 + \epsilon)t + 1$ \textsc{finish} messages, out of which (at least) $(1 + \epsilon)t + 1$ messages are sent by correct processes.
Therefore, every correct process eventually receives $(1 + \epsilon)t + 1 \geq t + 1$ \textsc{finish} messages and broadcasts its \textsc{finish} message.
Given that there are (at least) $n - t$ correct processes, every correct process eventually receives $n - t$ \textsc{finish} messages and completes the dissemination phase.
\end{proof}

We now prove that all correct processes eventually complete the dissemination phase of \commplus.

\begin{proposition} \label{lemma:everyone_completes_dispersion_plus}
Every correct process eventually completes the dissemination phase.
\end{proposition}
\begin{proof}
The proposition holds due to \cref{proposition:one_completes_dissemination_reducer_plus,proposition:if_one_completes_dispersion_everyone_completes_dispersion_reducer_plus}.
\end{proof}

Next, we prove that if all correct processes start any iteration $k$, all correct processes eventually start sub-iteration $(k, 1)$.

\begin{proposition} \label{lemma:first_sub_iteration_plus}
Let $k \in \mathbb{N}$ be any iteration such that all correct processes start iteration $k$.
Then, all correct processes eventually start sub-iteration $(k, 1)$.
\end{proposition}
\begin{proof}
As all correct processes start iteration $k$ and there are at least $n - t = (2 + \epsilon)t + 1$ correct processes, all correct processes eventually receive $n - t =  (2 + \epsilon)t + 1$ \textsc{stored} messages (line~\ref{line:wait_for_stored_messages_plus}) and broadcast a \textsc{suggest} message (line~\ref{line:broadcast_suggest_plus}).
Hence, all correct processes eventually receive $n - t =  (2 + \epsilon)t + 1$ \textsc{suggest} messages (line~\ref{line:wait_for_suggest_messages_plus}) and start sub-iteration $(k, 1)$ at line~\ref{line:for_plus}.
\end{proof}

We now prove that if all correct processes start any sub-iteration $(k, x)$, then all correct processes eventually complete sub-iteration $(k, x)$.

\begin{proposition} \label{lemma:sub_iteration_completion_plus}
Let $(k \in \mathbb{N}, x \in [1, 4C^2])$ be any sub-iteration such that all correct processes start sub-iteration $(k, x)$.
Then, all correct processes eventually complete sub-iteration $(k, x)$.
\end{proposition}
\begin{proof} 
Given that all correct processes start sub-iteration $(k, x)$, each correct process broadcasts a \textsc{reconstruct} message (line~\ref{line:reconstruct_broadcast_reducer_plus}), thus ensuring that every correct process eventually receives $n - t = (2 + \epsilon)t + 1$ \textsc{reconstruct} messages (line~\ref{line:wait_for_reconstruct_messages_plus}).
Hence, all correct processes eventually propose to $\mathcal{MBA}[k][x]$ (line~\ref{line:lmba_plus}).
The termination property of $\mathcal{MBA}[k][x]$ then ensures that all correct processes complete sub-iteration $(k, x)$.
\end{proof}

The following proposition proves progress for every iteration that is reached before deciding.

\begin{proposition} \label{lemma:every_sub_iteration_plus}
Every iteration started by a correct process is eventually started and completed by all correct processes.
\end{proposition}
\begin{proof}
By \Cref{lemma:everyone_completes_dispersion_plus}, all correct processes start iteration $1$.
By \cref{lemma:first_sub_iteration_plus,lemma:sub_iteration_completion_plus}, all correct processes complete its $4C^2$ sub-iterations.
MBA agreement gives them the same quasi-decision list; if it is non-empty, all invoke $\mathsf{Index}()$ and decide, while otherwise all start iteration $2$.
Inductively, if a correct process starts iteration $k+1$, its common quasi-decision list in iteration $k$ was empty, so all correct processes eventually start iteration $k+1$ and complete its sub-iterations by the same argument.
After deciding, processes start no new iterations but continue serving already invoked primitives, which permits slower correct processes to complete the deciding iteration.
\end{proof}

In the remainder of the proof, we consider only iterations started by a correct process.
We might not explicitly rely on \Cref{lemma:every_sub_iteration_plus} in every subsequent argument.
As in the proof of \comm's correctness, we say that a process $p_i$ \emph{suggests} a digest $z$ in iteration $k$ if it includes $z$ in its protocol-prescribed \textsc{suggest} broadcast (line~\ref{line:broadcast_suggest_plus}).
Let $\mathsf{suggested}_i(k)$ denote the set of digests in that broadcast.
The following bound also applies to processes corrupted after sending their suggestions.

\begin{proposition}\label{lemma:proposed_bound_plus}
For every process $p_i$ that is so-far-uncorrupted when it broadcasts its \textsc{suggest} message in iteration $k$, the following holds:
\[|\mathsf{suggested}_i(k)| \leq \lceil 1 + \frac{2}{\epsilon} \rceil.\]
\end{proposition}
\begin{proof}
For every digest $z \in \mathsf{suggested}_i(k)$, process $p_i$ receives (at least) $n - 3t = \epsilon t + 1$ \textsc{stored} messages in iteration $k$ (line~\ref{line:check_stored_plus}).
As $p_i$ receives $n - t = (2 + \epsilon) t + 1$ \textsc{stored} messages (line~\ref{line:wait_for_stored_messages_plus}) before broadcasting its \textsc{suggest} message, there can be at most $\frac{(2 + \epsilon)t + 1}{\epsilon t + 1}$ suggested digests.
We now prove that
\[
\frac{(2 + \epsilon)t + 1}{\epsilon t + 1} \leq 1 + \frac{2}{\epsilon}.
\]
Indeed, we have
\begin{align*}
(2 + \epsilon)t + 1 
&\leq \Bigl(1 + \frac{2}{\epsilon}\Bigr)(\epsilon t + 1) \\
&= (\epsilon t + 1) + \frac{2}{\epsilon}(\epsilon t + 1) \\
&= \epsilon t + 1 + 2t + \frac{2}{\epsilon} \\
&= (2 + \epsilon)t + 1 + \frac{2}{\epsilon}.
\end{align*}
Clearly, we have
\[
(2 + \epsilon)t + 1 \leq (2 + \epsilon)t + 1 + \frac{2}{\epsilon},
\]
which concludes the proof.
\end{proof}
We say that a correct process \emph{commits} a digest $z$ in an iteration $k \in \mathbb{N}$ if and only if $z \in \mathit{candidates}_i$ when process $p_i$ reaches line~\ref{line:for_plus} in iteration $k$.
Let $\mathsf{committed}_i(k)$ denote the set of digests committed by correct process $p_i$ in iteration $k \in \mathbb{N}$.
We now prove that, for every correct process $p_i$ and every iteration $k \in \mathbb{N}$, $\mathsf{committed}_i(k) \subseteq \mathsf{suggested}_i(k)$.

\begin{proposition}\label{lemma:committed_precommited_plus}
For every correct process $p_i$ and every iteration $k \in \mathbb{N}$, the following holds: (1) $\mathsf{committed}_i(k) \subseteq \mathsf{suggested}_i(k)$, and (2) $|\mathsf{committed}_i(k)| \leq \lceil 1 + \frac{2}{\epsilon} \rceil$.
\end{proposition}
\begin{proof}
We have that $\mathsf{committed}_i(k) \subseteq \mathsf{suggested}_i(k)$ directly from the fact that correct process $p_i$ only removes digests from its $\mathit{candidates}_i$ list after broadcasting its \textsc{suggest} message in iteration $k$ (line~\ref{line:candidates_remove_plus}).
Moreover, as $|\mathsf{suggested}_i(k)| \leq \lceil 1 + \frac{2}{\epsilon} \rceil$ (by \Cref{lemma:proposed_bound_plus}), $|\mathsf{committed}_i(k)| \leq \lceil 1 + \frac{2}{\epsilon} \rceil$.
\end{proof}

To prove \commplus's termination, we rely on the notion of good iterations.
Recall that, by \Cref{definition:good_iteration_plus}, an iteration $k \in \mathbb{N}$ is said to be good if and only if $\mathsf{leader}(k) \in \dfirst$.
The first so-far-uncorrupted \textsc{finish} broadcaster cannot be triggered by $t+1$ \textsc{finish} messages: one would have been sent earlier by a so-far-uncorrupted process.
It therefore receives $n-t$ \textsc{done} messages, at least $n-2t$ of whose senders are uncorrupted when $\dfirst$ is frozen.
Every election disclosure requires a so-far-uncorrupted requester that has completed dissemination, so this first \textsc{finish} broadcast precedes every election disclosure.
Moreover, for every good iteration $k$, (1) $v^{\star}(k)$ denotes the valid proposal disseminated by $\mathsf{leader}(k)$ while it is uncorrupted, and (2) $z^{\star}(k)$ denotes its digest.
The following proposition proves that every correct process suggests $z^{\star}(k)$ in any good iteration $k \in \mathbb{N}$.

\begin{proposition}\label{lemma:good_iteration_all_correct_suggest_plus}
Let $k \in \mathbb{N}$ be any good iteration.
Then, every correct process suggests $z^{\star}(k)$ in iteration $k$.
\end{proposition}
\begin{proof}
As $\mathsf{leader}(k) \in \dfirst$, it broadcasts \textsc{done} while uncorrupted and has previously stored $z^{\star}(k)$ at $(2 + \epsilon)t + 1$ processes, out of which at most $t$ are faulty in the complete execution.
Thus, $\mathsf{leader}(k)$ stores $z^{\star}(k)$ at $\geq (1 + \epsilon) t + 1$ correct processes.
This implies that each correct process receives $z^{\star}(k)$ in \textsc{stored} messages from at least $\epsilon t + 1$ processes, which further means that each correct process broadcasts a \textsc{suggest} message with $z^{\star}(k)$ and thus suggests $z^{\star}(k)$ in iteration $k$.
\end{proof}

Next, we show that every correct process commits $z^{\star}(k)$ in any good iteration $k$.

\begin{proposition}\label{lemma:good_iteration_all_correct_commit_plus}
Let $k \in \mathbb{N}$ be any good iteration.
Then, every correct process commits $z^{\star}(k)$ in iteration $k$.
\end{proposition}
\begin{proof}
By \Cref{lemma:good_iteration_all_correct_suggest_plus}, all correct processes suggest $z^{\star}(k)$ in iteration $k$.
Therefore, each correct process receives a \textsc{suggest} message with $z^{\star}(k)$ from at least $n - 2t = (1 + \epsilon)t + 1$ processes, which means that each correct process commits $z^{\star}(k)$.
\end{proof}

For every iteration $k \in \mathbb{N}$, we define the $\mathsf{committed}(k)$ set:
\begin{equation*}
    \mathsf{committed}(k) = \{ z \,|\, \text{$z$ is committed by a correct process in iteration $k$} \}.
\end{equation*}
The following proposition proves that $|\mathsf{committed}(k)| \leq C$ in any good iteration $k$.
Recall that $C = \lceil 1 + 2 / \epsilon \rceil^2$.

\begin{proposition}\label{lemma:good_iteration_all_precommitted_bound_commit_plus}
For every good iteration $k \in \mathbb{N}$, $|\mathsf{committed}(k)| \leq C$.
\end{proposition}
\begin{proof}
For every digest $z \in \mathsf{committed}(k)$, at least $(1 + \epsilon)t + 1 - f$ correct processes suggest $z$ in iteration $k$, where $f$ denotes the number of faulty processes.
Moreover, by \Cref{lemma:proposed_bound_plus}, each correct process suggests at most $\lceil 1 + \frac{2}{\epsilon} \rceil$ digests.
Hence, we consider
\[
E(x) = \frac{x \cdot \lceil 1 + \frac{2}{\epsilon} \rceil}{(1 + \epsilon)t + 1 - (n - x)},
\]
where $x$ denotes the number of correct processes and $x \in [(2 + \epsilon)t + 1, n]$.
Simplifying the denominator, we get
\[
(1+\epsilon)t + 1 - (n-x) = (1+\epsilon)t + 1 - ((3+\epsilon)t + 1 - x) = x - 2t,
\]
so that
\[
E(x) = \lceil 1 + 2/\epsilon \rceil \cdot \frac{x}{x-2t}.
\]
Taking the derivative with respect to $x$, we have
\[
E'(x) = \lceil 1 + 2/\epsilon \rceil \cdot \frac{d}{dx}\left(\frac{x}{x-2t}\right)
= \lceil 1 + 2/\epsilon \rceil \cdot \frac{-2t}{(x-2t)^2} < 0 \quad \text{for all relevant } x.
\]
Thus, $E(x)$ is strictly decreasing over $x \in [(2+\epsilon)t + 1, n]$, and its maximum occurs at the left endpoint $x_{\min} = (2+\epsilon)t + 1$:
\[
E(x) \leq E(x_{\min}) = \lceil 1 + 2/\epsilon \rceil \cdot \frac{(2+\epsilon)t + 1}{(2+\epsilon)t + 1 - 2t}
= \lceil 1 + 2/\epsilon \rceil \cdot \frac{(2+\epsilon)t + 1}{\epsilon t + 1}.
\]
Finally, since
\[
\frac{(2+\epsilon)t + 1}{\epsilon t + 1} = 1 + \frac{2t}{\epsilon t + 1} < 1 + \frac{2}{\epsilon},
\]
we obtain
\[
E(x) \leq \lceil 1 + 2/\epsilon \rceil \cdot \left(1 + \frac{2}{\epsilon}\right) \leq \lceil 1 + 2/\epsilon \rceil^2,
\]
which concludes the proof.
\end{proof}

Next, we prove a crucial proposition about the $\mathcal{MBA}[k][x]$ instance employed in any sub-iteration $(k, x)$.

\begin{proposition}\label{lemma:weird_lemma_mba}
Consider any sub-iteration $(k \in \mathbb{N}, x \in [1, 4C^2])$.
Suppose a correct process $p_i$ decides some value $v'$ from $\mathcal{MBA}[k][x]$.
Moreover, suppose all correct processes that proposed to $\mathcal{MBA}[k][x]$ before $p_i$ decides $v'$ do so with the same value $v$.
Then, $v' = v$.
\end{proposition}
\begin{proof}
Project the execution prefix ending at $p_i$'s decision onto the standalone MBA instance.
Since inputs are supplied only upon invoking $\mathsf{propose}$, this prefix has a continuation preserving the eventual corruption set in which every remaining correct proposer supplies $v$.
All correct processes then propose $v$, while the unchanged prefix contains the decision $v'$; strong unanimity therefore implies $v'=v$.
\end{proof}

We say that a sub-iteration $(k \in \mathbb{N}, x \in [1, 4C^2])$ \emph{starts} at the moment when the first correct process invokes the $\mathsf{Noise}()$ request in sub-iteration $(k, x)$ (line~\ref{line:random_noise_plus}).
Similarly, we say that a sub-iteration $(k \in \mathbb{N}, x \in [1, 4C^2])$ \emph{ends} at the moment when the first correct process decides from $\mathcal{MBA}[k][x]$ in sub-iteration $(k, x)$ (line~\ref{line:lmba_plus}).
For any sub-iteration $(k \in \mathbb{N}, x \in [1, 4C^2])$, we define the $\mathsf{start}(k, x)$ set:
\begin{equation*}
    \mathsf{start}(k, x) = \{ z \,|\, \text{$z$ is committed by a correct process before sub-iteration $(k, x)$ starts} \}.
\end{equation*}
Next, for every sub-iteration $(k \in \mathbb{N}, x \in [1, 4C^2])$, we define the $\mathsf{end}(k, x)$ set:
\begin{equation*}
    \mathsf{end}(k, x) = \{ z  \,|\, \text{$z$ is committed by a correct process before sub-iteration $(k, x)$ ends} \}.
\end{equation*}
Note that the following holds:
\begin{compactitem}
    \item For every iteration $k \in \mathbb{N}$ and every $x \in [1, 4C^2]$, $\mathsf{start}(k, x) \subseteq \mathsf{committed}(k)$.

    \item For every iteration $k \in \mathbb{N}$ and every $x \in [1, 4C^2]$, $\mathsf{end}(k, x) \subseteq \mathsf{committed}(k)$.

    \item For every iteration $k \in \mathbb{N}$, $\mathsf{start}(k, 1) \subseteq \mathsf{end}(k, 1) \subseteq \mathsf{start}(k, 2) \subseteq \mathsf{end}(k, 2) \subseteq ... \subseteq \mathsf{start}(k, 4C^2) \subseteq \mathsf{end}(k, 4C^2) \subseteq \mathsf{committed}(k)$.
    This is true as (1) each sub-iteration $(k, \cdot)$ starts before it ends, and (2) each sub-iteration $(k, x)$ ends before sub-iteration $(k, x + 1)$ starts.
\end{compactitem}
The following proposition proves that $z^{\star}(k) \in \mathsf{start}(k, x)$, for any sub-iteration $(k, x)$ of a good iteration $k$.

\begin{proposition}\label{lemma:every_round_commit_good}
Let $k \in \mathbb{N}$ be any good iteration.
Then, for every sub-iteration $(k, x \in [1, 4C^2])$, $z^{\star}(k) \in \mathsf{start}(k, x)$.
\end{proposition}
\begin{proof}
By \Cref{lemma:good_iteration_all_correct_commit_plus}, all correct processes commit $z^{\star}(k)$ in iteration $k$.
Hence, the correct process that ``starts'' sub-iteration $(k, 1)$ (i.e., invokes the first $\mathsf{Noise}()$ request) commits $z^{\star}(k)$ in sub-iteration $k$.
Thus, $z^{\star}(k) \in \mathsf{start}(k, 1)$.
As $\mathsf{start}(k, 1) \subseteq \mathsf{start}(k, x)$, for every $x \in [2, 4C^2]$, the proposition holds.
\end{proof}

We say that a correct process $p_i$ \emph{adopts} a digest $z$ in a sub-iteration $(k \in \mathbb{N}, x \in [1, 4C^2])$ if and only if $\mathit{adopted\_digest}_i = z$ when process $p_i$ reaches line~\ref{line:reconstruct_broadcast_reducer_plus} in sub-iteration $(k, x)$.
The sets $\mathsf{start}(k,x)$ and $\mathsf{end}(k,x)$ refer to eventual correctness and are used only for deterministic counting.
For probabilities, immediately before $\mathsf{Noise}()$ discloses the random value $\phi(k,x)$, let $\mathcal U$ be the set of then-so-far-uncorrupted processes.
Freeze $\mathsf{pool}(k,x)$ as the union of the candidate lists already fixed for iteration $k$ at line~\ref{line:for_plus} by processes in $\mathcal U$.
This disclosure follows $t+1$ distinct requests, including a correct request, and precedes the first correct MBA decision in the trial.
Consequently, $\mathsf{start}(k,x)\subseteq\mathsf{pool}(k,x)$, and the sequential execution of trials places all preceding trial disclosures before the current one.
Writing $u=|\mathcal U|$ and $d=\lceil1+2/\epsilon\rceil$, every digest in the pool has at least $u-2t$ suggestions from $\mathcal U$, whereas each member of $\mathcal U$ suggests at most $d$ digests by \Cref{lemma:proposed_bound_plus}.
Thus,
\[
|\mathsf{pool}(k,x)|\leq\frac{ud}{u-2t}
\leq d\frac{n-t}{n-3t}\leq d^2=C.
\]
For a good iteration $k$, a trial $(k,x)$ is \emph{won} if $\mathsf{hash}(z^\star(k),\phi(k,x))$ is strictly smaller than $\mathsf{hash}(z,\phi(k,x))$ for every other digest $z$ in the pool.

\begin{proposition}\label{lemma:won_probability}
Except for a negligible random-oracle event, every trial of a good iteration is won with conditional probability at least $1/(2C)$, given its pre-disclosure history and the outcomes of preceding trials, for sufficiently large $\kappa$.
\end{proposition}
\begin{proof}
Let $N=2^\kappa$.
For a fixed set $S$ of $m\leq C$ digests and $z\in S$, the fraction of salts for which $z$ has the strictly smallest hash has mean at least $1/m-(m-1)/(2N)$, by symmetry and a union bound on ties.
For sufficiently large $N$, this mean is at least $3/(4C)$.
Hoeffding's inequality and a union bound over at most $C^2N^C$ pairs $(S,z)$ imply that every such fraction is at least $1/(2C)$, except with probability at most
\[
C^2N^C\exp\big(-N/(8C^2)\big)=\mathsf{negl}(\kappa).
\]
Condition on this oracle event's complement.
The pool is fixed before disclosure, contains $z^\star(k)$ by \Cref{lemma:every_round_commit_good}, and has size at most $C$.
The current salt is uniform and independent of the oracle and preceding history, including earlier won indicators, which proves the bound.
\end{proof}

Throughout the remaining probabilistic analysis, we condition on this overwhelmingly likely oracle event and take $\kappa$ sufficiently large.
We now separate the deterministic consequence of winning a trial from its probability.

\begin{proposition}\label{lemma:start_end_probability}
Let $k$ be a good iteration and let $(k,x)$ be a won trial with $\mathsf{start}(k,x)=\mathsf{end}(k,x)$.
Then, every correct process proposing to $\mathcal{MBA}[k][x]$ before the trial ends proposes $v^\star(k)$.
\end{proposition}
\begin{proof}
Every such process $p_i$ commits $z^\star(k)$ by \Cref{lemma:good_iteration_all_correct_commit_plus}, and
\[
\mathsf{committed}_i(k)\subseteq\mathsf{end}(k,x)
=\mathsf{start}(k,x)\subseteq\mathsf{pool}(k,x).
\]
Since the trial is won, $p_i$ therefore adopts $z^\star(k)$.
Correctly encoded symbols of $v^\star(k)$ are stored at at least $n-2t$ correct processes, so among its $n-t$ \textsc{reconstruct} messages, $p_i$ receives at least $n-3t=\epsilon t+1$ such symbols.
It consequently reconstructs and proposes $v^\star(k)$.
\end{proof}

The following proposition bounds the number of trials in which a new digest can be committed.

\begin{proposition}\label{lemma:exists_round_start_end}
In a good iteration $k$, at most $C-1$ trials satisfy $\mathsf{start}(k,x)\neq\mathsf{end}(k,x)$.
\end{proposition}
\begin{proof}
The inclusion chain above starts with $z^\star(k)\in\mathsf{start}(k,1)$ and ends within $\mathsf{committed}(k)$, whose size is at most $C$.
Each trial with unequal start and end sets adds at least one digest along this chain, so there can be at most $C-1$ such trials.
\end{proof}

The following proposition proves that if all correct processes that propose to $\mathcal{MBA}[k][x]$ before sub-iteration $(k, x)$ (of a good iteration $k$) ends do propose $v^{\star}(k)$, then all correct processes quasi-decide $v^{\star}(k)$ in sub-iteration $(k, x)$.

\begin{proposition}\label{lemma:good_iteration_and_agreement_plus}
Let $k \in \mathbb{N}$ be any good iteration. 
Moreover, let $(k \in \mathbb{N}, x \in [1, 4C^2])$ be any sub-iteration of iteration $k$.
Suppose all correct processes that propose to $\mathcal{MBA}[k][x]$ before sub-iteration $(k, x)$ ends do propose $v^{\star}(k)$.
Then, all correct processes quasi-decide $v^{\star}(k)$ in sub-iteration $(k, x)$.
\end{proposition}
\begin{proof}
By \Cref{lemma:weird_lemma_mba}, the first correct process that decides from $\mathcal{MBA}[k][x]$ does decide $v^{\star}(k)$.
As $\mathcal{MBA}[k][x]$ ensures agreement and termination, all correct processes eventually decide $v^{\star}(k)$ from $\mathcal{MBA}[k][x]$, which proves the proposition.
\end{proof}

Finally, we are ready to prove the termination property of \commplus.

\begin{lemma} [\commplus satisfies termination]\label{theorem:termination_plus}
Given $n = (3 + \epsilon)t + 1$, for any fixed constant $\epsilon > 0$, and the existence of a hash function modeled as a random oracle, \commplus (see \Cref{algorithm:reducer_plus}) satisfies termination in the presence of a computationally bounded adversary.
Precisely, every correct process decides within at most six iterations in expectation.
\end{lemma}
\begin{proof}
Fix a prefix immediately after electing a good leader in iteration $k$.
Call a trial \emph{quiet} if $\mathsf{start}(k,x)=\mathsf{end}(k,x)$.
By \cref{lemma:start_end_probability,lemma:good_iteration_and_agreement_plus}, every won, quiet trial makes all correct processes quasi-decide $v^\star(k)$.
By \Cref{lemma:every_sub_iteration_plus}, they complete the iteration and decide.
If $W$ denotes the number of won trials, failure to decide in iteration $k$ therefore implies $W<C$: otherwise, \Cref{lemma:exists_round_start_end} ensures that at least one of the $W$ won trials is quiet.

Let $Y_x$ indicate whether trial $x$ is won.
Conditioning immediately before each salt is disclosed, and including earlier won indicators in the history, \Cref{lemma:won_probability} gives
\[
\mathbb E[e^{-Y_x}\mid\text{history}]
=1-(1-e^{-1})\Pr[Y_x=1\mid\text{history}]
\leq1-\frac{1-e^{-1}}{2C}.
\]
Repeated conditional expectation then yields
\[
\mathbb E[e^{-W}]\leq\left(1-\frac{1-e^{-1}}{2C}\right)^{4C^2}.
\]
Thus, by Markov's inequality and $C\geq4$,
\[
\Pr[W<C]\leq e^C\mathbb E[e^{-W}]
\leq\exp\big(-(1-2/e)C\big)<\frac12.
\]
This argument uses conditional probability bounds, not independence between the won indicators.

The set $\dfirst$ is frozen before election disclosure and has at least $n-2t$ members.
Therefore, after every pre-election history in which no decision has occurred, the iteration is good with probability at least
\[
\frac{n-2t}{n}=\frac{(1+\epsilon)t+1}{(3+\epsilon)t+1}\geq\frac13,
\]
and succeeds with probability at least $1/6$.
Let $K$ denote the deciding iteration.
Before the election of iteration $k$ is disclosed, at least one correct requester has completed iteration $k-1$; MBA agreement fixes its common quasi-decision list, so whether the preceding iteration produces a decision is already determined.
The conditional bound consequently gives
\[
\Pr[K>k]\leq(5/6)^k,
\qquad
\mathbb E[K]=\sum_{k\geq0}\Pr[K>k]\leq6.
\]
The tail tends to zero, and \Cref{lemma:every_sub_iteration_plus} ensures that all correct processes complete the deciding iteration.
\end{proof}

\smallskip
\noindent \textbf{Quality.}
To conclude the proof of \commplus's correctness, we prove that \commplus satisfies the quality property.

\begin{lemma} [\commplus satisfies quality]
Given $n = (3 + \epsilon)t + 1$, for any fixed constant $\epsilon > 0$, and the existence of a hash function modeled as a random oracle, \commplus (see \Cref{algorithm:reducer_plus}) satisfies quality in the presence of a computationally bounded adversary.
\end{lemma}
\begin{proof}
Suppose the first iteration is good, and let $A$ count occurrences of its non-adversarial value $v^\star(1)$ in the common quasi-decision list.
Every won, quiet trial contributes one occurrence, so $A\geq W-(C-1)$, irrespective of how quiet trials are selected.
By \Cref{lemma:won_probability}, $\mathbb E[W]\geq4C^2/(2C)=2C$, and hence $\mathbb E[A]\geq C+1$.

The list is fixed before $\mathsf{Index}()$ is disclosed: among its $t+1$ requesters, a correct process has completed all trials, whose outputs are fixed by MBA agreement.
For a non-empty list of length $q\leq4C^2$, every position has at least one preimage under $I\mapsto(I\bmod q)+1$, for $I\in[1,4C^2]$.
Thus, the independent uniform index selects $v^\star(1)$ with probability at least $A/(4C^2)$; an empty list contributes zero.
The probability of a non-adversarial decision is consequently at least
\[
\frac{(1+\epsilon)t+1}{(3+\epsilon)t+1}\cdot\frac{C+1}{4C^2}
\geq\frac{(1+\epsilon)t+1}{(3+\epsilon)t+1}\cdot\frac1{C^2}
\geq\frac1{3C^2},
\]
as $C\geq4$.
Therefore, the probability of an adversarial decision is at most $1-1/(3C^2)<1$.
\end{proof}

\subsection{Proof of Complexity}

This subsection proves the complexity of \commplus, i.e., it formally proves \Cref{theorem:reducer_plus_complexity}.
Let $K$ denote the deciding iteration, and let $\mathcal H(k)$ be the history immediately before $\mathsf{Election}(k)$ is disclosed, for every reached iteration $k$.
No \textsc{stored}, \textsc{suggest}, \textsc{reconstruct}, or MBA message for iteration $k$ is sent by a correct process before this disclosure.
As in the termination proof, whether $K\geq k$ is already fixed at that point.
For a non-negative iteration cost $X_k$ with a uniform conditional bound $\mathbb E[X_k\mid\mathcal H(k)]\leq c_k$, we consequently have
\[
\mathbb E\left[\sum_{k=1}^{K}X_k\right]
\leq\sum_{k\geq1}c_k\Pr[K\geq k].
\]
This does not assume independence between $K$ and the iteration costs.
By \Cref{theorem:termination_plus},
\[
\sum_{k\geq1}\Pr[K\geq k]\leq6,
\qquad
\sum_{k\geq1}\Pr[K\geq k]\log k
\leq\sum_{k\geq1}(5/6)^{k-1}\log k\in O(1).
\]
We charge the full cost of every invoked primitive, including service after a process decides, to its iteration; the primitive's complexity bound includes this prescribed service.
Since no new iteration is started after deciding, no iteration beyond $K$ contributes such a cost.

\smallskip
\noindent \textbf{Message complexity.}
We start by proving \commplus's message complexity.

\begin{lemma} [\commplus's expected message complexity]
Given $n = (3 + \epsilon)t + 1$, for any fixed constant $\epsilon > 0$, and the existence of a hash function modeled as a random oracle, the expected message complexity of \commplus (see \Cref{algorithm:reducer_plus}) is $O(n^2 C^2)$ in the presence of a computationally bounded adversary.
\end{lemma}
\begin{proof}
The dissemination phase exchanges $O(n^2)$ messages.
Every iteration contains $4C^2$ sub-iterations, each exchanging $O(n^2)$ messages in expectation, uniformly over its preceding history.
Thus, the conditional per-iteration bound is $c_k=O(n^2C^2)$, and the aggregation inequality gives
\[
O(n^2)+O(n^2C^2)\sum_{k\geq1}\Pr[K\geq k]\subseteq O(n^2C^2).
\]
\end{proof}

\smallskip
\noindent \textbf{Bit complexity.}
We start by proving that correct processes send $O(\frac{n\ell}{\epsilon}+n^2\kappa\log n)$ bits in the dissemination phase.

\begin{proposition} \label{lemma:dispersion_cost_plus}
Correct processes send $O(\frac{n\ell}{\epsilon}+n^2\kappa\log n)$ bits in the dissemination phase.
\end{proposition}
\begin{proof}
Recall that values are treated as polynomials of degree $\epsilon t$, so each RS symbol is of size $O(\frac{\ell}{\epsilon n} + \log n)$ bits.
Let $p_i$ be any correct process.
Process $p_i$ sends $n \cdot  O(\frac{\ell}{\epsilon n} + \log n + \kappa + \kappa \log n) \subseteq O(\frac{\ell}{\epsilon} + n \kappa \log n)$ bits via \textsc{init} messages.
Moreover, process $p_i$ sends $O(n)$ bits via \textsc{ack}, \textsc{done} and \textsc{finish} messages.
Therefore, process $p_i$ sends
\begin{equation*}
    \begin{aligned}
        & \underbrace{O(\frac{\ell}{\epsilon} + n \kappa \log n)}_\text{\textsc{init}}
        {}+{}
        \underbrace{O(n)}_\text{\textsc{ack}, \textsc{done} \& \textsc{finish}}
    \\
    & \subseteq O(\frac{\ell}{\epsilon} + n \kappa \log n) \text{ bits in the dissemination phase.}
    \end{aligned}
\end{equation*}
This implies that correct processes send $O(\frac{n\ell}{\epsilon} + n^2 \kappa \log n)$ bits in the dissemination phase.
\end{proof}

Next, we bound the expected bit cost of an iteration, including the bits needed to identify its sub-iterations.

\begin{proposition} \label{lemma:iteration_cost_plus}
In every reached iteration $k$, correct processes send
\[
O\left(C^2\left(\frac{n\ell}{\epsilon}+n^2\kappa\log n+n^2\log k+n^2\log C\right)\right)
\]
bits in expectation, uniformly over $\mathcal H(k)$.
\end{proposition}
\begin{proof}
Correct processes send $O(n^2 \kappa \cdot \lceil 1 + \frac{2}{\epsilon} \rceil + n^2\log k)$ bits via \textsc{stored} and \textsc{suggest} messages.
Recall that \textsc{stored} messages include a single digest and \textsc{suggest} messages include at most $\lceil 1 + \frac{2}{\epsilon} \rceil$ digests due to \Cref{lemma:proposed_bound_plus}. \textsc{stored} and \textsc{suggest} messages also include the number of the iteration to which they refer.

Now, consider any sub-iteration $(k, x \in [1, 4C^2])$.
Correct processes send $O(\frac{n \ell}{\epsilon} + n^2 \kappa \log n + n^2 \log k + n^2 \log C)$ bits via \textsc{reconstruct} messages.
Moreover, correct processes send $O(n \ell + n^2 \kappa \log n + n^2 \log k + n^2 \log C)$ bits in expectation while executing $\mathcal{MBA}[k][x]$.
The $O(n^2\log k + n^2 \log C)$ term exists as (1) each message of $\mathcal{MBA}[k][x]$ needs to be tagged with ``$k, x \in [1, 4C^2]$'', and (2) $\mathcal{MBA}[k][x]$ exchanges $O(n^2)$ messages in expectation.
Thus, correct processes send
\begin{equation*}
    \begin{aligned}
        & \underbrace{O(\frac{n\ell}{\epsilon} + n^2 \kappa \log n + n^2 \log k + n^2 \log C)}_\text{\textsc{reconstruct}}
        {}+{}
        \underbrace{O(n\ell + n^2 \kappa \log n + n^2 \log k + n^2 \log C)}_\text{$\mathcal{MBA}[k][x]$} 
        \\
        & \subseteq O(\frac{n \ell}{\epsilon} + n^2 \kappa \log n + n^2 \log k + n^2 \log C) \text{ bits in sub-iteration $(k, x)$.}
    \end{aligned}
\end{equation*}
Given that iteration $k$ has $4C^2$ sub-iterations, correct processes collectively send
\begin{equation*}
    \begin{aligned}
        & \underbrace{O(n^2 \kappa \cdot \lceil 1 + \frac{2}{\epsilon} \rceil + n^2 \log k)}_\text{\textsc{stored} \& \textsc{suggest}} 
        {}+{}
        4C^2 \cdot \underbrace{O(\frac{n \ell}{\epsilon} + n^2 \kappa \log n + n^2 \log k + n^2 \log C)}_\text{sub-iterations}
        \\
        & \subseteq O\big( C^2 (\frac{n \ell}{\epsilon} + n^2 \kappa \log n + n^2 \log k + n^2 \log C) \big) \text{ bits in iteration $k$.}
    \end{aligned}
\end{equation*}
The MBA bounds hold uniformly over the preceding history, so the claimed conditional expected bound follows.
\end{proof}

Finally, we are ready to prove \commplus's bit complexity.

\begin{lemma} [\commplus's expected bit complexity]
Given $n = (3 + \epsilon)t + 1$, for any fixed constant $\epsilon > 0$, and the existence of a hash function modeled as a random oracle, the expected bit complexity of \commplus (see \Cref{algorithm:reducer_plus}) is
\[
O\left(C^2\left(\frac{n\ell}{\epsilon}+n^2\kappa\log n+n^2\log C\right)\right)
\]
in the presence of a computationally bounded adversary.
\end{lemma}
\begin{proof}
Let $B$ denote the total number of bits sent by correct processes.
Applying the aggregation inequality with the conditional costs from \Cref{lemma:iteration_cost_plus}, and adding the dissemination cost from \Cref{lemma:dispersion_cost_plus}, gives
\begin{align*}
\mathbb E[B]
&\leq O\left(\frac{n\ell}{\epsilon}+n^2\kappa\log n\right)
+O\left(C^2\left(\frac{n\ell}{\epsilon}+n^2\kappa\log n+n^2\log C\right)\right)
\sum_{k\geq1}\Pr[K\geq k] \\
&\quad+O(C^2n^2)\sum_{k\geq1}\Pr[K\geq k]\log k \\
&\subseteq O\left(C^2\left(\frac{n\ell}{\epsilon}+n^2\kappa\log n+n^2\log C\right)\right).
\end{align*}
For sufficiently large $\kappa$, $\log C\leq\kappa$; hence, the $n^2\log C$ term is absorbed by $n^2\kappa\log n$, giving the bound stated in \Cref{theorem:reducer_plus_complexity}.
\end{proof}

\smallskip
\noindent \textbf{Time complexity.}
Lastly, we prove \commplus's time complexity.

\begin{lemma} [\commplus's expected time complexity]
Given $n = (3 + \epsilon)t + 1$, for any fixed constant $\epsilon > 0$, and the existence of a hash function modeled as a random oracle, the expected time complexity of \commplus (see \Cref{algorithm:reducer_plus}) is $O(C^2)$ in the presence of a computationally bounded adversary.
\end{lemma}
\begin{proof}
Measure each communication or primitive stage from the last correct participant's entry until the last correct participant's completion.
Sequential local execution bounds the elapsed time by the dissemination time plus the sum of these charged stage durations.
Let $T(k)$ be their sum for iteration $k$.
There are $4C^2$ sub-iterations, each containing constantly many such stages, so the uniform expected-time guarantees give $\mathbb E[T(k)\mid\mathcal H(k)]\in O(C^2)$ for every reached iteration.
Adding the dissemination phase's $O(1)$ time and applying the aggregation inequality yields
\[
\mathbb E[T]\leq O(1)+O(C^2)\sum_{k\geq1}\Pr[K\geq k]\subseteq O(C^2).
\]
\end{proof}

\end{document}